\documentclass[12pt]{article}
\UseRawInputEncoding

\usepackage[T1]{fontenc}
\usepackage{tgpagella}
\usepackage{appendix}

\usepackage[
  top=1.0in,
  bottom=1.0in,
  left=1.0in,
  right=1.0in
]{geometry}
\usepackage{amssymb, amsmath, amsfonts, amsthm}
\usepackage{bm}
\usepackage{bbm}
\usepackage[mathscr]{eucal}
\usepackage{changepage}

\usepackage{graphicx}
\usepackage[font={footnotesize, singlespacing}]{caption}\usepackage{subcaption}
\usepackage{float}
\usepackage{tikz}
\usetikzlibrary{arrows.meta, positioning}

\usepackage{booktabs}
\usepackage{threeparttable}

\usepackage{pdflscape}
\usepackage{rotating}
\usepackage{comment}
\usepackage{afterpage}
\usepackage{placeins}
\usepackage[ruled,lined]{algorithm2e}
\usepackage{natbib}
\usepackage[colorlinks=true,linkcolor=blue,citecolor=blue,urlcolor=blue]{hyperref}
\usepackage{xr}
\newtheorem{theorem}{Theorem}
\newtheorem{lemma}{Lemma}

\newtheorem{assumption}{Assumption}[section]
\theoremstyle{definition}

\newtheorem{remark}{Remark}

\newcommand{\bo}{\boldsymbol}

\DeclareMathOperator*{\argmax}{arg\,max}

\newcommand{\Ups}{\boldsymbol\Upsilon}
\newcommand{\R}{\mathbb R}
\newcommand{\op}{\operatorname}

\renewcommand{\footnotesize}{\fontsize{9}{11}\selectfont}

\allowdisplaybreaks

\bibpunct{(}{)}{,}{a}{}{,}

\input{ee.sty}

\usepackage{titlesec}
\titleformat{\section}{\fontsize{13}{15.6}\selectfont\bfseries}{\thesection}{1em}{}
\titleformat{\subsection}{\normalsize\bfseries}{\thesubsection}{1em}{}
\titleformat{\subsubsection}{\normalsize\itshape}{\thesubsubsection}{1em}{}
\titleformat{\paragraph}[runin]{\normalsize\itshape}{}{}{}[.]

\begin{document}
\raggedbottom
\widowpenalty=100
\clubpenalty=100

\thispagestyle{empty}
\begin{center}

\vspace{-0.3in}
{\large A Pairwise Differencing Distribution Regression Approach for Network Models\footnote{Acknowledgements: I would like to thank Frank Kleibergen, 
Art\={u}ras Juodis and Bo Honor\'{e} for their advice and support. 
I also thank Jad Beyhum, Otilia Boldea, 
Pavel \v{C}\'\i\v{z}ek, Geert Dhaene,
Michal Koles\'{a}r, Louise Laage, 
Elena Manresa, Konrad Menzel, Chris Muris, Cavit Pakel, 
Mikkel Plagborg-M\o{}ller, Stephen Redding, 
Sebastian Roelsgaard, Yassine Sbai Sassi, Timo Schenk, 
Sami Stouli, Mark Watson, 
Martin Weidner, Kaspar W\"{u}thrich, Andrei Zeleneev, 
Lina Zhang and seminar participants at Oxford University, 
University of Michigan, New York University, University 
of Bristol, University of Exeter, University of 
Melbourne, Erasmus University Rotterdam, University of 
Groningen, Fordham University, University of Queensland, 
Tilburg University, University of Manchester, Princeton 
University, and participants at various workshops and 
conferences for comments and discussions. Any errors are my own.}}


\vspace{0.1in}

{ Gabriela M. Miyazato Szini$^{\dag}$}\\
{\today}

\end{center}

\renewcommand{\thefootnote}{\dag}
\footnotetext{Tilburg University, e-mail: \href{mailto:g.m.miyazatoszini@tilburguniversity.edu}{g.m.miyazatoszini@tilburguniversity.edu}}
\renewcommand{\thefootnote}{\arabic{footnote}}
\setcounter{footnote}{0}

\begin{center}\textbf{Abstract}\end{center}
\small\noindent 
I develop an estimation and inference framework for 
distribution regression in dyadic network settings with 
two-way fixed effects that vary across thresholds of the outcome. I show that identification of the structural 
parameters is achieved 
through binarization of the outcome at each threshold, 
and estimate the model by conditional maximum likelihood, 
which ``differences out'' the fixed effects and circumvents 
the incidental parameter problem. The estimator remains 
asymptotically unbiased under sparsity, whether 
from the network structure or binarization 
at extreme thresholds. The second 
novelty is to establish the joint asymptotic distribution of the 
estimators across multiple thresholds with different convergence rates, and to develop simultaneous confidence bands and tests for equality of coefficients across thresholds. Monte Carlo simulations confirm small bias, valid inference, and correct simultaneous coverage under sparsity. An application to bilateral trade finds that coefficients vary substantially across the distribution, with equality rejected for key trade barriers.

\vspace{0.1in}



\clearpage


\section{Introduction}
  
The vast majority of studies, especially for network models, estimate the effects of covariates on the mean of an outcome variable. However, in many settings, interest also lies in how the effects of covariates vary at different locations of the outcome distribution. For instance, in gravity models of international trade, the effects of trade barriers on bilateral exports may differ substantially across quantiles of trade flows. The distribution regression (DR) approach, initially proposed by \citet{foresi1995conditional} and developed further by \citet{chernozhukov2013inference}, addresses this by directly modelling the conditional distribution function through a sequence of binary choice models. At each threshold $y$, defined as a value of the outcome variable $y_{ij}$, the outcome is binarized into an indicator of whether it falls at or below $y$, and a binary response model is estimated. Varying the threshold characterizes how covariate effects differ across the distribution. Because the conditional distribution is approximated pointwise, the approach accommodates outcomes with point masses, such as variables bounded below at zero, without requiring smoothness of the conditional density.

A broad range of economic relationships involve bilateral 
interactions between agents, naturally giving rise to network 
data. Examples include models for international trade flows \citep{helpman2008estimating}, firm-level trade networks \citep{alfaro2023firm,bernard2022origins}, and earnings in employee-employer data \citep{bonhomme2019distributional}. In many of these settings, heterogeneity in the coefficients 
of covariates across different locations of the outcome 
distribution is of direct economic interest. Importantly, many of these networks are sparse, with only a small fraction of potential bilateral links realized. 

Accordingly, I consider a directed network structure through a dyadic model (where the outcomes reflect pairwise interactions among the sampled units, \citeauthor{graham2020dyadic}, \citeyear{graham2020dyadic}) with additively separable two-way fixed effects that capture unobserved heterogeneity of senders and receivers. In many network settings, nodes (agents) differ substantially in their unobserved characteristics, and failing to account for this heterogeneity can bias the estimated effects of covariates. In addition, in network formation models, the fixed effects capture degree heterogeneity, that is, the tendency of some nodes to form more or stronger connections than others. In the distribution regression setting, the fixed effects play an analogous role: capturing node-specific shifts in the conditional 
distribution at each threshold, allowing for systematic differences across senders and receivers in the level of the outcome. Both the structural parameters and the fixed effects are allowed to vary across levels of the outcome, providing 
a flexible characterization of the conditional distribution. Moreover, the fixed effects are treated as unrestricted 
parameters, with no distributional assumption imposed, 
neither on their distribution across nodes nor on their 
relationship with the covariates.

Each threshold of the distribution regression involves a nonlinear binary choice model with two-way fixed effects. Estimating these models by maximum likelihood gives rise to the incidental parameter problem \citep{neyman1948consistent}, since the number of fixed effects grows with the sample size, yielding asymptotically biased 
estimates and invalid inference. This problem is compounded by network sparsity, which occurs through two channels in the distribution regression setting. First, in sparse networks with outcomes bounded below at, for instance, 
zero, a large fraction of observations sit at the bound, so that the first threshold at which the binarized outcome varies 
already corresponds to an extreme quantile, compressing 
the estimable range into a sparse regime. At thresholds just above the bound, most binarized 
outcomes equal one, leaving many units without 
informative variation, rendering their fixed effects 
unidentified. Second, even 
in denser networks or with unbounded outcomes, 
binarization at extreme thresholds generates sparsity, 
since few observations fall in the tails.

To address these challenges, I develop a distribution regression framework based on the conditional maximum likelihood estimator (CMLE) of \citet{charbonneau2017multiple} and \citet{jochmans2018semiparametric}, originally proposed for network formation models. At each threshold, the approach relies on conditioning the likelihood function on a specific set of conditions for quadruples of nodes of the network, so that under a logistic specification the fixed effects are ``differenced out'' from the likelihood. Therefore, it yields estimates that are free of the incidental parameter problem and remain valid under sparsity. 

The contributions of this paper are threefold. First, I 
show that the structural parameters are identified in the 
flexible specification where 
both the coefficients and the fixed effects vary across 
thresholds, and that identification remains valid under 
network sparsity. The key mechanism is the binarization 
of the outcome at each threshold, which serves not only 
as an estimation tool but as an identification strategy. 
The pointwise asymptotic properties at each threshold 
follow from \citet{jochmans2018semiparametric}, and Monte 
Carlo simulations confirm low bias and valid pointwise 
inference across different degrees of network sparsity, 
including at extreme thresholds. 

Second, I develop methods for inference across thresholds. The object of interest in distribution regression is the entire profile of coefficients, and the empirical question is not only whether the effect of a covariate differs across the distribution, but where. Pointwise confidence intervals, plotted threshold by threshold, under-cover when interpreted jointly, while a joint equality test such as a Wald test delivers a single rejection decision without by itself indicating which parts of the distribution drive it. To address this, I derive the joint asymptotic distribution of the estimator across a finite number of thresholds, accommodating both the cross-threshold dependence (since the binary indicators at different thresholds are constructed from the same underlying outcome) and the varying convergence rates induced by the different degrees of sparsity at different locations of the distribution, without restricting these rates relative to one another. The main result is stated as a Gaussian approximation whose correlation matrix varies with the 
sample size, which requires the dependence across thresholds to be non-degenerate but it does not impose it to have a fixed limit. Using this result, I construct simultaneous confidence bands via the sup-$t$ approach of \citet{montiel2019simultaneous}, which cover the entire coefficient path with prescribed probability and thereby indicate where in the distribution the effects differ; inverting the bands yields a test of equality of the coefficients across thresholds. The Wald test is also asymptotically valid and complements the sup-$t$ test, as the two have power against different alternatives. In finite samples, however, the sup-$t$ bands deliver correct simultaneous coverage and the equality test maintains size close to the nominal level across designs, while the Wald test shows size distortions for some covariates as the number of thresholds grows.

Third, I apply the method to gravity models of international trade, testing whether the effects of trade barriers vary 
across the distribution of bilateral trade flows, a question with 
direct implications for welfare analysis in the international trade 
literature \citep{arkolakis2012new, melitz2015new, 
bergstrand2025quantile}. The setting is a natural one for 
distribution regression, since the outcome is bounded below at zero, 
approximately 55\% of country pairs record no bilateral trade, and 
the distribution among those with positive trade is heavily 
right-skewed. Moreover, recent evidence suggests that the effects of 
trade determinants vary substantially across the conditional 
distribution of trade flows \citep{baltagi2016estimationmain, 
bergstrand2025tails}, motivating methods that go beyond conditional 
mean estimation. The estimated coefficients indeed vary substantially 
across the distribution: the coefficient on distance increases from 
approximately $1.05$ at the median to $1.78$ at the 99th percentile, 
indicating that distance is a stronger barrier for the largest 
bilateral trade relationships, and the sup-$t$ test rejects the null 
of constant coefficients for distance, common legal system, and 
border, providing formal evidence of heterogeneity in the structural 
parameters across the distribution.

In general, conditional maximum likelihood methods eliminate the fixed effects by conditioning the likelihood on specific sets or sufficient statistics, avoiding their estimation entirely. In the network setting, \citet{charbonneau2017multiple} introduces the approach for directed networks, \citet{jochmans2018semiparametric} establishes its asymptotic theory under sparsity, and \citet{graham2017econometric} develops a related approach for undirected networks. Recent extensions include triadic network formation with dyad-level fixed effects \citep{muris2025triadic}, unified frameworks for static and dynamic binary choice in panels 
and networks \citep{dano2025binary}, estimation for ordered outcomes in networks \citep{muris2025dyadic}, and two-way duration models \citep{roelsgaard2025essays}. Relatedly, \citet{bonhomme2024functional} derive moment 
restrictions that eliminate fixed effects through functional 
differencing.\footnote{Approaches that relax the logistic specification or 
the additive structure of the fixed effects have also been 
developed for network formation models. 
\citet{candelaria2020semiparametric}, 
\citet{toth2017semiparametric}, and 
\citet{gao2020nonparametric} study identification of the 
structural parameters without a known parametric form for the 
disturbance term, while \citet{zeleneev2026identification} 
allows for nonparametric structure in the unobserved 
heterogeneity. These methods are developed for single binary choice models, 
primarily in undirected networks, and impose alternative 
assumptions such as special regressors, continuously 
distributed covariates, or compact support of the unobserved 
heterogeneity.} 

The present paper extends the framework from a single binary choice model to distribution regression, studying how structural parameters vary across the outcome distribution. The binarization at each threshold, which underpins 
identification in this paper, is related to the strategy 
used in the fixed-effects ordered logit literature to 
identify a common coefficient under the proportional odds 
assumption \citep{das1999panel, baetschmann2015consistent, 
muris2017estimation}. More broadly, the identification argument of this paper extends to generalized ordered choice 
models with fixed effects, including dyadic and network structures, where coefficients vary 
across categories; to my knowledge no estimator is currently available for that setting.

The setting considered in this paper builds on 
\citet{chernozhukov2020network}, who, to my knowledge, are the 
first to propose distribution regression for network models with 
two-way fixed effects. The key difference lies in the estimation 
method employed at each threshold. They address the incidental 
parameter problem at each threshold through analytical bias 
corrections \citep{fernandez2016individual}. Several related approaches have 
been proposed for single binary choice models in networks, including 
analytical corrections for probit specifications 
\citep{dzemski2019empirical} and jackknife bias 
corrections \citep{hughes2026jackknife}. These methods 
require dense network asymptotics, meaning that the link 
probabilities must be bounded away from zero and one.\footnote{\citet{yan2019statistical} allow fixed 
effects to grow at rate $\log n$, permitting some degree of 
sparsity, but expected degrees must still grow nearly proportionally 
to $n$.} In the distribution regression setting, this translates to requiring that the probability of the outcome falling below 
a given threshold is bounded away from zero and one at each 
threshold, a condition that is necessarily violated in the 
tails of the distribution, where the binarized outcome has 
little or no variation. In applications where, for instance, the outcome is bounded below at zero, and there is a high prevalence of zeros, the threshold at which the 
binarized outcome first exhibits variation already corresponds 
to an extreme quantile, so that this condition fails 
broadly rather than only in the tails.

The conditional maximum likelihood approach does not require this assumption. The fixed effects are eliminated from the likelihood entirely, and the estimator remains consistent under sequences where the link probabilities can approach zero or one. The Monte Carlo simulations confirm this contrast in finite samples. An advantage of the bias correction approach, however, is that it delivers estimates of the fixed 
effects, enabling the construction of counterfactual distributions 
and average partial effects, objects that the conditional 
likelihood approach cannot recover, since the fixed effects are 
eliminated by conditioning. The two approaches are therefore complementary, 
with the choice depending on whether the application requires 
inference on structural parameters under sparsity or estimation of fixed effects for constructing 
counterfactual distributions and average effects, which 
requires a dense setting.

\vspace{0.25cm}
\noindent \textbf{Plan of the paper.} Section \ref{model_estimation} outlines the model and presents the estimation method; Section \ref{asymptotics} shows the pointwise asymptotic properties of the proposed estimator and discusses the sparsity conditions and the estimable quantile (threshold) range; Section \ref{joint_distribution} derives the joint distribution of the estimators across a finite number of thresholds, and constructs the sup-$t$ confidence bands and equality tests; Section \ref{simulations} presents the different settings for the Monte Carlo simulations and the obtained results; Section \ref{application} applies the method to gravity 
models of international trade; and Section \ref{conclusion} concludes.

\section{Model and Estimation}\label{model_estimation} 
\subsection{Distribution Regression Model for Networks}\label{subsection_model}

This section introduces the distribution regression model for a directed network structure formed through bilateral ties of units. To accommodate the network framework, a general dyadic setting is considered. More specifically, the conditional distribution function is parametrized as a function of dyad-specific characteristics and fixed effects for each unit in the observed pair of nodes.

The model follows \citet{chernozhukov2020network}, who 
introduced the distribution regression framework with two-way 
fixed effects for network data. Let $\{(y_{ij}, \boldsymbol{x}_{ij}) : (i,j) \in \mathcal{D} \}$ be the observed dataset, where $y_{ij}$ is a scalar outcome variable that can be discrete, continuous or mixed for a dyad $(i,j)$, and $\boldsymbol{x}_{ij}$ is a vector of dyad-specific covariates, such as measures of distance or similarity between units. Let $\mathcal{Y}$ be a region of interest in the support 
of the outcome and $\mathcal{X} \subseteq \mathbb{R}^p$ 
the support of the covariates. The asymptotic theory is 
developed for a finite collection of thresholds in 
$\mathcal{Y}$; when 
the outcome is discrete with finite support, $\mathcal{Y}$ 
can be taken as a subset of the support, while for 
continuous outcomes, $\mathcal{Y}$ is a finite grid of 
values in the support. The set of nodes\footnote{Thoughout the paper, I use units, nodes, or individuals interchangeably.} in the network is given by $\mathcal{N} = \{1,2, \dots, n\}$, and the set of observed 
dyads is $\mathcal{D} = \{(i,j) : i \neq j,\; i,j \in 
\mathcal{N}\}$, with $|\mathcal{D}| = n(n-1)$.\footnote{We consider that all the nodes are senders and receivers, but the method in this paper also allows for cases where the nodes that are senders differs from the nodes that are receivers, i.e., $i = 1, \dots, I$ and $j = 1, \dots J$, with $I \neq J$; and also for self-links to be formed.}
\enlargethispage{\baselineskip}

The unobserved heterogeneity of units $i$ and $j$ is captured by vectors $\boldsymbol{\nu}_i$ and $\boldsymbol{\omega}_j$ of unspecified dimension, whose relationship with the covariates 
$\boldsymbol{x}_{ij}$ is left unrestricted. I assume that the conditional distribution of $y_{ij}$ given $(\boldsymbol{x}_{ij}, \boldsymbol{\nu}_i, \boldsymbol{\omega}_j)$ is given by:
\begin{align} \label{eq:model}
    F_{y_{i j}}\left(y \mid \boldsymbol{x}_{i j}, \boldsymbol{\nu}_i, \boldsymbol{\omega}_j\right)=\Lambda \left(\boldsymbol{x}_{i j}^{\prime} \boldsymbol{\theta}_0(y)+\alpha\left(\boldsymbol{\nu}_{i}, y\right)+\gamma\left(\boldsymbol{\omega}_{j}, y\right)\right), \quad y \in \mathcal{Y}, \quad(i, j) \in \mathcal{D},
\end{align}
{\looseness=-1 \noindent where $\Lambda(\cdot)$ is a known link function, and $\boldsymbol{\theta}_0(y)$ is an unknown parameter vector of interest varying with $y$. $\alpha\left(\boldsymbol{\nu}_{i}, y\right)$ and $\gamma\left(\boldsymbol{\omega}_{j}, y\right)$ are unspecified measurable functions that can be seen as the unobserved individual fixed effects at a given level of $y$. This model is naturally semiparametric, not only because the parameters are allowed to vary with the output levels but also because it does not restrict how the individual unobserved effects correlate with the covariates. For notational simplicity, and without loss of generality, I denote $\boldsymbol{\theta}_0(y) = \boldsymbol{\theta}_{y,0}$, $\alpha_{i,y} = \alpha(\boldsymbol{\nu}_{i}, y)$ and $\gamma_{j,y} = \gamma(\boldsymbol{\omega}_{j}, y)$. When referring to a finite collection of thresholds 
$\{y_1, \ldots, y_K\}$, the parameter vector at the 
$k$-th threshold is denoted 
$\boldsymbol{\theta}_{y_k}$, with $\theta_{y_k,d}$ 
its $d$-th element.\par}

Throughout this paper, $\Lambda(\cdot)$ is the logistic 
distribution. While the distribution regression framework 
accommodates general link functions 
\citep{foresi1995conditional, chernozhukov2013inference}, 
the logistic specification is standard in fixed-effects 
settings \citep{charbonneau2017multiple, 
jochmans2018semiparametric, chernozhukov2020network}. This specification is essential for the conditional likelihood approach in Section \ref{subsection_estimation_method}, since it is the only link under which conditioning on sufficient statistics yields a likelihood free of nuisance parameters.

A key feature of this model is the two-way fixed effects, 
which account for part of the dependence across 
dyads. For instance, the outcomes for dyads $(i,j)$ and $(i,k)$ can be correlated through the shared sender effect 
$\alpha_{i,y}$ and possible correlations in covariates 
sharing index $i$. Allowing sender and receiver effects 
to differ, together with $y_{ij}$ and $\boldsymbol{x}_{ij}$ not 
necessarily equal to $y_{ji}$ and 
$\boldsymbol{x}_{ji}$, accommodates directed 
networks. However, notice that the model outlined in this section and the estimator proposed in the following section can be easily modified to accommodate undirected and bipartite networks. The additive separability of the fixed effects, 
which is standard in the network formation literature 
\citep{charbonneau2017multiple, jochmans2018semiparametric, 
graham2017econometric}, is what enables the conditional 
likelihood approach developed in the next section.
\enlargethispage{\baselineskip}

Finally, the conditional  distribution $F_{y_{ij}} (y \mid \boldsymbol{x}_{ij}, \boldsymbol{\nu}_i, \boldsymbol{\omega}_j)$ can be written as:
\begin{align} \label{eq:binary}
    F_{y_{ij}}\left(y \mid \boldsymbol{x}_{ij}, \boldsymbol{\nu}_i, \boldsymbol{\omega}_j \right) &= \mathbb{E} [{1} \{y_{ij} \leq y \} \mid \boldsymbol{x}_{ij}, \boldsymbol{\nu}_i, \boldsymbol{\omega}_j] \nonumber \\
    &= \text{Pr}[\tilde{y}_{ij,y} = 1 \mid \boldsymbol{x}_{ij}, \boldsymbol{\nu}_i, \boldsymbol{\omega}_j] \nonumber \\
    &=\Lambda \left(\boldsymbol{x}_{i j}^{\prime} \boldsymbol{\theta}_{y,0}+\alpha_{i,y}+\gamma_{j,y}\right),
\end{align}
\noindent where $\tilde{y}_{ij,y} = 1\{y_{ij} \leq y \}$ is the binary indicator that the outcome falls at or below threshold $y$. The parameters $\boldsymbol{\theta}_{y,0}$ can therefore be estimated for each $y \in \mathcal{Y}$ as a sequence of binary (logistic) regressions with two-way fixed effects.\footnote{Distribution regression is preferred over 
quantile regression (QR) in this setting because the 
linear-in-parameters QR may poorly approximate the 
conditional distribution when the outcome does not have a 
smooth conditional density. In contrast, the DR 
approximates the conditional distribution pointwise at 
each threshold in the support of the outcome, making it 
well-suited for outcomes with point masses without 
requiring strong assumptions on how such masses are 
generated \citep{chernozhukov2020network}. Moreover, to 
my knowledge, no QR methods for two-way fixed effects in 
dyadic network settings are currently available.} At each threshold, the maximum likelihood estimator is subject to the incidental parameter problem. Moreover, identification of individual fixed effects requires sufficient variation in the binary outcomes at each threshold, which may fail under sparsity. Both concerns motivate the conditional maximum likelihood approach developed in the next section, which eliminates the fixed effects from the likelihood by conditioning, and therefore does not require their identification or estimation. 
\subsection{Conditional Maximum Likelihood 
Estimation}\label{subsection_estimation_method}

The main challenge in the estimation of the sequence of binary regressions given by Equation \eqref{eq:binary} is that, even for a single binary regression, the incidental parameter problem \citep{neyman1948consistent} arises from estimating 
models with two-way fixed effects by maximum likelihood. To circumvent this problem, I propose to estimate the parameters of the model $\boldsymbol{\theta}(y)$ for each threshold (for a given level $y$) independently, using the conditional maximum-likelihood method for network formation models suggested by \cite{charbonneau2017multiple} (for directed networks) and concurrently by \cite{graham2017econometric} (for undirected networks). This is applicable here because each binarized 
threshold yields a binary choice model with two-way fixed 
effects, identical in structure to a network formation 
model. The approach extends the conditional maximum likelihood method for logistic panel data models with a single fixed effect \citep{rasch1960studies, chamberlain2010binary}\footnote{Also refer to \cite{arellano2001panel} for a survey.} to models with two-way fixed effects in dyadic structures, avoiding any 
distributional assumption on the fixed effects.

The method relies on the existence of a set of conditions for quadruples of nodes in the observed network such that the fixed effects drop out of the conditional likelihood under a logistic link. The logistic specification is 
essential: it is the only link under which the 
conditional likelihood is entirely free of nuisance 
parameters and takes a known closed 
form.\footnote{In standard panel settings, 
\citet{chamberlain2010binary} shows that for $T=2$, 
parametric-rate estimation is only possible under 
logistic errors. For $T \geq 3$, 
\citet{davezies2023fixed} show that identification 
beyond logistic is possible through conditional moment 
restrictions estimated by GMM, though these do not yield 
a closed-form conditional likelihood.}

From the conditional distribution $F_{y_{ij}}$ and the constructed binary variables $\tilde{y}_{ij,y}$, it follows that
$$ \tilde{y}_{ij,y} = 1\{\boldsymbol{x}_{ij}'\boldsymbol{\theta}_{y,0} + {\alpha}_{i,y} + {\gamma}_{j,y} + \varepsilon_{ij,y} \geq 0\}, \quad (i,j) \in \mathcal{D} $$

\noindent where $\varepsilon_{ij,y} \sim \text{i.i.d.}\ 
\text{Logistic}(0,1)$ across dyads for each threshold 
$y$, with $\varepsilon_{ij,y} \perp 
\{\boldsymbol{x}_{ij}, \alpha_{i,y}, \gamma_{j,y}\}$. 
Under the logistic assumption, the conditional 
probability takes the explicit form:
\begin{align} \label{eq:main_charbonneau}
    \mathbb{E} [1\{y_{ij} \leq y\} \mid \boldsymbol{x}_{ij}, {\alpha}_{i,y}, {\gamma}_{j,y}] &= \text{Pr}[\tilde{y}_{ij,y} = 1 \mid \boldsymbol{x}_{ij}, {\alpha}_{i,y}, {\gamma}_{j,y}] \nonumber\\
    &= \frac{\text{exp}(\boldsymbol{x}_{ij}'\boldsymbol{\theta}_{y,0} + {\alpha}_{i,y} + {\gamma}_{j,y})}{1+\text{exp}(\boldsymbol{x}_{ij}'\boldsymbol{\theta}_{y,0} + {\alpha}_{i,y} + {\gamma}_{j,y})}
\end{align}

The logistic specification ensures the existence of 
sufficient statistics for the fixed effects, which is formalized in the Lemma below.

\begin{lemma}\label{lemma_sufficient}
    Under the model specification given by Equation \eqref{eq:main_charbonneau}, the sums across each dimension of the pseudo panel, $\sum_{j=1}^n \tilde{y}_{ij,y}$ and $\sum_{i=1}^n \tilde{y}_{ij,y}$, are sufficient statistics for ${\alpha}_{i,y}$ and ${\gamma}_{j,y}$.
\end{lemma}
\noindent \textit{Proof.} See \ref{appendix_sufficient}.

This result is well-known for the standard panel case with one fixed effect, and \cite{graham2017econometric} establishes sufficiency of the degree sequence for undirected networks with a single set of node effects. Lemma \ref{lemma_sufficient} extends it to directed networks with two-way fixed effects, where both the row and column sums serve as sufficient statistics.

Even though one could construct a conditional maximum likelihood estimator based on the sufficient statistics, the maximization is intractable. \cite{charbonneau2017multiple} provides a tractable alternative by showing that it is possible to eliminate the two-way fixed effects by conditioning the above probability on the set of events $\{\tilde{y}_{ij,y} + \tilde{y}_{ik,y} = 1, \tilde{y}_{lj,y} + \tilde{y}_{lk,y} = 1,  \tilde{y}_{ij,y} + \tilde{y}_{lk,y} = 1\}$ for different indices of senders and receivers $\{i,l;j,k\}$ (quadruples of nodes), such that:
\begin{align}
    \text{Pr}[\tilde{y}_{ij,y} &= 1 \mid \boldsymbol{x}_{ij}, {\alpha}_{i,y}, {\gamma}_{j,y}, \tilde{y}_{ij,y} + \tilde{y}_{ik,y} = 1, \tilde{y}_{lj,y} + \tilde{y}_{lk,y} = 1,  \tilde{y}_{ij,y} + \tilde{y}_{lk,y} = 1] \nonumber \\ &= \frac{\text{exp}(((\boldsymbol{x}_{ij} - \boldsymbol{x}_{ik}) - (\boldsymbol{x}_{lj} - \boldsymbol{x}_{lk})) '\boldsymbol{\theta}_{y,0})}{1+\text{exp}(((\boldsymbol{x}_{ij} - \boldsymbol{x}_{ik}) - (\boldsymbol{x}_{lj} - \boldsymbol{x}_{lk})) '\boldsymbol{\theta}_{y,0})},
\end{align}
\noindent which no longer depends on the fixed effects. This result follows from applying the classical conditional logit argument for static panel data models with a single fixed effect \citep{rasch1960studies} sequentially, first to eliminate the sender fixed effects, then the receiver fixed effects.

Summing over all quadruples that satisfy the conditioning 
events, the CMLE maximizes:
\begin{align} \label{eq:main}
    \sum_{i=1}^n \sum_{j=1, j \neq i}^n \sum_{l,k \in \mathbb{Z}_{ij,y}} \text{log} \left( \frac{\text{exp}(((\boldsymbol{x}_{ij} - \boldsymbol{x}_{ik}) - (\boldsymbol{x}_{lj} - \boldsymbol{x}_{lk})) '\boldsymbol{\theta}_y)}{1 + \text{exp}(((\boldsymbol{x}_{ij} - \boldsymbol{x}_{ik}) - (\boldsymbol{x}_{lj} - \boldsymbol{x}_{lk})) '\boldsymbol{\theta}_y)}\right) ,
\end{align}
where $\mathbb{Z}_{ij,y}$ is the set of nodes $k$ and 
$l$ satisfying the conditioning events for pair $(i,j)$. 
In practice, this reduces to a standard logit estimation 
on pairwise-differenced outcomes and covariates, as 
described in the next section.

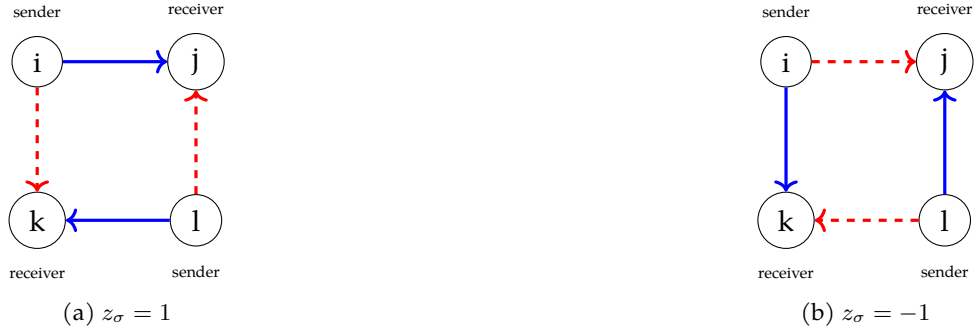
\begin{figure}[h]
\centering
\captionsetup{skip=0pt}
\captionsetup[subfigure]{font=footnotesize, skip=0pt}
\begin{subfigure}[b]{0.40\textwidth}
\centering
\begin{tikzpicture}[scale=1.4]
\node[circle,draw,minimum size=0.5cm] (i) at (0,1.5) {\small i};
\node[circle,draw,minimum size=0.5cm] (j) at (1.5,1.5) {\small j};
\node[circle,draw,minimum size=0.5cm] (k) at (0,0) {\small k};
\node[circle,draw,minimum size=0.5cm] (l) at (1.5,0) {\small l};
\node[above=0.1cm of i, font=\tiny] {sender};
\node[above=0.1cm of j, font=\tiny] {receiver};
\node[below=0.1cm of k, font=\tiny] {receiver};
\node[below=0.1cm of l, font=\tiny] {sender};
\draw[->,thick,blue,line width=1.2pt] (i) -- (j);
\draw[->,thick,blue,line width=1.2pt] (l) -- (k);
\draw[->,dashed,red,line width=1.2pt] (i) -- (k);
\draw[->,dashed,red,line width=1.2pt] (l) -- (j);
\end{tikzpicture}
\caption{$z_\sigma = 1$}
\end{subfigure}
\hfill
\begin{subfigure}[b]{0.40\textwidth}
\centering
\begin{tikzpicture}[scale=1.4]
\node[circle,draw,minimum size=0.5cm] (i) at (0,1.5) {\small i};
\node[circle,draw,minimum size=0.5cm] (j) at (1.5,1.5) {\small j};
\node[circle,draw,minimum size=0.5cm] (k) at (0,0) {\small k};
\node[circle,draw,minimum size=0.5cm] (l) at (1.5,0) {\small l};
\node[above=0.1cm of i, font=\tiny] {sender};
\node[above=0.1cm of j, font=\tiny] {receiver};
\node[below=0.1cm of k, font=\tiny] {receiver};
\node[below=0.1cm of l, font=\tiny] {sender};
\draw[->,dashed,red,line width=1.2pt] (i) -- (j);
\draw[->,dashed,red,line width=1.2pt] (l) -- (k);
\draw[->,thick,blue,line width=1.2pt] (i) -- (k);
\draw[->,thick,blue,line width=1.2pt] (l) -- (j);
\end{tikzpicture}
\caption{$z_\sigma = -1$}
\end{subfigure}

\caption{Informative quadruple configurations with 
\{i,l\} as senders and \{j,k\} as receivers. Blue solid 
arrows: links present ($\tilde{y}=1$). Red dashed arrows: 
links absent ($\tilde{y}=0$). The two senders connect to 
opposite receivers, yielding $z_\sigma \in \{-1, 1\}$. 
Values: (a) $\tilde{y}_{ij,y}=1, \tilde{y}_{ik,y}=0, 
\tilde{y}_{lj,y}=0, \tilde{y}_{lk,y}=1$; (b) 
$\tilde{y}_{ij,y}=0, \tilde{y}_{ik,y}=1, 
\tilde{y}_{lj,y}=1, \tilde{y}_{lk,y}=0$.}
\label{figure1:correct}
\end{figure}

\begin{figure}[h]
\centering
\captionsetup{skip=0pt}
\captionsetup[subfigure]{font=footnotesize, skip=0pt}
\begin{subfigure}[b]{0.30\textwidth}
\centering
\begin{tikzpicture}[scale=1.4]
\node[circle,draw,minimum size=0.5cm] (i) at (0,1.5) {\small i};
\node[circle,draw,minimum size=0.5cm] (j) at (1.5,1.5) {\small j};
\node[circle,draw,minimum size=0.5cm] (k) at (0,0) {\small k};
\node[circle,draw,minimum size=0.5cm] (l) at (1.5,0) {\small l};
\node[above=0.1cm of i, font=\tiny] {sender};
\node[above=0.1cm of j, font=\tiny] {receiver};
\node[below=0.1cm of k, font=\tiny] {receiver};
\node[below=0.1cm of l, font=\tiny] {sender};
\draw[->,dashed,red,line width=1.2pt] (i) -- (j);
\draw[->,dashed,red,line width=1.2pt] (l) -- (k);
\draw[->,dashed,red,line width=1.2pt] (i) -- (k);
\draw[->,dashed,red,line width=1.2pt] (l) -- (j);
\end{tikzpicture}
\caption{All absent ($z_\sigma = 0$)}
\end{subfigure}
\hfill
\begin{subfigure}[b]{0.30\textwidth}
\centering
\begin{tikzpicture}[scale=1.4]
\node[circle,draw,minimum size=0.5cm] (i) at (0,1.5) {\small i};
\node[circle,draw,minimum size=0.5cm] (j) at (1.5,1.5) {\small j};
\node[circle,draw,minimum size=0.5cm] (k) at (0,0) {\small k};
\node[circle,draw,minimum size=0.5cm] (l) at (1.5,0) {\small l};
\node[above=0.1cm of i, font=\tiny] {sender};
\node[above=0.1cm of j, font=\tiny] {receiver};
\node[below=0.1cm of k, font=\tiny] {receiver};
\node[below=0.1cm of l, font=\tiny] {sender};
\draw[->,thick,red,line width=1.2pt] (i) -- (j);
\draw[->,thick,red,line width=1.2pt] (l) -- (k);
\draw[->,thick,red,line width=1.2pt] (i) -- (k);
\draw[->,thick,red,line width=1.2pt] (l) -- (j);
\end{tikzpicture}
\caption{All present ($z_\sigma = 0$)}
\end{subfigure}
\hfill
\begin{subfigure}[b]{0.30\textwidth}
\centering
\begin{tikzpicture}[scale=1.4]
\node[circle,draw,minimum size=0.5cm] (i) at (0,1.5) {\small i};
\node[circle,draw,minimum size=0.5cm] (j) at (1.5,1.5) {\small j};
\node[circle,draw,minimum size=0.5cm] (k) at (0,0) {\small k};
\node[circle,draw,minimum size=0.5cm] (l) at (1.5,0) {\small l};
\node[above=0.1cm of i, font=\tiny] {sender};
\node[above=0.1cm of j, font=\tiny] {receiver};
\node[below=0.1cm of k, font=\tiny] {receiver};
\node[below=0.1cm of l, font=\tiny] {sender};
\draw[->,thick,blue,line width=1.2pt] (i) -- (j);
\draw[->,thick,blue,line width=1.2pt] (l) -- (j);
\draw[->,dashed,red,line width=1.2pt] (i) -- (k);
\draw[->,dashed,red,line width=1.2pt] (l) -- (k);
\end{tikzpicture}
\caption{Same direction ($z_\sigma = 0$)}
\end{subfigure}

\caption{Examples of non-informative quadruple configurations with 
\{i,l\} as senders and \{j,k\} as receivers. In (a), all 
links are absent; in (b), all links are present; in (c), 
both senders connect to the same receiver. In all three 
cases $z_\sigma = 0$, so the quadruple provides no 
information for estimation.}
\label{Figure3:exampletetrads}
\end{figure}
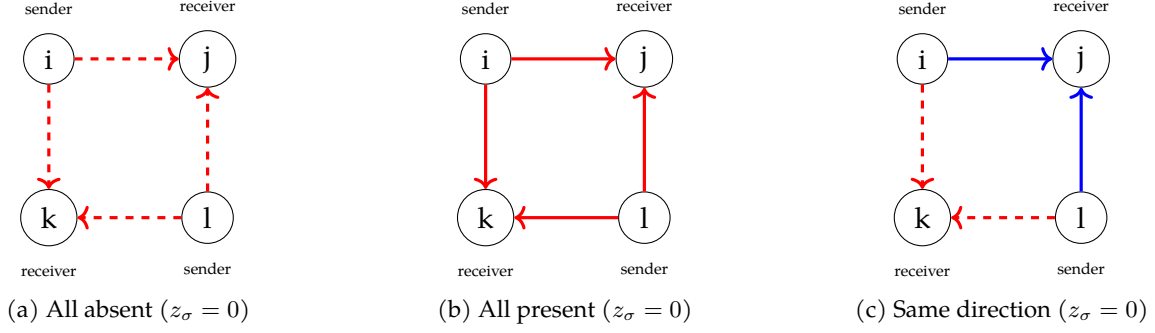
As in the static panel logit considered in \cite{rasch1960studies} and \cite{chamberlain2010binary}, only 
quadruples whose outcomes exhibit variation (the 
analogues of movers in the panel data literature) 
contribute to the likelihood: a quadruple $\{i,l;j,k\}$ is informative only if each 
node's outcomes vary across its two potential links. To illustrate this argument, Figure~\ref{figure1:correct} shows the two informative 
configurations for a fixed sender-receiver assignment 
$\{i,l;j,k\}$ ($i$ and $l$ are fixed to be senders; and $j$ and $k$ are fixed to be receivers). Note that there is variation in the outcomes for each node, i.e., each node has exactly one link 
present and one absent, yielding $z_\sigma \in \{-1, 1\}$. In a directed network, the sender-receiver roles can 
also be reversed; Figure~\ref{figure2:correct} in 
\ref{identification} shows all twelve informative 
configurations across all possible assignments. Figure~\ref{Figure3:exampletetrads} illustrates 
three non-informative configurations, maintaining the same 
sender-receiver assignment. In Subfigure (a), all links are 
absent; in Subfigure (b), all links are present; in both 
cases, no node's outcomes vary. In Subfigure (c), both 
senders connect to the same receiver, so that while each 
sender's outcomes vary, the receivers' do not, and the 
pairwise difference yields $z_\sigma = 0$. None of these 
quadruples contribute to the likelihood. For further 
intuition on the identification of the common parameters, I 
refer to \ref{identification}.

\begin{remark}[Conditioning across thresholds]
If the fixed effects were constant across different levels of the outcome $y$ (thresholds), 
one could also condition on events combining information across thresholds, such as $\tilde{y}_{ij,y_1} + \tilde{y}_{ij,y_2} = 1$. However, this approach is not pursued in this paper. Allowing the fixed effects to vary flexibly across thresholds, as done in the current framework, provides a more general specification that avoids possible misspecification.
\end{remark}

A property of the model central to the asymptotic theory of this estimator is that the binary indicators are conditionally independent across dyads:
$$
\Pr(\tilde{y}_{ij,y}, \tilde{y}_{kl,y} \mid 
\{\boldsymbol{x}_{ij}\}_{n,n}, \{\alpha_{i,y}, \gamma_{j,y}\}_n) = 
\Pr(\tilde{y}_{ij,y} \mid \boldsymbol{x}_{ij}, \alpha_{i,y}, 
\gamma_{j,y}) \cdot \Pr(\tilde{y}_{kl,y} \mid \boldsymbol{x}_{kl}, 
\alpha_{k,y}, \gamma_{l,y})
$$
for all $(i,j) \neq (k,l)$, where $\{\boldsymbol{x}_{ij}\}_{n,n}$ denotes the entire set of covariates, and $\{\alpha_{i,y}, \gamma_{j,y}\}_{n}$ denotes the full set of fixed effects. The conditional independence follows from the errors $\varepsilon_{ij,y}$ being independent across dyads for each threshold $y$, even when 
dyads share a node (with $\varepsilon_{ij,y}$ and 
$\varepsilon_{ji,y}$ treated as distinct and independent).\footnote{One drawback is that, in a network formation model context, transitivity across the probabilities is not taken into account by this model. It rules out interdependent link preferences, where individuals' preferences over a link may vary with the presence or absence of links elsewhere in the network. However, \cite{dzemski2019empirical} shows that this dyadic structure can still reproduce the transitivity patterns observed in some datasets.} The model thus belongs to the class of conditionally independent dyad (CID) models \citep{fafchamps2007formation, 
graham2020dyadic}. 

While the conditional independence holds across dyads within each threshold, the model allows for correlation of $\varepsilon_{ij,y}$ and $\varepsilon_{ij,y'}$ across thresholds $y \neq y'$ for the same dyad, since estimation is performed 
separately at each threshold. This 
cross-threshold dependence becomes relevant when 
deriving the joint distribution in Section 
\ref{joint_distribution}.

\section{Pointwise Asymptotics}\label{asymptotics}

The pointwise asymptotic properties of the estimator $\boldsymbol{\theta}_{n,y}$ for a single threshold value $y$ of the conditional distribution follow from results in \cite{jochmans2018semiparametric} for the estimator of \cite{charbonneau2017multiple}. The proofs in \ref{appendix_asymptotics} adapt the same broad structure, with modifications that facilitate the extension to the joint asymptotic distribution across thresholds in Section \ref{joint_distribution}. Subsection \ref{subsection_sparsity} discusses their implications for the estimable range of thresholds in the distribution regression setting. 

Throughout this section, the sequence of individual 
effects $\{\alpha_{i,y}, \gamma_{j,y}\}_n$ are treated 
as fixed, since the analysis conditions on them. The 
asymptotic framework considers $n \to \infty$, so that 
both dimensions of the network grow at the same rate. For an ordered quadruple of distinct nodes $\{i,l;j,k\}$ from $\mathcal{N}$, define:
$$ z_y(\sigma\{i,l;j,k\}) = \frac{(\tilde{y}_{ij,y} - \tilde{y}_{ik,y}) - (\tilde{y}_{lj,y} - \tilde{y}_{lk,y})}{2} $$
$$ \boldsymbol{r}(\sigma\{i,l;j,k\}) = (\boldsymbol{x}_{ij} - \boldsymbol{x}_{ik}) - (\boldsymbol{x}_{lj} - \boldsymbol{x}_{lk}) ,$$
\noindent where $\sigma(\cdot)$ maps an ordered quadruple to the index set $\mathcal{N}_{m_n} = \{1,2,\dots, m_n \}$, with $m_n$ denoting the number of distinct ordered quadruples from $\mathcal{N}$, i.e, $m_n = n (n-1) (n-2) (n-3)$.\footnote{\citet{jochmans2018semiparametric} exploits the permutation invariance of the score contributions in senders $(i,l)$ and receivers $(j,k)$, thus, considering combinations of senders and receivers and defining $m_n = n (n-1) (n-2) (n-3)/4$. I depart from this and consider ordered quadruples of nodes throughout. Since the score contributions are invariant under these permutations, the two formulations yield identical estimators. This choice is made for consistency with the projections of the scores defined later in this Section and with the proofs, which, in this paper, follow the ordered structure of indices.} Hereafter, I use the shortcut notation $z_{\sigma,y}$ and $\boldsymbol{r}_\sigma$.

The transformed variable $z_{\sigma,y}$ takes values in 
$\{-1,-1/2,0,1/2,1\}$, with $z_{\sigma,y} \in \{-1,1\}$ 
corresponding to the conditioning set 
$\{\tilde{y}_{ij,y} + \tilde{y}_{ik,y} = 1,\; 
\tilde{y}_{lj,y} + \tilde{y}_{lk,y} = 1,\; 
\tilde{y}_{ij,y} + \tilde{y}_{lk,y} = 1\}$ from 
Section~\ref{subsection_estimation_method}. Collecting $\boldsymbol{x} = (\boldsymbol{x}_{ij}, 
\boldsymbol{x}_{ik}, \boldsymbol{x}_{lj}, 
\boldsymbol{x}_{lk})$, the conditioning argument from the previous section yields:

\begin{lemma}\label{lemma:sufficiency}(Sufficiency) $$ \operatorname{Pr}[z_{\sigma,y} = 1 \mid \boldsymbol{x}, z_{\sigma,y} \in \{-1,1\}] = \frac{\exp(\boldsymbol{r}_\sigma' \boldsymbol{\theta}_{y,0})}{1 + \exp(\boldsymbol{r}_\sigma' \boldsymbol{\theta}_{y,0})}$$
\end{lemma}
\textit{Proof.} Follows from the conditioning argument in 
Subsection~\ref{subsection_estimation_method}. In particular, it follows immediately from the exponential 
family structure of the logistic specification in 
Equation~\eqref{eq:main_charbonneau}.

Lemma \ref{lemma:sufficiency} implies that, after conditioning, the fixed effects are eliminated from the likelihood, circumventing the incidental parameter problem, and each 
informative quadruple contributes a standard logistic 
term. Summing over all quadruples in $\mathcal{N}_{m_n}$, for which $z_{\sigma,y} \in \{-1,1\}$, the estimator is defined as
$$ {\boldsymbol{\theta}}_{n,y} = \argmax_{\boldsymbol{\theta}_y \in \Theta} L_n (\boldsymbol{\theta}_y),$$
\noindent where $\Theta$ is the parameter space searched over, and the conditional log-likelihood takes the form
$$ L_n (\boldsymbol{\theta}_y) = \sum_{\sigma \in \mathcal{N}_{m_n}} 1 \{z_{\sigma,y} = 1 \} \text{log} \Lambda(\boldsymbol{r}_\sigma' \boldsymbol{\theta}_y) + 1 \{z_{\sigma,y} = -1 \} \text{log} (1 -\Lambda(\boldsymbol{r}_\sigma' \boldsymbol{\theta}_y)).$$
The number of quadruples with $z_{\sigma,y} \in \{-1,1\}$ is denoted by $m_{n,y}^* = \sum_{\sigma \in \mathcal{N}_{m_n}} 1 \{z_{\sigma,y} \in \{-1,1\}\}$. 

The following standard assumptions are needed to establish consistency of the estimator: 

\begin{assumption}(Sampling) The n nodes in $\mathcal{N}$ are sampled independently.
    \label{assumption1} 
\end{assumption}
\begin{assumption}(Parameter space) $\boldsymbol{\theta}_{y,0}$ is interior to $\Theta$, a compact subset of $\mathbb{R}^{dim (\boldsymbol{\theta}_y)}$.
    \label{assumption2} 
\end{assumption}
\begin{assumption}(Moments) For all $(i,j) \in \mathcal{D}$, $\mathbb{E} (||\boldsymbol{x}_{ij}||^2) < C_1$, where $C_1$ is a finite constant.
    \label{assumption3} 
\end{assumption}
Define the expected fraction of quadruples that contribute to the log-likelihood at each threshold as:
$$ p_{n,y} = \frac{\mathbb{E}(m^*_{n,y})}{m_n} = \frac{\sum_{\sigma \in \mathcal{N}_{m_n}} \Pr \{z_{\sigma,y} \in \{-1,1\}\}}{m_n}. $$
\begin{assumption} (Identification) $n p_{n,y} \xrightarrow[]{} \infty$ as $n \xrightarrow[]{} \infty$ and the matrix 
$$ \lim_{n \xrightarrow[]{} \infty} (m_n p_{n,y})^{-1} \sum_{\sigma \in \mathcal{N}_{m_n}} \mathbb{E}(-\boldsymbol{r}_\sigma\boldsymbol{r}_\sigma' f(\boldsymbol{r}_\sigma'\boldsymbol{\theta}_{y,0}) 1\{z_{\sigma,y}\in \{-1,1\}\}), $$ where $f$ is the logistic density function, has maximal rank. 
\label{assumption4} 
\end{assumption}

Assumption \ref{assumption1} is a standard sampling scheme for network data, requiring independent sampling of nodes. Importantly, it does not impose independence across dyads, since covariates for dyads sharing a node may be correlated, which, together with the same unrestricted fixed effects appearing across different pairs sharing a node, generates the network dependence. This accommodates for settings where covariates take the form $x_{ij} = g(x_i, x_j)$ \citep{graham2017econometric}, where $g(\cdot)$ is a measurable function, but it is more general than that. Moreover, the assumption does not require that the nodes are identically distributed.

Assumptions \ref{assumption2}, \ref{assumption3} and the rank condition in Assumption \ref{assumption4} are standard for establishing consistency in non-linear models \citep{newey1994chapter}. The matrix in Assumption \ref{assumption4} is the limit of the normalized expected Hessian, so that the rank condition is equivalent to requiring the limiting Hessian to be negative definite. The key condition in Assumption \ref{assumption4} is that $np_{n,y} \to \infty$, where $p_{n,y}$ measures the expected fraction of informative quadruples at the fixed threshold $y$. In the network formation setting, for which this estimator was originally proposed, this allows $p_{n,y}$ to shrink as $n$ grows, so that the probability of link formation can approach zero or one. In the DR context, this condition allows the probabilities $\Pr(y_{ij} \leq y)$ to approach zero or one, accommodating 
sparsity both from the underlying network structure and from 
binarization at extreme thresholds. Since $p_{n,y}$ is the 
average of $\Pr(z_{\sigma,y} \in \{-1,1\})$ over all 
quadruples $\sigma$, it depends on the fixed effects of each 
quadruple, so that sequences of fixed effects growing without 
bound are allowed as long as $p_{n,y}$ does not shrink faster 
than $n^{-1}$. The condition $np_{n,y} \to \infty$ thus 
requires that the expected number of informative quadruples 
continues to grow with $n$, even as the fraction of 
informative quadruples may vanish. The formal 
characterization of these two sources of sparsity and their 
implications for the estimable range of thresholds are 
developed in Subsection~\ref{subsection_sparsity}.

The following theorem establishes pointwise consistency at 
each threshold.

\begin{theorem} \label{theorem1} (Consistency) Let Assumptions \ref{assumption1}-\ref{assumption4} hold. Then ${\boldsymbol{\theta}}_{n,y} \overset{p}{\to} \boldsymbol{\theta}_{y,0}$ as $n \xrightarrow[]{} \infty$ for each fixed $y \in \mathcal{Y}$.
\end{theorem}
\noindent \textit{Proof.} See \ref{appendix_asymptotics}.

Conventional logit standard errors are not valid for the estimated ${\boldsymbol{\theta}}_{n,y}$. The score vector sums over quadruples of nodes, and quadruples sharing common nodes induce dependence across summands (such that each node participates in $O(n^3)$ quadruples), so the information matrix equality does not hold, and a sandwich-type variance estimator is needed. To derive the asymptotic distribution of the estimator, the moment requirements must be 
strengthened: 

\begin{assumption} \label{assumption5} (Moments) For all $(i,j) \in \mathcal{D}$, $\mathbb{E} (||\boldsymbol{x}_{ij}||^6) < C_2$, where $C_2$ is a finite constant. 
\end{assumption}
Each summand of the score vector takes the form
$$ \boldsymbol{s}_y(\sigma, \boldsymbol{\theta}_y) = \boldsymbol{r}_\sigma \{ 1 \{z_{\sigma,y} = 1 \} (1-\Lambda(\boldsymbol{r}_\sigma' \boldsymbol{\theta}_y)) - 1 \{z_{\sigma,y} = -1 \} \Lambda(\boldsymbol{r}_\sigma' \boldsymbol{\theta}_y)  \},$$
and the score vector is
$$ \boldsymbol{S}_{n,y} (\boldsymbol{\theta}_y) = \sum_{i}^n \sum_{j \neq i} \sum_{\substack{l \neq i,j}} \sum_{\substack{k \neq  i,j,l}} \boldsymbol{s}_y(\sigma\{i,l;j,k\}; \boldsymbol{\theta}_y).$$ 
The asymptotic distribution of the estimator is characterized by the result that $$\boldsymbol{\Upsilon}_{n,y} (\boldsymbol{\theta}_{y,0})^{-1/2} \boldsymbol{S}_{n,y}(\boldsymbol{\theta}_{y,0}) \overset{d}{\to} N(\boldsymbol{0},\boldsymbol{I}),$$
\noindent where $\boldsymbol{\Upsilon}_{n,y}(\boldsymbol{\theta}_y)$ 
is defined as follows
\begin{align}
    \boldsymbol{\Upsilon}_{n,y}({\boldsymbol{\theta}}_y) = \sum_{i} \sum_{j \neq i} \sum_{i' \neq i,j} \sum_{j' \neq i,j,i'} \sum_{i'' \neq i,j,i'} \sum_{j'' \neq i,j,j',i''} 16 \times \left[ \boldsymbol{s}_y(\sigma\{i,i';j,j'\}; {\boldsymbol{\theta}}_y) \boldsymbol{s}_y(\sigma\{i,i'';j,j''\}; {\boldsymbol{\theta}}_y)'\right]. 
\end{align}
Its expectation at 
$\boldsymbol{\theta}_{y,0}$ gives the leading term of the 
variance of the score vector, comprising the $O(n^6)$ terms 
corresponding to pairs of quadruples that share exactly one 
dyad.\footnote{The factor $16$ arises from fixing the shared 
dyad $(i,j)$ to be the first sender-receiver pair in both 
quadruples: there are four positions in each quadruple where 
$(i,j)$ can appear, giving $4 \times 4 = 16$ ordered pairs 
that share $(i,j)$ in the same position. 
\citet[Supplement, p.~28]{jochmans2018semiparametric} obtains 
the same result by multiplying each score contribution by 
$4$.}
Theorem \ref{theorem2} also requires that the dyad-clustered score outer-product matrix be 
nondegenerate at the scale of its leading term, ruling out first-order degeneracy of the score in 
any parameter direction.

\begin{assumption} \label{assumption_scorevar} (Score-covariance nondegeneracy) For each fixed $y\in\mathcal Y$, there exists a constant $c_y>0$ such that
\[
\Pr\!\left\{
\lambda_{\min}\!\left[(n^6p_{n,y})^{-1}
\Ups_{n,y}(\boldsymbol\theta_{y,0})\right]\ge c_y
\right\}\longrightarrow1.
\]
\end{assumption}

An analogous condition is required for the corresponding result in 
\citet{jochmans2018semiparametric}, and \citet{muris2025dyadic} impose an analogous eigenvalue 
condition on their dyad-clustered score outer-product matrix. Defining the Hessian
$$ \boldsymbol{H}_{n,y}(\boldsymbol{\theta}_y) = - \sum_{\sigma \in \mathcal{N}_{m_n}} \boldsymbol{r}_\sigma\boldsymbol{r}_\sigma' f(\boldsymbol{r}_\sigma'\boldsymbol{\theta}_y) 1\{z_{\sigma,y} \in \{-1,1\}\}$$
and the sandwich variance estimator
$$ \boldsymbol{\Omega}_{n,y}(\boldsymbol{\theta}_{n,y}) = \boldsymbol{H}_{n,y}({\boldsymbol{\theta}}_{n,y})^{-1} \boldsymbol{\Upsilon}_{n,y} ({\boldsymbol{\theta}}_{n,y})\boldsymbol{H}_{n,y}({\boldsymbol{\theta}}_{n,y})^{-1},$$
the following result holds, which establishes that pointwise (for each threshold $y$), the estimator converges to the true parameter value, and the sandwich estimator for the asymptotic variance delivers valid inference.

\begin{theorem} (Asymptotic distribution) Let Assumptions \ref{assumption1}-\ref{assumption_scorevar} hold. Then 
$$|| {\boldsymbol{\theta}}_{n,y} - \boldsymbol{\theta}_{y,0} || = O_p (1/ \sqrt{n(n-1)p_{n,y}})$$ and
$$ \boldsymbol{\Omega}_{n,y}(\boldsymbol{\theta}_{n,y})^{-1/2} ({\boldsymbol{\theta}}_{n,y} - \boldsymbol{\theta}_{y,0}) \overset{d}{\to} N({0},\boldsymbol{I}) $$ as $n \xrightarrow[]{} \infty$, for each fixed $y \in \mathcal{Y}$.
\label{theorem2}
\end{theorem}
\noindent \textit{Proof.} See \ref{appendix_derivation}.

The proof follows a four-step structure, with steps akin to those used when establishing the limit distribution of a U-statistic \citep{graham2017econometric}: (i) a projection of the score vector is proposed\footnote{The proposed projection resembles a H\'{a}jek projection, but it is not formally one, since the kernel of this projection is not symmetric, and we do not only condition on observable and unobservable attributes of a specific dyad $(i,j)$, but also on node attributes.} and its asymptotic equivalence to the score evaluated at the true parameter $\boldsymbol{\theta}_{y,0}$ is established; (ii) the limit distribution of the projection is derived via a conditional CLT; (iii) the uniform convergence of the Hessian is established; and (iv) the results are combined via a mean-value expansion. 

The main departure from \citet{jochmans2018semiparametric} is in step (i): while their approach explicitly computes the projection of the scores, I leverage the conditional independence structure, following \citet{graham2017econometric}, to show that the scores and projections are asymptotically equivalent. This approach extends naturally to the cross-threshold setting of Section \ref{joint_distribution}, avoiding the need to compute explicit cross-threshold joint probabilities. Moreover, I derive that the score variance is $O(n^6p_{n,y})$, which determines the convergence rate of the estimator and is used to establish the joint distribution across thresholds. 

\begin{remark}[Computation]
Evaluating the objective function, score, and Hessian requires 
summation over all informative quadruples, which can be 
computationally costly for large networks. Appendix~2.A of 
\citet{roelsgaard2025essays} describes implementation strategies for 
a related conditional likelihood estimator, more specifically, a two-way duration model. The proposed implementation includes storing the covariates in contiguous memory, pre-computing the covariate differences and storing only the relevant ones (only the ones referring to informative quadruples) in a vector, speeding up access, and reusing intermediate calculations within 
each Newton step, all of which apply directly to the present setting. Moreover, since 
estimation at each threshold is performed independently, the 
$K$ threshold-level optimizations can be run in parallel. 
\end{remark}

\subsection{Sparsity Conditions and Estimable Threshold Range}\label{subsection_sparsity}

The identification condition 
$np_{n,y} \to \infty$ accommodates sparsity from two distinct 
sources. This subsection characterizes these sources formally, 
translates the conditions of 
\citet{jochmans2018semiparametric} to the DR setting, and 
proposes a parameterization of fixed effects across thresholds 
that captures the variation in sparsity across the 
distribution.

\begin{figure}[htbp]
  \centering
  \captionsetup[subfigure]{font=footnotesize, skip=0pt}
  \captionsetup{skip=0pt}
  \begin{subfigure}[t]{0.49\textwidth}
    \centering
    \includegraphics[width=\textwidth]{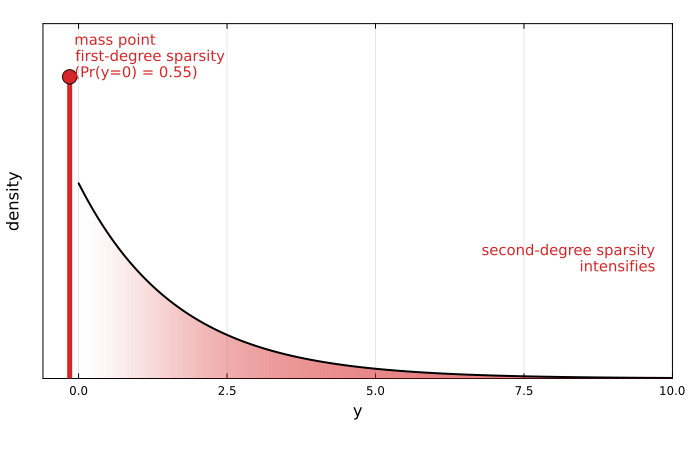}
    \caption{Density, mild first-degree sparsity}
    \label{fig:sparsity_density_mild}
  \end{subfigure}
  \hfill
  \begin{subfigure}[t]{0.49\textwidth}
    \centering
    \includegraphics[width=\textwidth]{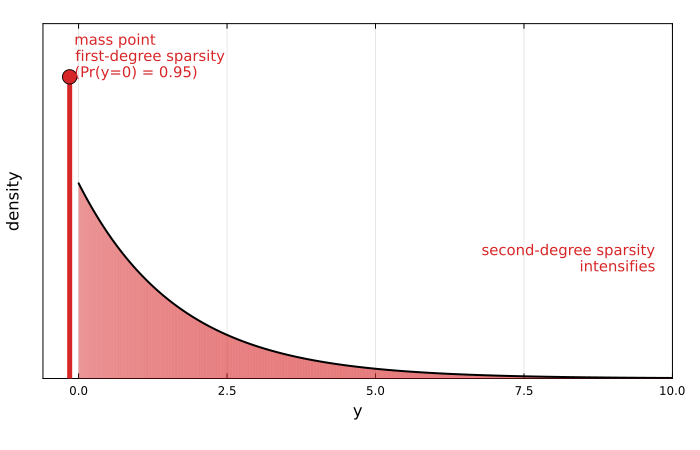}
    \caption{Density, strong first-degree sparsity}
    \label{fig:sparsity_density_strong}
  \end{subfigure}

  \vspace{0.5em}

  \begin{subfigure}[t]{0.49\textwidth}
    \centering
    \includegraphics[width=\textwidth]{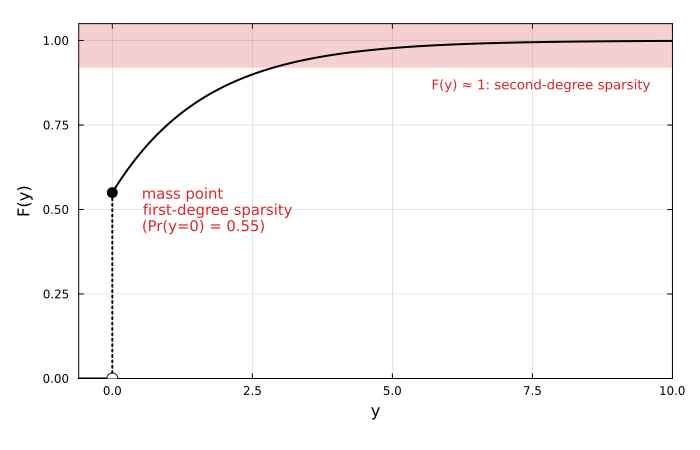}
    \caption{CDF, mild first-degree sparsity}
    \label{fig:sparsity_cdf_mild}
  \end{subfigure}
  \hfill
  \begin{subfigure}[t]{0.49\textwidth}
    \centering
    \includegraphics[width=\textwidth]{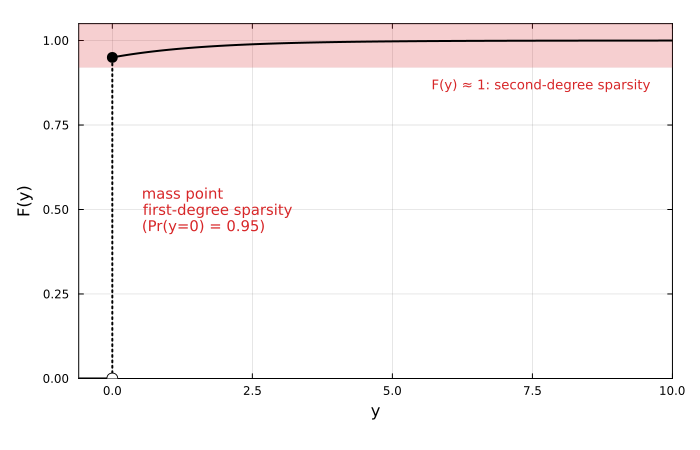}
    \caption{CDF, strong first-degree sparsity}
    \label{fig:sparsity_cdf_strong}
  \end{subfigure}

\caption{Two types of sparsity in distribution regression, for an outcome
bounded below at zero, under mild (left) and strong (right) concentration
at the bound. The mass point at zero is the source of first-degree
sparsity. Second-degree
sparsity, shown by the shading intensity in the density panels and the band
at $F(y) \approx 1$ in the CDF panels, arises where the binary indicators
$\tilde{y}_{ij,y}$ vary little across dyads. The mild and strong columns show
how a larger mass point places $F(y)$ near one over a wider range of thresholds.}
  \label{fig:two_types_sparsity}
\end{figure}

\paragraph{Two Types of Sparsity in Distribution Regression} The distribution
regression framework for network data involves two distinct but related
notions of sparsity, illustrated in Figure~\ref{fig:two_types_sparsity}.

\textit{First-degree sparsity} arises when the outcome variable is bounded
and a positive fraction of observations sit at the bound, creating a mass
point at the boundary of the support. Estimation begins
at the bound, the lowest threshold at which the binary indicators
$\tilde{y}_{ij,y} = \mathbf{1}\{y_{ij} \leq y\}$ vary across dyads. When the
mass point is large, as in firm-level trade networks, patent citation
networks, or venture capital flows, where 90--99\% of dyads take the
boundary value, this first threshold already places most indicators at one,
so the bound itself lies in a sparse regime and the estimable range is
compressed to a narrow interval. 

\textit{Second-degree sparsity} refers to
the proportion of zeros or ones in the binary indicators at each threshold,
which determines whether sufficient variation exists for identification. It
arises at thresholds in either tail of the conditional distribution, as
$\Pr(y_{ij} \leq y)$ approaches zero or one, including in outcomes with unbounded support, which exhibit no first-degree sparsity.\footnote{Analogously, sparsity
can be induced through few links or few non-links in a network formation
setting.} The two are related: when first-degree sparsity is extreme, it
places $\Pr(y_{ij} \leq y)$ near one across most thresholds, inducing
second-degree sparsity at nearly all thresholds. Whether sparsity arises from a mass point or from the 
tails of the conditional distribution, bias correction 
methods require $\Pr(y_{ij} \leq y)$ bounded away from 
zero and one, a condition violated under both regimes. 
The approach in this paper accommodates both through 
Assumption~\ref{assumption4}: $p_{n,y} \to 0$ is 
allowed provided $np_{n,y} \to \infty$.

\paragraph{Identification condition and the estimable quantile range} The
identification condition $np_{n,y} \to \infty$ can be translated into
restrictions on (i) the average probability of falling below threshold $y$,
$q_{n,y} = \sum_{i=1}^n \sum_{j \neq i} \Pr\{y_{ij} \leq y\} / n(n-1)$,
providing a direct characterization of the estimable threshold range; and
(ii) on the growth rate of the fixed effects. Following the argument in
\citet{jochmans2018semiparametric}, the exponential tails of the logistic
distribution imply $q_{n,y} \sim \sqrt{p_{n,y}}$ in the left tail and
$1 - q_{n,y} \sim \sqrt{p_{n,y}}$ in the right tail. Combined with the
condition $np_{n,y} \to \infty$, this yields rate restrictions on the
estimable thresholds:
$$
\sqrt{n}\, q_{n,y} \to \infty \text{ in the left tail,} 
\qquad \sqrt{n}\,(1-q_{n,y}) \to \infty \text{ in the 
right tail.}
$$
Accessing very extreme quantiles therefore requires
correspondingly larger samples, regardless of the distribution of node
heterogeneity. \ref{appendix_sparsity_rates} gives the full derivation.

\paragraph{Fixed effect parameterization across thresholds} To understand
the growth rates of fixed effects that satisfy Assumption~\ref{assumption4}
and their connection to sparsity across thresholds, I adopt the
parameterization
$$
\alpha_{i,y} = a_{n,i} \cdot g_y(n), \qquad
\gamma_{j,y} = b_{n,j} \cdot g_y(n),
$$
where $a_{n,i}, b_{n,j} \in [0,1]$ are bounded sequences in $n$ capturing
node-level heterogeneity, and $g_y(n)$ is a real-valued sequence that scales
all fixed effects proportionately, governing how sparsity varies with the
threshold $y$. Specifically, $g_y(n) < 0$ at low thresholds (sparsity through many zeros),
$g_y(n) \approx 0$ at intermediate thresholds (dense network), and
$g_y(n) > 0$ at high thresholds (sparsity through many ones). The larger
$|g_y(n)|$, the more the binary indicators concentrate near zero or one.
This parameterization describes the fixed effects at each threshold, not a
data generating process for the outcome $y_{ij}$ itself. When the outcome
is bounded below, for instance, at zero, the parameterization is defined only
for thresholds from the bound onward, where sparsity varies across
thresholds as governed by $g_y(n)$.

The sequences $a_{n,i}$ and $b_{n,j}$ capture each node's sensitivity to the changes in sparsity generated by $g_y(n)$: nodes with high values
have fixed effects that change substantially across thresholds, causing
sharp transitions from mostly zeros to mostly ones; nodes with low values
have more stable fixed effects, maintaining more stable probabilities
across thresholds. This parameterization extends the framework of
\citet{jochmans2018semiparametric} by allowing $g_y(n)$ to vary across
thresholds, capturing how sparsity changes across the distribution. This
cross-threshold structure underlies the joint asymptotic analysis of
Section~\ref{joint_distribution} and the specific parameterizations
explored in the Supplemental \ref{sim_single_threshold}.

The choice of $a_{n,i}$ and $b_{n,j}$ also affects how far into the
sparse region estimation remains feasible. The largest admissible
$|g_y(n)|$ depends on the heterogeneity pattern, and so, through the
induced link probability, does the set of estimable thresholds. Holding
the mean of $a_{n,i}$ and $b_{n,j}$ fixed, more dispersed sequences
admit larger $|g_y(n)|$. The
identification condition admits, however, a single characterization in
terms of the link probability that holds whatever the heterogeneity
pattern: a threshold is estimable when $\sqrt{n}\, q_{n,y} \to \infty$ 
in the left tail and $\sqrt{n}\,(1 - q_{n,y}) \to \infty$ 
in the right tail. The
heterogeneity pattern enters this criterion only through the value of
$q_{n,y}$ it induces at each threshold; the criterion itself is
unchanged. \ref{appendix_sparsity_heterogeneity} develops this
distinction.

\section{Joint Inference Across Thresholds}\label{joint_distribution}

This section extends the pointwise results of 
Section~\ref{asymptotics} to the joint distribution of 
the estimators across multiple thresholds, enabling 
simultaneous confidence bands and formal tests of 
coefficient equality. Subsection~\ref{joint} establishes 
the joint asymptotic distribution. 
Subsections~\ref{bands} and \ref{testing} develop the 
inference tools.

\subsection{Joint Asymptotic Distribution Across Thresholds}\label{joint}

Since the estimators at different thresholds are computed from the same underlying outcome variable $y_{ij}$, they are not independent: informative quadruples at different thresholds are correlated since they potentially share the same dyadic outcomes, and the sequences of fixed effects and idiosyncratic errors are allowed to be correlated across thresholds. Moreover, the convergence rates differ across thresholds through the threshold-specific parameters $p_{n,y_k}$, reflecting the expected fraction of informative quadruples at each threshold and the different degrees of sparsity at different thresholds. The joint distribution established below accounts for both the cross-threshold dependence and the different convergence rates. 

Let $\mathbf{y} = (y_1, ..., y_K)'$ denote the vector of $K$ ordered thresholds with $y_1 < y_2 < ... < y_K$, and each $y_k \in \mathcal{Y}$. For each threshold $y_k$, the parameter vector $\boldsymbol{\theta}_{y_k} \in \mathbb{R}^p$ represents the $p$-dimensional coefficient vector at that threshold. Define the stacked parameter vector $\boldsymbol{\theta}_{\mathbf{y}} = (\boldsymbol{\theta}_{y_1}', ..., \boldsymbol{\theta}_{y_K}')' \in \mathbb{R}^{Kp}$. Stacking the threshold-specific scores gives
$$\boldsymbol{S}_{n,\mathbf{y}}(\boldsymbol{\theta}_{\mathbf{y},0}) = 
\big(\boldsymbol{S}_{n,1}(\boldsymbol{\theta}_{y_{1,0}})', ..., 
\boldsymbol{S}_{n,K}(\boldsymbol{\theta}_{y_{K,0}})'\big)'.$$
Define the $(Kp) \times (Kp)$ dyad-clustered outer-product matrix 
$\boldsymbol{\Upsilon}_{n, \mathbf{y}}(\boldsymbol{\theta}_{\mathbf{y}})$ with blocks 
$\left[\boldsymbol{\Upsilon}_{n, \mathbf{y}}\right]_{kl} = 
\boldsymbol{\Upsilon}_{n, kl}\left(\boldsymbol{\theta}_{y_k}, 
\boldsymbol{\theta}_{y_l}\right)$, where
\begin{adjustwidth}{-\oddsidemargin}{-\oddsidemargin}
\centering
$$\boldsymbol{\Upsilon}_{n, k l}
(\boldsymbol{\theta}_{y_k}, \boldsymbol{\theta}_{y_l})=\sum_{i=1}^n \sum_{j \neq i} \sum_{i^{\prime} \neq i, j} \sum_{j^{\prime} \neq i, j, i^{\prime}} \sum_{i^{\prime \prime} \neq i, j, i^{\prime}} \sum_{j^{\prime \prime} \neq i, j, j^{\prime}, i^{\prime \prime}} 16 \times \left[\boldsymbol{s}_k\left(\sigma\left\{i, i^{\prime} ; j, j^{\prime}\right\}; \boldsymbol{\theta}_{y_{k}}\right) \boldsymbol{s}_l\left(\sigma\left\{i, i^{\prime \prime} ; j, j^{\prime \prime}\right\}; \boldsymbol{\theta}_{y_l}\right)^{\prime}\right]$$
\end{adjustwidth}
whose expectation at $\boldsymbol{\theta}_{\mathbf{y},0}$ gives the leading term of the score 
covariance for thresholds $y_k$ and $y_l$. The diagonal blocks coincide with the pointwise matrices $\boldsymbol{\Upsilon}_{n, y_k}(\boldsymbol{\theta}_{y_k})$, and the off-diagonal blocks capture the cross-threshold dependence.

The joint limit distribution is driven by that of the stacked score vector, and the first two 
steps of Theorem \ref{theorem2} carry over. In step (i), the stacked projections are shown to be 
asymptotically equivalent to the stacked scores by comparing their leading covariance terms, 
which is where the conditional independence structure replaces the explicit cross-threshold joint 
probabilities of the binary indicators, as anticipated in Section \ref{asymptotics}. In step 
(ii), conditioning on the joint information set containing the covariates and all 
threshold-specific fixed effects preserves cross-dyad independence, and a conditional central 
limit theorem applied to the projections delivers joint asymptotic normality of 
$\boldsymbol{S}_{n,\mathbf{y}}(\boldsymbol{\theta}_{n,\mathbf{y},0})$ after normalization by its 
own covariance.

Normalizing by the covariance itself absorbs the threshold-specific rates automatically, but it 
requires that the covariance stays invertible as $n$ grows, and this cannot be stated without a reference scale: $\boldsymbol{\Upsilon}_{n,\mathbf{y}}$ 
grows without bound, at rates that differ across blocks. Define the score-rate matrix
$$\boldsymbol{D}_{n,\mathbf{y}} := \operatorname{diag}\left((n^6p_{n,y_1})^{1/2}\boldsymbol{I}_p, 
\ldots, (n^6p_{n,y_K})^{1/2}\boldsymbol{I}_p\right),$$
so that the $(k,l)$ block of 
$\boldsymbol{D}_{n,\mathbf{y}}^{-1}\boldsymbol{\Upsilon}_{n,\mathbf{y}}
\boldsymbol{D}_{n,\mathbf{y}}^{-1}$ is 
$\boldsymbol{\Upsilon}_{n,kl}/(n^6\sqrt{p_{n,y_k}p_{n,y_l}})$, which is $O_p(1)$ for every 
pair of thresholds. 

A scalar normalization must be matched to one threshold or another: chosen for the sparsest, the 
densest block diverges; chosen for the densest, the sparsest block vanishes and the bound below 
fails for reasons of rate alone. Dividing each block by its own scale instead leaves only the 
shape of the covariance, so that the condition below is a restriction on how the score behaves 
across directions and not on how the $p_{n,y_k}$ compare across thresholds. The same blockwise scaling is used throughout the proof, so that no condition is required on the 
relative convergence rates across thresholds.

\begin{assumption} (Joint score-covariance nondegeneracy) \label{assumption_jointscorevar}
For every fixed finite collection $\mathbf{y} = (y_1, \ldots, y_K)'$ of distinct thresholds, 
there exists a constant $c_{\mathbf{y}} > 0$ such that
$$\Pr\left\{\lambda_{\min}\left[\boldsymbol{D}_{n,\mathbf{y}}^{-1}
\boldsymbol{\Upsilon}_{n,\mathbf{y}}(\boldsymbol{\theta}_{\mathbf{y},0})
\boldsymbol{D}_{n,\mathbf{y}}^{-1}\right] \geq c_{\mathbf{y}}\right\} \longrightarrow 1.$$
\end{assumption}

Assumption \ref{assumption_jointscorevar} requires that no nonzero linear combination of the 
threshold-specific scores become asymptotically degenerate after this normalization. For $K = 1$ it reduces to Assumption \ref{assumption_scorevar}. Assumption \ref{assumption_jointscorevar} fails, for instance, when two thresholds are close 
enough that the corresponding binary indicators nearly coincide, which suggests that the 
intervals between selected thresholds should contain a non-negligible fraction of observations.

The $K p \times K p$ joint Hessian inverse 
$\boldsymbol{H}_{n, \mathbf{y}}^{-1}\left(\boldsymbol{\theta}_{\mathbf{y}}\right) = 
\text{diag} \left(\boldsymbol{H}_{n, 1}\left(\boldsymbol{\theta}_{y_1}\right)^{-1}, \ldots, 
\boldsymbol{H}_{n, K}\left(\boldsymbol{\theta}_{y_K}\right)^{-1}\right)$ is block-diagonal since each threshold is estimated separately. The sandwich variance estimator is
$$\boldsymbol{\Omega}_{n, \mathbf{y}}\left(\boldsymbol{\theta}_{n,\mathbf{y}}\right)=\boldsymbol{H}_{n, \mathbf{y}}^{-1}\left(\boldsymbol{\theta}_{n,\mathbf{y}}\right) \boldsymbol{\Upsilon}_{n, \mathbf{y}}\left(\boldsymbol{\theta}_{n,\mathbf{y}}\right) \boldsymbol{H}_{n, \mathbf{y}}^{-1}\left(\boldsymbol{\theta}_{n,\mathbf{y}}\right),$$
where each block is given by $
\left[\boldsymbol{\Omega}_{n, \mathbf{y}}\left(\boldsymbol{\theta}_{n,\mathbf{y}}\right)\right]_{k l}=\boldsymbol{H}_{n, k}\left(\boldsymbol{\theta}_{n,y_k}\right)^{-1} \boldsymbol{\Upsilon}_{n, k l}\left(\boldsymbol{\theta}_{n,y_k}, \boldsymbol{\theta}_{n,y_l}\right) \boldsymbol{H}_{n, l}\left(\boldsymbol{\theta}_{n,y_l}\right)^{-1}
$. 

For any $\boldsymbol\theta$, define
\[
\boldsymbol\sigma_{n,\mathbf y}(\boldsymbol\theta)
:=\op{diag}\!\left(
\sqrt{[\boldsymbol\Omega_{n,\mathbf y}(\boldsymbol\theta)]_{11}},
\ldots,
\sqrt{[\boldsymbol\Omega_{n,\mathbf y}(\boldsymbol\theta)]_{Kp,Kp}}
\right),
\]
\[
\boldsymbol P_{n,\mathbf y}(\boldsymbol\theta)
:=\boldsymbol\sigma_{n,\mathbf y}(\boldsymbol\theta)^{-1}
\boldsymbol\Omega_{n,\mathbf y}(\boldsymbol\theta)
\boldsymbol\sigma_{n,\mathbf y}(\boldsymbol\theta)^{-1},
\]
and
\[
\boldsymbol T_{n,\mathbf y}
:=\boldsymbol\sigma_{n,\mathbf y}
(\boldsymbol\theta_{n,\mathbf y})^{-1}
(\boldsymbol\theta_{n,\mathbf y}-\boldsymbol\theta_{\mathbf y,0}).
\]
Thus $\boldsymbol P_{n,\mathbf y}(\boldsymbol\theta_{n,\mathbf y})$ is the feasible plug-in version of the sample correlation matrix. Under Assumption \ref{assumption_jointscorevar}, the diagonal entries of $\boldsymbol{\sigma}_{n,\mathbf{y}}(\boldsymbol{\theta}_{n,\mathbf{y}})$ are positive and $\boldsymbol{P}_{n,\mathbf{y}}(\boldsymbol{\theta}_{n,\mathbf{y}})$ is nondegenerate with probability approaching one (see \ref{joint_distribution_appendix}).

For probability measures $\mu, \nu$ on $\mathbb{R}^{Kp}$, let
$$d_{\mathrm{BL}}(\mu,\nu) := \sup_{f}\left|\int f \, d\mu - \int f \, d\nu\right|,$$
where the supremum is over all $f: \mathbb{R}^{Kp} \to \mathbb{R}$ with $\|f\|_\infty \leq 1$ and 
$|f(\boldsymbol{u}) - f(\boldsymbol{v})| \leq \|\boldsymbol{u} - \boldsymbol{v}\|$ for all 
$\boldsymbol{u}, \boldsymbol{v}$. This is the bounded--Lipschitz distance, which metrizes weak 
convergence and remains meaningful when the approximating law varies with $n$. Write 
$\mathcal{F}_{n,\mathbf{y}}$ for the covariates together with the fixed effects at all $K$ 
thresholds, and $\mathcal{L}(\boldsymbol{X} \mid \mathcal{F}_{n,\mathbf{y}})$ for the conditional 
law of $\boldsymbol{X}$ given these.

\begin{theorem} (Joint asymptotic distribution) \label{theorem3} Let Assumptions 
\ref{assumption1}--\ref{assumption_scorevar} and \ref{assumption_jointscorevar} hold, and fix a 
finite collection $\mathbf y=(y_1,\ldots,y_K)'$ of distinct thresholds in $\mathcal Y$. Then, as 
$n\to\infty$:
\begin{enumerate}
\item[(i)] The coordinatewise studentized vector $\boldsymbol T_{n,\mathbf y}$ satisfies
\[
d_{\mathrm{BL}}\!\left(
\mathcal{L}\!\left(\boldsymbol T_{n,\mathbf y} \mid \mathcal{F}_{n,\mathbf{y}}\right),\;
N\!\left(\boldsymbol 0, \boldsymbol P_{n,\mathbf y}(\boldsymbol\theta_{n,\mathbf y})\right)
\right) \overset p\longrightarrow 0.
\]
\item[(ii)] For every fixed nonzero $\boldsymbol a\in\R^{Kp}$,
\[
\frac{\boldsymbol a'
(\boldsymbol\theta_{n,\mathbf y}-\boldsymbol\theta_{\mathbf y,0})}
{\sqrt{\boldsymbol a'
\boldsymbol\Omega_{n,\mathbf y}(\boldsymbol\theta_{n,\mathbf y})
\boldsymbol a}}
\overset d\longrightarrow N(0,1).
\]
\item[(iii)] For every fixed $q\times Kp$ matrix $\boldsymbol R$ of full row rank,
\[
\begin{aligned}
&\big[\boldsymbol R(\boldsymbol\theta_{n,\mathbf y}-\boldsymbol\theta_{\mathbf y,0})\big]'
\big[\boldsymbol R\boldsymbol\Omega_{n,\mathbf y}
(\boldsymbol\theta_{n,\mathbf y})\boldsymbol R'\big]^{-1}\\
&\hspace{3.5cm}\times
\big[\boldsymbol R(\boldsymbol\theta_{n,\mathbf y}-\boldsymbol\theta_{\mathbf y,0})\big]
\overset d\longrightarrow\chi_q^2.
\end{aligned}
\]
\end{enumerate}
\end{theorem}
\noindent \textit{Proof.} See \ref{joint_distribution_appendix}.

The block structure of the sandwich covariance is what accommodates the different rates: the 
Hessian blocks satisfy $\boldsymbol{H}_{n,k}(\boldsymbol{\theta}_{n,y_k}) = O_p(n^4 p_{n,y_k})$ 
and the score covariance blocks 
$\boldsymbol{\Upsilon}_{n,kl}(\boldsymbol{\theta}_{n,y_k}, \boldsymbol{\theta}_{n,y_l}) = 
O_p(n^6 \sqrt{p_{n,y_k} p_{n,y_l}})$, so that
$$\left[\boldsymbol{\Omega}_{n,\mathbf{y}}(\boldsymbol{\theta}_{n,\mathbf{y}})\right]_{kl} = 
O_p\left(\frac{1}{n^2\sqrt{p_{n,y_k}\, p_{n,y_l}}}\right).$$
The diagonal blocks recover the threshold-specific rates $(n(n-1)p_{n,y_k})^{-1/2}$ of Theorem 
\ref{theorem2}, and the off-diagonal blocks scale with the geometric mean of the two 
informativeness parameters, so that the implied cross-threshold correlations are unaffected by 
the relative magnitudes of $p_{n,y_k}$ and $p_{n,y_l}$. Dividing each coefficient by its own 
standard error, as in $\boldsymbol{T}_{n,\mathbf{y}}$, normalizes these rates away. What it does 
not remove is the dependence across thresholds, which stays in the correlation matrix 
$\boldsymbol{P}_{n,\mathbf{y}}(\boldsymbol{\theta}_{n,\mathbf{y}})$, and no limit is imposed on 
that matrix. Doing so would be a condition on the joint informativeness structure across 
thresholds, which determines the off-diagonal blocks of 
$\boldsymbol{P}_{n,\mathbf{y}}(\boldsymbol{\theta}_{n,\mathbf{y}})$ and concerns which quadruples 
are informative at two thresholds at once: a restriction on the joint design rather than on each 
threshold separately. 

Part (i) of Theorem \ref{theorem3} is therefore stated as a Gaussian approximation, with a 
correlation matrix that varies with $n$, rather than as convergence to a fixed law. Normalizing 
the whole vector instead would give a fixed limit without any such restriction, but it would 
replace the individual coefficients by linear combinations of them across thresholds, which is 
not the form the bands of Section \ref{bands} require; parts (ii) and (iii) take that route, and 
their limits are correspondingly fixed.

\subsection{Sup-$t$ Joint Confidence Bands}\label{bands}

Simultaneous confidence bands across thresholds can be constructed via the sup-$t$ method \citep{montiel2019simultaneous}. Since typically the question of interest is on inference on each separate covariate (whether the effect of covariate $d$ varies across the distribution of the outcome), the bands are constructed per covariate.\footnote{Joint bands across all covariates and thresholds can be constructed analogously by replacing $\boldsymbol{\Omega}_{n,d}$ with $\boldsymbol{\Omega}_{n,\mathbf{y}}(\boldsymbol{\theta}_{n,\mathbf{y}})$ in Algorithm \ref{alg:supt_bands}; the sup-$t$ critical value adjusts automatically to the dimension of the chosen index set.} Let $\boldsymbol{\Omega}_{n,d}(\boldsymbol{\theta}_{n,\mathbf{y}})$ denote the $K \times K$ 
submatrix of $\boldsymbol{\Omega}_{n,\mathbf{y}}(\boldsymbol{\theta}_{n,\mathbf{y}})$ 
corresponding to the $d$-th coefficient across all $K$ thresholds, and let 
$\sigma_{n,y_k,d} = \sqrt{[\boldsymbol{\Omega}_{n,d}]_{kk}}$, so that the 
$\sigma_{n,y_k,d}$ are the corresponding diagonal entries of 
$\boldsymbol{\sigma}_{n,\mathbf{y}}(\boldsymbol{\theta}_{n,\mathbf{y}})$ and the correlation 
matrix implied by $\boldsymbol{\Omega}_{n,d}$ is the corresponding submatrix of 
$\boldsymbol{P}_{n,\mathbf{y}}(\boldsymbol{\theta}_{n,\mathbf{y}})$, denoted by $\boldsymbol{P}_{n,d}(\boldsymbol{\theta}_{n,\mathbf{y}})$.

The sup-$t$ confidence band is the set
$${C}_{n,d} = \left\{ \boldsymbol{\theta}_d \in \mathbb{R}^K : 
\max_{k=1,\ldots,K} 
\frac{|{\theta}_{n,y_k,d} - \theta_{y_k,d}|}
{{\sigma}_{n,y_k,d}} \leq {c}_{n,d,1-\alpha} \right\}$$
where ${c}_{n,d,1-\alpha}$ is the empirical $(1-\alpha)$ quantile 
of $\max_{k=1,\ldots,K} |{Z}_{n,y_k,d}| / {\sigma}_{n,y_k,d}$ 
with ${\boldsymbol{Z}}_{n,d} \sim N(0, {\boldsymbol{\Omega}}_{n,d})$. This is equivalent to the Cartesian product of intervals $C_{n,y_k,d} = [\theta_{n,y_k,d} \pm 
c_{n,d,1-\alpha} \sigma_{n,y_k,d}]$. Unlike pointwise confidence intervals, which cover each $\theta_{y_k,d}$ individually at level 
$(1-\alpha)$, the joint confidence bands cover the entire parameter vector $\boldsymbol{\theta}_d$ simultaneously at level $(1-\alpha)$. This allows for valid confidence statements that compare across thresholds, accounting for the cross-threshold dependence.

The critical value ${c}_{n,d,1-\alpha}$ 
generally exceeds the pointwise critical value $z_{1-\alpha/2}$. 
However, as shown in \citet{montiel2019simultaneous}, the sup-$t$ band achieves the smallest critical value among the one-parameter class of simultaneous confidence bands with coverage at least $(1 - \alpha)$, which includes Bonferroni and Wald projection bands as special cases. The gains are particularly sizeable when estimates are highly correlated, as is the case for nearby thresholds in the DR setting, and the sup-$t$ critical value increases slowly with the number of thresholds $K$. Moreover, unlike confidence ellipsoids, the rectangular structure of the bands is easily visualized and permits visual hypothesis testing, regardless of the dimension of the parameter vector.  

The critical value ${c}_{n,d,1-\alpha}$ is obtained by simulation, as described in Algorithm \ref{alg:supt_bands}.
\begin{algorithm}[h]
   \setlength{\interspacetitleruled}{0pt}
  \setlength{\interspacealgoruled}{2pt}
  \setlength{\algomargin}{1em}
  \footnotesize
  \linespread{1.2}\selectfont
\caption{sup-$t$ Confidence Bands for Covariate $d$}
\label{alg:supt_bands}
\KwIn{Estimates $\theta_{n,y_k,d}$ for $k=1,\ldots,K$; 
  covariance matrix $\boldsymbol{\Omega}_{n,d}(\boldsymbol{\theta}_{n,\mathbf{y}})$; 
  significance level $\alpha$; 
  number of draws $B$}
\KwOut{Simultaneous confidence band}
\For{$b = 1, \ldots, B$}{
  Draw ${\boldsymbol{Z}}^{(b)}_{n,d} \sim N(0, {\boldsymbol{\Omega}}_{n,d})$\;
  Compute $T^{(b)}_{n,d} = \max_{k=1,\ldots,K} 
    |{Z}^{(b)}_{n,y_k,d}| / {\sigma}_{n,y_k,d}$\;
}
Set ${c}_{n,d,1-\alpha} = (1-\alpha)$ empirical quantile of 
  $\{T^{(1)}_{n,d}, \ldots, T^{(B)}_{n,d}\}$\;
\For{$k = 1, \ldots, K$}{
  ${C}_{n,y_k,d} = [\theta_{n,y_k,d} - 
    {c}_{n,d,1-\alpha} {\sigma}_{n,y_k,d}, \; 
    \theta_{n,y_k,d} + 
    {c}_{n,d,1-\alpha} {\sigma}_{n,y_k,d}]$\;
}
\end{algorithm}
The sup-$t$ band requires (i) a joint Gaussian approximation for the coordinatewise studentized 
estimates, with (ii) the correlation structure of the approximating law consistently estimated in the 
relative sense that its difference from the correlation matrix conditional on the covariates and fixed effects at all thresholds vanishes, and (iii)
strictly positive marginal variances. Theorem \ref{theorem3}(i) delivers the first two directly, 
since the approximating law is built from 
$\boldsymbol{P}_{n,\mathbf{y}}(\boldsymbol{\theta}_{n,\mathbf{y}})$, which the proof shows to 
satisfy this, and the third holds because the feasible covariance matrix is positive definite with probability approaching one under 
Assumption \ref{assumption_jointscorevar}; see \ref{joint_distribution_appendix}. Algorithm 
\ref{alg:supt_bands} is a direct application of Algorithm 1 in 
\citet{montiel2019simultaneous}. The draws enter only through the standardized maxima 
$\max_k |{Z}^{(b)}_{n,y_k,d}|/{\sigma}_{n,y_k,d}$, and since the ${\sigma}_{n,y_k,d}$ are the 
square roots of the diagonal entries of $\boldsymbol{\Omega}_{n,d}$, drawing from 
$\boldsymbol{\Omega}_{n,d}$ and dividing by ${\sigma}_{n,y_k,d}$ is equivalent to drawing from the 
implied correlation matrix $\boldsymbol{P}_{n,d}(\boldsymbol{\theta}_{n,\mathbf{y}})$ directly; the 
former is retained to match Algorithm 1 of \citet{montiel2019simultaneous}. 

Because the draws are standardized, each coordinate has unit variance and the different 
convergence rates across thresholds are absorbed without restricting how they compare. The standardization also makes the anti-concentration bound for the maximum of Gaussian 
coordinates \citep{chernozhukov2015comparison} hold with a constant that does not depend on the 
correlation matrix, which is what converts the bounded--Lipschitz approximation of Theorem 
\ref{theorem3}(i) into a statement about the probability of the rectangle defining ${C}_{n,d}$. 
This step is not needed in \citet{montiel2019simultaneous}, where the correlation matrix has a 
fixed limit; here it is what allows the correlation matrix to vary with $n$. The resulting coverage statement holds 
conditionally and hence unconditionally, since coverage probabilities are bounded.

\subsection{Joint Testing}\label{testing}

While simultaneous confidence bands provide visual evidence on whether coefficients vary across thresholds, formal tests of the null hypothesis of coefficient equality are of interest. For a given covariate $d$, the null hypothesis is 
$$H_{0,d}: \theta_{y_1,d} = \theta_{y_2,d} = \cdots 
= \theta_{y_K,d},$$
which can be expressed as $\boldsymbol{R} {\boldsymbol{\theta}}_d = 
{0}$, where $\boldsymbol{R}$ is any $(K-1) \times K$ 
matrix of full row rank satisfying $\boldsymbol{R} \mathbf{1}_K = 
\boldsymbol{0}$. The choice of $\boldsymbol{R}$ does not affect the 
null hypothesis, but determines the parameterization 
of the deviations from equality. Common choices include adjacent differences 
($\boldsymbol{R}{\boldsymbol{\theta}}_d = (\theta_{y_2,d} - 
\theta_{y_1,d}, \ldots, \theta_{y_K,d} - 
\theta_{y_{K-1},d})'$) and differences from a 
reference threshold ($\boldsymbol{R}{\boldsymbol{\theta}}_d = 
(\theta_{y_2,d} - \theta_{y_1,d}, \ldots, 
\theta_{y_K,d} - \theta_{y_1,d})'$). I focus on two possible tests: the standard Wald test and a sup-$t$ test obtained by band inversion. The Wald test is 
invariant to the choice of $\boldsymbol{R}$, while 
the sup-$t$ test depends on the specific contrasts 
used, which may affect power but not validity.

The sup-$t$ equality test is obtained by inverting a sup-$t$ confidence band for the contrast 
vector $\boldsymbol{R}{\boldsymbol{\theta}}_d$ \citep{montiel2019simultaneous}. Let 
$\boldsymbol{R}$ select the $d$th coefficient across thresholds and form the contrasts. The 
argument of Theorem \ref{theorem3}(iii), applied to this $\boldsymbol{R}$, gives a joint Gaussian 
approximation for $\boldsymbol{R}({\boldsymbol{\theta}}_{n,d} - \boldsymbol{\theta}_{d,0})$ with 
covariance $\boldsymbol{R}\boldsymbol{\Omega}_{n,d}(\boldsymbol{\theta}_{n,\mathbf{y}})
\boldsymbol{R}'$. Studentizing each contrast by its own standard error (the square root of the 
corresponding diagonal entry of that matrix), then puts the vector on the scale used by the 
sup-$t$ construction, with correlation structure implied by the same matrix. Under $H_{0,d}$ the 
centering vanishes, since $\boldsymbol{R}\boldsymbol{\theta}_{d,0} = \boldsymbol{0}$, so the same 
approximation applies to $\boldsymbol{R}{\boldsymbol{\theta}}_{n,d}$ itself, and the test 
statistic is
$$T_{n,d}^{\text{eq}} = \max_{k=1,\ldots,K-1} 
\frac{|(\boldsymbol{R}{\boldsymbol{\theta}}_{n,d})_k|}
{{\sigma}_{n,\delta,k}}$$
where ${\sigma}_{n,\delta,k} = 
\sqrt{[\boldsymbol{R}{\boldsymbol{\Omega}}_{n,d}
\boldsymbol{R}']_{kk}}$, and $H_{0,d}$ is rejected 
if $T_{n,d}^{\text{eq}} > {c}^{\text{eq}}_{n,d,1-\alpha}$. The critical value ${c}^{\text{eq}}_{n,d,1-\alpha}$ 
is computed as in Algorithm \ref{alg:supt_bands}, 
replacing ${\boldsymbol{\Omega}}_{n,d}$ with 
$\boldsymbol{R}{\boldsymbol{\Omega}}_{n,d}
\boldsymbol{R}'$. Since the contrast vector is studentized by the diagonal of 
$\boldsymbol{R}\boldsymbol{\Omega}_{n,d}(\boldsymbol{\theta}_{n,\mathbf{y}})\boldsymbol{R}'$, its 
approximating law again has unit diagonal, and the argument of Section \ref{bands} applies 
unchanged.

As a complement, the standard Wald test for $H_{0,d}$ is valid by Theorem \ref{theorem3}(iii), 
with $\boldsymbol{R}$ there replaced by the $(K-1) \times Kp$ matrix mapping 
$\boldsymbol{\theta}_{\mathbf{y}}$ to the $K-1$ contrasts of the $d$th coefficient across 
thresholds. The test does not depend on which contrast matrix is used, since any two full-rank matrices with 
the same null space differ by a nonsingular transformation, which cancels in the quadratic form. The Wald test rejects when $(\boldsymbol{R}{\boldsymbol{\theta}}_{n,d})' 
(\boldsymbol{R}{\boldsymbol{\Omega}}_{n,d}\boldsymbol{R}')^{-1}(\boldsymbol{R}
{\boldsymbol{\theta}}_{n,d})$, which under $H_{0,d}$ coincides with the centered quadratic form of 
Theorem \ref{theorem3}(iii), exceeds the $(1-\alpha)$-quantile of the $\chi^2_{(K-1)}$ 
distribution. The two tests have complementary power properties: the sup-$t$ test has higher power against sparse alternatives (when a small number of thresholds have different coefficients), while the Wald test has higher power against diffuse alternatives (where all coefficients differ slightly across thresholds). In 
particular, neither test uniformly dominates the 
other in terms of local power \citep{montiel2019simultaneous}. However, the sup-$t$ framework allows immediate 
identification of which thresholds drive the 
rejection, either through the simultaneous bands for 
$\boldsymbol{\theta}_d$ or through bands for the 
contrasts $\boldsymbol{R}\boldsymbol{\theta}_d$, 
whereas the Wald test provides only a joint rejection 
decision.\footnote{The joint distribution established in 
Theorem \ref{theorem3} also permits testing more 
refined hypotheses, such as the equality of adjacent 
coefficients ($H_k: \theta_{y_k,d} = 
\theta_{y_{k+1},d}$) or monotonicity restrictions 
($\theta_{y_1,d} \leq \cdots \leq 
\theta_{y_K,d}$). I focus on the global test 
and joint confidence bands, which suffice to 
characterize heterogeneity across the distribution.}

\begin{remark}[Extension to a continuum of thresholds]
\label{rem:continuum}
The asymptotic results in this paper are established for 
finitely many thresholds. The main payoff of extending them 
to a continuum $y \in \mathcal{Y}$ would be to obtain uniform confidence 
bands for counterfactual distributions, and by inversion, for quantile 
functions and effects \citep{chernozhukov2013inference, 
chernozhukov2020network}. In the sparse network setting, where consistent estimation of the fixed effects is not possible, the construction of these objects is unavailable. Additionally, average effects that aggregate over 
the fixed effects are at best set-identified under sparsity. What remains possible is inference on the structural parameter process $y \mapsto \boldsymbol{\theta}_0(y)$ as a functional object, but the sup-$t$ test developed in 
Section~\ref{joint_distribution} already provides 
simultaneous inference across multiple thresholds, which suffices for the 
question of whether covariate effects vary across the 
distribution. Moreover, existing uniform results for dyadic and network data do not directly apply to the setting in this paper.\footnote{The empirical process theory for exchangeable arrays of \citet{davezies2021empirical} does not accommodate triangular array structures with the varying rates of convergence across thresholds. The functional central limit theorem in \citet{chernozhukov2020network} is developed for bias-corrected fixed effects estimators in dense networks where the rate of convergence is uniform across thresholds, and does not apply to the conditional maximum likelihood estimator, whose asymptotic structure (involving  projections over conditionally independent quadruples with threshold-dependent rates) differs fundamentally from theirs. Among existing approaches, the 
strong approximation methods of 
\citet{cattaneo2024uniform} for dyadic kernel density 
estimation share key structural features with the present 
setting: the dyad-level contributions are conditionally 
independent given node-specific attributes (and in the case of this paper, dyad-specific attributes as well), and the rate of convergence is not uniform across evaluation points (due to degeneracy of the H\'{a}jek projection in their setting, and due to varying network sparsity across thresholds in the present one). Adapting 
their approach to the present setting requires establishing 
that the linearization of the CMLE holds uniformly over 
the continuum of thresholds (in particular, uniform 
convergence of the Hessian over $\mathcal{Y}$ and uniform negligibility of 
the higher-order remainder term) which is not required in 
\citet{cattaneo2024uniform} because their estimator is a sample average that can be directly expressed as a sum over dyads, whereas the CMLE is defined implicitly as the solution to a score equation. One should also verify that the strong 
approximation error vanishes uniformly over $\mathcal{Y}$ 
when the sparsity parameter $p_{n,y}$ varies continuously 
across thresholds.} Developing this extension is left for 
future research.
\end{remark}

\section{Monte Carlo Simulations}\label{simulations}

This section presents Monte Carlo simulation studies evaluating the finite-sample performance of the CMLE, relative to the maximum likelihood estimator (MLE) and the analytical bias correction method (BC, \citet{chernozhukov2020network}). I consider two sets of simulations: (i) for the full distribution regression setting calibrated to the international trade application; and (ii) for joint inference across thresholds, including sup-$t$ confidence bands and equality tests. Appendix ... additionally provides a Monte Carlo simulation for a single threshold (a network formation model) with varying degrees of sparsity, extending the DGP of \citet{jochmans2018semiparametric} with additional node heterogeneity specifications and right-tail sparsity (the left-tail case is studied in their paper). The results confirm that the CMLE maintains smaller bias than BC under extreme sparsity, consistent with the theoretical predictions of Section~\ref{asymptotics}.

\subsection{Monte Carlo Simulations for the Distribution Regression}\label{monte_carlo_DR}

Using the bilateral trade dataset \citep{helpman2008estimating,jochmans2018semiparametric,chernozhukov2020network} described in Section \ref{application}, which is standard in the gravity model literature, a logit model with two-way fixed effects is estimated by MLE at each threshold $y_k$, and the estimates are set as the true parameters in the simulation. Following the notation of Section~\ref{model_estimation}, 
the parameter vector at the $k$-th threshold is denoted 
$\boldsymbol{\theta}_{y_k}$. When thresholds correspond 
to empirical quantiles, $\tau$ denotes the quantile level 
associated with threshold $y$, so that 
$\theta_{y,d} = \theta_d(\tau)$ for the $d$-th covariate. At each threshold, the true DGP is
$$\Pr(\tilde{y}_{ij,k} = 1 \mid \boldsymbol{x}_{ij}, 
\alpha_{i,y_k}^{\text{MLE}}, \gamma_{j,y_k}^{\text{MLE}}) = 
\Lambda\big(\boldsymbol{x}_{ij}'
\boldsymbol{\theta}_{y_k}^{\text{MLE}} 
+ \alpha_{i,y_k}^{\text{MLE}} 
+ \gamma_{j,y_k}^{\text{MLE}}\big), \quad 
k = 1, \dots K, \quad (i,j) \in \mathcal{D},$$
where $\boldsymbol{x}_{i j}$ are the values of the covariates for the observational unit $(i, j)$ in the trade data set, $\tilde{y}_{ij,k} = 1 (y_{ij} \leq y_k)$, and the thresholds $y_1 < \dots < y_K$ correspond to the 
empirical quantile levels $\tau_1 < \dots < \tau_K$ of the trade outcome at 
$\tau_k \in \{0.545, 0.550, \ldots, 0.990\}$. 
The true parameter vector includes all coefficients and fixed effects:
$$\boldsymbol{\beta}_k^{\text{MLE}}=\left(
\boldsymbol{\theta}_{y_k}^{\text{MLE}}, 
\alpha_{1,y_k}^{\text{MLE}}, \ldots, 
\alpha_{n,y_k}^{\text{MLE}}, 
\gamma_{2,y_k}^{\text{MLE}}, \ldots, 
\gamma_{n,y_k}^{\text{MLE}}\right).$$
The simulated data are generated by drawing a single vector of $n(n-1)$ 
logistic errors $\varepsilon_{ij} \sim \text{Logistic}(0,1)$, shared 
across all thresholds, so that the binary indicators at different thresholds are generated from 
common shocks, as they are in real distribution regression data. The parameters remain 
threshold-specific: $\boldsymbol{\theta}_{y_k}^{\text{MLE}}$, $\alpha_{i,y_k}^{\text{MLE}}$ and 
$\gamma_{j,y_k}^{\text{MLE}}$ are estimated separately at each threshold, so the design imposes 
no restriction on how they vary across the distribution, which is what DR is intended to 
accommodate. This cross-threshold dependence is accounted for by the joint inference procedures 
in Section~\ref{joint_distribution}. The collection of binary variables at the different thresholds is generated as
$$\tilde{y}_{ij,k}^{sim} = {1}\big(\boldsymbol{W}_{ij}'\boldsymbol{\beta}_k^{\text{MLE}} + \varepsilon_{ij} \geq 0\big), \quad k = 1, \ldots, K, \quad(i, j) \in \mathcal{D},$$
where $\boldsymbol{W}_{ij}$ includes covariates and dummy variables for the dimensions $i$ and $j$ (fixed effects indicators).\footnote{This design choice differs from the Tobit-based DGP of \citet{chernozhukov2020network}. In their DGP, a single latent variable, together with a single sequence of fixed effects (estimated by Tobit on the underlying continuous outcome), generates the binarized outcomes at all thresholds. Since the fixed effects do not vary with $y$, 
link probabilities become increasingly homogeneous 
across nodes at extreme thresholds. The threshold-specific design adopted here allows the fixed effects to reflect the data at each threshold separately, capturing the cross-threshold 
heterogeneity patterns observed in the trade data.}

The dataset contains information on $n=157$ countries, and the covariates included are log of distance, common legal system, contiguous border, common language, and common religion.\footnote{Three additional covariates (colonial ties, currency union, and regional free trade agreement) are available in the dataset but excluded from the analysis because each takes the value one for fewer than 2\% of country pairs (Table~\ref{tab:descriptive} in the Supplemental Appendix), resulting in insufficient variation in the pairwise-differenced covariates at many thresholds.} Results are based on 500 simulations.

\paragraph{Results}

The average probability $\Pr(\tilde{y}_{ij,k} = 1)$ at each 
threshold increases approximately linearly from around 0.55 
at the 54.5th percentile to nearly 1.0 at the 99th 
percentile, recovering the empirical quantile structure of 
the trade data on which the DGP is calibrated 
(Figure~\ref{fig:sparsity} in the Supplemental \ref{simulation_appendix}).

\begin{figure}
\centering
\captionsetup[subfigure]{font=footnotesize, skip=0pt}
\captionsetup{skip=0pt}
\begin{subfigure}[b]{0.45\textwidth}
    \centering
    \includegraphics[width=\textwidth]{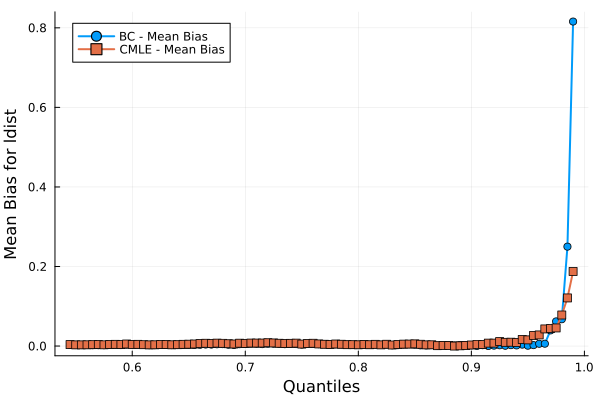}
    \caption{Mean bias: Log distance}
    \label{fig:sim_bias_ldist}
\end{subfigure}
\hfill
\begin{subfigure}[b]{0.45\textwidth}
    \centering
    \includegraphics[width=\textwidth]{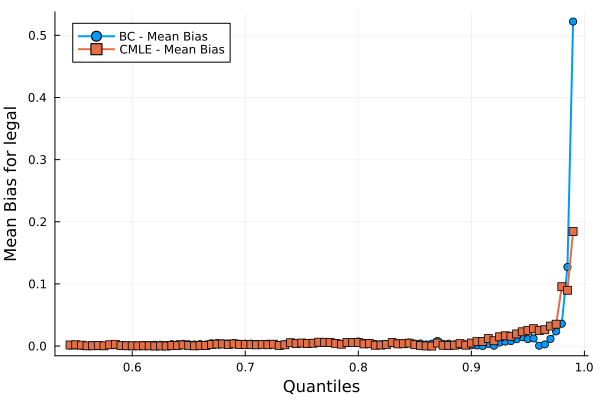}
    \caption{Mean bias: Legal system}
    \label{fig:sim_bias_legal}
\end{subfigure}

\vspace{0.15cm}

\begin{subfigure}[b]{0.45\textwidth}
    \centering
    \includegraphics[width=\textwidth]{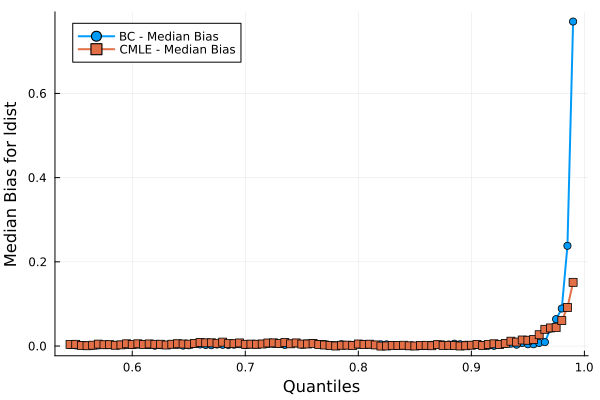}
    \caption{Median bias: Log distance}
    \label{fig:sim_medbias_ldist}
\end{subfigure}
\hfill
\begin{subfigure}[b]{0.45\textwidth}
    \centering
    \includegraphics[width=\textwidth]{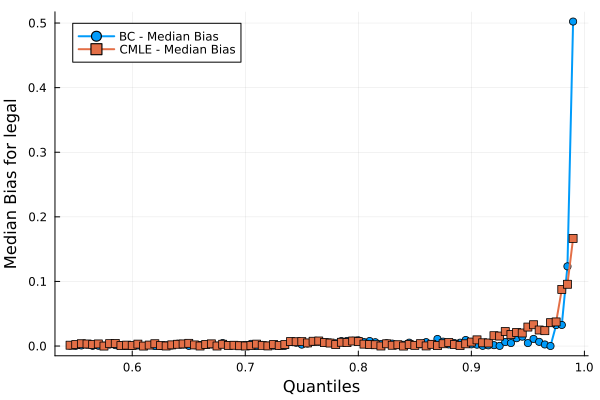}
    \caption{Median bias: Legal system}
    \label{fig:sim_medbias_legal}
\end{subfigure}
\caption{Mean bias (top) and median bias (bottom) of the 
bias-corrected estimator (BC) and the conditional maximum 
likelihood estimator (CMLE) across quantiles of the trade 
distribution, based on 500 replications using the empirically 
calibrated DGP. Thresholds correspond to empirical quantiles 
in [0.545, 0.990] at intervals of 0.005.}
\label{fig:sim_bias}
\label{fig:sim_medbias}
\end{figure}

Figure~\ref{fig:sim_medbias} 
plots the mean and median bias of BC and CMLE across 
quantiles for log distance and common legal system; results 
for contiguous border, common language, and common religion 
are reported in Figures \ref{fig:sim_bias_appendix} and \ref{fig:sim_medbias_appendix} in the Supplemental Appendix. For most of the quantile range, both estimators show negligible bias. In a narrow window before the more extreme quantiles, the BC estimator shows a slightly smaller bias for most covariates. This pattern reflects the fact that the BC uses all observations to estimate the bias correction, while the CMLE uses only the informative quadruples, which constitute a smaller effective sample at moderate sparsity levels. At the most extreme thresholds, in particular beyond the 98th percentile, the BC estimator shows sharply increasing bias for most covariates, while the CMLE's bias remains substantially smaller. On median bias, which is robust to outlier replications, CMLE has smaller bias than BC at extreme thresholds across all covariates. This pattern is consistent with the single-threshold 
simulations in Tables~\ref{tab:est_jochmans_right_uniform_full} and \ref{tab:est_jochmans_right_beta_full} in the Supplemental \ref{sim_single_threshold}.

\begin{figure}
\centering
\captionsetup[subfigure]{font=footnotesize, skip=0pt}
\captionsetup{skip=0pt}
\begin{subfigure}[b]{0.45\textwidth}
    \centering
    \includegraphics[width=\textwidth]{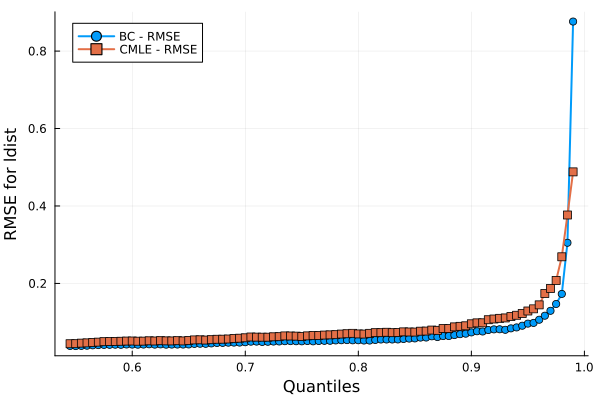}
    \caption{RMSE: Log distance}
    \label{fig:sim_rmse_ldist}
\end{subfigure}
\hfill
\begin{subfigure}[b]{0.45\textwidth}
    \centering
    \includegraphics[width=\textwidth]{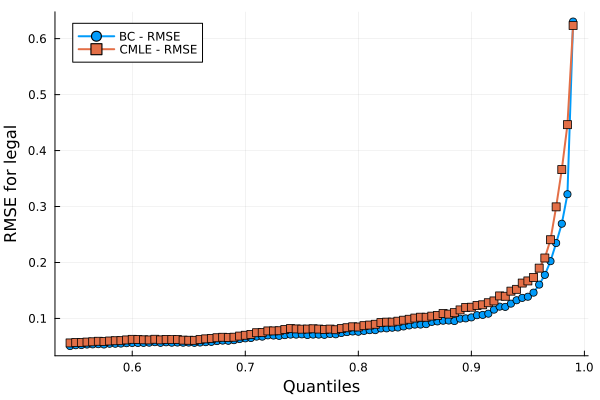}
    \caption{RMSE: Legal system}
    \label{fig:sim_rmse_legal}
\end{subfigure}

\vspace{0.15cm}

\begin{subfigure}[b]{0.45\textwidth}
    \centering
    \includegraphics[width=\textwidth]{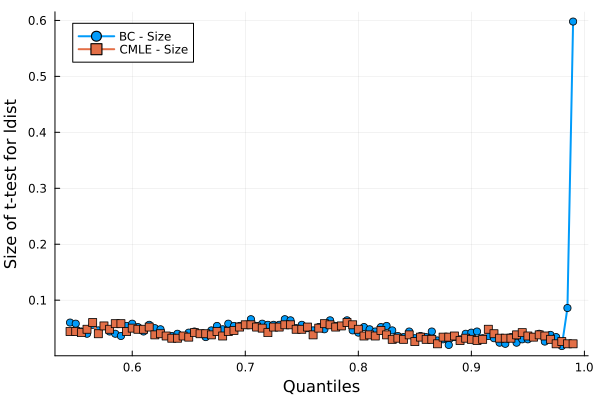}
    \caption{Size: Log distance}
    \label{fig:sim_size_ldist}
\end{subfigure}
\hfill
\begin{subfigure}[b]{0.45\textwidth}
    \centering
    \includegraphics[width=\textwidth]{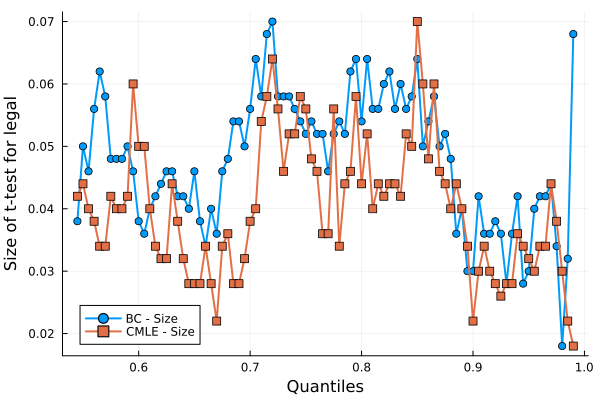}
    \caption{Size: Legal system}
    \label{fig:sim_size_legal}
\end{subfigure}
\caption{RMSE (top) and rejection frequency of the two-sided 
$t$-test at the 5\% nominal level (bottom) for the 
bias-corrected estimator (BC) and the conditional maximum 
likelihood estimator (CMLE) across quantiles of the trade 
distribution, based on 500 replications using the empirically 
calibrated DGP. Thresholds correspond to empirical quantiles 
in [0.545, 0.990] at intervals of 0.005.}
\label{fig:sim_rmse}
\label{fig:sim_size}
\end{figure}

Figure~\ref{fig:sim_rmse} reports the RMSE (top) and 
rejection frequencies (bottom) across quantiles. Through 
most of the quantile range, the RMSE of both estimators is 
nearly identical. In the narrow range before the extremes, the CMLE's RMSE increases somewhat more than BC's, 
reflecting the higher variance from relying on only the 
informative quadruples. The RMSE comparison at the extreme 
quantiles depends on which effect dominates. For log 
distance, where BC's bias is large, the CMLE has 
substantially smaller RMSE. For common legal system, where 
BC's bias is more moderate, the two estimators have 
comparable RMSE. For some covariates for which results are in the Appendix (notably 
contiguous border), the CMLE's variance from the smaller 
effective sample outweighs its bias advantage, giving BC 
smaller RMSE. This bias-variance trade-off is structural: 
the CMLE uses only observations that identify $\boldsymbol{\theta}_{y_k}$, 
which gives smaller bias at the cost of higher variance.

Turning to size control, for log distance, both estimators 
maintain adequate rejection frequencies through most of the 
quantile range. At the extreme quantiles, the BC rejection 
rate rises sharply, reaching over 0.50 at the 99th 
percentile, as the bias dominates the test statistic. The 
CMLE remains close to the nominal level throughout. For 
common legal system, both estimators fluctuate around the 
nominal level, with both becoming conservative at some of 
the highest quantiles. Across covariates, BC displays over-rejection at the extremes 
when the bias dominates (log distance, contiguous border) and 
under-rejection when the standard error estimate becomes 
inflated (common language at the extreme). The CMLE generally maintains better size control, 
though it too becomes conservative for some covariates at 
the most extreme quantiles.

\subsection{Monte Carlo Simulations for Joint Confidence Bands and Testing}\label{monte_carlo_simultaneous}

The simultaneous confidence bands and equality tests across 
thresholds developed in Section~\ref{joint_distribution} are 
evaluated using the empirically calibrated DGP described in 
the previous Subsection, with parameters varied across 
specifications to test specific null and alternative 
hypotheses.

The simulations vary along three dimensions: the number of thresholds $K \in \{5, 15, 50\}$, the location of the thresholds in the distribution (capturing different sequences of convergence rates of threshold-specific estimators), and the values of the parameters. For the location of the thresholds, three configurations are considered: (i) \textit{equally spaced}, with $K$ thresholds spread across the full range of available quantiles, from $\tau = 0.545$ to $\tau = 0.95$ (with intervals of 0.005), capturing variation across different sparsity levels; (ii) \textit{first $K$}, which considers the $K$ lowest thresholds, where the network is relatively dense ($\Pr(\tilde{y}_{ij,k} = 1)$ is moderate); and (iii) \textit{last $K$}, which considers the $K$ highest thresholds, where the network is sparse ($\Pr(\tilde{y}_{ij,k} = 1)$ is very high).

The different designs for the values of the parameters are 
distinguished by the specification of the structural 
parameter vector $\boldsymbol{\theta}_{y_k}$ across thresholds; 
in both designs, the fixed effects $\alpha_{i,y_k}, 
\gamma_{j,y_k}$ are as in the previous Subsection and remain 
threshold-specific. Under \textit{constant 
parameters}, the parameter vector is held constant 
across thresholds at $\boldsymbol{\theta}_{y_k} = \boldsymbol{\theta}_{y_1}^{\text{MLE}}$ for all $k$, so that the null hypothesis $H_{0,d}: \theta_{y_1,d} = \cdots = \theta_{y_K,d}$ holds and the empirical rejection rates should be close to the 
nominal $5\%$ level. Under \textit{varying parameters}, 
the coefficients are set to the empirical estimates 
$\boldsymbol{\theta}_{y_k} = \boldsymbol{\theta}_{y_k}^{\text{MLE}}$, which 
vary across thresholds, so that $H_{0,d}$ is false.

For each simulation and configuration, I compute the CMLE at each of 
the $K$ thresholds, estimate the joint covariance matrix 
$\boldsymbol{\Omega}_{n,\mathbf{y}}(\boldsymbol{\theta}_{n,\mathbf{y}})$, and obtain sup-$t$ critical values from 100000 draws from the estimated 
Gaussian distribution. The equality tests use deviations from the first 
threshold as the contrast specification, i.e., 
$\delta_{y_k,d} = \theta_{y_k,d} - \theta_{y_1,d}$ for $k = 2, \ldots, K$ and all covariates $d$. Results are based on $500$ simulations.

\paragraph{Results}

Table \ref{tab:main_coverage_095} reports the empirical coverage of the sup-$t$ joint confidence bands and the empirical simultaneous coverage of the pointwise confidence intervals (the fraction of simulations where all $K$ pointwise intervals simultaneously contain the true values) for the configuration considering equally spaced thresholds, for constant and varying parameters, with the highest threshold considered at 
$\tau_{\max} = 0.95$. Supplemental \ref{simulation_appendix} provides the simulation results for the remaining designs and specifications considered above and extends the threshold range 
to $\tau_{\max} = 0.99$, assessing whether the sup-$t$ bands maintain correct coverage 
when thresholds reach into the sparse tail where fewer 
informative quadruples are available.

\begin{table}[htbp]
\centering
\caption{Coverage Comparison: Pointwise vs Sup-$t$ Bands ($\tau_{\max} = 0.95$, Equally Spaced)}
\label{tab:main_coverage_095}
\begin{threeparttable}
\footnotesize
\begin{tabular}{ll cc cc cc cc cc}
\toprule
& & \multicolumn{2}{c}{Distance} & \multicolumn{2}{c}{Legal} & \multicolumn{2}{c}{Border} & \multicolumn{2}{c}{Language} & \multicolumn{2}{c}{Religion} \\
\cmidrule(lr){3-4} \cmidrule(lr){5-6} \cmidrule(lr){7-8} \cmidrule(lr){9-10} \cmidrule(lr){11-12}
$K$ & DGP & PW & Sup-$t$ & PW & Sup-$t$ & PW & Sup-$t$ & PW & Sup-$t$ & PW & Sup-$t$ \\
\midrule
5 & Varying & 83.60 & 97.40 & 81.80 & 96.80 & 85.40 & 96.00 & 81.20 & 96.60 & 83.20 & 95.20 \\
5 & Constant & 81.40 & 95.20 & 83.60 & 97.40 & 82.00 & 94.80 & 82.80 & 94.40 & 85.60 & 97.00 \\
\addlinespace
15 & Varying & 69.00 & 96.60 & 69.00 & 96.80 & 73.20 & 95.40 & 66.80 & 95.60 & 69.60 & 96.20 \\
15 & Constant & 68.00 & 94.80 & 71.40 & 96.20 & 71.80 & 95.80 & 68.00 & 95.20 & 70.80 & 97.60 \\
\addlinespace
50 & Varying & 56.20 & 96.60 & 56.20 & 97.00 & 61.60 & 96.00 & 54.40 & 95.40 & 56.20 & 96.00 \\
50 & Constant & 56.40 & 96.20 & 58.60 & 97.40 & 57.80 & 95.00 & 55.00 & 95.20 & 59.80 & 97.20 \\
\bottomrule
\end{tabular}
\begin{tablenotes}
\smallskip\footnotesize
\item \textit{Notes:} Empirical coverage rates (in \%) of 95\% simultaneous confidence bands. ``PW'' (Pointwise) uses $z_{0.975} = 1.96$ at each threshold; ``Sup-$t$'' uses simulated critical values. Target coverage is 95\%. ``Constant'' = coefficients identical across thresholds ($H_0$ true); ``Varying'' = coefficients follow empirical trade data pattern ($H_0$ false). Based on 500 Monte Carlo simulations.
\end{tablenotes}
\end{threeparttable}
\end{table}

Table \ref{tab:main_coverage_095} shows that while the results for the pointwise confidence intervals show under-coverage that worsens with $K$, the sup-$t$ bands maintain coverage at or near the $95\%$ 
level. The Supplemental \ref{app_monte_carlo_simultaneous} shows that this 
pattern holds across $K$, covariates, $\tau_{\max}$, and DGP 
design, with sup-$t$ coverage remaining near the nominal level 
throughout. This pattern shows the need for simultaneous 
inference when constructing coverage statements across multiple 
thresholds: pointwise intervals constructed at the marginal 
$95\%$ level under-cover when interpreted simultaneously, 
while the sup-$t$ critical value uses the joint distribution 
to deliver correct simultaneous coverage. For pointwise 
coverage, extending to $\tau_{\max} = 0.99$ further reduces 
simultaneous coverage substantially.

\begin{table}[htbp]
\centering
\caption{Size and Power by Covariate: Sup-$t$ vs Wald ($\tau_{\max} = 0.95$, Equally Spaced)}
\label{tab:main_by_covariate_095}
\begin{threeparttable}
\footnotesize
\begin{tabular}{l cc cc cc cc cc}
\toprule
& \multicolumn{2}{c}{Distance} & \multicolumn{2}{c}{Legal} & \multicolumn{2}{c}{Border} & \multicolumn{2}{c}{Language} & \multicolumn{2}{c}{Religion} \\
\cmidrule(lr){2-3} \cmidrule(lr){4-5} \cmidrule(lr){6-7} \cmidrule(lr){8-9} \cmidrule(lr){10-11}
$K$ & Sup-$t$ & Wald & Sup-$t$ & Wald & Sup-$t$ & Wald & Sup-$t$ & Wald & Sup-$t$ & Wald \\
\midrule
\multicolumn{11}{l}{\textit{Panel A: Constant Coefficients (Size, target: 5\%)}} \\
\addlinespace
5 & 4.60 & 3.80 & 2.80 & 2.60 & 3.60 & 3.00 & 4.20 & 4.20 & 3.60 & 3.80 \\
15 & 3.20 & 3.20 & 3.60 & 1.60 & 4.00 & 5.00 & 3.60 & 3.40 & 3.20 & 3.80 \\
50 & 2.80 & 2.60 & 2.80 & 1.60 & 5.80 & 13.80 & 3.60 & 2.20 & 2.40 & 2.40 \\
\midrule
\multicolumn{11}{l}{\textit{Panel B: Varying Coefficients (Power)}} \\
\addlinespace
5 & 100.00 & 100.00 & 100.00 & 99.80 & 99.40 & 99.00 & 39.80 & 40.00 & 52.80 & 60.40 \\
15 & 100.00 & 100.00 & 100.00 & 100.00 & 98.80 & 97.80 & 58.80 & 53.40 & 77.80 & 91.20 \\
50 & 100.00 & 99.60 & 100.00 & 100.00 & 99.00 & 98.20 & 54.40 & 79.00 & 79.20 & 92.40 \\
\addlinespace
\bottomrule
\end{tabular}
\begin{tablenotes}
\smallskip\footnotesize
\item \textit{Notes:} Rejection rates (in \%) for testing $H_0: \theta_d(\tau_1) = \cdots = \theta_d(\tau_K)$ at the 5\% level. Sup-$t$ = Sup-$t$ test; Wald = Wald test. ``Constant'' = coefficients identical across thresholds ($H_0$ true); ``Varying'' = coefficients follow empirical trade data pattern ($H_0$ false). Based on 500 Monte Carlo simulations.
\end{tablenotes}
\end{threeparttable}
\end{table}

Table~\ref{tab:main_by_covariate_095} reports size and power 
of the sup-$t$ and Wald equality tests for the equally 
spaced configuration at $\tau_{\max} = 0.95$. Under the null, the sup-$t$ test stays at or near the nominal 
$5\%$ level across all values of $K$, with rejection rates 
between $2.4\%$ and $5.8\%$, indicating size control without 
substantial over-rejection but some conservativeness in 
particular cases. The Wald 
test, by contrast, shows size distortion that varies by 
covariate and worsens with $K$: for Border at $K = 50$, it 
over-rejects at 13.8\%, while for Legal it under-rejects at 
1.6\%. The sup-$t$ test exhibits high power for Distance, 
Legal, and Border, consistent with the large heterogeneity in the coefficients for these covariates, shown in 
Figure~\ref{fig:coef_paths_main} in the Supplemental \ref{app_monte_carlo_simultaneous}. Language and Religion show 
more moderate power, reflecting smaller heterogeneity in their coefficients across thresholds. The Wald test shows higher power than the sup-$t$ test for Language and Religion 
at large $K$, but at the cost of substantial size distortion 
for Border at the same $K$; the sup-$t$ test provides more 
uniform performance across covariates. The Supplemental \ref{app_monte_carlo_simultaneous} 
shows that these patterns hold at $\tau_{\max} = 0.99$.

\section{Application to gravity models of international trade}\label{application}

There are two important features that models for bilateral international trade should take into account. First, the outcome of interest (the volume of bilateral trade) is bounded below at zero, with a mass at zero: in the dataset used in this application, approximately 55\% of country pairs have zero bilateral trade in a given year. Second, consistent estimation typically requires controlling for country-specific terms that capture unobservable barriers each country faces with all its trading partners. Such terms are known in the international trade literature as multilateral resistance terms \citep{anderson2003gravity}, and are typically treated as two-way (exporter and importer) fixed effects.

Several approaches to handling these features focus on the conditional mean of trade. The Poisson pseudo-maximum likelihood (PPML) estimator \citep{silva2006log} retains zero observations and delivers consistent estimates of the gravity coefficients in two-way fixed effects gravity specifications \citep{fernandez2016individual}. Another approach is the Heckman-type sample selection model of \cite{helpman2008estimating}, which models the decision to trade and the volume of trade jointly via a two-stage procedure with a first-stage selection equation and a second-stage equation for log positive trade.\footnote{\cite{helpman2008estimating} estimate the first-stage selection equation by standard probit, which is subject to the incidental parameter problem in two-way fixed effects settings. Bias-corrected estimators for binary outcome models with two-way fixed effects are available \citep{fernandez2016individual}, and in the dataset used here (where approximately 45\% of country pairs have positive trade), where there is no first-degree sparsity, they could be applied. In settings with sparser first-stage outcomes, the fixed effects cannot be consistently estimated even with bias correction, making the Heckman-type selection correction infeasible. \cite{sakamoto2024dyadic} develops a semiparametric sample selection estimator for dyadic data extending 
\cite{kyriazidou1997estimation}, but this approach requires panel data with time-varying covariates, which is not available in this application since key gravity determinants (such as distance) are time-invariant.} However, both methods estimate effects on the conditional mean, and neither characterizes how covariates might affect different parts of the trade distribution differently. 

There are theoretical reasons to expect that the effects of trade determinants are not constant across the distribution of bilateral trade flows \citep{novy2013international, bas2017micro, carrere2020gravity}. Several econometric methods have been used to study this heterogeneity. Quantile regression on log positive trade \citep{baltagi2016estimationmain} documents that gravity-covariate effects vary across quantiles but drops zero-trade observations entirely. More recently, \cite{bergstrand2025quantile} use censored quantile regression that retains zeros. However, they address the incidental parameter problem from the high-dimensional fixed effects by imposing a parametric restriction on the unobserved heterogeneity via a Chamberlain-Mundlak correlated random effects parameterization \citep{abrevaya2008effects}. Expectile regression provides an alternative distributional approach \citep{bergstrand2025tails} that extends Poisson pseudo-maximum likelihood to recover expectile-specific coefficients across the conditional distribution. This approach handles zeros and documents heterogeneous effects, but does not inherit the robustness of Poisson-based estimation to the incidental parameter problem, with no bias correction currently available.\footnote{Their Monte Carlo simulations show coverage 
degradation at extreme expectiles, particularly with shorter time 
dimensions. The authors leave the development of bias corrections 
for future work.} 

The distribution regression framework developed in this paper is used to study how the effects of standard gravity covariates vary across the conditional distribution of bilateral trade. At each threshold $y_k$, the binary indicator $\tilde{y}_{ij,k} = \mathbf{1}\{y_{ij} \leq y_k\}$ is well-defined whether or not country pair $(i,j)$ has positive trade, and the estimator is consistent without imposing a parametric specification on the unobserved heterogeneity. The joint inference framework developed in Section~\ref{joint_distribution} allows formal tests of whether covariate effects are constant across the full threshold range, rather than pairwise comparisons at selected points.

The dataset, used previously by \cite{helpman2008estimating}, \cite{jochmans2018semiparametric}, and \cite{chernozhukov2020network}, contains information on bilateral trade flows and covariates for 157 countries in 1986, where $i$ and $j$ index each country as an exporter and an importer, respectively. Descriptive statistics are reported in Table~\ref{tab:descriptive} in the Supplemental Appendix. The outcome $y_{ij}$ is the volume of trade in thousands of constant 2000 U.S. dollars from country $i$ to country $j$.\footnote{The bilateral trade flows data are from Feenstra's ``World Trade Flows, 1970--1992,'' transformed to constant 2000 U.S. dollars using the U.S. CPI by \cite{helpman2008estimating}.} The covariates $\boldsymbol{x}_{ij}$ include the following bilateral determinants of trade flows: the logarithm of distance between capitals, and binary indicators for shared legal system, contiguous border, common language, and common religion.

\begin{figure}[htbp]
\centering
\captionsetup[subfigure]{font=footnotesize, skip=0pt}
\captionsetup{skip=0pt}
\begin{subfigure}[b]{0.48\textwidth}
    \centering
    \includegraphics[width=\textwidth]{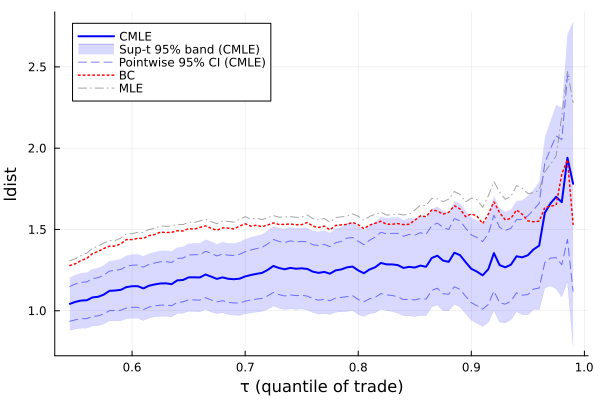}
    \caption{Log distance}
    \label{fig:band_ldist}
\end{subfigure}
\hfill
\begin{subfigure}[b]{0.48\textwidth}
    \centering
    \includegraphics[width=\textwidth]{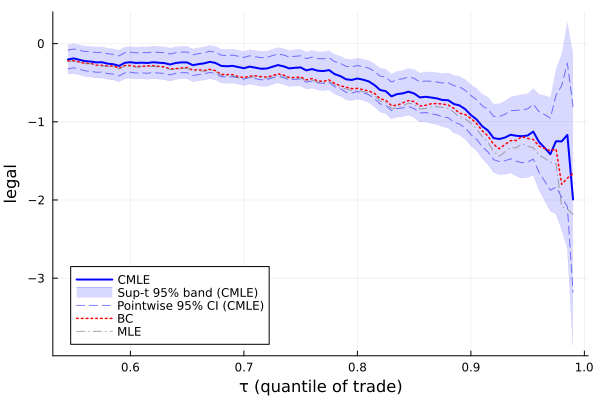}
    \caption{Legal system}
    \label{fig:band_legal}
\end{subfigure}
\caption{Distribution regression estimates of the effect of log distance and 
legal system on bilateral trade. Solid blue: CMLE; dotted red: BC; dash-dotted 
gray: MLE. Shaded region and dashed lines show the 95\% simultaneous sup-$t$ 
confidence band and pointwise confidence intervals for the CMLE, respectively. 
Thresholds correspond to empirical quantiles in $[0.545, 0.990]$ at intervals 
of $0.005$.}
\label{fig:bands_main}
\end{figure}

\begin{figure}[htbp]
\centering
\captionsetup[subfigure]{font=footnotesize, skip=0pt}
\captionsetup{skip=0pt}
\begin{subfigure}[b]{0.48\textwidth}
    \centering
    \includegraphics[width=\textwidth]{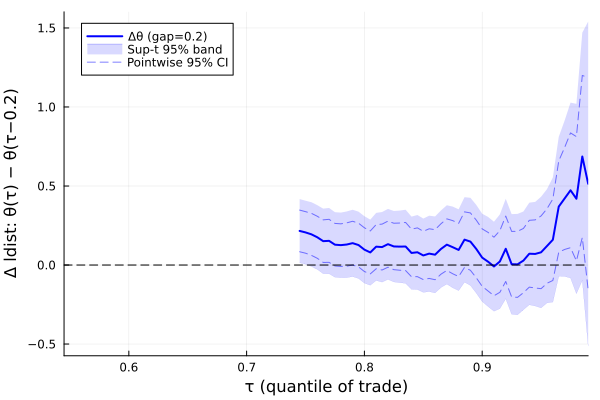}
    \caption{Log distance}
    \label{fig:diff_ldist}
\end{subfigure}
\hfill
\begin{subfigure}[b]{0.48\textwidth}
    \centering
    \includegraphics[width=\textwidth]{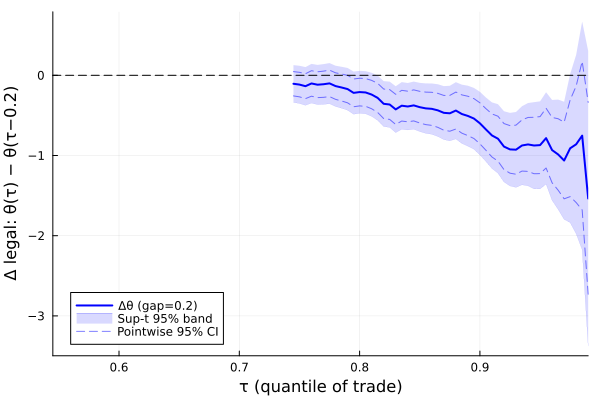}
    \caption{Legal system}
    \label{fig:diff_legal}
\end{subfigure}
\caption{Differences $\theta_{n,d}(\tau) - \theta_{n,d}(\tau - 0.20)$ from the CMLE 
estimator for log distance and legal system. Shaded region and dashed lines show 
the 95\% simultaneous sup-$t$ confidence band and pointwise confidence intervals 
for the differences, respectively. The horizontal dashed line marks zero.}
\label{fig:diffs_main}
\end{figure}

\begin{table}[htbp]\label{table_application}
\centering
\caption{Joint tests of coefficient equality across thresholds}
\label{tab:equality_tests}
    \footnotesize
\begin{tabular}{l cc c cc c cc}
\toprule
 & \multicolumn{2}{c}{$\tau_{\max} = 0.95$} & & \multicolumn{2}{c}{$\tau_{\max} = 0.97$} & & \multicolumn{2}{c}{$\tau_{\max} = 0.99$} \\
\cmidrule{2-3} \cmidrule{5-6} \cmidrule{8-9}
Covariate & Sup-$t$ & Wald & & Sup-$t$ & Wald & & Sup-$t$ & Wald \\
\midrule
Log distance & $3.83{}^{**}$ & $82.7$ & & $3.83{}^{**}$ & $88.8$ & & $3.83{}^{**}$ & $98.2$ \\
Legal system & $6.75{}^{**}$ & $141.3{}^{**}$ & & $6.75{}^{**}$ & $150.4{}^{**}$ & & $6.75{}^{**}$ & $156.5{}^{**}$ \\
Border & $3.95{}^{**}$ & $91.2$ & & $3.95{}^{**}$ & $95.2$ & & $3.95{}^{**}$ & $104.9$ \\
Language & $2.91$ & $79.1$ & & $2.91$ & $80.3$ & & $2.91$ & $81.1$ \\
Religion & $2.68$ & $90.7$ & & $2.68$ & $95.6$ & & $2.68$ & $96.5$ \\
\midrule
\textit{K} & \multicolumn{2}{c}{82} & & \multicolumn{2}{c}{86} & & \multicolumn{2}{c}{90} \\
\bottomrule
\end{tabular}
\begin{minipage}{\textwidth}
\smallskip\footnotesize
\textit{Notes:} The sup-$t$ statistic tests $H_0\colon \theta_{y_1,d} = \cdots = \theta_{y_K,d}$ using the
supremum of standardized deviations from the first threshold, with critical values
obtained from $100{,}000$ draws of the estimated Gaussian distribution implied by the joint
asymptotic distribution (Theorem~\ref{theorem3}). The Wald statistic uses
adjacent differences $R = R_{\text{adj}}$ and is distributed as $\chi^2_{K-1}$ under $H_0$.
${}^{**}$~Rejection at the $5\%$ level. $\tau_{\max}$ denotes the largest quantile included in the test.
\end{minipage}
\end{table}

Figure~\ref{fig:bands_main} displays the DR coefficient estimates 
for log distance and legal system. A positive DR 
coefficient indicates a negative effect on trade volume: a 
covariate that increases the probability of trade falling below 
a threshold $y$ reduces the likelihood of large bilateral flows. The CMLE coefficient $\theta_d(y)$ has a direct interpretation as a log-odds ratio for $\Pr(Y_{ij} \leq y \mid \boldsymbol{x}_{ij}, \nu_i, \omega_j)$, such that variation across thresholds characterizes how this log-odds effect varies across the distribution.

The CMLE estimates show substantial heterogeneity in the 
structural parameters across the distribution. For log 
distance, the coefficient increases from approximately $1.05$ 
at the 55th percentile to around $1.26$ at the 90th percentile, and $1.78$ at the 99th percentile. The increasing coefficient indicates that distance is a stronger barrier at the upper 
end of the trade distribution.
For legal system, the coefficient becomes increasingly negative 
across the distribution, suggesting that shared legal origins 
facilitate trade more strongly among the largest bilateral 
relationships.
Table~\ref{tab:equality_tests} confirms that the null of 
constant coefficients is rejected for both covariates.

Turning to the comparison across estimators, for log distance, the CMLE estimates are systematically lower than both the BC and uncorrected MLE across most of the distribution, with a persistent gap of approximately $0.3$ in the area where the network is denser (until approximately the 95th percentile). Part of this gap may reflect higher-order incidental parameter bias in the BC estimator, or differences in pseudo-true parameters under misspecification of the logistic link \citep{white1982maximum, hughes2026jackknife}. In the upper tail of log distance, the uncorrected MLE diverges from the other two, while both BC and CMLE become noisier. 

The pattern across the other four covariates is qualitatively different: BC and CMLE track each other closely in the interior of the distribution, then diverge meaningfully in the upper tail, where the estimates also become considerably noisier. For legal system, the divergence emerges from approximately the 93rd percentile; for border, language, and religion (Supplemental \ref{application_appendix}), it emerges earlier, from approximately the 85th-90th percentile, with larger gaps at extreme thresholds. This pattern across four of five covariates is consistent with the BC bias correction relying on the assumption of a dense network, which fails at extreme thresholds. As expected, the joint sup-$t$ confidence bands are wider than the pointwise intervals, and both widen toward the extremes as the effective sample size for estimation decreases.

Table~\ref{tab:equality_tests} reports the sup-$t$ and Wald 
test statistics for $H_0\colon \theta_{y_1,d} = \cdots = \theta_{y_K,d}$, evaluated at three choices of the maximum considered
quantile: $\tau_{\max} \in \{0.95, 0.97, 0.99\}$. These vary how far the analysis extends into the sparse tail. The sup-$t$ 
test rejects at 5\% level for log distance, legal system, and border for all specifications, and does not reject the null for language and religion. The Wald test does not reject for any covariate apart from legal system. With $K - 1$ ranging from 81 to 89 restrictions, localized departures from equality contribute relatively little to the joint Wald statistic, while the sup-$t$ test is driven by the largest standardized deviation across thresholds. The rejection decisions are stable across $\tau_{\max}$: for sup-$t$, the supremum is attained at or below $\tau = 0.95$ for all covariates, so no rejection is overturned despite the larger critical value at higher $\tau_{\max}$; for the Wald test, the test statistic itself changes with $\tau_{\max}$ but the rejection decisions are unaffected.

To characterize where in the distribution the variation in the distribution regression coefficients occurs, Figure~\ref{fig:diffs_main} displays simultaneous 95\% confidence bands for the differences $\theta_{n,d}(\tau) - \theta_{n,d}(\tau - 0.20)$ for log distance and legal system. These bands use the feasible joint covariance of the CMLE and the Gaussian approximation of Theorem~\ref{theorem3}(i), applied to the contrasts. Since the bands have simultaneous coverage at the 95\% level, any threshold $\tau$ where the band excludes zero refers to a rejection of $H_0\colon \theta_d(\tau) = \theta_d(\tau - 0.20)$, with family-wise error rate controlled across all thresholds jointly. For log distance, the differences are mostly positive but the band excludes zero only at the two lowest threshold pairs, indicating that few 20-percentile intervals show a significant change. Combined with the rejection in Table~\ref{tab:equality_tests}, this indicates that thresholds far apart in the distribution differ significantly, but neighboring thresholds do not: the effect of distance varies gradually across the distribution. For legal system, by contrast, the band excludes zero across much of the upper range, consistent with shared legal origins becoming increasingly important for the largest bilateral relationships.

Supplemental \ref{application_appendix} presents additional results. The coefficient paths with sup-$t$ bands and pointwise confidence intervals are shown for the remaining covariates. Pointwise comparisons between the CMLE and BC for all covariates show that CMLE intervals are wider throughout, reflecting the efficiency cost of conditioning that yields the CMLE's robustness in sparse settings. Sup-$t$ bands evaluated at $\tau_{\max} = 0.97$ yield similar conclusions, and the simultaneous bands for the coefficient differences are reported for all covariates at both $\tau_{\max} = 0.97$ and $0.99$.

The CMLE does not deliver estimates of the fixed effects, which limits the recovery of the conditional distribution and the computation of counterfactual distributions or related functionals that aggregate over the fixed effects. Moreover, under sparsity, average (marginal) effects are at best set-identified, analogous to short panel data settings.\footnote{Methods for bounding partially identified average effects have been developed for panel data models with fixed effects \citep{honore2006bounds, chernozhukov2013average, davezies2025identification, botosaru2024adversarial, dobronyi2024identification, pakel2023bounds, aguirregabiria2024identification}; extending these methods to network models (under a dyadic structure and sparsity) remains an open question.} Nevertheless, beyond the log-odds interpretation of individual coefficients, ratios of coefficients at a given threshold correspond to ratios of partial derivatives of the conditional quantile function $Q(u \mid \boldsymbol{x}_{ij}, \boldsymbol{\nu}_i, \boldsymbol{\omega}_j)$, where $u \in (0,1)$ denotes the quantile index \citep{chernozhukov2020network}:
\begin{align}
    \left.\frac{\theta_{\ell,y}}{\theta_{k,y}}\right|_{y=Q(u \mid \boldsymbol{x}_{ij}, \boldsymbol{\nu}_i, \boldsymbol{\omega}_j)} = \frac{\partial_{x_{ij}^{\ell}} Q(u \mid \boldsymbol{x}_{ij}, \boldsymbol{\nu}_i, \boldsymbol{\omega}_j)}{\partial_{x_{ij}^k} Q(u \mid \boldsymbol{x}_{ij}, \boldsymbol{\nu}_i, \boldsymbol{\omega}_j)}.
\end{align}
This relationship applies only at thresholds above the zero mass, where the conditional distribution is continuous. The ratios measure the relative importance of different covariates at each point in the distribution, even when absolute marginal effects are not point-identified. For instance, the ratio $\theta_{n,\text{ldist},y} / \theta_{n,\text{legal},y}$ indicates how much larger the marginal effect of distance is relative to that of legal system on the conditional quantile of trade at threshold $y$.

The findings in this section provide evidence of heterogeneity in covariate effects across the conditional distribution of trade. The directional patterns differ from earlier evidence based on quantile regression \citep{baltagi2016estimationmain, carrere2020gravity} and expectile regression \citep{bergstrand2025tails}, which find covariate effects typically larger at the lower part of the conditional distribution of trade flows. The bias-corrected distribution regression estimates of \cite{chernozhukov2020network} on the same data, by contrast, deliver coefficient paths qualitatively similar to those obtained with the approach proposed in this paper. The contrast across methods reflects that DR, QR, and expectile regression target different distributional objects, although differences in dataset and time period across studies may also contribute. 
\section{Conclusion}\label{conclusion}

I develop a framework for distribution regression in network settings with two-way fixed effects that are treated as incidental parameters. The conditional maximum likelihood method of \citet{charbonneau2017multiple} and \citet{jochmans2018semiparametric}, originally proposed for network formation models to address the incidental parameter problem, is extended to a distribution regression setting. The estimator is applied at multiple thresholds of the outcome distribution (after a binarization of the outcome variable). The proposed method provides asymptotically unbiased pointwise estimates, in particular in sparse settings and regions of the distribution (extreme quantiles), filling a gap in the literature. I establish joint asymptotic results across multiple thresholds with heterogeneous convergence rates, and construct simultaneous sup-$t$ confidence bands and equality tests. Monte Carlo simulations confirm that the proposed estimator has reduced bias and valid inference in sparse settings, in particular at the extreme quantiles of the distribution, where methods based on analytical bias corrections are not well suited. The simulations also 
confirm that the sup-$t$ bands provide correct simultaneous coverage across all configurations considered, while the Wald test exhibits size distortions that worsen with the 
number of considered thresholds. An empirical application to international trade documents heterogeneity in the structural parameters of the gravity model across the distribution of trade flows and identifies where it is most pronounced.

The method proposed in this paper complements the bias-corrected distribution regression of \citet{chernozhukov2020network}. The two approaches target different objects: in dense settings, the bias-corrected DR recovers counterfactual distributions and functionals that aggregate over the fixed effects, while the DR with CMLE delivers asymptotically unbiased estimates of the structural parameter $\theta_0(y)$ in sparse settings (both in terms of first and second degree) where the fixed effects cannot be consistently estimated. Developing informative bounds on partially identified average effects in this setting, which are typically only set-identified under sparsity, is a natural direction for future research.

While the application in this paper is to international trade, the framework applies broadly to dyadic network settings with two-way fixed effects. Examples include bilateral migration \citep{anderson2011gravity, beine2016practitioners, grogger2011income}, foreign direct investment, bilateral patent flows, and bilateral trade in services. In these settings, researchers have been interested in the relative importance of different covariates; for instance, in migration, the relative importance of geographic, cultural, and policy-related barriers \citep{ortega2013effect, grogger2011income}; in trade, the tariff equivalents of various frictions. The ratio interpretation introduced in Section~\ref{application} provides a formal tool for such questions. More specifically, in bilateral migration, the ratio of the coefficient on a visa waiver indicator to the coefficient on log distance can be interpreted as the geographic distance offset by a visa waiver between origin and destination countries. Moreover, the framework can also be adapted to bipartite settings, such as worker-firm matching with wage outcomes. Notably, these applications typically feature sparse 
networks with substantial mass at zero, precisely the 
setting where the framework developed in this paper 
provides valid inference.

\bibliography{library}


\appendix
\renewcommand{\thesection}{Appendix~\Alph{section}}
\section{Proof of Lemma \ref{lemma_sufficient}}
\label{appendix_sufficient}

\noindent \textit{Proof.} Denote by 
$\tilde{\boldsymbol{y}}_y$ the vector of all binary 
indicators 
$(\tilde{y}_{12,y}, \tilde{y}_{13,y}, \ldots, 
\tilde{y}_{n,n-1,y})$ at threshold $y$; by 
$\boldsymbol{r}_y = (r_{1,y}, \ldots, r_{n,y})$ the 
vector of row sums where 
$r_{i,y} = \sum_{j=1}^n \tilde{y}_{ij,y}$; and by 
$\boldsymbol{c}_y = (c_{1,y}, \ldots, c_{n,y})$ the 
vector of column sums where 
$c_{j,y} = \sum_{i=1}^n \tilde{y}_{ij,y}$.

Sufficiency holds if 
$\Pr(\tilde{\boldsymbol{y}}_y \mid \boldsymbol{r}_y, 
\boldsymbol{c}_y, \{\alpha_{i,y}\}_n, \{\gamma_{j,y}\}_n, 
\{\boldsymbol{x}_{ij}\}_{n,n}; \boldsymbol{\theta}_y)$ 
does not depend on the fixed effects 
$\{\alpha_{i,y}\}_n$ and $\{\gamma_{j,y}\}_n$. It follows that:
$$ \Pr(\tilde{\boldsymbol{y}}_y \mid \boldsymbol{r}_y, 
\boldsymbol{c}_y, \{\alpha_{i,y}\}_n, \{\gamma_{j,y}\}_n, 
\{\boldsymbol{x}_{ij}\}_{n,n}; \boldsymbol{\theta}_y) = 
\frac{\Pr\left(\tilde{\boldsymbol{y}}_y \mid 
\{\alpha_{i,y}\}_n, \{\gamma_{j,y}\}_n, 
\{\boldsymbol{x}_{ij}\}_{n,n}; 
\boldsymbol{\theta}_y\right)}{\Pr\left(\boldsymbol{r}_y, 
\boldsymbol{c}_y \mid \{\alpha_{i,y}\}_n, 
\{\gamma_{j,y}\}_n, \{\boldsymbol{x}_{ij}\}_{n,n}; 
\boldsymbol{\theta}_y\right)}, $$

\noindent where 
$\Pr\left(\boldsymbol{r}_y, \boldsymbol{c}_y \mid 
\{\alpha_{i,y}\}_n, \{\gamma_{j,y}\}_n, 
\{\boldsymbol{x}_{ij}\}_{n,n}; 
\boldsymbol{\theta}_y\right) = \sum_{\bar{\tilde{\boldsymbol{y}}}_y \in \mathbb{Q}} 
\Pr(\tilde{\boldsymbol{y}}_y = 
\bar{\tilde{\boldsymbol{y}}}_y \mid \{\alpha_{i,y}\}_n, 
\{\gamma_{j,y}\}_n, \{\boldsymbol{x}_{ij}\}_{n,n}; 
\boldsymbol{\theta}_y)$, and $\mathbb{Q}$ is the set of 
all binary matrices with the same row sums 
$\boldsymbol{r}_y$ and column sums $\boldsymbol{c}_y$.

From the model specification in Equation 
\eqref{eq:main_charbonneau}, the probability of a single 
binary indicator is:
$$
\Pr(\tilde{y}_{ij,y} \mid \boldsymbol{x}_{ij}, 
\alpha_{i,y}, \gamma_{j,y}; \boldsymbol{\theta}_y) = 
\frac{\exp(\boldsymbol{x}_{ij}'\boldsymbol{\theta}_y + 
\alpha_{i,y} + \gamma_{j,y})^{\tilde{y}_{ij,y}}}{1 + 
\exp(\boldsymbol{x}_{ij}'\boldsymbol{\theta}_y + 
\alpha_{i,y} + \gamma_{j,y})}. $$

By conditional independence across dyads, the joint 
probability of all indicators at threshold $y$ is:
\begin{adjustwidth}{-0.25in}{-0.25in}
\begin{align}
\Pr\left(\tilde{\boldsymbol{y}}_y \mid 
\{\alpha_{i,y}\}_n, \{\gamma_{j,y}\}_n, 
\{\boldsymbol{x}_{ij}\}_{n,n}; 
\boldsymbol{\theta}_y\right) &= \prod_{(i,j) \in 
\mathcal{D}} \Pr\left(\tilde{y}_{ij,y} \mid 
\boldsymbol{x}_{ij}, \alpha_{i,y}, \gamma_{j,y}; 
\boldsymbol{\theta}_y\right) \nonumber \\
&= \frac{\exp\left(\sum_{(i,j) \in \mathcal{D}} 
\tilde{y}_{ij,y}\, \boldsymbol{x}_{ij}' 
\boldsymbol{\theta}_y + \sum_{(i,j) \in \mathcal{D}} 
\tilde{y}_{ij,y}(\alpha_{i,y} + 
\gamma_{j,y})\right)}{\prod_{(i,j) \in \mathcal{D}} 
\left[1 + \exp(\boldsymbol{x}_{ij}' 
\boldsymbol{\theta}_y + \alpha_{i,y} + 
\gamma_{j,y})\right]} \nonumber \\
&= \frac{\exp\left(\sum_{(i,j) \in \mathcal{D}} 
\tilde{y}_{ij,y}\, \boldsymbol{x}_{ij}' 
\boldsymbol{\theta}_y\right) \exp\left(\sum_{i=1}^n 
\alpha_{i,y}\, r_{i,y}\right) \exp\left(\sum_{j=1}^n 
\gamma_{j,y}\, c_{j,y}\right)}{\prod_{(i,j) \in 
\mathcal{D}} \left[1 + \exp(\boldsymbol{x}_{ij}' 
\boldsymbol{\theta}_y + \alpha_{i,y} + 
\gamma_{j,y})\right]} \nonumber
\end{align}
\end{adjustwidth}

\noindent where the last equality uses 
$\sum_{(i,j) \in \mathcal{D}} \tilde{y}_{ij,y}\, 
\alpha_{i,y} = \sum_{i=1}^n \alpha_{i,y} \sum_{j \neq i} 
\tilde{y}_{ij,y} = \sum_{i=1}^n \alpha_{i,y}\, r_{i,y}$, 
and analogously for $\gamma_{j,y}$. Taking the ratio:
\begin{align}
&\frac{\Pr\left(\tilde{\boldsymbol{y}}_y \mid 
\{\alpha_{i,y}\}_n, \{\gamma_{j,y}\}_n, 
\{\boldsymbol{x}_{ij}\}_{n,n}; 
\boldsymbol{\theta}_y\right)}{\sum_{\bar{\tilde{\boldsymbol{y}}}_y \in \mathbb{Q}} 
\Pr\left(\bar{\tilde{\boldsymbol{y}}}_y \mid 
\{\alpha_{i,y}\}_n, \{\gamma_{j,y}\}_n, 
\{\boldsymbol{x}_{ij}\}_{n,n}; 
\boldsymbol{\theta}_y\right)} = 
\frac{\exp\left(\sum_{(i,j) \in \mathcal{D}} 
\tilde{y}_{ij,y}\, \boldsymbol{x}_{ij}' 
\boldsymbol{\theta}_y\right)}{\sum_{\bar{\tilde{\boldsymbol{y}}}_y \in \mathbb{Q}} 
\exp\left(\sum_{(i,j) \in \mathcal{D}} 
\bar{\tilde{y}}_{ij,y}\, \boldsymbol{x}_{ij}' 
\boldsymbol{\theta}_y\right)} \nonumber
\end{align}

\noindent since for any 
$\bar{\tilde{\boldsymbol{y}}}_y \in \mathbb{Q}$, by 
construction, $\sum_{j \neq i} \bar{\tilde{y}}_{ij,y} = 
r_{i,y}$ and $\sum_{i \neq j} \bar{\tilde{y}}_{ij,y} = 
c_{j,y}$, so the terms involving $\alpha_{i,y}$ and 
$\gamma_{j,y}$ cancel between numerator and denominator. 
The resulting expression depends only on 
$\boldsymbol{\theta}_y$ and the covariates, confirming that 
$(\boldsymbol{r}_y, \boldsymbol{c}_y)$ are sufficient 
statistics for $(\{\alpha_{i,y}\}_n, \{\gamma_{j,y}\}_n)$.

\section{Identification}\label{identification}

The intuition for the identification of the common parameters is analogous to that of \cite{graham2017econometric} for the undirected case. The heterogeneity parameters (fixed effects) account for the in-degree and out-degree distributions of the network (the number of ones for a given node when the node is a sender or a receiver). Therefore, the precise location of the ones (or links) is driven by the variation provided by the covariates and the common parameters ($\boldsymbol{x}_{ij}'\boldsymbol{\theta}_y$). Thus, conditioning on the set $\{\tilde{y}_{ij,y} + \tilde{y}_{ik,y} = 1, \tilde{y}_{lj,y} + \tilde{y}_{lk,y} = 1,  \tilde{y}_{ij,y} + \tilde{y}_{lk,y} = 1\}$ provides the ground for an estimator that is based on the relative probability of different types of subgraphs configurations with identical degree sequences, giving the necessary variation to identify the common parameters. 

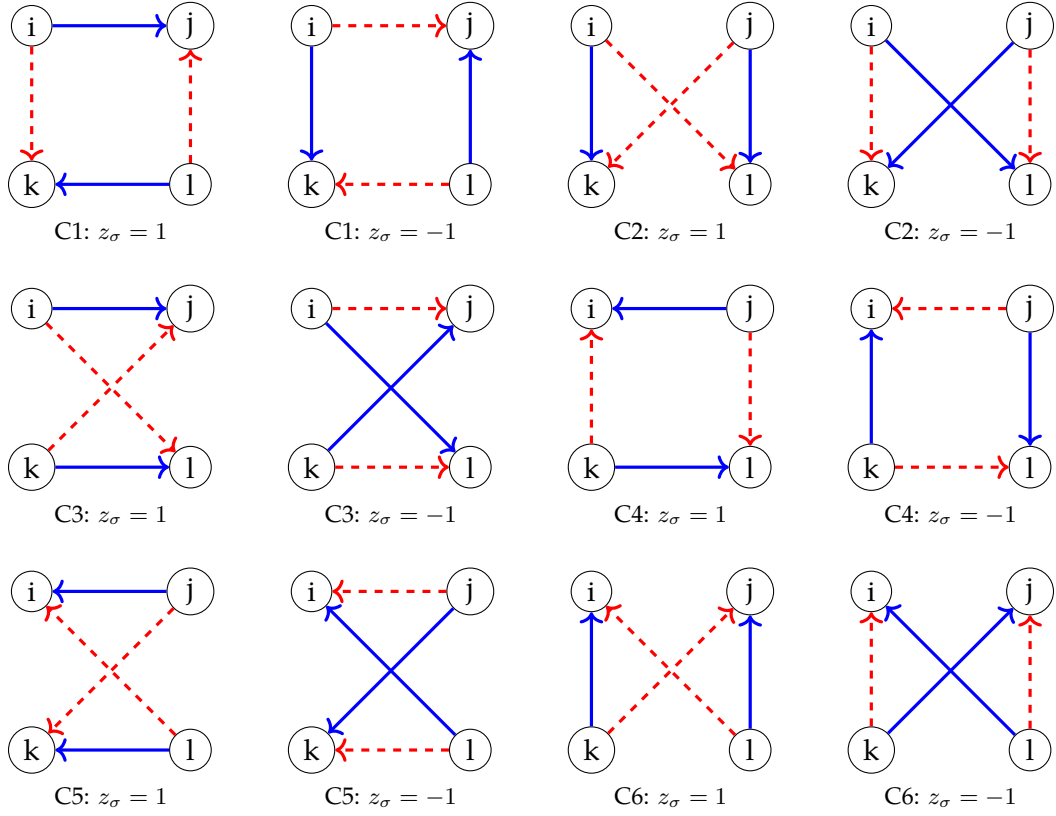
\begin{figure}
\centering
\scriptsize
\captionsetup[subfigure]{font=footnotesize, skip=0pt}

\begin{subfigure}[b]{0.22\textwidth}
\centering
\begin{tikzpicture}[scale=1.4]
\node[circle,draw,minimum size=0.5cm] (i) at (0,1.5) {\small i};
\node[circle,draw,minimum size=0.5cm] (j) at (1.5,1.5) {\small j};
\node[circle,draw,minimum size=0.5cm] (k) at (0,0) {\small k};
\node[circle,draw,minimum size=0.5cm] (l) at (1.5,0) {\small l};
\draw[->,thick,blue,line width=1.2pt] (i) -- (j);
\draw[->,thick,blue,line width=1.2pt] (l) -- (k);
\draw[->,dashed,red,line width=1.2pt] (i) -- (k);
\draw[->,dashed,red,line width=1.2pt] (l) -- (j);
\end{tikzpicture}
\caption*{C1: $z_\sigma = 1$}
\end{subfigure}
\begin{subfigure}[b]{0.22\textwidth}
\centering
\begin{tikzpicture}[scale=1.4]
\node[circle,draw,minimum size=0.5cm] (i) at (0,1.5) {\small i};
\node[circle,draw,minimum size=0.5cm] (j) at (1.5,1.5) {\small j};
\node[circle,draw,minimum size=0.5cm] (k) at (0,0) {\small k};
\node[circle,draw,minimum size=0.5cm] (l) at (1.5,0) {\small l};
\draw[->,dashed,red,line width=1.2pt] (i) -- (j);
\draw[->,dashed,red,line width=1.2pt] (l) -- (k);
\draw[->,thick,blue,line width=1.2pt] (i) -- (k);
\draw[->,thick,blue,line width=1.2pt] (l) -- (j);
\end{tikzpicture}
\caption*{C1: $z_\sigma = -1$}
\end{subfigure}
\begin{subfigure}[b]{0.22\textwidth}
\centering
\begin{tikzpicture}[scale=1.4]
\node[circle,draw,minimum size=0.5cm] (i) at (0,1.5) {\small i};
\node[circle,draw,minimum size=0.5cm] (j) at (1.5,1.5) {\small j};
\node[circle,draw,minimum size=0.5cm] (k) at (0,0) {\small k};
\node[circle,draw,minimum size=0.5cm] (l) at (1.5,0) {\small l};
\draw[->,thick,blue,line width=1.2pt] (i) -- (k);
\draw[->,thick,blue,line width=1.2pt] (j) -- (l);
\draw[->,dashed,red,line width=1.2pt] (i) -- (l);
\draw[->,dashed,red,line width=1.2pt] (j) -- (k);
\end{tikzpicture}
\caption*{C2: $z_\sigma = 1$}
\end{subfigure}
\begin{subfigure}[b]{0.22\textwidth}
\centering
\begin{tikzpicture}[scale=1.4]
\node[circle,draw,minimum size=0.5cm] (i) at (0,1.5) {\small i};
\node[circle,draw,minimum size=0.5cm] (j) at (1.5,1.5) {\small j};
\node[circle,draw,minimum size=0.5cm] (k) at (0,0) {\small k};
\node[circle,draw,minimum size=0.5cm] (l) at (1.5,0) {\small l};
\draw[->,dashed,red,line width=1.2pt] (i) -- (k);
\draw[->,dashed,red,line width=1.2pt] (j) -- (l);
\draw[->,thick,blue,line width=1.2pt] (i) -- (l);
\draw[->,thick,blue,line width=1.2pt] (j) -- (k);
\end{tikzpicture}
\caption*{C2: $z_\sigma = -1$}
\end{subfigure}

\vspace{0.4cm}

\begin{subfigure}[b]{0.22\textwidth}
\centering
\begin{tikzpicture}[scale=1.4]
\node[circle,draw,minimum size=0.5cm] (i) at (0,1.5) {\small i};
\node[circle,draw,minimum size=0.5cm] (j) at (1.5,1.5) {\small j};
\node[circle,draw,minimum size=0.5cm] (k) at (0,0) {\small k};
\node[circle,draw,minimum size=0.5cm] (l) at (1.5,0) {\small l};
\draw[->,thick,blue,line width=1.2pt] (i) -- (j);
\draw[->,thick,blue,line width=1.2pt] (k) -- (l);
\draw[->,dashed,red,line width=1.2pt] (i) -- (l);
\draw[->,dashed,red,line width=1.2pt] (k) -- (j);
\end{tikzpicture}
\caption*{C3: $z_\sigma = 1$}
\end{subfigure}
\begin{subfigure}[b]{0.22\textwidth}
\centering
\begin{tikzpicture}[scale=1.4]
\node[circle,draw,minimum size=0.5cm] (i) at (0,1.5) {\small i};
\node[circle,draw,minimum size=0.5cm] (j) at (1.5,1.5) {\small j};
\node[circle,draw,minimum size=0.5cm] (k) at (0,0) {\small k};
\node[circle,draw,minimum size=0.5cm] (l) at (1.5,0) {\small l};
\draw[->,dashed,red,line width=1.2pt] (i) -- (j);
\draw[->,dashed,red,line width=1.2pt] (k) -- (l);
\draw[->,thick,blue,line width=1.2pt] (i) -- (l);
\draw[->,thick,blue,line width=1.2pt] (k) -- (j);
\end{tikzpicture}
\caption*{C3: $z_\sigma = -1$}
\end{subfigure}
\begin{subfigure}[b]{0.22\textwidth}
\centering
\begin{tikzpicture}[scale=1.4]
\node[circle,draw,minimum size=0.5cm] (i) at (0,1.5) {\small i};
\node[circle,draw,minimum size=0.5cm] (j) at (1.5,1.5) {\small j};
\node[circle,draw,minimum size=0.5cm] (k) at (0,0) {\small k};
\node[circle,draw,minimum size=0.5cm] (l) at (1.5,0) {\small l};
\draw[->,thick,blue,line width=1.2pt] (j) -- (i);
\draw[->,thick,blue,line width=1.2pt] (k) -- (l);
\draw[->,dashed,red,line width=1.2pt] (j) -- (l);
\draw[->,dashed,red,line width=1.2pt] (k) -- (i);
\end{tikzpicture}
\caption*{C4: $z_\sigma = 1$}
\end{subfigure}
\begin{subfigure}[b]{0.22\textwidth}
\centering
\begin{tikzpicture}[scale=1.4]
\node[circle,draw,minimum size=0.5cm] (i) at (0,1.5) {\small i};
\node[circle,draw,minimum size=0.5cm] (j) at (1.5,1.5) {\small j};
\node[circle,draw,minimum size=0.5cm] (k) at (0,0) {\small k};
\node[circle,draw,minimum size=0.5cm] (l) at (1.5,0) {\small l};
\draw[->,dashed,red,line width=1.2pt] (j) -- (i);
\draw[->,dashed,red,line width=1.2pt] (k) -- (l);
\draw[->,thick,blue,line width=1.2pt] (j) -- (l);
\draw[->,thick,blue,line width=1.2pt] (k) -- (i);
\end{tikzpicture}
\caption*{C4: $z_\sigma = -1$}
\end{subfigure}

\vspace{0.4cm}

\begin{subfigure}[b]{0.22\textwidth}
\centering
\begin{tikzpicture}[scale=1.4]
\node[circle,draw,minimum size=0.5cm] (i) at (0,1.5) {\small i};
\node[circle,draw,minimum size=0.5cm] (j) at (1.5,1.5) {\small j};
\node[circle,draw,minimum size=0.5cm] (k) at (0,0) {\small k};
\node[circle,draw,minimum size=0.5cm] (l) at (1.5,0) {\small l};
\draw[->,thick,blue,line width=1.2pt] (j) -- (i);
\draw[->,thick,blue,line width=1.2pt] (l) -- (k);
\draw[->,dashed,red,line width=1.2pt] (j) -- (k);
\draw[->,dashed,red,line width=1.2pt] (l) -- (i);
\end{tikzpicture}
\caption*{C5: $z_\sigma = 1$}
\end{subfigure}
\begin{subfigure}[b]{0.22\textwidth}
\centering
\begin{tikzpicture}[scale=1.4]
\node[circle,draw,minimum size=0.5cm] (i) at (0,1.5) {\small i};
\node[circle,draw,minimum size=0.5cm] (j) at (1.5,1.5) {\small j};
\node[circle,draw,minimum size=0.5cm] (k) at (0,0) {\small k};
\node[circle,draw,minimum size=0.5cm] (l) at (1.5,0) {\small l};
\draw[->,dashed,red,line width=1.2pt] (j) -- (i);
\draw[->,dashed,red,line width=1.2pt] (l) -- (k);
\draw[->,thick,blue,line width=1.2pt] (j) -- (k);
\draw[->,thick,blue,line width=1.2pt] (l) -- (i);
\end{tikzpicture}
\caption*{C5: $z_\sigma = -1$}
\end{subfigure}
\begin{subfigure}[b]{0.22\textwidth}
\centering
\begin{tikzpicture}[scale=1.4]
\node[circle,draw,minimum size=0.5cm] (i) at (0,1.5) {\small i};
\node[circle,draw,minimum size=0.5cm] (j) at (1.5,1.5) {\small j};
\node[circle,draw,minimum size=0.5cm] (k) at (0,0) {\small k};
\node[circle,draw,minimum size=0.5cm] (l) at (1.5,0) {\small l};
\draw[->,thick,blue,line width=1.2pt] (k) -- (i);
\draw[->,thick,blue,line width=1.2pt] (l) -- (j);
\draw[->,dashed,red,line width=1.2pt] (k) -- (j);
\draw[->,dashed,red,line width=1.2pt] (l) -- (i);
\end{tikzpicture}
\caption*{C6: $z_\sigma = 1$}
\end{subfigure}
\begin{subfigure}[b]{0.22\textwidth}
\centering
\begin{tikzpicture}[scale=1.4]
\node[circle,draw,minimum size=0.5cm] (i) at (0,1.5) {\small i};
\node[circle,draw,minimum size=0.5cm] (j) at (1.5,1.5) {\small j};
\node[circle,draw,minimum size=0.5cm] (k) at (0,0) {\small k};
\node[circle,draw,minimum size=0.5cm] (l) at (1.5,0) {\small l};
\draw[->,dashed,red,line width=1.2pt] (k) -- (i);
\draw[->,dashed,red,line width=1.2pt] (l) -- (j);
\draw[->,thick,blue,line width=1.2pt] (k) -- (j);
\draw[->,thick,blue,line width=1.2pt] (l) -- (i);
\end{tikzpicture}
\caption*{C6: $z_\sigma = -1$}
\end{subfigure}

\caption{All 12 informative quadruple configurations for 
nodes \{i,j,k,l\}. Blue solid arrows: links present 
($\tilde{y}=1$). Red dashed: links absent ($\tilde{y}=0$). 
Each pair corresponds to a sender-receiver assignment: 
C1: \{i,l\} send to \{j,k\}; C2: \{i,j\} send to \{k,l\}; 
C3: \{i,k\} send to \{j,l\}; C4: \{j,k\} send to \{i,l\}; 
C5: \{j,l\} send to \{i,k\}; C6: \{k,l\} send to \{i,j\}. 
Within each assignment, the two senders connect to opposite 
receivers, yielding $z_\sigma \in \{-1, 1\}$.}
\label{figure2:correct}
\end{figure}

For instance, assuming that only the links represented by 
the solid blue lines are present in Figure 
\ref{figure2:correct}, the different configurations provide 
the same contribution of the unobserved heterogeneity to 
the likelihood, such that the conditional frequency to 
which each is observed depends only on the variation given 
by the covariates associated with each. In other words, in 
conditioning on the degree sequences of quadruples (since 
they are the same in all subgraphs, given that only the 
links represented by the solid blue lines are present), the 
only variation is the location of the links. This intuition 
aligns with Lemma~\ref{lemma_sufficient}: the sums across 
each dimension are sufficient statistics for the fixed 
effects. At the same time, the conditioning events guarantee 
that for a node there is variation in the outcomes such that 
the common parameters can be identified. This feature cannot 
be seen in Figure~\ref{Figure3:exampletetrads}, where none 
of the configurations are informative: in (a) all links are 
absent, in (b) all links are present, and in (c) both 
senders connect to the same receiver. In all three cases 
$z_\sigma = 0$, so that no variation remains after 
conditioning.
\section{Proofs of Section \ref{asymptotics}}\label{appendix_asymptotics} \label{appendix_derivation}

In the following, for the sake of simplification of notation, I omit the subscript $y$ that denotes the threshold of the outcome variable. The proofs in this appendix follow the broad structure of \cite{jochmans2018semiparametric} for the estimator of \cite{charbonneau2017multiple}, with more detailed derivations (particularly of the variance order) and modifications that facilitate the extension to the joint asymptotic distribution across thresholds in Section \ref{joint_distribution}. References to \cite{jochmans2018semiparametric} below highlight specific points where my derivation departs from or extends theirs.

\subsection{Proof of Theorem \ref{theorem1}}

Consider the limit of the objective function: 
$$ \lim_{n \xrightarrow[]{} \infty} (m_n p_n)^{-1}  \mathbb{E}(L_n (\boldsymbol{\theta})), $$
where $L_n(\bo \theta)=\sum_{\sigma \in \mathcal{N}_{m_n}} 1\left\{z_\sigma=1\right\}
\log F\left(\bo r_\sigma^{\prime} \bo \theta\right)+1\left\{z_\sigma=-1\right\} \log \left(1-F\left(\bo r_\sigma^{\prime} \bo \theta\right)\right)$. From Assumption 4, $\bo \theta_0$ is the unique global maximizer of the limit quantity above. This follows due the rank condition, which ensures that the Hessian at $\bo \theta_0$ is negative definite, combined with the concavity of the function. Using Theorem 2.7 in \citet{newey1994chapter}, since the function is concave and has a unique maximizer, the consistency of the estimator, $\bo \theta_n \xrightarrow{p} \bo \theta_0$, will follow from pointwise convergence in probability of the normalized objective function, $(m_n^*)^{-1} L_n(\bo \theta)$, to $(m_n p_n)^{-1} \mathbb{E}(L_n(\bo \theta))$. 

To show that this is the case, start by considering the log-likelihood function:
$$
L_n(\bo \theta)=\sum_{\sigma \in \mathcal{N}_{m_n}} \ell_\sigma(\bo \theta),
$$
where $\ell_\sigma(\bo \theta)$ denotes the log-likelihood contribution of quadruple $\sigma$, which from the above, is given by:
$$\ell_\sigma(\bo \theta)=
 1\left\{z_\sigma=1\right\}
\log F\left(\bo r_\sigma^{\prime} \bo \theta\right)+1\left\{z_\sigma=-1\right\} \log \left(1-F\left(\bo r_\sigma^{\prime} \bo \theta\right)\right)$$
Remembering that $m_n^* = \sum_{\sigma \in \mathcal{N}_{m_n}} 1 \{z_\sigma \in\{-1,1\}\}$, $m_n$ is the total number of quadruples, and $
p_n=\frac{E\left(m_n^*\right)}{m_n}$, then:
$$
\frac{L_n(\bo \theta)}{m_n^*}-\frac{\mathbb{E}\left(L_n(\bo \theta)\right)}{\mathbb{E}\left(m_n^*\right)}=\frac{\sum_{\sigma \in \mathcal{N}_m} \ell_\sigma(\bo \theta)-\mathbb{E}\left(\ell_\sigma(\bo \theta)\right)}{\mathbb{E}\left(m_n^*\right)}+\frac{\sum_{\sigma \in \mathcal{N}_m} \ell_\sigma(\bo \theta)}{\mathbb{E}\left(m_n^*\right)}\left(\frac{\mathbb{E}\left(m_n^*\right)}{m_n^*}-1\right)
$$
The proof proceeds with showing that each of the right-hand side terms converges to zero in probability.\\
\\
\textbf{First term.} Consider the different cases, for $z_\sigma = 1$:
$$ \operatorname{Pr} (z_\sigma = 1 \mid z_\sigma \in\{-1,1\}, \bo x_\sigma, \left\{\alpha_i, \gamma_i\right\}_n) = \frac{1}{1 + \exp(-\bo r_\sigma'\bo \theta)}$$
$$ \ell_\sigma (\bo \theta) = - \log (1 + \exp (-\bo r_\sigma'\bo \theta))$$
And, for $z_\sigma = -1$:
$$ \operatorname{Pr} (z_\sigma = -1 \mid z_\sigma \in\{-1,1\}, \bo x_\sigma, \left\{\alpha_i, \gamma_i\right\}_n) = \frac{\exp(-\bo r_\sigma'\bo \theta)}{1 + \exp(-\bo r_\sigma'\bo \theta)}$$
$$ \ell_\sigma (\bo \theta) = \log (\exp (-\bo r_\sigma'\bo \theta)) - \log (1 + \exp (-\bo r_\sigma'\bo \theta)) = - \bo r_\sigma'\bo \theta - \log (1 + \exp (-\bo r_\sigma'\bo \theta))$$
I proceed with finding bounds for each case, such that, for $z_\sigma = 1$, $|\ell_\sigma (\bo \theta)| \leq |\bo r_\sigma'\bo \theta| + \log(2)$ and, for $z_\sigma = -1$, $|\ell_\sigma (\bo \theta)| \leq 2|\bo r_\sigma'\bo \theta| + \log(2)$ by the triangle inequality. Combining both cases, by an application of Cauchy-Schwarz inequality, it follows that: $|\ell_\sigma (\bo \theta)| \leq 2||\bo r_\sigma||\hspace{0.05cm} ||\theta|| + \log(2)$. Since $\mathbb{E}(||\bo r_\sigma||^2)$ is finite, following Assumption \ref{assumption3}, and $\Theta$ is compact, the variance of $\ell_\sigma$ exists and is uniformly bounded in $\sigma$, since: $ \mathbb{E}(\ell_\sigma^2(\bo \theta)) \leq (\log(2))^2 + 4 \log(2) \mathbb{E}(||\bo r_\sigma||) \hspace{0.05cm} ||\bo \theta|| + 4 \mathbb{E}(||\bo r_\sigma||^2) \hspace{0.05cm} ||\bo \theta||^2$. \\
\\
Therefore, by Chebyshev's inequality, it holds that, for any $\varepsilon > 0$, 
$$
\operatorname{Pr}\left(\left|\frac{\sum_{\sigma \in \mathcal{N}_{m_n}} \ell_\sigma(\bo \theta)-\mathbb{E}\left(\ell_\sigma(\bo \theta)\right)}{\mathbb{E}\left(m_n^*\right)}\right|>\varepsilon\right) \leq \frac{1}{\varepsilon^2} \frac{\mathbb{E}\left(\left|\sum_{\sigma \in \mathcal{N}_{m_n}} \ell_\sigma(\bo \theta)-\mathbb{E}\left(\ell_\sigma(\bo \theta)\right)\right|^2\right)}{\mathbb{E}\left(m_n^*\right)^2},
$$
for each $\bo \theta \in \Theta$. The next step is to focus on the term $\mathbb{E}\left(\left|\sum_{\sigma \in \mathcal{N}_{m_n}} \ell_\sigma(\bo \theta)-E\left(\ell_\sigma(\bo \theta)\right)\right|^2\right)$, by expanding the squared sum to analyse the variance and covariances structures, where the uniform bound on $|\ell_\sigma (\bo\theta)|$ becomes crucial. Now, it follows that $\mathbbm{E} \left( \left \rvert  \sum_{\sigma \in \mathcal{N}_{m_n}} \ell_\sigma(\boldsymbol{\theta}) - \mathbbm{E}(\ell_\sigma(\boldsymbol{\theta}))  \right \rvert^2 \right)$ expands to $\sum_{\sigma \in \mathcal{N}_{m_n}} \sum_{\sigma' \in \mathcal{N}_{m_n}} \operatorname{Cov} (\ell_\sigma(\bo \theta), \ell_{\sigma'}(\bo \theta))$, that is:
\begin{adjustwidth}{-0.25in}{-0.25in}
$$ \mathbbm{E} \left( \left \rvert  \sum_{\sigma \in \mathcal{N}_{m_n}} \ell_\sigma(\boldsymbol{\theta}) - \mathbbm{E}(\ell_\sigma(\boldsymbol{\theta}))  \right \rvert^2 \right) = \mathbbm{E} \left( \left(  \sum_{\sigma \in m_n} \ell_\sigma(\boldsymbol{\theta}) - \mathbbm{E}(\ell_\sigma(\boldsymbol{\theta}))  \right) \left(  \sum_{\sigma' \in m_n} \ell_{\sigma'}(\boldsymbol{\theta}) - \mathbbm{E}(\ell_{\sigma'}(\boldsymbol{\theta}))  \right) \right). $$
\end{adjustwidth}
Note that a pair of quadruples $\sigma = \sigma\{i,l;j,k\}$ and $\sigma' = \sigma \{i',l';j',k'\}$ only contributes to the covariance if they share at least one node in common, i.e., a pair of quadruples with only distinct nodes are independent by Assumption \ref{assumption1}. 

More specifically, the product above contains $O(n^8)$ terms. Consider $\sigma$ and $\sigma'$ that share indices in common, and therefore, having dependence and contributing to the covariance. If they share the index $i = i'$, there are $O(n^2)$ choices for $l$ and $l'$ to be distinct, and $O(n^4)$ choices of $j,j',k,k'$ to be distinct. Since there are $n$ different combinations for $i=i'$, the number of quadruples that share at least one node in common is of the order $O(n^7)$ (moreover, the quadruples that share two or more nodes in common are of smaller order than the ones sharing one node in common). Put differently, given the $n$ choices for the common node $i = i'$, there are $O(n^3)$ different ways of completing the first quadruple, and $O(n^3)$ different ways of completing the second quadruple, yielding $O(n^7)$ pairs in total. This also implies that for a fixed quadruple $\sigma = \sigma\{i,l;j,k\}$ there are $O(n^3)$ other quadruples sharing at least one index. 

By indicating the quadruples that share common indices with $1 \{\sigma \cap \sigma' \neq \emptyset \}$, and an application of the Cauchy-Schwarz inequality:
\begin{align}
    \mathbbm{E} &\left( \left(  \sum_{\sigma \in \mathcal{N}_{m_n}} \ell_\sigma(\boldsymbol{\theta}) - \mathbbm{E}(\ell_\sigma(\boldsymbol{\theta}))  \right) \left(  \sum_{\sigma' \in \mathcal{N}_{m_n}} \ell_{\sigma'}(\boldsymbol{\theta}) - \mathbbm{E}(\ell_{\sigma'}(\boldsymbol{\theta}))  \right) \right) \nonumber \\
    &= \sum_{\sigma \in \mathcal{N}_{m_n}} \sum_{\sigma' \in \mathcal{N}_{m_n}}\mathbbm{E} \left( \left( \ell_\sigma(\boldsymbol{\theta}) - \mathbbm{E}(\ell_\sigma(\boldsymbol{\theta}))  \right) \left( \ell_{\sigma'}(\boldsymbol{\theta}) - \mathbbm{E}(\ell_{\sigma'}(\boldsymbol{\theta}))  \right) \right) \nonumber \\
    &= \sum_{\sigma \in \mathcal{N}_{m_n}} \sum_{\sigma' \in \mathcal{N}_{m_n}} 1 \{\sigma \cap \sigma' \neq \emptyset \} \mathbbm{E} \left( \left( \ell_\sigma(\boldsymbol{\theta}) - \mathbbm{E}(\ell_\sigma(\boldsymbol{\theta}))  \right) \left( \ell_{\sigma'}(\boldsymbol{\theta}) - \mathbbm{E}(\ell_{\sigma'}(\boldsymbol{\theta}))  \right) \right) \nonumber \\
    &= \sum_{\sigma \in \mathcal{N}_{m_n}} \sum_{\sigma^{\prime} \in \mathcal{N}_{m_n}} 1\left\{\sigma \cap \sigma^{\prime} \neq \emptyset\right\} \operatorname{Cov}\left(\ell_\sigma(\bo \theta), \ell_{\sigma^{\prime}}(\bo \theta)\right) \nonumber \\
    &\leq \sum_{\sigma \in \mathcal{N}_{m_n}} \sum_{\sigma^{\prime} \in \mathcal{N}_{m_n}} 1\left\{\sigma \cap \sigma^{\prime} \neq \emptyset\right\}\left|\operatorname{Cov}\left(\ell_\sigma(\bo \theta), \ell_{\sigma^{\prime}}(\bo \theta)\right)\right|  \nonumber \\
    &\leq \sum_{\sigma \in \mathcal{N}_{m_n}} \sum_{\sigma' \in \mathcal{N}_{m_n}} 1 \{\sigma \cap \sigma' \neq \emptyset \} \sqrt{\text{Var} (\ell_\sigma(\boldsymbol{\theta}))} \sqrt{\text{Var} (\ell_{\sigma'}(\boldsymbol{\theta}))} \nonumber \\
    &\leq \sum_{\sigma \in \mathcal{N}_{m_n}} \sqrt{\text{Var} (\ell_\sigma(\boldsymbol{\theta}))}\sum_{\sigma' \in \mathcal{N}_{m_n}} 1 \{\sigma \cap \sigma' \neq \emptyset \} \sqrt{\text{Var} (\ell_{\sigma'}(\boldsymbol{\theta}))} \nonumber \\
    &= \sum_{\sigma \in \mathcal{N}_{m_n}} O(\sqrt{p_\sigma}) \sum_{\sigma' \in \mathcal{N}_{m_n}} 1 \{\sigma \cap \sigma' \neq \emptyset \} O(\sqrt{p_{\sigma'}})  \nonumber \\
    &= O(n^3 m_n p_n) = O(n^7 p_n),
\end{align}

\noindent The second to last equality follows from the fact that, when considering the variance term, $\text{Var} (\ell_\sigma(\boldsymbol{\theta})) = \mathbb{E}(\ell_\sigma^2(\bo \theta)) - \mathbb{E}(\ell_\sigma(\bo \theta))^2$, $\ell_\sigma = 0$, when a quadruple is not informative, and hence $\mathbb{E}(\ell_\sigma^2(\bo \theta)) = \operatorname{Pr} (z_\sigma \in \{-1,1\}) \mathbb{E} (\ell_\sigma^2(\bo \theta) \mid z_\sigma \in \{-1,1\})$, and similarly for $\mathbb{E}(\ell_\sigma(\bo \theta))$. Since the conditional expectations are uniformly bounded, by considering a finite constant $B$, $\text{Var} (\ell_\sigma(\boldsymbol{\theta})) \leq p_\sigma B$, where $p_\sigma = \operatorname{Pr} (z_\sigma \in \{-1,1\})$ (conditional on fixed effects). Finally, remembering that  $\sum_{\sigma \in \mathcal{N}_{m_n}} p_\sigma = \mathbb{E}(m_n^*) = m_n p_n$, the final equality follows by considering the two sums separately.

First, focusing on the term $\sum_{\sigma \in \mathcal{N}_{m_n}} \sqrt{\operatorname{Var}\left(\ell_\sigma\left(\boldsymbol{\theta}\right)\right)}=\sum_{\sigma \in \mathcal{N}_{m_n}} O\left(\sqrt{p_\sigma}\right)$. By Jensen's inequality (since square root is concave) $\frac{1}{m_n} \sum_\sigma \sqrt{p_\sigma} \leq \sqrt{\frac{1}{m_n} \sum_\sigma p_\sigma}=\sqrt{p_n}$, and therefore $\sum_{\sigma \in \mathcal{N}_{m n}} \sqrt{p_\sigma} \leq m_n \sqrt{p_n}=O\left(n^4 \sqrt{p_n}\right)$ .

For the term $\sum_{\sigma^{\prime} \in \mathcal{N}_{m n}} 1 \left\{\sigma \cap \sigma^{\prime} \neq \emptyset\right\} \sqrt{\operatorname{Var}\left(\ell_{\sigma^{\prime}}\left(\boldsymbol{\theta}\right)\right)}$, note that for any fixed $\sigma$, there are $O\left(n^3\right)$ quadruples $\sigma^{\prime}$ sharing a node with $\sigma$. Then, similarly, applying Jensen's inequality to this subset: $\sum_{\sigma^{\prime} \in \mathcal{N}_{m n}} 1\left\{\sigma \cap \sigma^{\prime} \neq \emptyset\right\} \sqrt{p_{\sigma^{\prime}}}=O\left(n^3 \sqrt{p_n}\right)$. Which leads to the desired result.

Another way to see this result is that, the variance decomposes as:
$$
\mathbb{E}\left(\left|\sum_{\sigma \in \mathcal{N}_{m_n}} \ell_\sigma(\bo \theta)-\mathbb{E}\left(\ell_\sigma(\bo \theta)\right)\right|^2\right)=\sum_{\sigma \in \mathcal{N}_{m_n}} \operatorname{Var}\left(\ell_\sigma(\bo \theta)\right)+\sum_{\sigma \neq \sigma^{\prime}} \operatorname{Cov}\left(\ell_\sigma(\bo \theta), \ell_{\sigma^{\prime}}(\bo \theta)\right)
$$
The first term (variances) contributes $O\left(m_n p_n\right)$ since $\sum_\sigma p_\sigma=m_n p_n$.
The second term (covariances) is non-zero only for dependent quadruples. With $O\left(n^7\right)$ pairs sharing nodes and typical covariance of order $O\left(p_n\right)$, this contributes $O\left(n^7 p_n\right)=O\left(n^3 m_n p_n \right)$.

Moreover, as $\mathbbm{E}(m^*_n) = m_n p_n$ and $m_n = O(n^4)$, it follows that:
$$ \frac{\mathbbm{E} \left( \left \rvert  \sum_{\sigma \in \mathcal{N}_{m_n}} \ell_\sigma(\boldsymbol{\theta}) - \mathbbm{E}(\ell_\sigma(\boldsymbol{\theta}))  \right \rvert^2 \right)}{\mathbbm{E}(m^*_n)^2} = O\left( \frac{1}{n p_n}\right), $$
which converges to zero by Assumption \ref{assumption4}, as $n \xrightarrow[]{} \infty$. Therefore:
$$ \lim_{n \xrightarrow[]{} \infty} \text{Pr} \left( \left \lvert \frac{\sum_{\sigma \in \mathcal{N}_{m_n}} \ell_\sigma(\boldsymbol{\theta}) - \mathbbm{E}(\ell_\sigma(\boldsymbol{\theta}))}{\mathbbm{E}(m^*_n)} \right \lvert > \varepsilon \right) \leq \frac{1}{\varepsilon^2} \frac{\mathbbm{E} \left( \left \rvert  \sum_{\sigma \in \mathcal{N}_{m_n}} \ell_\sigma(\boldsymbol{\theta}) - \mathbbm{E}(\ell_\sigma(\boldsymbol{\theta}))  \right \rvert^2 \right)}{\mathbbm{E}(m^*_n)^2} = 0 $$
for any $\varepsilon > 0$ and all $\boldsymbol{\theta} \in \Theta$.\\
\\
\textbf{Second term.} For the second term, the summands are bounded uniformly in $\sigma$ and do not depend on $\boldsymbol{\theta}$. Following the same arguments as before, it is easy to verify that $ \left( \frac{m^*_n}{m_n} - p_n \right) \xrightarrow[]{p} 0$, and therefore, $\frac{m^*_n}{\mathbbm{E}(M^*_n)} \xrightarrow[]{p} 1$. The first step to verify this statement is to apply the Chebyshev's inequality:
$$ \operatorname{Pr} \left( \left \lvert \frac{m_n^*}{m_n} - p_n \right \lvert > \varepsilon \right) \leq \frac{\operatorname{Var}(m_n^*/m_n)}{\varepsilon^2} $$
which holds for any $\varepsilon > 0$. It is important to reiterate here that $p_n$ is defined as: $p_n = \frac{\mathbb{E}(m_n^*)}{m_n} = \frac{1}{m_n} \sum_{\sigma \in \mathcal{N}_{m_n}} \operatorname{Pr} (z_\sigma \in \{-1,1\})$, since $m_n^* = \sum_{\sigma \in \mathcal{N}_{m_n}} 1 \{z_\sigma \in \{-1,1\}\}$. Moreover, the variance term in the numerator of the right-hand side term is defined as:
$$ \operatorname{Var} \left( \frac{m_n^*}{m_n}\right) = \frac{1}{m_n^2} \operatorname{Var} (m_n^*) $$ 
where
$$ \operatorname{Var} (m_n^*) = \sum_{\sigma \in \mathcal{N}_{m_n}} \operatorname{Var} ( 1 \{z_\sigma \in \{-1,1\}\}) + \sum_{\sigma \neq \sigma'} \operatorname{Cov} (1 \{z_\sigma \in \{-1,1\}\}, 1 \{z_\sigma' \in \{-1,1\}\})$$
Note that the variance terms are uniformly bounded, since $\operatorname{Var} ( 1 \{z_\sigma \in \{-1,1\}\}) = p_\sigma (1-p_\sigma)$, and similarly for the covariance terms. Similarly to before, quadruples that are independent have contribution zero to the covariance, and there are $O(n^7)$ dependent quadruples. Therefore, $\operatorname{Var} (m_n^*) = O(m_n p_n) + O(n^3 m_n p_n) = O(n^3 m_n p_n)$. Thus,
$$ \operatorname{Var} \left( \frac{m_n^*}{m_n}\right) = \frac{O(p_n)}{n} \xrightarrow{}0$$
due to Assumption \ref{assumption4}. Such that:
$$ \lim_{n \xrightarrow[]{} \infty} \operatorname{Pr} \left( \left \lvert \frac{m_n^*}{m_n} - p_n \right \lvert > \varepsilon \right) \leq \frac{\operatorname{Var}(m_n^*/m_n)}{\varepsilon^2} $$
for any $\varepsilon >0$, which completes the proof. The same variance bound gives $\operatorname{Var}(m_n^*/(m_n p_n)) = O(1)/(n p_n) \xrightarrow{} 0$, and since $\mathbb{E}(m_n^*/(m_n p_n)) = 1$, it also follows that $m_n^*/(m_n p_n) \xrightarrow{p} 1$, which is used in the proof of Theorem \ref{theorem2}.

\subsection{Proof of Theorem 2}

The proof proceeds in 4 steps:

\noindent \textbf{Step 1:} Show that the score vector, evaluated at the true parameter $\bo \theta_0$, is asymptotically equivalent to its projection.

\noindent \textbf{Step 2:} Obtain the limit distribution of the projection and express it under the empirical normalization by $\boldsymbol{\Upsilon}_n(\bo \theta_0)$.

\noindent \textbf{Step 3:} Prove that the Hessian of the conditional likelihood, normalized by $m_n^*$ converges to a well behaved limit uniformly on $\Theta$.

\noindent \textbf{Step 4:} Collect results and combine them with a mean-value expansion of the first-order condition around the true value, to obtain the limit distribution of the estimator $\bo \theta_n$.

Remembering that the score vector given by:
\begin{align} \label{score_rew}
    \boldsymbol{S}_n(\boldsymbol{\theta}) = \sum_{i}^n \sum_{j \neq i} \sum_{i' \neq i, j} \sum_{j' \neq i,j,i'} \underbrace{\frac{\partial}{\partial \boldsymbol{\theta}} \ell_\sigma(\boldsymbol{\theta})}_{\boldsymbol{s}_\sigma},
\end{align}
where, as before,
\begin{align} \label{summand_score}
    \bo s(\sigma\{i,i';j,j'\}; \boldsymbol{\theta}) = \boldsymbol{r}_\sigma \{ 1 \{ z_\sigma = 1 \} (1 - \Lambda(\boldsymbol{r}_\sigma' \boldsymbol{\theta})) - 1 \{ z_\sigma = -1 \} \Lambda(\boldsymbol{r}_\sigma' \boldsymbol{\theta}) \}
\end{align}
This sum counts all $m_n=n(n-1)(n-2)(n-3)$ ordered quadruples without imposing permutation invariance of sender-receiver roles, which differs from \citet{jochmans2018semiparametric}.\\
\\
\noindent \textbf{Step 1.} \\
By defining the information set $\mathcal{F}_n = \{\{\boldsymbol{x}_{ij}\}_{n,n}, \{\alpha_{i}\}_n, \{\gamma_{j}\}_n\}$, the projection of the score, evaluated at the true parameter, is given by:
\begin{align} \label{hajek}
    \boldsymbol{V}_n (\boldsymbol{\theta}_{0}) &=\sum_{i}^n \sum_{j \neq i} \sum_{i' \neq i, j} \sum_{j' \neq i,j,i'} \mathbbm{E}(\boldsymbol{s}(\sigma \{i,i';j,j'\}; \boldsymbol{\theta}_{0}) \mid \tilde{y}_{ij}, \mathcal{F}_n) \nonumber \\
    &= \sum_{i} \sum_{j\neq i} \boldsymbol{v}_{ij}(\boldsymbol{\theta}_{0})
\end{align}
where $\boldsymbol{v}_{ij}(\boldsymbol{\theta}_{0}) = \sum_{i' \neq i,j} \sum_{j' \neq i,j,i'} \mathbbm{E}(\boldsymbol{s}(\sigma\{i,i';j,j'\}; \boldsymbol{\theta}_{0}) \mid \tilde{y}_{ij}, \mathcal{F}_n)$.
Note that while this projection resembles a H\'{a}jek projection, as in  \citep{serfling2009approximation}, it is not formally one, since (i) the kernel is not symmetric in the arguments, and (ii) the conditioning terms are more involved - since they are not only related to the dyad $i,j$, but also to node-specific characteristis ($\alpha_i$, $\gamma_j$).

While \citet{jochmans2018semiparametric} derives asymptotic normality by explicitly computing the projection, I follow an approach inspired by \citet{graham2017econometric}, establishing asymptotic equivalence through tools from U-statistic theory without detailed projection calculations, leveraging on the conditional independence structure of link formation.
 
Before moving to the next steps of the proof, I present some intermediate results that will become relevant.\\
\\
\noindent \textit{Intermediate result 1.} \\
Note that $\mathbbm{E}(\boldsymbol{s}(\sigma;\boldsymbol{\theta}_{0})) = 0$, since, from sufficiency, and by defining $\mathcal{F}_\sigma$ to be the collection of covariates and fixed effects for the nodes in the quadruple $\sigma$, it follows that:
$$ \operatorname{Pr}(z_\sigma = 1 \mid \mathcal{F}_\sigma) = \Lambda(\boldsymbol{r}_\sigma'\boldsymbol{\theta}_{0}) Pr(z_\sigma \{-1,1\} \mid \mathcal{F}_\sigma) $$
$$ \operatorname{Pr}(z_\sigma = -1 \mid \mathcal{F}_\sigma) = (1-\Lambda(\boldsymbol{r}_\sigma'\boldsymbol{\theta}_{0})) Pr(z_\sigma \{-1,1\} \mid \mathcal{F}_\sigma) $$
Then,
\begin{align}
    \mathbbm{E}[\boldsymbol{s}(\sigma,\boldsymbol{\theta}_{0})] &= \mathbbm{E}\left[ \mathbbm{E} [ 1 \{z_\sigma = 1\} (1-\Lambda(\boldsymbol{r}_\sigma'\boldsymbol{\theta}_{0}))\boldsymbol{r}_\sigma' -  1 \{z_\sigma = -1\}  \Lambda(\boldsymbol{r}_\sigma'\boldsymbol{\theta}_{0}) \boldsymbol{r}_\sigma' \mid \mathcal{F}_\sigma  ]\right] \nonumber \\
    &= \mathbbm{E}\left[ \mathbbm{E} [ 1 \{z_\sigma = 1\} \mid \mathcal{F}_\sigma] (1-\Lambda(\boldsymbol{r}_\sigma'\boldsymbol{\theta}_{0}))\boldsymbol{r}_\sigma' - \mathbbm{E} [ 1 \{z_\sigma = -1\} \mid \mathcal{F}_\sigma ] \Lambda(\boldsymbol{r}_\sigma'\boldsymbol{\theta}_{0}) \boldsymbol{r}_\sigma'\right] \nonumber \\
    &= 0.
\end{align}
\noindent \textit{Intermediate result 2.} \\
From the first intermediate result, it also follows that:
\begin{align}
    \mathbbm{E} [\boldsymbol{v}_{ij}(\boldsymbol{\theta}_{0})] &= \sum_{i' \neq i,j} \sum_{j' \neq i,j,i'} \mathbbm{E} \left[ \mathbbm{E}[\boldsymbol{s}(\sigma\{i,i';j,j'\}, \boldsymbol{\theta}_{0}) \mid \tilde{y}_{ij}, \mathcal{F}_n] \right] \nonumber \\
    &= \sum_{i' \neq i,j} \sum_{j' \neq i,j,i'} \mathbbm{E} [\boldsymbol{s}(\sigma\{i,i';j,j'\}, \boldsymbol{\theta}_{0})] =0
\end{align}
Importantly, from this result, it also follows that $\mathbbm{E} [\boldsymbol{V}_n (\boldsymbol{\theta}_{0})] = 0$. \\
\\
\noindent \textit{Intermediate result 3.} \\
First, note that, similarly to \citet{graham2017econometric} link decisions are conditionally independent, that is, considering nodes $i$, $j$, $i'$, and $j'$, conditional on these agents observed and unobserved characteristics, the events that \textit{i and j are connected}, \textit{i and j' are connected}, \textit{i' and j are connected}, and \textit{i' and j' are connected} are independent of each other. Then, it follows that:
\begin{align}
    \mathbbm{E} [\boldsymbol{v}_{ij}(\boldsymbol{\theta}_{0}) \boldsymbol{v}_{i'j'}(\boldsymbol{\theta}_{0})] &= \mathbbm{E} \left[ \mathbbm{E} [\boldsymbol{v}_{ij}(\boldsymbol{\theta}_{0}) \boldsymbol{v}_{i'j'}(\boldsymbol{\theta}_{0}) \mid \mathcal{F}_n] \right] \nonumber \\
    &= \mathbbm{E} \left[ \mathbbm{E} [\boldsymbol{v}_{ij}(\boldsymbol{\theta}_{0}) \mid \mathcal{F}_n] \mathbbm{E} [\boldsymbol{v}_{i'j'}(\boldsymbol{\theta}_{0}) \mid \mathcal{F}_n] \right]  = 0
\end{align}
unless $i=i'$ and $j=j'$. \\
\\
\noindent \textit{Intermediate result 4.} \\
Also from the conditional independence of link formation, $$\mathbbm{E} [\boldsymbol{s}_\sigma(\boldsymbol{\theta}_{0}) \boldsymbol{s}_{\sigma'}(\boldsymbol{\theta}_{0})'] = \mathbbm{E} \left[ \mathbbm{E} [\boldsymbol{s}_\sigma(\boldsymbol{\theta}_{0}) \boldsymbol{s}_{\sigma'}(\boldsymbol{\theta}_{0})' \mid \mathcal{F}_n] \right]  = 0, $$ unless $\sigma$ and $\sigma'$ share at least one dyad in common.

To see this, note that after conditioning on $\mathcal{F}_n$, each score contribution $\bo s_\sigma\left(\theta_0\right)$ depends only on the four dyadic outcomes $\left(\tilde{y}_{i_1 j_1}, \tilde{y}_{i_1 j_2}, \tilde{y}_{i_2 j_1}, \tilde{y}_{i_2 j_2}\right)$ through $z_\sigma$. Since the errors $\epsilon_{i j}$ are independent across dyads, these dyadic outcomes are conditionally independent given $\mathcal{F}_n$. Therefore, $\boldsymbol{s}_\sigma$ and $\boldsymbol{s}_{\sigma^{\prime}}$ are conditionally independent whenever they involve entirely distinct sets of dyads. Therefore, the diagonal structure comes from the conditional independence argument.\\
\\
At this point, it is important to discuss a distinction from the proof strategy of Theorem \ref{theorem1}. In the proof of Theorem \ref{theorem1}, when analyzing the convergence of the objective function, I did not condition on covariates and fixed effects. There, dependencies arose whenever quadruples shared nodes because: (i) the covariate vectors $\bo x_{i j}$ can be dependent across dyads sharing nodes; (ii) common fixed effects also generate dependency; (iii) the expectation $\mathbb{E}\left[\ell_\sigma(\bo \theta)\right]$ averages over both covariate and error distributions; (iv) the likelihood contributions $\ell_\sigma$ have non-zero expectations (unlike the score contributions which have zero expectation).

This distinction highlights that: (i) when conditioning on $\mathcal{F}_n$, independence requires distinct dyads (Intermediate Results 3 and 4); (ii) without conditioning on covariates  Independence requires no shared nodes (as in Theorem \ref{theorem1}). The conditioning structure fundamentally changes the dependence patterns, which is crucial for understanding the different convergence rates and variance calculations throughout the proofs.

Now, I proceed to show the asymptotic equivalence between the normalized score vector and the normalized projection. That is, to show that $\boldsymbol{\Upsilon}^{-1/2} \boldsymbol{V}_n (\boldsymbol{\theta}_{0})$ and $\boldsymbol{\Upsilon}^{-1/2} \boldsymbol{S}_n (\boldsymbol{\theta}_{0})$ are asymptotically equivalent, where $\boldsymbol{\Upsilon}$ is the covariance matrix of the projection. The covariance matrix of the projection is given by, given the Intermediate results 1, 2 and 3:
\begin{adjustwidth}{-0.25in}{-0.25in}
\begin{align}
    \boldsymbol{\Upsilon} &= \mathbbm{E} [\boldsymbol{V}_n(\boldsymbol{\theta}_{0}) \boldsymbol{V}_n(\boldsymbol{\theta}_{0})'] \nonumber \\&= \sum_{i} \sum_{j \neq i} \mathbbm{E} [\boldsymbol{v}_{ij}(\boldsymbol{\theta}_{0}) \boldsymbol{v}_{ij}(\boldsymbol{\theta}_{0})] \nonumber \\
    &= \sum_{i} \sum_{j \neq i} \mathbbm{E} \left[ \sum_{i' \neq i,j} \sum_{j' \neq i,j,i'} \mathbbm{E} [\boldsymbol{s}(\sigma\{i,i';j,j'\}, \boldsymbol{\theta}_{0}) \mid \tilde{y}_{ij}, \mathcal{F}_n] \sum_{i'' \neq i,j} \sum_{j'' \neq i,j,i''} \mathbbm{E} [\boldsymbol{s}(\sigma\{i,i'';j,j''\}, \boldsymbol{\theta}_{0}) \mid \tilde{y}_{ij}, \mathcal{F}_n]' \right] \nonumber \\
    &= \sum_{i} \sum_{j \neq i} \sum_{i' \neq i,j} \sum_{j' \neq i,j,i'} \sum_{i'' \neq i,j} \sum_{j'' \neq i,j,i''} \mathbbm{E} \left[ \mathbbm{E} [\boldsymbol{s}(\sigma\{i,i';j,j'\}, \boldsymbol{\theta}_{0}) \mid \tilde{y}_{ij}, \mathcal{F}_n] \mathbbm{E} [\boldsymbol{s}(\sigma\{i,i'';j,j''\}, \boldsymbol{\theta}_{0}) \mid \tilde{y}_{ij}, \mathcal{F}_n]'\right] 
\end{align}
\end{adjustwidth}
The first equality follows since $\mathbbm{E} [\boldsymbol{V}_n(\boldsymbol{\theta}_{0})] = 0$. The second equality follows since the cross terms vanish unless $i = i'$ and $j = j'$. Therefore, only same-dyad terms contribute to the variance.
The matrix $\boldsymbol{\Upsilon}$ is positive semidefinite by construction, being a sum of outer products. Its positive definiteness at the $n^6 p_n$ scale is established at the beginning of Step 2, as a consequence of Assumption \ref{assumption_scorevar} and the score--projection covariance equivalence derived there. On that basis, the normalizations by $\boldsymbol{\Upsilon}^{-1/2}$ used below are well defined for sufficiently large $n$.

Furthermore, analogously to the Intermediate result 3, we know that the terms of the projections, say $\mathbbm{E} [\boldsymbol{s}(\sigma\{i,i';j,j'\}, \boldsymbol{\theta}_{0}) \mid \tilde{y}_{ij}, \mathcal{F}_n]$ and  $\mathbbm{E} [\boldsymbol{s}(\sigma\{i,i'';j,j''\}, \boldsymbol{\theta}_{0}) \mid \tilde{y}_{ij}, \mathcal{F}_n]$ are conditionally uncorrelated, unless they share one dyad in common. In the expression above there are $O(n^6)$ terms with only one dyad in common. This follows from the fact that one can first choose the shared dyad $(i,j)$ in $O(n^2)$ different ways, then complete the first quadruple by choosing $(i',j')$ in $O(n^2)$ ways, and analogously complete the second quadruple, leading to the term $O(n^6)$. Configurations in which the two quadruples share more than one dyad have at most five free node indices and therefore number $O(n^5)$. By Cauchy--Schwarz and the uniform conditional moment bounds implied by Assumption \ref{assumption3}, the aggregate contribution of these configurations is $O(n^5)$. Since the leading term is of order $n^6 p_n$, as established in the rate calculation at the end of this proof, these configurations are negligible:
$$(n^6 p_n)^{-1} O(n^5) = O\left((n p_n)^{-1}\right) = o(1)$$
by Assumption \ref{assumption4}. Therefore, the leading term of $\mathbbm{E} [\boldsymbol{V}_n(\boldsymbol{\theta}_{0}) \boldsymbol{V}_n(\boldsymbol{\theta}_{0})']$ is comprised of correlations between $\mathbbm{E} [\boldsymbol{s}(\sigma\{i,i';j,j'\}, \boldsymbol{\theta}_{0}) \mid \tilde{y}_{ij}, \mathcal{F}_n]$ and $\mathbbm{E} [\boldsymbol{s}(\sigma\{i,i'';j,j''\}, \boldsymbol{\theta}_{0}) \mid \tilde{y}_{ij}, \mathcal{F}_n]$ for which quadruples $\{i,i';j,j'\}$, $\{i,i'';j,j''\}$ share exactly one dyad in common.

Given that, conditional on $\mathcal{F}_n$ and ${y}_{ij}$, $\boldsymbol{s}(\sigma\{i,i';j,j'\}, \boldsymbol{\theta}_{0})$ and $\boldsymbol{s}(\sigma\{i,i'';j,j''\}, \boldsymbol{\theta}_{0})$ are independent if $i' \neq i''$ and $j' \neq j''$, the leading term of the variance of the projection is characterized as:
\begin{adjustwidth}{-0.25in}{-0.25in}
\begin{align}
    \boldsymbol{\Upsilon}_l &= \sum_{i} \sum_{j \neq i} \sum_{i' \neq i,j} \sum_{j' \neq i,j,i'} \sum_{i'' \neq i,j,i'} \sum_{j'' \neq i,j,j',i''} \mathbbm{E} \left[ \mathbbm{E} [\boldsymbol{s}(\sigma\{i,i';j,j'\}, \boldsymbol{\theta}_{0}) \mid \tilde{y}_{ij}, \mathcal{F}_n] \mathbbm{E} [\boldsymbol{s}(\sigma\{i,i'';j,j''\}, \boldsymbol{\theta}_{0}) \mid \tilde{y}_{ij}, \mathcal{F}_n]'\right] \nonumber \\
    &= \sum_{i} \sum_{j \neq i} \sum_{i' \neq i,j} \sum_{j' \neq i,j,i'} \sum_{i'' \neq i,j,i'} \sum_{j'' \neq i,j,j',i''}  \mathbbm{E} \left[ \mathbbm{E} [\boldsymbol{s}(\sigma\{i,i';j,j'\}, \boldsymbol{\theta}_{0}) \boldsymbol{s}(\sigma\{i,i'';j,j''\}, \boldsymbol{\theta}_{0}) \mid \tilde{y}_{ij}, \mathcal{F}_n]'\right] \nonumber \\
    &= \sum_{i} \sum_{j \neq i} \sum_{i' \neq i,j} \sum_{j' \neq i,j,i'} \sum_{i'' \neq i,j,i'} \sum_{j'' \neq i,j,j',i''}  \mathbbm{E} \left[ \boldsymbol{s}(\sigma\{i,i';j,j'\}, \boldsymbol{\theta}_{0}) \boldsymbol{s}(\sigma\{i,i'';j,j''\}, \boldsymbol{\theta}_{0})'\right] 
\end{align}
\end{adjustwidth}
\noindent where the second equality follows from conditional independence. As mentioned before, note that for the second quadruple, it is required that $i'' \neq i'$, and that $j'' \neq j'$, since, otherwise, the dyads $(i',j)$ and $(i'',j)$, and the dyads $(i,j')$ and $(i,j'')$ are the same. The condition $i^{\prime \prime} \neq i^{\prime}$ and $j^{\prime \prime} \neq j^{\prime}$ ensures that the quadruples $\sigma\left\{i, i^{\prime} ; j, j^{\prime}\right\}$ and $\sigma\left\{i, i^{\prime \prime} ; j, j^{\prime \prime}\right\}$ share only the dyad $(i, j)$, making them conditionally independent given $\left(\tilde{y}_{i j}, \mathcal{F}_n\right)$. This conditional independence is crucial for the second equality above. 

To establish asymptotic equivalence between $\boldsymbol{S}_n\left(\boldsymbol{\theta}_0\right)$ and $\boldsymbol{V}_n\left(\boldsymbol{\theta}_0\right)$, one needs to show that their covariance structures coincide asymptotically. Specifically, following \citet{van2000asymptotic}, I establish:
$$ \lim_{n \xrightarrow[]{} \infty} \Upsilon^{-1/2} \mathbbm{E} \left[ (\boldsymbol{V}_n(\boldsymbol{\theta}_{0}) - \boldsymbol{S}_n (\boldsymbol{\theta}_{0}))(\boldsymbol{V}_n(\boldsymbol{\theta}_{0}) - \boldsymbol{S}_n (\boldsymbol{\theta}_{0}))'\right]\Upsilon^{-1/2} = 0.$$
Or, equivalently,
$$ \lim_{n \xrightarrow[]{} \infty} \Upsilon^{-1/2} \mathbbm{E} \left[\boldsymbol{V}_n(\boldsymbol{\theta}_{0})\boldsymbol{V}_n(\boldsymbol{\theta}_{0})'- \boldsymbol{S}_n (\boldsymbol{\theta}_{0}) \boldsymbol{V}_n(\boldsymbol{\theta}_{0})' - \boldsymbol{V}_n(\boldsymbol{\theta}_{0})\boldsymbol{S}_n (\boldsymbol{\theta}_{0})' + \boldsymbol{S}_n (\boldsymbol{\theta}_{0})\boldsymbol{S}_n (\boldsymbol{\theta}_{0})'\right]\Upsilon^{-1/2} = 0.$$

The next step is to show that the (leading term of the) covariance matrix of the score $\boldsymbol{S}_n (\boldsymbol{\theta}_{0})$ coincides with the (leading term of the) covariance matrix for the projection, that is given above. The covariance of the score is given by $\mathbb{E}[\boldsymbol{S}_n (\boldsymbol{\theta}_{0}) \boldsymbol{S}_n (\boldsymbol{\theta}_{0})']$. 

From the previous intermediate results, because (i) $\mathbb{E}[\bo s(\sigma; \bo \theta_0) \mid \mathcal{F}_\sigma] = 0$ for all $\sigma \in \mathcal{N}_{m_n}$, where $\mathcal{F}_\sigma$ is the collection of covariates and fixed effects for the nodes in the quadruple $\sigma$, as in Intermediate result 1, (ii) and link decisions are conditionally independent, it follows that:
$$E\left(\bo s\left(\sigma ; \bo \theta_0\right) \bo s\left(\sigma^{\prime} ; \bo \theta_0\right)^{\prime} \mid \mathcal{F}_\sigma, \mathcal{F}_{\sigma^{\prime}}\right)=0$$
unless $\sigma$ and $\sigma^{\prime}$ have at least one dyad in common. Due to the same reasoning as before, there are $O(n^6)$ terms with only one dyad in common, while configurations sharing more than one dyad number $O(n^5)$ and, by Cauchy--Schwarz and the moment bounds of Assumption \ref{assumption5}, contribute $O(n^5)$ in aggregate. Since the leading term is of order $n^6 p_n$, as established in the rate calculation at the end of this proof, these configurations are negligible relative to it, so that the factor-of-16 and factor-of-4 relations between the leading covariance terms below hold after normalization by the full covariance scale. Therefore, the leading term of the variance of the scores is given by correlations between $\boldsymbol{s}\left(\sigma ; \boldsymbol{\theta}_0\right)$, and and $\boldsymbol{s}\left(\sigma^{\prime} ; \boldsymbol{\theta}_0\right)$ for which the quadruples $\sigma, \sigma^{\prime}$ have exactly one dyad in common. Similarly to \citet{jochmans2018semiparametric}, one can fix the dyad in common to be the first sender-receiver dyad $(i,j)$, and then multiply the expression for the score $\boldsymbol{s}\left(\sigma ; \boldsymbol{\theta}_0\right)$ by 4. The factor 4 comes from the fact that in the expression for the score of the first quadruple, the shared dyad can be in 4 different positions, and similarly for the second dyad. Therefore, fixing the dyad in common to be the first sender-receiver in each quadruple captures only 1 out of 16 possibilities. The leading term of the variance of the score is then given by:
\begin{align}
    \boldsymbol{\Upsilon}_{s,l} = \sum_{i} \sum_{j \neq i} \sum_{i' \neq i,j} \sum_{j' \neq i,j,i'} \sum_{i'' \neq i,j,i'} \sum_{j'' \neq i,j,j',i''} 16 \times \mathbbm{E} \left[ \boldsymbol{s}(\sigma\{i,i';j,j'\}, \boldsymbol{\theta}_{0}) \boldsymbol{s}(\sigma\{i,i'';j,j''\}, \boldsymbol{\theta}_{0})'\right] 
\end{align}
Thus, an immediate result is that $\boldsymbol{\Upsilon}_{s,l} = 16 \times \boldsymbol{\Upsilon}_l$. Note that the discrepancy by a factor of 16 follows from the fact that in the variance of the scores, the shared dyads can be in any of the four positions in each quadruple, while in the variance of the projection, $(i,j)$ is always set to the first position due to the definition of the projection.

It follows not only that $\boldsymbol{\Upsilon}^{-1/2} \mathbb{E} [\boldsymbol{V}_n(\boldsymbol{\theta}_{0}) \boldsymbol{V}_n(\boldsymbol{\theta}_{0})']\boldsymbol{\Upsilon}^{-1/2} = \bo I + o(1)$, but also from the equivalence of the leading terms, it also follows that $\boldsymbol{\Upsilon}^{-1/2} \mathbb{E} [\tilde{\boldsymbol{S}}_n(\boldsymbol{\theta}_{0}) \tilde{\boldsymbol{S}}_n(\boldsymbol{\theta}_{0})']\boldsymbol{\Upsilon}^{-1/2} = \bo I + o(1)$, where $ \tilde{\boldsymbol{S}}_n(\boldsymbol{\theta}_{0}) = 1/4 \boldsymbol{S}_n(\boldsymbol{\theta}_{0})$, a rescaled score such that the factor of 16 is absorbed, and standard asymptotic equivalence follows directly. Note that this rescaling is merely for notational convenience; the original score $\boldsymbol{S}_n(\boldsymbol{\theta}_{0})$ and projection $\boldsymbol{V}_n(\boldsymbol{\theta}_{0})$ remain asymptotically equivalent after proper normalization. From analogous arguments, given:
\begin{adjustwidth}{-0.25in}{-0.25in}
\begin{align}
    \mathbbm{E} [\boldsymbol{S}_n (\boldsymbol{\theta}_{0})\boldsymbol{V}_n (\boldsymbol{\theta}_{0})'] \nonumber 
    &= \sum_{i} \sum_{j \neq i} \sum_{i' \neq i,j} \sum_{j' \neq i,j,i'} \sum_{i'' \neq i,j} \sum_{j'' \neq i,j,i''} \mathbbm{E} \left[ 4 \times \boldsymbol{s}(\sigma\{i,i';j,j'\}, \boldsymbol{\theta}_{0}) \mathbbm{E} [\boldsymbol{s}(\sigma\{i,i'';j,j''\}, \boldsymbol{\theta}_{0}) \mid {y}_{ij}, \mathcal{F}_n]'\right], \nonumber 
\end{align}
\end{adjustwidth}
$\boldsymbol{\Upsilon}^{-1/2} \mathbb{E} [\tilde{\boldsymbol{S}}_n(\boldsymbol{\theta}_{0}) \boldsymbol{V}_n(\boldsymbol{\theta}_{0})']\boldsymbol\Upsilon^{-1/2} = \bo I + o(1)$, and analogously $\boldsymbol{\Upsilon}^{-1/2} \mathbb{E} [ \boldsymbol{V}_n(\boldsymbol{\theta}_{0}) \tilde{\boldsymbol{S}}_n(\boldsymbol{\theta}_{0})']\boldsymbol{\Upsilon}^{-1/2} = \bo I + o(1)$, which proves the asymptotic equivalence above (after proper normalization).

The factor of 4 in this expression arises again from the positional asymmetry between the score and projection terms. In the score $\boldsymbol{S}_n\left(\bo \theta_0\right)$, when summing over all quadruples, each pair of quadruples sharing dyad $(i, j)$ can have this shared dyad appearing in any of the 4 possible positions within the quadruple structure. However, in the projection $\boldsymbol{V}_n\left(\bo \theta_0\right)$, the term $\boldsymbol{v}_{i j}$ specifically conditions on $\tilde{y}_{i j}$ and only includes quadruples where dyad $(i, j)$ occupies the first sender-first receiver position. Therefore, when computing the cross-product $\mathbb{E}\left[\boldsymbol{S}_n\left(\bo \theta_0\right) \boldsymbol{V}_n\left(\bo \theta_0\right)^{\prime}\right]$, each score contribution with the shared dyad in any of its 4 possible positions gets matched with projection terms where that same dyad is fixed in the first position. This results in the multiplicative factor of 4 , which is consistent with the earlier finding that the variance of the score has a factor of 16 (from $4 \times 4$ possible position combinations) while the projection variance has no such factor.\\
\\
\noindent \textbf{Step 2.} \\
\noindent \textit{Projection covariance and the score--projection bridge.}\\
Recall that $\boldsymbol{v}_{ij}$ are zero mean and independent conditional on the sequence of covariates and fixed effects (as shown in intermediate result 3). Let 
$$ \boldsymbol{\Upsilon}_X=\sum_{i=1}^N \sum_{j \neq i} \mathbb{E} \left(\boldsymbol{v}_{i j} \boldsymbol{v}_{i j}^{\prime} \mid\left\{\boldsymbol{x}_{i j}\right\}_{n, n}, \{\alpha_i\}_n, \{\gamma_j\}_n \right). $$
Recall the matrix $\boldsymbol{\Upsilon}_n(\boldsymbol{\theta})$ from the main text, which for any $\boldsymbol{\theta}$ is given by
\begin{align} \label{upsilon_n}
    \boldsymbol{\Upsilon}_{n}({\boldsymbol{\theta}}) = \sum_{i} \sum_{j \neq i} \sum_{i' \neq i,j} \sum_{j' \neq i,j,i'} \sum_{i'' \neq i,j,i'} \sum_{j'' \neq i,j,j',i''} 16 \times \left[ \boldsymbol{s}(\sigma\{i,i';j,j'\}; \boldsymbol{\theta}) \boldsymbol{s}(\sigma\{i,i'';j,j''\}; \boldsymbol{\theta})'\right].
\end{align}
Assumption \ref{assumption_scorevar} is imposed on this matrix, evaluated at $\boldsymbol{\theta}_0$. It is connected to the conditional covariance of the projection by the following result:
\begin{align} \label{bridge}
(n^6 p_n)^{-1} \left\| \boldsymbol{\Upsilon}_n(\boldsymbol{\theta}_0) - 16 \boldsymbol{\Upsilon}_X \right\| \xrightarrow{p} 0.
\end{align}
To establish \eqref{bridge}, decompose
$$\boldsymbol{\Upsilon}_n(\boldsymbol{\theta}_0) - 16 \boldsymbol{\Upsilon}_X = \Big\{ \boldsymbol{\Upsilon}_n(\boldsymbol{\theta}_0) - \mathbb{E}[\boldsymbol{\Upsilon}_n(\boldsymbol{\theta}_0) \mid \mathcal{F}_n] \Big\} + \Big\{ \mathbb{E}[\boldsymbol{\Upsilon}_n(\boldsymbol{\theta}_0) \mid \mathcal{F}_n] - 16 \boldsymbol{\Upsilon}_X \Big\}.$$
For the second term, note that the index restrictions $i'' \neq i'$ and $j'' \neq j'$ in \eqref{upsilon_n} ensure that the quadruples $\sigma\{i,i';j,j'\}$ and $\sigma\{i,i'';j,j''\}$ share only the dyad $(i,j)$, so that, conditional on $(\tilde{y}_{ij}, \mathcal{F}_n)$, the two score contributions are independent. By iterated expectations,
\begin{align}
\mathbb{E}&\left[ \boldsymbol{s}(\sigma\{i,i';j,j'\}; \boldsymbol{\theta}_0) \boldsymbol{s}(\sigma\{i,i'';j,j''\}; \boldsymbol{\theta}_0)' \mid \mathcal{F}_n \right] \nonumber
\\ &= \mathbb{E}\left[ \mathbb{E}[\boldsymbol{s}(\sigma\{i,i';j,j'\}; \boldsymbol{\theta}_0) \mid \tilde{y}_{ij}, \mathcal{F}_n] \, \mathbb{E}[\boldsymbol{s}(\sigma\{i,i'';j,j''\}; \boldsymbol{\theta}_0) \mid \tilde{y}_{ij}, \mathcal{F}_n]' \mid \mathcal{F}_n \right], \nonumber
\end{align}
which is exactly the corresponding term in the expansion of $\boldsymbol{\Upsilon}_X$. Hence $\mathbb{E}[\boldsymbol{\Upsilon}_n(\boldsymbol{\theta}_0) \mid \mathcal{F}_n]$ coincides with $16$ times the restriction of $\boldsymbol{\Upsilon}_X$ to configurations satisfying $i'' \neq i'$ and $j'' \neq j'$. The discrepancy therefore consists only of the configurations excluded by these restrictions. As established in Step 1, there are $O(n^5)$ such configurations, each contributing $O(1)$ in expectation, so that their total contribution is $O_p(n^5)$.

For the first term, $\boldsymbol{\Upsilon}_n(\boldsymbol{\theta}_0)$ is a sum of $O(n^6)$ matrix-valued terms, each a product of two score contributions. Conditional on $\mathcal{F}_n$, each such term depends only on the outcomes of the dyads in its two quadruples, and these outcomes are independent across dyads. Two terms therefore have zero conditional covariance unless their sets of directed dyads overlap. For a fixed configuration, selecting a shared directed dyad fixes two node indices of the second configuration, leaving at most four free, so that there are $O(n^4)$ configurations overlapping any given one, and $O(n^{10})$ pairs with potentially non-zero conditional covariance in total. Using the moment bounds of Assumption \ref{assumption5} and, for fixed parameter dimension, a conditional Chebyshev inequality,
$$\mathbb{E}\left[ \left\| \boldsymbol{\Upsilon}_n(\boldsymbol{\theta}_0) - \mathbb{E}[\boldsymbol{\Upsilon}_n(\boldsymbol{\theta}_0) \mid \mathcal{F}_n] \right\|_F^2 \right] = O(n^{10}), \qquad \text{so that} \qquad \left\| \boldsymbol{\Upsilon}_n(\boldsymbol{\theta}_0) - \mathbb{E}[\boldsymbol{\Upsilon}_n(\boldsymbol{\theta}_0) \mid \mathcal{F}_n] \right\| = O_p(n^5).$$
Note that, unlike the unconditional covariance counting used for the Hessian in Step 3, where pairs of quadruples sharing a node are dependent through the covariates and fixed effects, the present argument conditions on $\mathcal{F}_n$, so that only dyad-sharing configurations contribute.

Both terms in the decomposition above are therefore $O_p(n^5)$. Since $np_n \to \infty$ by Assumption \ref{assumption4},
$$(n^6 p_n)^{-1} O_p(n^5) = O_p\left((n p_n)^{-1}\right) = o_p(1),$$
which proves \eqref{bridge}.\\

\noindent \textit{Eigenvalue transfers.}\\
The eigenvalue condition in Assumption \ref{assumption_scorevar} transfers to the covariance matrices used below. By \eqref{bridge} and Weyl's inequality, and since $\tfrac{1}{16}\boldsymbol{\Upsilon}_n(\boldsymbol{\theta}_0)$ satisfies the eigenvalue bound of Assumption \ref{assumption_scorevar} with constant $c/16$,
$$\lambda_{\min}\left[(n^6 p_n)^{-1} \boldsymbol{\Upsilon}_X \right] \geq \frac{c}{16} - o_p(1)$$
with probability approaching one, so that $\boldsymbol{\Upsilon}_X$ is positive definite with probability approaching one, its symmetric inverse square root is well defined on an event whose probability tends to one, and $\|\boldsymbol{\Upsilon}_X^{-1/2}\| = O_p((n^6p_n)^{-1/2})$.

The same holds for the population matrix $\boldsymbol{\Upsilon} = \mathbb{E}[\boldsymbol{\Upsilon}_X]$. Let $A_n = \{\lambda_{\min}(\boldsymbol{\Upsilon}_X) \geq c_X n^6 p_n\}$ for a sufficiently small constant $c_X > 0$, so that $\Pr(A_n) \to 1$. Since $\boldsymbol{\Upsilon}_X$ is positive semidefinite, for every unit vector $\boldsymbol{a}$,
$$\boldsymbol{a}' \boldsymbol{\Upsilon} \boldsymbol{a} = \mathbb{E}[\boldsymbol{a}' \boldsymbol{\Upsilon}_X \boldsymbol{a}] \geq \mathbb{E}[\boldsymbol{a}' \boldsymbol{\Upsilon}_X \boldsymbol{a} \, 1\{A_n\}] \geq c_X n^6 p_n \Pr(A_n),$$
so that $\lambda_{\min}(\boldsymbol{\Upsilon}) \geq c_\Upsilon n^6 p_n$ for some $c_\Upsilon > 0$ and all sufficiently large $n$. Together with the upper bound $\|\boldsymbol{\Upsilon}\| = O(n^6 p_n)$ from the leading-term calculation, the matrix $\boldsymbol{\Upsilon}$ is positive definite for sufficiently large $n$ and $\|\boldsymbol{\Upsilon}^{-1/2}\| = O((n^6p_n)^{-1/2})$, as used in Step 1.\\

\noindent \textit{Conditional central limit theorem.}\\
Conditional on $\mathcal{F}_n$, the summands $\boldsymbol{v}_{ij}(\boldsymbol{\theta}_0)$ are independent and have conditional mean zero. The fourth-moment Lyapunov condition can be verified directly from the definition of the projection. Write $\mathcal{N}_{m_n,ij}$ for the set of quadruples containing the dyad $(i,j)$ in the first sender-receiver position, so that
$$\sum_{\sigma \in \mathcal{N}_{m_n,ij}} (\cdot) = \sum_{i' \neq i,j} \sum_{j' \neq i,j,i'} (\cdot), \qquad |\mathcal{N}_{m_n,ij}| = (n-2)(n-3) = O(n^2).$$
Since the scalar logistic residual in each score contribution is bounded by one in absolute value, $\|\boldsymbol{s}(\sigma;\boldsymbol{\theta}_0)\| \leq \|\boldsymbol{r}_\sigma\|$, and $\boldsymbol{r}_\sigma$ is $\mathcal{F}_n$-measurable. Conditional Jensen's inequality therefore gives
\begin{align} \label{lyapunov_envelope}
\|\boldsymbol{v}_{ij}(\boldsymbol{\theta}_0)\| \leq \sum_{\sigma \in \mathcal{N}_{m_n,ij}} \mathbb{E}\left[\|\boldsymbol{s}(\sigma;\boldsymbol{\theta}_0)\| \mid \tilde{y}_{ij}, \mathcal{F}_n\right] \leq \sum_{\sigma \in \mathcal{N}_{m_n,ij}} \|\boldsymbol{r}_\sigma\|.
\end{align}
By the power-mean inequality and $|\mathcal{N}_{m_n,ij}| = O(n^2)$,
$$\left( \sum_{\sigma \in \mathcal{N}_{m_n,ij}} \|\boldsymbol{r}_\sigma\| \right)^4 \leq |\mathcal{N}_{m_n,ij}|^3 \sum_{\sigma \in \mathcal{N}_{m_n,ij}} \|\boldsymbol{r}_\sigma\|^4,$$
and the uniformly bounded sixth moments of Assumption \ref{assumption5} imply uniformly bounded fourth moments of $\boldsymbol{r}_\sigma$. Summing over the $O(n^2)$ directed dyads, the resulting bound has unconditional expectation of order $n^{10}$, so that by Markov's inequality
$$\sum_{i=1}^n \sum_{j \neq i} \mathbb{E}\left[\|\boldsymbol{v}_{ij}(\boldsymbol{\theta}_0)\|^4 \mid \mathcal{F}_n\right] = O_p(n^{10}).$$
Combining this with the eigenvalue bound established above,
$$\sum_{i=1}^n \sum_{j \neq i} \mathbb{E}\left[\left\|\boldsymbol{\Upsilon}_X^{-1/2} \boldsymbol{v}_{ij}(\boldsymbol{\theta}_0)\right\|^4 \mid \mathcal{F}_n\right] \leq \|\boldsymbol{\Upsilon}_X^{-1/2}\|^4 \sum_{i=1}^n \sum_{j \neq i} \mathbb{E}\left[\|\boldsymbol{v}_{ij}(\boldsymbol{\theta}_0)\|^4 \mid \mathcal{F}_n\right] = O_p\left(\frac{n^{10}}{(n^6p_n)^2}\right) = O_p\left((np_n)^{-2}\right),$$
which converges to zero by Assumption \ref{assumption4}. The conditional fourth-moment Lyapunov condition therefore holds, and a conditional version of the Lyapunov CLT \citep{rao2009conditional} gives:
$$ \boldsymbol{\Upsilon}_X^{-1/2} \boldsymbol{V}_n(\boldsymbol{\theta}_{0}) \xrightarrow[]{d} N(0,\boldsymbol{I})$$
conditional on the covariates and fixed effects. Note that the convergence holds since conditional on covariates and fixed effects, the remaining randomness comes from the idiosyncratic errors in the outcomes $\{y_{ij}\}$, such that the projection $\boldsymbol{V}_n(\boldsymbol{\theta}_{0})$ is a summation over conditionally independent summands. Specifically, $\boldsymbol{V}_n\left(\theta_0\right)=\sum_{i, j} \boldsymbol{v}_{i j}\left(\theta_0\right)$ where each $\boldsymbol{v}_{i j}$ depends only on $\tilde{y}_{i j}$, and the $\tilde{y}_{i j}$'s are conditionally independent given $\mathcal{F}_n$ due to the independence of the error terms $\epsilon_{i j}$. In particular, the convergence result holds for any sequence of fixed effects, and any configuration of covariates. For an explicit fourth-moment verification in a closely related dyadic conditional-likelihood setting, see \citet{muris2025dyadic}. Their argument computes the projection weights explicitly; the envelope in \eqref{lyapunov_envelope} follows here directly from the abstract definition of the projection.\\

\noindent \textit{Empirical normalization.}\\
Define a matrix $\boldsymbol{\Upsilon}_n({\boldsymbol{\theta}}_n)$ to be the plug-in estimator of the leading term of the score variance $\boldsymbol{\Upsilon}_{s,l}$ above, such that:
\begin{align}
    \boldsymbol{\Upsilon}_{n}({\boldsymbol{\theta}}_n) = \sum_{i} \sum_{j \neq i} \sum_{i' \neq i,j} \sum_{j' \neq i,j,i'} \sum_{i'' \neq i,j,i'} \sum_{j'' \neq i,j,j',i''} 16 \times \left[ \boldsymbol{s}(\sigma\{i,i';j,j'\}, \boldsymbol{\theta}_n) \boldsymbol{s}(\sigma\{i,i'';j,j''\}, \boldsymbol{\theta}_n)'\right] 
\end{align}

The limit distribution is now expressed directly in terms of the empirical matrix $\boldsymbol{\Upsilon}_n(\boldsymbol{\theta}_0)$. Recall that $\tilde{\boldsymbol{S}}_n(\boldsymbol{\theta}_0) = \frac{1}{4}\boldsymbol{S}_n(\boldsymbol{\theta}_0)$ and write $\boldsymbol{U}_n = \tilde{\boldsymbol{S}}_n(\boldsymbol{\theta}_0) - \boldsymbol{V}_n(\boldsymbol{\theta}_0)$ for the remainder from Step 1. Denote the rate-normalized matrices by
$$\overline{\boldsymbol{\Upsilon}}_n(\boldsymbol{\theta}_0) := \frac{\boldsymbol{\Upsilon}_n(\boldsymbol{\theta}_0)}{16 \, n^6 p_n}, \qquad \overline{\boldsymbol{\Upsilon}}_X := \frac{\boldsymbol{\Upsilon}_X}{n^6 p_n}.$$
By \eqref{bridge}, $\|\overline{\boldsymbol{\Upsilon}}_n(\boldsymbol{\theta}_0) - \overline{\boldsymbol{\Upsilon}}_X\| \xrightarrow{p} 0$, and by the eigenvalue transfers established above both matrices have smallest eigenvalue bounded away from zero with probability approaching one. Writing
$$\overline{\boldsymbol{\Upsilon}}_n(\boldsymbol{\theta}_0)^{-1/2} - \overline{\boldsymbol{\Upsilon}}_X^{-1/2} = \overline{\boldsymbol{\Upsilon}}_n(\boldsymbol{\theta}_0)^{-1/2}\left(\overline{\boldsymbol{\Upsilon}}_X^{1/2} - \overline{\boldsymbol{\Upsilon}}_n(\boldsymbol{\theta}_0)^{1/2}\right)\overline{\boldsymbol{\Upsilon}}_X^{-1/2},$$
and using the perturbation bound for matrix square roots, which for symmetric positive definite matrices $\boldsymbol{A}$ and $\boldsymbol{B}$ gives $\|\boldsymbol{B}^{1/2} - \boldsymbol{A}^{1/2}\| \leq \|\boldsymbol{B}-\boldsymbol{A}\|/(\lambda_{\min}(\boldsymbol{A})^{1/2} + \lambda_{\min}(\boldsymbol{B})^{1/2})$, the eigenvalue bounds established above imply
$$\left\|\overline{\boldsymbol{\Upsilon}}_n(\boldsymbol{\theta}_0)^{-1/2} - \overline{\boldsymbol{\Upsilon}}_X^{-1/2}\right\| \xrightarrow{p} 0.$$
Moreover, $(n^6p_n)^{-1/2}\boldsymbol{V}_n(\boldsymbol{\theta}_0) = O_p(1)$, since its conditional variance is $\overline{\boldsymbol{\Upsilon}}_X = O_p(1)$. Therefore
$$\left[\left(\tfrac{1}{16}\boldsymbol{\Upsilon}_n(\boldsymbol{\theta}_0)\right)^{-1/2} - \boldsymbol{\Upsilon}_X^{-1/2}\right] \boldsymbol{V}_n(\boldsymbol{\theta}_0) = \left[\overline{\boldsymbol{\Upsilon}}_n(\boldsymbol{\theta}_0)^{-1/2} - \overline{\boldsymbol{\Upsilon}}_X^{-1/2}\right] \frac{\boldsymbol{V}_n(\boldsymbol{\theta}_0)}{(n^6p_n)^{1/2}} = o_p(1),$$
and combining this with the limit distribution of $\boldsymbol{\Upsilon}_X^{-1/2}\boldsymbol{V}_n(\boldsymbol{\theta}_0)$ obtained above, Slutsky's theorem gives
$$\left(\tfrac{1}{16}\boldsymbol{\Upsilon}_n(\boldsymbol{\theta}_0)\right)^{-1/2}\boldsymbol{V}_n(\boldsymbol{\theta}_0) \xrightarrow{d} N(\boldsymbol{0}, \boldsymbol{I}).$$
For the remainder, the eigenvalue bound on $\boldsymbol{\Upsilon}_n(\boldsymbol{\theta}_0)$ and the upper bound $\|\boldsymbol{\Upsilon}\| = O(n^6p_n)$ give $\|(\tfrac{1}{16}\boldsymbol{\Upsilon}_n(\boldsymbol{\theta}_0))^{-1/2}\boldsymbol{\Upsilon}^{1/2}\| = O_p(1)$, while the asymptotic equivalence established in Step 1 gives $\boldsymbol{\Upsilon}^{-1/2}\mathbb{E}[\boldsymbol{U}_n\boldsymbol{U}_n']\boldsymbol{\Upsilon}^{-1/2} = o(1)$, obtained by combining the three normalized second-moment results, so that $\boldsymbol{\Upsilon}^{-1/2}\boldsymbol{U}_n = o_p(1)$ by Markov's inequality. Hence
$$\left(\tfrac{1}{16}\boldsymbol{\Upsilon}_n(\boldsymbol{\theta}_0)\right)^{-1/2}\boldsymbol{U}_n = \left[\left(\tfrac{1}{16}\boldsymbol{\Upsilon}_n(\boldsymbol{\theta}_0)\right)^{-1/2}\boldsymbol{\Upsilon}^{1/2}\right]\left[\boldsymbol{\Upsilon}^{-1/2}\boldsymbol{U}_n\right] = o_p(1).$$
Since $\tilde{\boldsymbol{S}}_n(\boldsymbol{\theta}_0) = \boldsymbol{V}_n(\boldsymbol{\theta}_0) + \boldsymbol{U}_n$ and $(\tfrac{1}{16}\boldsymbol{\Upsilon}_n(\boldsymbol{\theta}_0))^{-1/2} \tilde{\boldsymbol{S}}_n(\boldsymbol{\theta}_0) = \boldsymbol{\Upsilon}_n(\boldsymbol{\theta}_0)^{-1/2} \boldsymbol{S}_n(\boldsymbol{\theta}_0)$, it follows that
$$\boldsymbol{\Upsilon}_n(\boldsymbol{\theta}_0)^{-1/2}\boldsymbol{S}_n(\boldsymbol{\theta}_0) \xrightarrow{d} N(\boldsymbol{0}, \boldsymbol{I}).$$
This convergence is unconditional: $\boldsymbol{\Upsilon}_n(\boldsymbol{\theta}_0)$ is outcome-dependent and therefore not $\mathcal{F}_n$-measurable, so Slutsky's theorem is applied to the unconditional limit obtained above.

The feasible version replaces $\boldsymbol{\theta}_0$ by $\boldsymbol{\theta}_n$ in $\boldsymbol{\Upsilon}_n(\cdot)$; the required plug-in consistency and the resulting invertibility are established in the final part of this proof. Note that the original score is considered here, and accordingly, the empirical matrix takes into account the factor 16.\\
\\
\noindent \textbf{Step 3.} \\
The proof proceeds by showing the uniform convergence of the Hessian. Recall that the Hessian is
$$\boldsymbol{H}_n(\boldsymbol{\theta})= - \sum_{\sigma \in \mathcal{N}_{m_n}} \boldsymbol{r}_\sigma \boldsymbol{r}_\sigma^{\prime} f\left(\boldsymbol{r}_\sigma^{\prime} \boldsymbol{\theta}\right) 1\left\{\boldsymbol{z}_\sigma \in\{-1,1\}\right\}.$$
The goal is to show that
$$\sup_{\theta \in \Theta}\left\|\frac{\boldsymbol{H}_n(\boldsymbol{\theta})}{m_n^*}-\frac{\mathbb{E}\left(\boldsymbol{H}_n(\boldsymbol{\theta})\right)}{m_n p_n}\right\| \xrightarrow{p} 0$$
as $n \rightarrow \infty$. The matrix $\lim _{n \rightarrow \infty}\left(m_n p_n\right)^{-1} E\left(\boldsymbol{H}_n(\boldsymbol{\theta}_{0})\right)$ is the matrix given in Assumption \ref{assumption4}. Because it was shown in the proof of Theorem \ref{theorem1} that $\left(m_n^* / m_n-p_n\right) \xrightarrow{p} 0$ as $n \rightarrow \infty$ it suffices to show that
$$
\frac{\sup _{\theta \in \Theta}\left\|\boldsymbol{H}_n(\boldsymbol{\theta})-E\left(\boldsymbol{H}_n(\boldsymbol{\theta})\right)\right\|}{m_n p_n} \xrightarrow{p} 0
$$
as $n \rightarrow \infty$. In this way, the randomness in the random variable $m_n^*$ is separated from the randomness in the quadruple values in the expression for the Hessian. Moreover, the uniform convergence over $\Theta$ ensures that the result holds at the random $\boldsymbol{\theta}_n$. To show that the above expression holds, the conditions of Lemma 2.9 of \citet{newey1994chapter} need to be verified, which essentially requires pointwise convergence in probability and stochastic equicontinuity of the Hessian.\\
\\
\textit{Stochastic equicontinuity.}\\
First, note that:
$$\boldsymbol{H}_n(\boldsymbol{\theta}_{1})-\boldsymbol{H}_n(\boldsymbol{\theta}_{2})= - \sum_{\sigma \in \mathcal{N}_{m_n}} \boldsymbol{r}_\sigma \boldsymbol{r}_\sigma^{\prime} \left[f(\boldsymbol{r}_\sigma^{\prime} \boldsymbol{\theta}_1) - f(\boldsymbol{r}_\sigma^{\prime} \boldsymbol{\theta}_2)\right] 1\left\{\boldsymbol{z}_\sigma \in\{-1,1\}\right\}$$
Then, it follows that:
$$ f(\boldsymbol{r}_\sigma^{\prime} \boldsymbol{\theta}_1) - f(\boldsymbol{r}_\sigma^{\prime} \boldsymbol{\theta}_2) = f'(\boldsymbol{r}_\sigma^{\prime} \boldsymbol{\theta}_\sigma^*) \boldsymbol{r}_{\sigma}' (\boldsymbol{\theta}_1 - \boldsymbol{\theta}_2)$$
By an application of the mean-value theorem to $\left[f(\boldsymbol{r}_\sigma^{\prime} \boldsymbol{\theta}_1) - f(\boldsymbol{r}_\sigma^{\prime} \boldsymbol{\theta}_2)\right]$, in the sense that for each term, there exists a $\boldsymbol{\theta}_\sigma^*$ between $\boldsymbol{\theta}_1$ and $\boldsymbol{\theta}_2$. Next, by an application of the Cauchy-Schwarz inequality:
\begin{align}
    \rvert \rvert \boldsymbol{H}_n(\boldsymbol{\theta}_{1})-\boldsymbol{H}_n(\boldsymbol{\theta}_{2}) \rvert \rvert &\leq \sum_{\sigma \in \mathcal{N}_{m_n}} \rvert \rvert \boldsymbol{r}_\sigma \boldsymbol{r}_\sigma^{\prime} \rvert \rvert \hspace{0.1cm} \rvert f'(\boldsymbol{r}_\sigma^{\prime} \boldsymbol{\theta}_\sigma^*) \rvert \hspace{0.1cm} \rvert \boldsymbol{r}_\sigma^{\prime} (\boldsymbol{\theta}_{1} - \boldsymbol{\theta}_{2}) \rvert \hspace{0.1cm} 1\left\{\boldsymbol{z}_\sigma \in\{-1,1\}\right\} \nonumber \\
&\leq \sum_{\sigma \in \mathcal{N}_{m_n}} \rvert \rvert \boldsymbol{r}_\sigma \rvert \rvert^3 \hspace{0.1cm} \rvert f'(\boldsymbol{r}_\sigma^{\prime} \boldsymbol{\theta}_\sigma^*) \rvert \hspace{0.1cm} \rvert \rvert \boldsymbol{\theta}_{1} - \boldsymbol{\theta}_{2}) \rvert \rvert \hspace{0.1cm} 1\left\{\boldsymbol{z}_\sigma \in\{-1,1\}\right\} \nonumber 
\end{align}
To obtain the final bound, since $\boldsymbol{\theta}_\sigma^*$ lies between $\boldsymbol{\theta}_1$ and $\boldsymbol{\theta}_2$, and both are in the compact set $\Theta$, it is possible to bound: 
$$ \rvert f'(\boldsymbol{r}_\sigma^{\prime} \boldsymbol{\theta}_\sigma^*) \rvert \leq \sup _{\varepsilon \in \mathcal{R}} \left \rvert \frac{\partial f(\varepsilon)}{\partial \varepsilon} \right \rvert $$
Leading to the expression:
\begin{align} \label{hessian_expression}
\frac{\left\|\boldsymbol{H}_n(\boldsymbol{\theta}_{1})-\boldsymbol{H}_n(\boldsymbol{\theta}_{2})\right\|}{m_n p_n} \leq\left(\left(m_n p_n\right)^{-1} \sum_{\sigma \in \mathcal{N}_{m_n}}\left\|\boldsymbol{r}_\sigma\right\|^3 1\left\{z_\sigma \in\{-1,1\}\right\}\right) \sup _{\varepsilon \in \mathcal{R}}\left|\frac{\partial f(\varepsilon)}{\partial \varepsilon}\right|\left\|\boldsymbol{\theta}_{1}-\boldsymbol{\theta}_{2}\right\|
\end{align}
for any $\boldsymbol{\theta}_{1}, \boldsymbol{\theta}_{2} \in \Theta$. 
Using the same arguments as those used to establish Theorem \ref{theorem1} it follows that
$$
\left(m_n p_n\right)^{-1} \sum_{\sigma \in \mathcal{N}_{m_n}}\left\|\boldsymbol{r}_\sigma\right\|^3 1\left\{z_\sigma \in\{-1,1\}\right\}=O_p(1)
$$
where the moment condition in Assumption \ref{assumption5} is used. The proof proceeds in formally showing this statement. By Chebyshev's inequality:
\begin{align}
&\operatorname{Pr} \left( \left \rvert \frac{\sum_{\sigma \in \mathcal{N}_{m_n}}\left\|\boldsymbol{r}_\sigma\right\|^3 1\left\{z_\sigma \in\{-1,1\}\right\}}{m_n p_n} - \mathbb{E} \left[ \frac{\sum_{\sigma \in \mathcal{N}_{m_n}}\left\|\boldsymbol{r}_\sigma\right\|^3 1\left\{z_\sigma \in\{-1,1\}\right\}}{m_n p_n}\right] \right \rvert > \varepsilon\right) \nonumber \\ &\leq \frac{\operatorname{Var}(\sum_{\sigma \in \mathcal{N}_{m_n}}\left\|\boldsymbol{r}_\sigma\right\|^3 1\left\{z_\sigma \in\{-1,1\}\right\})}{(m_n p_n)^2 \varepsilon^2}
\end{align}
I first focus on the term $\mathbb{E} \left[ \frac{\sum_{\sigma \in \mathcal{N}_{m_n}}\left\|\boldsymbol{r}_\sigma\right\|^3 1\left\{z_\sigma \in\{-1,1\}\right\}}{m_n p_n}\right]$. Note that $\sum_{\sigma \in \mathcal{N}_{m_n}} \mathbb{E}\left[1\left\{z_\sigma \in\{-1,1\}\right\}\right]=m_n p_n$. It is possible to rewrite:
\begin{align}
   \sum_{\sigma \in \mathcal{N}_{m_n}} \mathbb{E} [\left\|\boldsymbol{r}_\sigma\right\|^3 1\left\{z_\sigma \in\{-1,1\}\right\}] &= \sum_{\sigma \in \mathcal{N}_{m_n}}  \mathbb{E}\left[\left\|\boldsymbol{r}_\sigma\right\|^3 \cdot 1\left\{z_\sigma \in\{-1,1\}\right\}\right] \nonumber \\ &=\sum_{\sigma \in \mathcal{N}_{m_n}}  \mathbb{E}\left[\left\|\boldsymbol{r}_\sigma\right\|^3 \mid z_\sigma \in\{-1,1\}\right] \cdot \operatorname{Pr}\left(z_\sigma \in\{-1,1\}\right)
\end{align}
Since the conditional expectation $\mathbb{E}\left[\left\|\boldsymbol{r}_\sigma\right\|^3 \mid z_\sigma \in\{-1,1\}\right] \leq C^{\prime}$ for some finite constant $C^{\prime}$ (the third moment is bounded regardless of the link pattern), it follows that:
$$
\sum_{\sigma \in \mathcal{N}_{m_n}} \mathbb{E}\left[\left\|\boldsymbol{r}_\sigma\right\|^3 \cdot 1\left\{z_\sigma \in\{-1,1\}\right\}\right] \leq C^{\prime} \sum_{\sigma \in \mathcal{N}_{m_n}} \operatorname{Pr}\left(z_\sigma \in\{-1,1\}\right)=C^{\prime} \cdot m_n p_n
$$
Such that $$\mathbb{E} \left[ \frac{\sum_{\sigma \in \mathcal{N}_{m_n}}\left\|\boldsymbol{r}_\sigma\right\|^3 1\left\{z_\sigma \in\{-1,1\}\right\}}{m_n p_n}\right] = O(1)$$
For the variance term, I follow a similar argument as in Theorem \ref{theorem1}. Note that the term can be written as, by denoting $A_\sigma = \left\|\boldsymbol{r}_\sigma\right\|^3 1\left\{z_\sigma \in\{-1,1\}\right\} - \mathbb{E} \left( \left\|\boldsymbol{r}_\sigma\right\|^3 1\left\{z_\sigma \in\{-1,1\}\right\}\right)$:
\begin{align}
&\operatorname{Var}\left(\sum_{\sigma \in \mathcal{N}_{m_n}}\left\|\boldsymbol{r}_\sigma\right\|^3 1\left\{z_\sigma \in\{-1,1\}\right\}\right) = \mathbb{E} \left( \left \rvert \sum_{\sigma \in \mathcal{N}_{m_n}} A_\sigma \right \rvert^2  \right) \nonumber \\
&= \mathbb{E} \left( \left(\sum_{\sigma \in \mathcal{N}_{m_n}} A_{\sigma} \right) \left(\sum_{\sigma' \in \mathcal{N}_{m_n}}A_{\sigma'} \right)\right) \nonumber \\
&= \sum_{\sigma \in \mathcal{N}_{m_n}} \sum_{\sigma' \in \mathcal{N}_{m_n}} 1 \{\sigma \cap \sigma' \neq \emptyset \} \mathbb{E} \left( \left(A_\sigma \right) \left(A_{\sigma'} \right)\right) \nonumber \\
&= \sum_{\sigma \in \mathcal{N}_{m_n}} \sum_{\sigma' \in \mathcal{N}_{m_n}} 1 \{\sigma \cap \sigma' \neq \emptyset \} \operatorname{Cov} (\left\|\boldsymbol{r}_{\sigma}\right\|^3 1\left\{z_{\sigma} \in\{-1,1\}\right\}, \left\|\boldsymbol{r}_{\sigma'}\right\|^3 1\left\{z_{\sigma'} \in\{-1,1\}\right\})\nonumber \\
& \leq \sum_{\sigma \in \mathcal{N}_{m_n}} \sum_{\sigma' \in \mathcal{N}_{m_n}} 1 \{\sigma \cap \sigma' \neq \emptyset \} \rvert \operatorname{Cov} (\left\|\boldsymbol{r}_{\sigma}\right\|^3 1\left\{z_{\sigma} \in\{-1,1\}\right\}, \left\|\boldsymbol{r}_{\sigma'}\right\|^3 1\left\{z_{\sigma'} \in\{-1,1\}\right\}) \rvert \nonumber \\ 
&\leq \sum_{\sigma \in \mathcal{N}_{m_n}} \sum_{\sigma' \in \mathcal{N}_{m_n}} 1 \{\sigma \cap \sigma' \neq \emptyset \} \sqrt{\mathbb{E}\left(\left\|\boldsymbol{r}_{\sigma}\right\|^6 1\left\{z_{\sigma} \in\{-1,1\} \right\} \right)}  \sqrt{\mathbb{E}\left(\left\|\boldsymbol{r}_{\sigma'}\right\|^6 1\left\{z_{\sigma'} \in\{-1,1\} \right\} \right)} \nonumber \\
&\leq B'\sum_{\sigma \in \mathcal{N}_{m_n}} \sum_{\sigma' \in \mathcal{N}_{m_n}} 1 \{\sigma \cap \sigma' \neq \emptyset \} \sqrt{p_\sigma p_{\sigma'}} \nonumber \\
&= O(n^3 m_n p_n)
\end{align}
where the second inequality follows from an application of the Cauchy-Schwarz inequality. The expression $\left\|\boldsymbol{r}_{\sigma'}\right\|^6 1\left\{z_{\sigma'} \in\{-1,1\} \right\}$ follows from the fact that $1\left\{z_{\sigma'} \in\{-1,1\} \right\}^2 = 1\left\{z_{\sigma'} \in\{-1,1\} \right\}$. Moreover, by writing: $$\mathbb{E}\left(\left\|\boldsymbol{r}_{\sigma'}\right\|^6 1\left\{z_{\sigma'} \in\{-1,1\} \right\} \right) = \operatorname{Pr}(1\left\{z_{\sigma'} \in\{-1,1\}\right\}) \mathbb{E}\left(\left\|\boldsymbol{r}_{\sigma'}\right\|^6 \mid 1\left\{z_{\sigma'} \in\{-1,1\} \right\} \right),$$ and since the conditional expectation is bounded, since $\mathbb{E}\left(\left\|\boldsymbol{r}_{\sigma'}\right\|^6\right)$ is uniformly bounded from Assumption \ref{assumption5} by some constant $B'$, it follows that  $\mathbb{E}\left(\left\|\boldsymbol{r}_{\sigma}\right\|^6 1\left\{z_{\sigma} \in\{-1,1\} \right\} \right) \leq p_\sigma B'$. Finally, by applying the same reasoning as in the previous proof, by an application of Jensen's inequatility, and from the fact that $O(n^7)$ quadruples share a node, the result follows.

Furthermore, similarly to before, it follows that:
$$\frac{\operatorname{Var}(\sum_{\sigma \in \mathcal{N}_{m_n}}\left\|\boldsymbol{r}_\sigma\right\|^3 1\left\{z_\sigma \in\{-1,1\}\right\})}{(m_n p_n)^2} = O\left( \frac{1}{n p_n}\right)$$
which converges to zero by Assumption \ref{assumption4}. Thus:
\begin{align}
& \lim_{n \xrightarrow[]{} \infty} \operatorname{Pr} \left( \left \rvert \frac{\sum_{\sigma \in \mathcal{N}_{m_n}}\left\|\boldsymbol{r}_\sigma\right\|^3 1\left\{z_\sigma \in\{-1,1\}\right\}}{m_n p_n} - \mathbb{E} \left[ \frac{\sum_{\sigma \in \mathcal{N}_{m_n}}\left\|\boldsymbol{r}_\sigma\right\|^3 1\left\{z_\sigma \in\{-1,1\}\right\}}{m_n p_n}\right] \right \rvert > \varepsilon\right) \nonumber \\ &\leq \frac{\operatorname{Var}(\sum_{\sigma \in \mathcal{N}_{m_n}}\left\|\boldsymbol{r}_\sigma\right\|^3 1\left\{z_\sigma \in\{-1,1\}\right\})}{(m_n p_n)^2 \varepsilon^2} = 0
\end{align}
for any $\varepsilon > 0$, leading to:
$$
\left(m_n p_n\right)^{-1} \sum_{\sigma \in \mathcal{N}_{m_n}}\left\|\boldsymbol{r}_\sigma\right\|^3 1\left\{z_\sigma \in\{-1,1\}\right\}=O_p(1)
$$
Because the derivative of $f$ is bounded uniformly on $\mathcal{R}$, I gather these last results to obtain the following for Equation \ref{hessian_expression}:
$$
\frac{\left\|\boldsymbol{H}_n\left(\boldsymbol{\theta}_{1}\right)-\boldsymbol{H}_n\left(\boldsymbol{\theta}_{2}\right)\right\|}{m_n p_n}=O_p(1)\left\|\boldsymbol{\theta}_{1}-\boldsymbol{\theta}_{2}\right\|
$$
for any $\boldsymbol{{\theta}}_{1}, \boldsymbol{\theta}_{2} \in \Theta$. Thus, the Hessian matrix is stochastically equicontinuous. This implies that uniform convergence of the Hessian follows from pointwise convergence on $\Theta$, from an application of Lemma 2.9 in \citet{newey1994chapter}. \\
\\
\textit{Pointwise convergence.}\\
Assumption \ref{assumption5} implies that $E\left(\left\|\boldsymbol{r}_\sigma\right\|^4 \mid z_\sigma \in\{-1,1\}\right)$ is uniformly bounded in $\sigma$ while $f$ is bounded uniformly on $\mathcal{R}$. Therefore, in the following, using these results, the proof proceeds to show the convergence result
$$
\frac{\left\|\boldsymbol{H}_n(\boldsymbol{\theta})-E\left(\boldsymbol{H}_n(\boldsymbol{\theta})\right)\right\|}{m_n p_n} \xrightarrow{p} 0
$$
for all $\boldsymbol{\theta} \in \Theta$. By Chebyshev's inequality, it follows that, for any fixed $\boldsymbol{\theta}$ in $\Theta$ and $\varepsilon > 0$:
$$ \operatorname{Pr} \left( \frac{\left\|\boldsymbol{H}_n(\boldsymbol{\theta})-E\left(\boldsymbol{H}_n(\boldsymbol{\theta})\right)\right\|}{m_n p_n} > \varepsilon\right) \leq \frac{ \mathbb{E} \left( \left\|\boldsymbol{H}_n(\boldsymbol{\theta})-E\left(\boldsymbol{H}_n(\boldsymbol{\theta})\right)\right\|^2 \right)}{\varepsilon^2 (m_n p_n)^2} $$
For a fixed value of $\bo \theta$, define $\bo h_\sigma(\boldsymbol{\theta}) = - \boldsymbol{r}_\sigma \boldsymbol{r}_\sigma' f (\boldsymbol{r}_\sigma'\boldsymbol{\theta}) 1 \{ z_\sigma \in \{-1,1\} \}$, such that $\boldsymbol{H}_n(\boldsymbol{\theta}) = \sum_{\sigma \in \mathcal{N}_{m_n}} \bo h_\sigma(\boldsymbol{\theta})$.
Similarly to before, it follows that, using an application of the triangle inequality for norms and the Cauchy-Schwarz inequality:
\begin{adjustwidth}{-0.25in}{-0.25in}
\begin{align}
\mathbb{E} \left( \left\|\boldsymbol{H}_n(\boldsymbol{\theta})-E\left(\boldsymbol{H}_n(\boldsymbol{\theta})\right)\right\|^2 \right) &= \mathbb{E} \left( \left\|\sum_{\sigma \in \mathcal{N}_{m_n}} h_\sigma(\boldsymbol{\theta}) - \mathbb{E}(h_\sigma(\boldsymbol{\theta}))\right\|^2\right) \nonumber \\
& \leq \mathbb{E} \left( \left(\sum_{\sigma \in \mathcal{N}_{m_n}} \left\| h_\sigma(\boldsymbol{\theta}) - \mathbb{E}(h_\sigma(\boldsymbol{\theta}))\right\| \right)^2\right) \nonumber \\
& = \sum_{\sigma \in \mathcal{N}_{m_n}} \sum_{\sigma' \in \mathcal{N}_{m_n}} \mathbb{E}\left( \left\| h_\sigma(\boldsymbol{\theta}) - \mathbb{E}(h_\sigma(\boldsymbol{\theta}))\right\| \left\| h_{\sigma'}(\boldsymbol{\theta}) - \mathbb{E}(h_{\sigma'}(\boldsymbol{\theta}))\right\| \right) \nonumber \\
& = \sum_{\sigma \in \mathcal{N}_{m_n}} \sum_{\sigma' \in \mathcal{N}_{m_n}} 1 \{\sigma \cap \sigma' \neq \emptyset \} \mathbb{E}\left( \left\| h_\sigma(\boldsymbol{\theta}) - \mathbb{E}(h_\sigma(\boldsymbol{\theta}))\right\| \left\| h_{\sigma'}(\boldsymbol{\theta}) - \mathbb{E}(h_{\sigma'}(\boldsymbol{\theta}))\right\| \right) \\
& = \sum_{\sigma \in \mathcal{N}_{m_n}} \sum_{\sigma' \in \mathcal{N}_{m_n}} 1 \{\sigma \cap \sigma' \neq \emptyset \}  \operatorname{Cov}(h_\sigma(\boldsymbol{\theta}), h_{\sigma'}(\boldsymbol{\theta})) \nonumber \\
& \leq \sum_{\sigma \in \mathcal{N}_{m_n}} \sum_{\sigma' \in \mathcal{N}_{m_n}} 1 \{\sigma \cap \sigma' \neq \emptyset \} \rvert \operatorname{Cov}(h_\sigma(\boldsymbol{\theta}), h_{\sigma'}(\boldsymbol{\theta})) \rvert \nonumber \\
& \leq \sum_{\sigma \in \mathcal{N}_{m_n}} \sum_{\sigma' \in \mathcal{N}_{m_n}} 1 \{\sigma \cap \sigma' \neq \emptyset \} \sqrt{\operatorname{Var} (h_\sigma(\boldsymbol{\theta}))} \sqrt{\operatorname{Var} (h_{\sigma'}(\boldsymbol{\theta}))} \nonumber \\
&= O (n^3 m_n p_n)
\end{align}
\end{adjustwidth}
Where the final result follows the same arguments as in the proof of Theorem \ref{theorem1}, using the bound on sixth moments from Assumption \ref{assumption5} and the fact that there are $O\left(n^7\right)$ pairs of quadruples sharing nodes. Therefore, 
$$ \lim_{n \xrightarrow{} \infty} \operatorname{Pr} \left( \frac{\left\|\boldsymbol{H}_n(\boldsymbol{\theta})-E\left(\boldsymbol{H}_n(\boldsymbol{\theta})\right)\right\|}{m_n p_n} > \varepsilon\right) \leq \frac{ \mathbb{E} \left( \left\|\boldsymbol{H}_n(\boldsymbol{\theta})-E\left(\boldsymbol{H}_n(\boldsymbol{\theta})\right)\right\|^2 \right)}{\varepsilon^2 (m_n p_n)^2} = 0$$
for any $\varepsilon >0$.\\

\noindent \textit{Convergence to the limit matrix.}\\
Therefore, pointwise convergence and uniform convergence have been shown. In particular, evaluating at $\boldsymbol{\theta}_0$,
$$(m_n p_n)^{-1}\boldsymbol{H}_n(\boldsymbol{\theta}_0) \xrightarrow{p} \lim_{n \xrightarrow[]{} \infty} (m_n p_n)^{-1}\mathbb{E}(\boldsymbol{H}_n(\boldsymbol{\theta}_0)),$$
the matrix given in Assumption \ref{assumption4}.\\
\\
\noindent \textbf{Step 4.} \\
\textit{Limit distribution of the estimator.}\\
In this final step, the limit distribution of the estimator is derived. A first order Taylor expansion the first order condition to the log-likelihood optimization problem around the true value of the parameter, $\boldsymbol{\theta}_0$ yields:
$$ \boldsymbol{S}_n(\boldsymbol{\theta}_n) = \boldsymbol{S}_n (\boldsymbol{\theta}_0) + \boldsymbol{H}_n (\boldsymbol{\theta}_*) (\boldsymbol{\theta}_n - \boldsymbol{\theta}_0) $$
Where $\boldsymbol{\theta}_* \in \Theta$ is a value that lies between $\boldsymbol{{\theta}}_n$ and $\boldsymbol{\theta}_{0}$. Since from the first order condition it follows that $\boldsymbol{S}_n(\boldsymbol{\theta}_n) = 0$, and hence:
$$ \boldsymbol{\theta}_n - \boldsymbol{\theta}_0 = - \boldsymbol{H}_n (\boldsymbol{\theta}_*)^{-1} \boldsymbol{S}_n (\boldsymbol{\theta}_0) $$
To get to the limit distribution, a proper normalization is needed. Define the sandwich asymptotic variance matrix (given that the information matrix equality fails in this quasi-likelihood setting) $\boldsymbol{\Omega}_n ( \boldsymbol{\theta}_0)= \boldsymbol{H}_n( \boldsymbol{\theta}_0)^{-1} \boldsymbol{\Upsilon}_n ( \boldsymbol{\theta}_0) \boldsymbol{H}_n( \boldsymbol{\theta}_0)^{-1} $.

Since $\boldsymbol{\theta}_* \xrightarrow{p} \boldsymbol{\theta}_0$, the stochastic equicontinuity bound established in Step 3 gives $(m_n p_n)^{-1}\|\boldsymbol{H}_n(\boldsymbol{\theta}_*) - \boldsymbol{H}_n(\boldsymbol{\theta}_0)\| \xrightarrow{p} 0$, so that, on the event where both matrices are invertible, $(m_n p_n)\boldsymbol{H}_n(\boldsymbol{\theta}_*)^{-1} = (m_n p_n)\boldsymbol{H}_n(\boldsymbol{\theta}_0)^{-1} + o_p(1)$, the invertibility and the associated eigenvalue bounds being established in the final part of this proof. Consequently,
$$\boldsymbol{\Omega}_n(\boldsymbol{\theta}_0)^{-1/2}\left(\boldsymbol{\theta}_n-\boldsymbol{\theta}_{0}\right)=-\boldsymbol{\Omega}_n(\boldsymbol{\theta}_0)^{-1/2} \boldsymbol{H}_n\left(\boldsymbol{\theta}_{0}\right)^{-1} \boldsymbol{S}_n\left(\boldsymbol{\theta}_{0}\right) + o_p(1).$$
Passing from the limit distribution of the score to that of the estimator requires multiplying by
$$\boldsymbol{M}_n := \boldsymbol{\Omega}_n(\boldsymbol{\theta}_0)^{-1/2} \boldsymbol{H}_n(\boldsymbol{\theta}_0)^{-1} \boldsymbol{\Upsilon}_n(\boldsymbol{\theta}_0)^{1/2},$$
which satisfies $\boldsymbol{M}_n\boldsymbol{M}_n' = \boldsymbol{I}$ by the definition of $\boldsymbol{\Omega}_n(\boldsymbol{\theta}_0)$ and is therefore orthogonal for every $n$. Orthogonality alone is not enough: $\boldsymbol{M}_n$ is constructed from $\boldsymbol{\Upsilon}_n(\boldsymbol{\theta}_0)$ and is therefore not independent of the standardized score it multiplies, so the rotation invariance of the standard normal law cannot be invoked directly. The argument proceeds conditionally on $\mathcal{F}_n$, where the rotation is fixed. A deterministic rotation would serve equally well and would not need to converge, but establishing one would require the conditional score covariance to concentrate around its unconditional mean, which the assumptions do not deliver.

Let $\overline{\boldsymbol{H}}_n := (m_n p_n)^{-1}\boldsymbol{H}_n(\boldsymbol{\theta}_0)$ and $\overline{\boldsymbol{H}}_\infty := \lim_{n \to \infty}\mathbb{E}(\overline{\boldsymbol{H}}_n)$, which exists and is negative definite by Assumption \ref{assumption4}, and define
$$\boldsymbol{\Omega}_{X,n} := (m_n p_n)^{-2}\,\overline{\boldsymbol{H}}_\infty^{-1} \left(16\,\boldsymbol{\Upsilon}_X\right) \overline{\boldsymbol{H}}_\infty^{-1}, \qquad \boldsymbol{L}_{X,n} := \boldsymbol{\Omega}_{X,n}^{-1/2}\,(m_n p_n)^{-1}\overline{\boldsymbol{H}}_\infty^{-1}\left(16\,\boldsymbol{\Upsilon}_X\right)^{1/2},$$
so that $\boldsymbol{L}_{X,n}\boldsymbol{L}_{X,n}' = \boldsymbol{I}$ and $\boldsymbol{L}_{X,n}$ is a fixed matrix given $\mathcal{F}_n$, since $\boldsymbol{\Upsilon}_X$ is a conditional expectation given $\mathcal{F}_n$ and $\overline{\boldsymbol{H}}_\infty$ is deterministic. By the conditional central limit theorem of Step 2, together with $\boldsymbol{\Upsilon}^{-1/2}\boldsymbol{U}_n = o_p(1)$ from the asymptotic equivalence,
$$\left(16\,\boldsymbol{\Upsilon}_X\right)^{-1/2}\boldsymbol{S}_n(\boldsymbol{\theta}_0) \xrightarrow{d} N(\boldsymbol{0}, \boldsymbol{I})$$
conditionally on $\mathcal{F}_n$. Since an orthogonal transformation that is fixed given the conditioning variables leaves a conditionally standard normal vector standard normal, combining this with the expansion above and replacing $\boldsymbol{H}_n(\boldsymbol{\theta}_0)$ by $(m_n p_n)\overline{\boldsymbol{H}}_\infty$, which is negligible at the Hessian scale by Step 3, gives
$$\boldsymbol{\Omega}_{X,n}^{-1/2}\left(\boldsymbol{\theta}_n - \boldsymbol{\theta}_0\right) \xrightarrow{d} N(\boldsymbol{0}, \boldsymbol{I})$$
conditionally on $\mathcal{F}_n$. The limiting distribution does not depend on the conditioning variables, so the convergence also holds unconditionally.

Finally, $\boldsymbol{\Omega}_n(\boldsymbol{\theta}_0)$ and $\boldsymbol{\Omega}_{X,n}$ differ through $\overline{\boldsymbol{\Upsilon}}_n(\boldsymbol{\theta}_0)$ against $\overline{\boldsymbol{\Upsilon}}_X$ and $\overline{\boldsymbol{H}}_n$ against $\overline{\boldsymbol{H}}_\infty$, both of which vanish in norm by \eqref{bridge} and Step 3 respectively. The common scale $(n(n-1)p_n)^{-1}$ cancels in the ratio $\boldsymbol{\Omega}_n(\boldsymbol{\theta}_0)^{-1/2}\boldsymbol{\Omega}_{X,n}^{1/2}$, and since the eigenvalue bounds of Step 2 hold for both matrices, the perturbation bound for matrix square roots used there gives $\|\boldsymbol{\Omega}_n(\boldsymbol{\theta}_0)^{-1/2}\boldsymbol{\Omega}_{X,n}^{1/2} - \boldsymbol{I}\| \xrightarrow{p} 0$. Slutsky's theorem then yields
$$\boldsymbol{\Omega}_n(\boldsymbol{\theta}_0)^{-1/2} (\boldsymbol{\theta}_n - \boldsymbol{\theta}_0) \xrightarrow{d} N (0, \boldsymbol{I}).$$
Slutsky applies at this last step, unlike at the first, because the multiplying object is a ratio of two versions of the same covariance matrix rather than a product of distinct matrices, so that whatever drift the score covariance exhibits affects both and cancels.\\
\\
\noindent \textit{Determining the convergence rate.}\\
Now, the focus is on determining the convergence rate of the estimator $\boldsymbol{\theta}_n$. I have established before that $\boldsymbol{H}_n(\boldsymbol{\theta}_0) = O(m_n p_n) = O(n^4 p_n)$. The key remaining component is to establish the order of $\bo \Upsilon_n(\boldsymbol{\theta}_0) = O(n^6 p_n)$. Intuitively, it follows as: 
\begin{align} \label{variance_score}
    \operatorname{Var} (\boldsymbol{S}_n(\boldsymbol{\theta}_0)) &= \mathbb{E} [\boldsymbol{S}_n(\boldsymbol{\theta}_0) \boldsymbol{S}_n(\boldsymbol{\theta}_0)'] \nonumber \\
    &= \mathbb{E} \left[ \left( \sum_{\sigma \in \mathcal{N}_{m_n}}\bo s_\sigma(\boldsymbol{\theta}_0)\right) \left( \sum_{\sigma' \in \mathcal{N}_{m_n}}\bo s_{\sigma'}(\boldsymbol{\theta}_0)\right)\right] \nonumber \\
    &= \sum_{\sigma \in \mathcal{N}_{m_n}} \sum_{\sigma' \in \mathcal{N}_{m_n}} \mathbb{E} \left[ \bo s_\sigma(\boldsymbol{\theta}_0) \bo s_{\sigma'}(\boldsymbol{\theta}_0)'\right]
\end{align}
From the conditional independence structure, the intermediate result 4, and the results in Step 2, the leading tem of the variance is composed by pairs of quadruples that share one dyad in common, since for pairs sharing no dyads in common, $\mathbb{E} \left[ \bo s_\sigma(\boldsymbol{\theta}_0) \bo s_{\sigma'}(\boldsymbol{\theta}_0)'\right] = 0$. There are $O(n^6)$ of such pairs, and each contributes $O(p_n)$ to the variance, leading to the rate of convergence above.

To understand why each pair sharing a dyad in common contributes $O(p_n)$ to the variance, first, remember the expression for the summands of the score function, for a quadruple $\sigma=\left\{(i, j),\left(i, j^{\prime}\right),\left(i^{\prime}, j\right),\left(i^{\prime}, j^{\prime}\right)\right\}$ is given by $$\boldsymbol{s}_\sigma (\boldsymbol{\theta}_0)=\boldsymbol{r}_\sigma\left[1\left\{z_\sigma=1\right\}\left(1-F\left(\boldsymbol{r}_\sigma^{\prime} \boldsymbol{\theta}_0\right)\right)-1\left\{z_\sigma=-1\right\} F\left(\boldsymbol{r}_\sigma^{\prime} \boldsymbol{\theta}_0\right)\right].$$ Note further that:
$$\operatorname{Var}(\bo s_\sigma (\boldsymbol{\theta}_0)) = \mathbb{E}[\bo s_\sigma (\boldsymbol{\theta}_0) \bo s_\sigma(\boldsymbol{\theta}_0)' \mid 1 \{ z_{\sigma} \in \{-1, 1\}\}] \times \operatorname{Pr}(z_{\sigma} \in \{-1, 1\})$$

Such that: (i) if $z_\sigma=1: \bo s_\sigma (\boldsymbol{\theta}_0)=\boldsymbol{r}_\sigma\left(1-\boldsymbol{F}\left(\boldsymbol{r}_\sigma^{\prime} \theta_0\right)\right)$; and (ii) if $z_\sigma=-1: \bo s_\sigma (\boldsymbol{\theta}_0)=-\boldsymbol{r}_\sigma F\left(\boldsymbol{r}_\sigma^{\prime} \theta_0\right)$. Since $0<F(x)<1$ for the logistic function, both $(1-F)^2 \leq 1$ and $F^2 \leq 1$.
Therefore:
$$\|\bo s_\sigma (\boldsymbol{\theta}_0)\|^2 = \|\boldsymbol{r}_\sigma\|^2 \times \begin{cases}
(1-F(\boldsymbol{r}_\sigma'\theta_0))^2 & \text{if } z_\sigma = 1\\
F(\boldsymbol{r}_\sigma'\theta_0)^2 & \text{if } z\sigma = -1
\end{cases}$$
Such that:
$$ \mathbb{E}[ \| \bo s_\sigma (\boldsymbol{\theta}_0) \|^2 \mid z_{\sigma} \in \{-1, 1\}] \leq \mathbb{E} [\| \bo r_\sigma \|^2] \times O(1) = O(1) $$
by Assumption \ref{assumption3} (bounded second moments of covariates). Given this result, it follows that:
$$
\|\operatorname{Var}(\bo s_\sigma (\boldsymbol{\theta}_0))\| \leq \mathbb{E}\left[\|\bo s_\sigma (\boldsymbol{\theta}_0)\|^2\right]=O(1) \times \operatorname{Pr}\left(z_\sigma \in\{-1,1\}\right)=O\left(p_\sigma\right)
$$
leading to the expression $
\sqrt{\|\operatorname{Var}(\bo s_\sigma (\boldsymbol{\theta}_0))\|}=O\left(\sqrt{p_\sigma}\right)
$. Going back to the expression of the variance of the scores, given by Equation \ref{variance_score}, by an application of the Cauchy-Schwarz inequality:
\begin{align}
    \operatorname{Var} (\boldsymbol{S}_n(\boldsymbol{\theta}_0)) &= \mathbb{E} [\boldsymbol{S}_n(\boldsymbol{\theta}_0) \boldsymbol{S}_n(\boldsymbol{\theta}_0)'] \nonumber \\
    &= \mathbb{E} \left[ \left( \sum_{\sigma \in \mathcal{N}_{m_n}}s_\sigma(\boldsymbol{\theta}_0)\right) \left( \sum_{\sigma' \in \mathcal{N}_{m_n}}s_{\sigma'}(\boldsymbol{\theta}_0)\right)\right] \nonumber \\
    &= \sum_{\sigma \in \mathcal{N}_{m_n}} \sum_{\sigma' \in \mathcal{N}_{m_n}} \mathbb{E} \left[ \bo s_\sigma(\boldsymbol{\theta}_0) \bo s_{\sigma'}(\boldsymbol{\theta}_0)'\right] \nonumber \\
    &= \sum_{\sigma \in \mathcal{N}_{m_n}} \sum_{\sigma' \in \mathcal{N}_{m_n}} 1 \{\text{dyads}(\sigma) \cap \text{dyads}(\sigma') \neq \emptyset \} \mathbb{E} \left[ \bo s_\sigma(\boldsymbol{\theta}_0) \bo s_{\sigma'}(\boldsymbol{\theta}_0)'\right] \nonumber \\
    &= \sum_{\sigma \in \mathcal{N}_{m_n}} \sum_{\sigma' \in \mathcal{N}_{m_n}} 1 \{\text{dyads}(\sigma) \cap \text{dyads}(\sigma') \neq \emptyset \} \operatorname{Cov} (\bo s_\sigma(\boldsymbol{\theta}_0), \bo s_{\sigma'}(\boldsymbol{\theta}_0)) \nonumber \\
    &\leq \sum_{\sigma \in \mathcal{N}_{m_n}} \sum_{\sigma' \in \mathcal{N}_{m_n}} 1 \{\text{dyads}(\sigma) \cap \text{dyads}(\sigma') \neq \emptyset \} \| \operatorname{Cov} (\bo s_\sigma(\boldsymbol{\theta}_0), \bo s_{\sigma'}(\boldsymbol{\theta}_0)) \| \nonumber \\
    & \leq \sum_{\sigma \in \mathcal{N}_{m_n}} \sum_{\sigma' \in \mathcal{N}_{m_n}} 1 \{\text{dyads}(\sigma) \cap \text{dyads}(\sigma') \neq \emptyset \} \sqrt{\| \operatorname{Var} (\bo s_\sigma(\boldsymbol{\theta}_0))\|} \sqrt{\| \operatorname{Var} (\bo s_{\sigma'}(\boldsymbol{\theta}_0))\|} \nonumber \\
    &= \sum_{\sigma \in \mathcal{N}_{m_n}} \sqrt{\| \operatorname{Var} (\bo s_\sigma(\boldsymbol{\theta} _0))\|} \sum_{\sigma' \in \mathcal{N}_{m_n}} 1 \{\text{dyads}(\sigma) \cap \text{dyads}(\sigma') \neq \emptyset \} \sqrt{\| \operatorname{Var} (\bo s_{\sigma'}(\boldsymbol{\theta}_0))\|} \nonumber \\
    &= \sum_{\sigma \in \mathcal{N}_{m_n}} O(\sqrt{p_\sigma}) \sum_{\sigma' \in \mathcal{N}_{m_n}}1 \{\text{dyads}(\sigma) \cap \text{dyads}(\sigma') \neq \emptyset \} O(\sqrt{p_{\sigma'}}) \nonumber \\
    &= O(n^6 p_n)
\end{align}
Where $1 \{\text{dyads}(\sigma) \cap \text{dyads}(\sigma') \neq \emptyset \}$ denotes if quadruples $\sigma$ and $\sigma'$ share a dyad in common. The final equality follows from analysing the two sums separately. First, focusing on the term $\sum_{\sigma \in \mathcal{N}_{m_n}} \sqrt{\left\|\operatorname{Var}\left(\bo s_\sigma(\boldsymbol{\theta} _0)\right)\right\|}=\sum_{\sigma \in \mathcal{N}_{m_n}} O\left(\sqrt{p_\sigma}\right)
$. By Jensen's inequality (since square root is concave) $
\frac{1}{m_n} \sum_\sigma \sqrt{p_\sigma} \leq \sqrt{\frac{1}{m_n} \sum_\sigma p_\sigma}=\sqrt{p_n}
$, and therefore $\sum_{\sigma \in \mathcal{N}_{m_n}} \sqrt{p_\sigma} \leq m_n \sqrt{p_n}=O\left(n^4 \sqrt{p_n}\right)$.

For the term $\sum_{\sigma' \in \mathcal{N}_{m_n}} 1 \{\text{dyads}(\sigma) \cap \text{dyads}(\sigma') \neq \emptyset \} \sqrt{\| \operatorname{Var} (s_{\sigma'}(\boldsymbol{\theta}_0))\|}$, note that for any fixed $\sigma$, there are $O\left(n^2\right)$ quadruples $\sigma^{\prime}$ sharing a dyad with $\sigma$. Then, similarly, applying Jensen's inequality to this subset:$
\sum_{\sigma' \in \mathcal{N}_{m_n}} 1 \{\text{dyads}(\sigma) \cap \text{dyads}(\sigma') \neq \emptyset \} \sqrt{p_{\sigma^{\prime}}}=O\left(n^2 \sqrt{p_n}\right)
$. Which leads to the desired result.

From that, the rate of convergence of the estimator is determined by the order of
\begin{align}
    \boldsymbol{\Omega}_n ( \boldsymbol{\theta}_0)&= \boldsymbol{H}_n( \boldsymbol{\theta}_0)^{-1} \boldsymbol{\Upsilon}_n ( \boldsymbol{\theta}_0) \boldsymbol{H}_n( \boldsymbol{\theta}_0)^{-1} \nonumber \\
    &= O_p \left( \frac{1}{n^4p_n}\right) \times O_p(n^6p_n) \times O_p \left( \frac{1}{n^4p_n}\right) = O_p\left( \frac{1}{n^2 p_n}\right)
\end{align}
Such that $\boldsymbol{\Omega}_n ( \boldsymbol{\theta}_0)^{-1/2} = O_p(\sqrt{n^2 p_n})$, which is asymptotically equivalent to $O_p(\sqrt{n(n-1p_n)})$, delivering the desired result that $\| \bo \theta_n - \bo \theta_0 \| = O_p\left( \frac{1}{\sqrt{n(n-1)p_n}}\right)$.\\
\\
\noindent \textit{Feasible inference.}\\
The limit distribution derived above uses the infeasible normalization $\boldsymbol{\Omega}_n(\boldsymbol{\theta}_0)$. Replacing $\boldsymbol{\theta}_0$ by the estimator relies on the following results:
\begin{enumerate}
    \item The stochastic equicontinuity bound for the Hessian established in Step 3, which gives $(m_n p_n)^{-1}\|\boldsymbol{H}_n(\boldsymbol{\theta}_1) - \boldsymbol{H}_n(\boldsymbol{\theta}_2)\| \leq O_p(1)\|\boldsymbol{\theta}_1 - \boldsymbol{\theta}_2\|$ uniformly over $\Theta$, so that, combined with the consistency of $\boldsymbol{\theta}_n$, $(m_n p_n)^{-1}\|\boldsymbol{H}_n(\boldsymbol{\theta}_n) - \boldsymbol{H}_n(\boldsymbol{\theta}_0)\| \xrightarrow{p} 0$.
    \item The corresponding bound for the score covariance at its own scale, $(n^6 p_n)^{-1}\|\boldsymbol{\Upsilon}_n(\boldsymbol{\theta}_1) - \boldsymbol{\Upsilon}_n(\boldsymbol{\theta}_2)\| \leq O_p(1)\|\boldsymbol{\theta}_1 - \boldsymbol{\theta}_2\|$, obtained by the same argument, so that $(n^6 p_n)^{-1}\|\boldsymbol{\Upsilon}_n(\boldsymbol{\theta}_n) - \boldsymbol{\Upsilon}_n(\boldsymbol{\theta}_0)\| \xrightarrow{p} 0$.
    \item The consistency of $\boldsymbol{\theta}_n \xrightarrow{p} \boldsymbol{\theta}_0$ (Theorem \ref{theorem1}).
\end{enumerate}

By Assumption \ref{assumption_scorevar}, item 2 above, and Weyl's inequality
$$\lambda_{\min}\left[(n^6 p_n)^{-1} \boldsymbol{\Upsilon}_n(\boldsymbol{\theta}_n)\right] \geq c - o_p(1) \geq c/2$$
with probability approaching one, so that $\boldsymbol{\Upsilon}_n(\boldsymbol{\theta}_n)$ is positive definite and $\|\boldsymbol{\Upsilon}_n(\boldsymbol{\theta}_n)^{-1}\| = O_p((n^6 p_n)^{-1})$. For the Hessian, Assumption \ref{assumption4} states that $\lim_{n \to \infty}(m_n p_n)^{-1}\mathbb{E}(\boldsymbol{H}_n(\boldsymbol{\theta}_0))$ has maximal rank. This limit is negative semidefinite, being the limit of negative semidefinite matrices, and therefore negative definite; equivalently, $-\lim_{n \to \infty}(m_n p_n)^{-1}\mathbb{E}(\boldsymbol{H}_n(\boldsymbol{\theta}_0))$ is positive definite, with smallest eigenvalue $c_H > 0$. By the pointwise convergence of the normalized Hessian established in Step 3, together with $m_n^*/(m_n p_n) \xrightarrow{p} 1$ from the proof of Theorem \ref{theorem1}, item 1, and Weyl's inequality,
$$\lambda_{\min}\left[-(m_n p_n)^{-1}\boldsymbol{H}_n(\boldsymbol{\theta}_n)\right] \geq c_H - o_p(1) \geq c_H/2$$
with probability approaching one, so that $\boldsymbol{H}_n(\boldsymbol{\theta}_n)$ is invertible and $\|(m_n p_n)\boldsymbol{H}_n(\boldsymbol{\theta}_n)^{-1}\| = O_p(1)$.

The feasible sandwich matrix $\boldsymbol{\Omega}_n(\boldsymbol{\theta}_n)$ is therefore positive definite with probability approaching one, so that $\boldsymbol{\Omega}_n(\boldsymbol{\theta}_n)^{-1/2}$ is well defined and the required matrix inversions are valid. Writing the common scale $(n(n-1)p_n)^{-1}$ explicitly, it cancels in the ratio $\boldsymbol{\Omega}_n(\boldsymbol{\theta}_n)^{-1/2}\boldsymbol{\Omega}_n(\boldsymbol{\theta}_0)^{1/2}$, and items 1 and 2 together with the eigenvalue bounds give $\|\boldsymbol{\Omega}_n(\boldsymbol{\theta}_n)^{-1/2}\boldsymbol{\Omega}_n(\boldsymbol{\theta}_0)^{1/2} - \boldsymbol{I}\| \xrightarrow{p} 0$. By Slutsky's theorem, the limit distribution of the estimator can therefore be restated as:
$$\boldsymbol{\Omega}_n(\boldsymbol{\theta}_n)^{-1/2} (\boldsymbol{\theta}_n - \boldsymbol{\theta}_0) \xrightarrow{d} N (0, \boldsymbol{I})$$

The proof of Theorem \ref{theorem2} is thus complete.

\section{Sparsity Conditions: Detailed Derivations}\label{appendix_sparsity}

This appendix collects the derivations supporting the sparsity conditions and parameterization discussed in Subsection~\ref{subsection_sparsity}.
\ref{appendix_sparsity_rates} derives the relations
$q_{n,y} \sim \sqrt{p_{n,y}}$ in the left tail and
$1 - q_{n,y} \sim \sqrt{p_{n,y}}$ in the right tail, from which the rate
restrictions on the estimable threshold range follow.
\ref{appendix_sparsity_heterogeneity} analyzes how heterogeneity
patterns shape the admissible range of $|g_y(n)|$.

\subsection{Asymptotic Relations Between \texorpdfstring{$p_{n,y}$ and $q_{n,y}$}{p and q}}\label{appendix_sparsity_rates}

Following \citet{jochmans2018semiparametric}, consider sequences of fixed
effects where $\alpha_{i,y}$ and $\gamma_{j,y}$ tend to $-\infty$, so that
the probability of falling below the threshold $y$ vanishes (the left tail
of the distribution). With covariates of bounded support, the expected
fraction of informative quadruples satisfies
$$
p_{n,y} \sim
\left(\frac{\sum_{i=1}^n e^{\alpha_{i,y}}}{n}\right)^2
\left(\frac{\sum_{i=1}^n e^{\gamma_{i,y}}}{n}\right)^2
$$
as $n$ increases, by the exponential tails of the logistic distribution.
Moreover, it follows that
$$
q_{n,y} \sim
\frac{\sum_{i=1}^n e^{\alpha_{i,y}}}{n} \cdot
\frac{\sum_{i=1}^n e^{\gamma_{i,y}}}{n},
$$
so that $q_{n,y} \sim \sqrt{p_{n,y}}$ in the left tail.

\citet{jochmans2018semiparametric} notes that networks with many ones (when
$q_{n,y} \to 1$) are equally challenging. In the DR setting, this case is
often equally or more relevant. When outcomes are bounded below at zero
with positive support, the left tail becomes irrelevant, and the constraint
on the estimable range of thresholds comes from the right tail. By the
logistic symmetry, the same structure holds in the right tail, when fixed
effects tend to $+\infty$ as $n$ grows:
$$
p_{n,y} \sim
\left(\frac{\sum_{i=1}^n e^{-\alpha_{i,y}}}{n}\right)^2
\left(\frac{\sum_{i=1}^n e^{-\gamma_{i,y}}}{n}\right)^2,
$$
with $1 - q_{n,y} \sim \sqrt{p_{n,y}}$. Therefore, the rate at which
$p_{n,y}$ shrinks follows the same structure in both tails: whether the
fixed effects tend to $-\infty$ or $+\infty$, reduced variation in the
indicators generates fewer informative quadruples.

The relations $q_{n,y} \sim \sqrt{p_{n,y}}$ and
$1 - q_{n,y} \sim \sqrt{p_{n,y}}$ in the left and right tails, respectively,
combined with the condition $np_{n,y} \to \infty$ yield rate restrictions on
the estimable thresholds: $\sqrt{n}\, q_{n,y} \to \infty$ in the left tail, and
$\sqrt{n}\,(1 - q_{n,y}) \to \infty$ in the right tail. These conditions illustrate
that accessing very extreme quantiles requires larger samples, regardless of
the distribution of node heterogeneity.

\subsection{Heterogeneity Patterns and the Admissible Range of \texorpdfstring{$|g_y(n)|$}{g\_y(n)}}\label{appendix_sparsity_heterogeneity}

The condition $np_{n,y} \to \infty$ can be rewritten as
$A_{n,y}^2 G_{n,y}^2 / n^3 \to \infty$, by defining
$A_{n,y} = \sum_{i=1}^n e^{\alpha_{i,y}}$ and
$G_{n,y} = \sum_{j=1}^n e^{\gamma_{j,y}}$ in the left tail (where
$\alpha_{i,y}, \gamma_{i,y} \to -\infty$), and
$A_{n,y} = \sum_{i=1}^n e^{-\alpha_{i,y}}$ and
$G_{n,y} = \sum_{j=1}^n e^{-\gamma_{j,y}}$ in the right tail (where
$\alpha_{i,y}, \gamma_{i,y} \to +\infty$). In both cases, the condition
requires $A_{n,y}^2 G_{n,y}^2 / n^3 \to \infty$, placing a lower bound on
how fast $A_{n,y}$ and $G_{n,y}$ can shrink as $n \to \infty$.

Under the parameterization $\alpha_{i,y} = a_{n,i}\, g_y(n)$ and
$\gamma_{j,y} = b_{n,j}\, g_y(n)$ of
Subsection~\ref{subsection_sparsity}, the behavior of $A_{n,y}$ and
$G_{n,y}$ is governed by $g_y(n)$ jointly with the heterogeneity sequences
$a_{n,i}$ and $b_{n,j}$. As $|g_y(n)|$ grows, that is, as the threshold
moves into a sparser region, the fixed effects are driven toward
$\pm\infty$, so that $A_{n,y}$ and $G_{n,y}$, being sums of
$e^{\pm a_{n,i} g_y(n)}$ and $e^{\pm b_{n,j} g_y(n)}$, shrink toward zero.
For the condition $A_{n,y}^2 G_{n,y}^2 / n^3 \to \infty$ to hold, the
shrinkage of $A_{n,y}$ and $G_{n,y}$ cannot be too fast; the largest
$|g_y(n)|$ for which the condition still holds therefore depends on the
heterogeneity pattern.

The link probability satisfies $q_{n,y} \sim (A_{n,y}/n)(G_{n,y}/n)$, so
different heterogeneity patterns yield different levels of $q_{n,y}$ for
the same $g_y(n)$. Since $A_{n,y}$ and $G_{n,y}$ are sums of the convex
transformations $a \mapsto e^{a\, g_y(n)}$ and $b \mapsto e^{b\, g_y(n)}$
of the heterogeneity sequences, Jensen's inequality implies that, for a
fixed mean, greater dispersion in $a_{n,i}$ and $b_{n,j}$ yields larger
$A_{n,y}$ and $G_{n,y}$ for the same
$g_y(n)$. The condition
$A_{n,y}^2 G_{n,y}^2 / n^3 \to \infty$ therefore continues to hold at
larger $|g_y(n)|$, so the largest admissible $|g_y(n)|$ is larger under
more dispersed heterogeneity.

Both $q_{n,y}$ and $p_{n,y}$ are population quantities determined by the
fixed effects, hence by $a_{n,i}$, $b_{n,j}$, and $g_y(n)$ through
$A_{n,y}$ and $G_{n,y}$, so the identification condition is ultimately a
restriction on these primitives. It nonetheless admits a characterization
in the link probability alone. From the relations in
\ref{appendix_sparsity_rates}, $p_{n,y} \sim q_{n,y}^2$; since
this holds up to constants bounded away from zero and infinity and
$np_{n,y} \to \infty$ is a rate condition, it is equivalent to
$\sqrt{n}\, q_{n,y} \to \infty$, giving the estimable range
$\sqrt{n}\, q_{n,y} \to \infty$ in the left tail and $\sqrt{n}\,(1 - q_{n,y}) \to \infty$ in
the right tail. Thus the same restriction has a clean characterization in
the link probability it induces, while its reach in $|g_y(n)|$ ---
established above --- depends on the heterogeneity pattern through the
mapping $g_y(n) \mapsto q_{n,y}$; the two are not in tension. Specific
parameterization and their implications are explored in
Section~\ref{simulations}.

\section{Proof of Theorem \ref{theorem3}}\label{joint_distribution_appendix}

Throughout this appendix, threshold-specific objects are indexed by $k$ rather than by $y_k$ where this is more compact; in particular, $p_{n,k} \equiv p_{n,y_k}$ denotes the informativeness parameter at threshold $y_k$, and similarly for the score, Hessian, and covariance blocks.

\noindent \textit{Set up.}\\
Since the estimation for each threshold is done separately, I start with the first order Taylor expansions to the first order condition of the log-likelihood optimization problem around the true values of the parameters. Let $\mathbf{y} = (y_1, ..., y_K)'$ denote the vector of K ordered thresholds where $y_1 < y_2 < ... < y_K$, with each $y_k \in \mathcal{Y}$ (the support of the outcome variable). For each threshold $y_k$, the parameter vector $\boldsymbol{\theta}_{y_k} \in \mathbb{R}^p$ represents the $p$-dimensional coefficient vector at that threshold.

For each threshold $k = 1, ..., K$:
\begin{align}
\boldsymbol{\theta}_{n,y_k} - \boldsymbol{\theta}_{y_{k,0}} = -\boldsymbol{H}_{n,k}(\boldsymbol{\theta}_{y_k,*})^{-1} \boldsymbol{S}_{n,k}(\boldsymbol{\theta}_{y_{k,0}})
\end{align}
where $\boldsymbol{\theta}_{y_k,*}$ lies between $\boldsymbol{\theta}_{n,y_k}$ and $\boldsymbol{\theta}_{y_{k,0}}$ by the mean value theorem.

I proceed by stacking these K expansions into a single system. Define the stacked parameter vector $\boldsymbol{\theta}_{\mathbf{y}} = (\boldsymbol{\theta}_{y_1}', ..., \boldsymbol{\theta}_{y_K}')' \in \mathbb{R}^{Kp}$ and similarly for the estimates and true values:
\begin{adjustwidth}{-0.25in}{-0.25in}
\begin{align}
\begin{pmatrix} 
\boldsymbol{\theta}_{n,y_1} - \boldsymbol{\theta}_{y_{1,0}} \\ 
\vdots \\
\boldsymbol{\theta}_{n,y_K} - \boldsymbol{\theta}_{y_{K,0}}
\end{pmatrix}_{Kp \times 1} = 
-\begin{pmatrix}
\boldsymbol{H}_{n,1}(\boldsymbol{\theta}_{y_1,*})^{-1} & 0 & \cdots & 0 \\ 
0 & \boldsymbol{H}_{n,2}(\boldsymbol{\theta}_{y_2,*})^{-1} & \cdots & 0 \\
\vdots & \vdots & \ddots & \vdots \\
0 & 0 & \cdots & \boldsymbol{H}_{n,K}(\boldsymbol{\theta}_{y_K,*})^{-1}
\end{pmatrix}_{Kp \times Kp}
\begin{pmatrix} 
\boldsymbol{S}_{n,1}(\boldsymbol{\theta}_{y_{1,0}}) \\ 
\vdots \\
\boldsymbol{S}_{n,K}(\boldsymbol{\theta}_{y_{K,0}})
\end{pmatrix}_{Kp \times 1}
\end{align}
\end{adjustwidth}
The block diagonal structure of the Hessian matrix reflects the fact that the log-likelihood at each threshold is maximized independently, with no cross-threshold restrictions on the parameters.

The matrices $\boldsymbol{H}_{n,k}(\boldsymbol{\theta}_{y_k,*})^{-1}$ for $k = 1, ..., K$ are the inverse Hessian matrices evaluated at intermediate values $\boldsymbol{\theta}_{y_k,*}$. These intermediate values arise from applying the mean value theorem to the first-order conditions at each threshold separately. Specifically, for each $k$, the vector $\boldsymbol{\theta}_{y_k,*}$ lies between the true parameter $\boldsymbol{\theta}_{y_{k,0}}$ and the estimator $\boldsymbol{\theta}_{n,y_k}$. 

Importantly, these intermediate values are threshold-specific, that is, $\boldsymbol{\theta}_{y_1,*}, ..., \boldsymbol{\theta}_{y_K,*}$ are generally all different because they come from separate Taylor expansions of the first-order conditions for each threshold. Since $\boldsymbol{\theta}_{n,y_k} \xrightarrow{p} \boldsymbol{\theta}_{y_{k,0}}$ for each $k = 1, ..., K$ by the consistency established in Theorem \ref{theorem1}, it follows that $\boldsymbol{\theta}_{y_k,*} \xrightarrow{p} \boldsymbol{\theta}_{y_{k,0}}$ as $n \rightarrow \infty$. 

This convergence, combined with the stochastic equicontinuity bound for the normalized Hessian 
established in the proof of Theorem \ref{theorem2}, gives 
$(m_n p_{n,y_k})^{-1}\|\boldsymbol{H}_{n,k}(\boldsymbol{\theta}_{y_k,*}) - 
\boldsymbol{H}_{n,k}(\boldsymbol{\theta}_{y_{k,0}})\| \xrightarrow{p} 0$ for each $k$, so that, on 
the event where the matrices are invertible, each 
$\boldsymbol{H}_{n,k}(\boldsymbol{\theta}_{y_k,*})^{-1}$ may be replaced by 
$\boldsymbol{H}_{n,k}(\boldsymbol{\theta}_{y_{k,0}})^{-1}$ up to a remainder that is negligible at 
the block-specific scale. Invertibility at each of these arguments follows from the eigenvalue 
floor on the normalized Hessian established in the proof of Theorem \ref{theorem2}, applied 
threshold by threshold, together with the same equicontinuity bound. This substitution is valid 
simultaneously for all $K$ thresholds.

Note that it is also equivalent to write this expression as:
\begin{align} \label{joint_expansion}
    \begin{pmatrix} 
    \boldsymbol{\theta}_{n,y_1} - \boldsymbol{\theta}_{y_{1,0}} \\ 
    \vdots \\
    \boldsymbol{\theta}_{n,y_K} - \boldsymbol{\theta}_{y_{K,0}}
    \end{pmatrix} = -
    \begin{pmatrix}
    \boldsymbol{H}_{n,1}(\boldsymbol{\theta}_{y_1,0}) & 0 & \cdots & 0 \\ 
    0 & \boldsymbol{H}_{n,2}(\boldsymbol{\theta}_{y_2,0}) & \cdots & 0 \\
    \vdots & \vdots & \ddots & \vdots \\
    0 & 0 & \cdots & \boldsymbol{H}_{n,K}(\boldsymbol{\theta}_{y_K,0})
    \end{pmatrix}^{-1} 
    \begin{pmatrix} 
    \boldsymbol{S}_{n,1}(\boldsymbol{\theta}_{y_{1,0}}) \\ 
    \vdots \\
    \boldsymbol{S}_{n,K}(\boldsymbol{\theta}_{y_{K,0}}) 
    \end{pmatrix}
\end{align}
This equivalence will be useful when defining the joint Hessian matrix compactly.

Notice that the joint distribution of the estimators at different thresholds will be driven by the 
joint distributions of their score vectors. Let $\boldsymbol{\theta}_{\mathbf{y},0} = 
(\boldsymbol{\theta}_{y_1,0}', ..., \boldsymbol{\theta}_{y_K,0}')'$ denote the stacked true 
parameter vector, and write $\boldsymbol{H}_{n,\mathbf{y}}(\boldsymbol{\theta}_{\mathbf{y},0}) := 
\operatorname{diag}\left(\boldsymbol{H}_{n,1}(\boldsymbol{\theta}_{y_1,0}), \ldots, 
\boldsymbol{H}_{n,K}(\boldsymbol{\theta}_{y_K,0})\right)$ for the joint Hessian. The joint score 
covariance matrix is:

$$\boldsymbol{\Upsilon}_{n,\mathbf{y}}(\boldsymbol{\theta}_{\mathbf{y},0}) = 
\begin{pmatrix}
\boldsymbol{\Upsilon}_{n,11}(\boldsymbol{\theta}_{y_1,0}) & \boldsymbol{\Upsilon}_{n,12}(\boldsymbol{\theta}_{y_1,0}, \boldsymbol{\theta}_{y_2,0}) & \cdots & \boldsymbol{\Upsilon}_{n,1K}(\boldsymbol{\theta}_{y_1,0}, \boldsymbol{\theta}_{y_K,0})\\
\boldsymbol{\Upsilon}_{n,21}(\boldsymbol{\theta}_{y_2,0}, \boldsymbol{\theta}_{y_1,0}) & \boldsymbol{\Upsilon}_{n,22}(\boldsymbol{\theta}_{y_2,0}) & \cdots & \boldsymbol{\Upsilon}_{n,2K}(\boldsymbol{\theta}_{y_2,0}, \boldsymbol{\theta}_{y_K,0})\\
\vdots & \vdots & \ddots & \vdots \\
\boldsymbol{\Upsilon}_{n,K1}(\boldsymbol{\theta}_{y_K,0}, \boldsymbol{\theta}_{y_1,0}) & \boldsymbol{\Upsilon}_{n,K2}(\boldsymbol{\theta}_{y_K,0}, \boldsymbol{\theta}_{y_2,0}) & \cdots & \boldsymbol{\Upsilon}_{n,KK}(\boldsymbol{\theta}_{y_K,0})
\end{pmatrix}_{Kp \times Kp}$$

where each block is given by
\begin{adjustwidth}{-0.25in}{-0.25in}
\begin{align} \label{upsilon_n_kl}
    \boldsymbol{\Upsilon}_{n,kl}(\boldsymbol{\theta}_{y_k}, \boldsymbol{\theta}_{y_l}) = \sum_{i=1}^n \sum_{j \neq i} \sum_{i' \neq i,j} \sum_{j' \neq i,j,i'} \sum_{i'' \neq i,j,i'} \sum_{j'' \neq i,j,j',i''} 16 \times \Big[ \boldsymbol{s}_k(\sigma\{i,i';j,j'\}; \boldsymbol{\theta}_{y_k}) \boldsymbol{s}_l(\sigma\{i,i'';j,j''\}; \boldsymbol{\theta}_{y_l})' \Big],
\end{align}
\end{adjustwidth}
with $\boldsymbol{\Upsilon}_{n,kl}(\boldsymbol{\theta}_{y_k,0}, \boldsymbol{\theta}_{y_l,0}) = \boldsymbol{\Upsilon}_{n,lk}(\boldsymbol{\theta}_{y_l,0}, \boldsymbol{\theta}_{y_k,0})'$ by symmetry. The index restrictions $i'' \neq i'$ and $j'' \neq j'$ ensure that the two quadruples share only the dyad $(i,j)$. The joint sandwich variance matrix is:
$$\boldsymbol{\Omega}_{n,\mathbf{y}}(\boldsymbol{\theta}_{\mathbf{y},0}) = 
\boldsymbol{H}_{n,\mathbf{y}}^{-1}(\boldsymbol{\theta}_{\mathbf{y},0}) 
\boldsymbol{\Upsilon}_{n,\mathbf{y}}(\boldsymbol{\theta}_{\mathbf{y},0}) 
\boldsymbol{H}_{n,\mathbf{y}}^{-1}(\boldsymbol{\theta}_{\mathbf{y},0})'$$

Which delivers the following expression:
$$\boldsymbol{\Omega}_{n,\mathbf{y}}(\boldsymbol{\theta}_{\mathbf{y},0}) = 
\begin{pmatrix}
[\boldsymbol{\Omega}_{n,\mathbf{y}}]_{11} & [\boldsymbol{\Omega}_{n,\mathbf{y}}]_{12} & \cdots & [\boldsymbol{\Omega}_{n,\mathbf{y}}]_{1K} \\
[\boldsymbol{\Omega}_{n,\mathbf{y}}]_{21} & [\boldsymbol{\Omega}_{n,\mathbf{y}}]_{22} & \cdots & [\boldsymbol{\Omega}_{n,\mathbf{y}}]_{2K} \\
\vdots & \vdots & \ddots & \vdots \\
[\boldsymbol{\Omega}_{n,\mathbf{y}}]_{K1} & [\boldsymbol{\Omega}_{n,\mathbf{y}}]_{K2} & \cdots & [\boldsymbol{\Omega}_{n,\mathbf{y}}]_{KK}
\end{pmatrix}$$

where each block is given by:
$$[\boldsymbol{\Omega}_{n,\mathbf{y}}(\boldsymbol{\theta}_{\mathbf{y},0})]_{kl} = \boldsymbol{H}_{n,k}(\boldsymbol{\theta}_{y_k,0})^{-1}\boldsymbol{\Upsilon}_{n,kl}(\boldsymbol{\theta}_{y_k,0}, \boldsymbol{\theta}_{y_l,0})\boldsymbol{H}_{n,l}(\boldsymbol{\theta}_{y_l,0})^{-1}$$

The diagonal blocks $[\boldsymbol{\Omega}_{n,\mathbf{y}}(\boldsymbol{\theta}_{\mathbf{y},0})]_{kk} = \boldsymbol{H}_{n,k}(\boldsymbol{\theta}_{y_k,0})^{-1}\boldsymbol{\Upsilon}_{n,kk}(\boldsymbol{\theta}_{y_k,0})\boldsymbol{H}_{n,k}(\boldsymbol{\theta}_{y_k,0})^{-1}$ are the threshold-specific sandwich variances from the single-threshold case. The off-diagonal blocks $[\boldsymbol{\Omega}_{n,\mathbf{y}}(\boldsymbol{\theta}_{\mathbf{y},0})]_{kl}$ for $k \neq l$ capture the cross-threshold dependence, where $\boldsymbol{\Upsilon}_{n,kl}(\boldsymbol{\theta}_{y_k,0}, \boldsymbol{\theta}_{y_l,0})$ is the covariance between the score vectors at thresholds $y_k$ and $y_l$. 

Recall also the score-rate matrix $\boldsymbol{D}_{n,\mathbf{y}} = \operatorname{diag}\left((n^6 p_{n,1})^{1/2}\boldsymbol{I}_p, \ldots, (n^6 p_{n,K})^{1/2}\boldsymbol{I}_p\right)$ introduced in Section \ref{joint}, where $n^6 p_{n,k}$ is the order of the score covariance at threshold $y_k$ established in \ref{appendix_derivation}, arising from the $O(n^6)$ pairs of quadruples sharing exactly one dyad, each contributing $O(p_{n,k})$. The cross-threshold score covariance blocks are of order $n^6\sqrt{p_{n,k}p_{n,l}}$, as derived in the convergence-rate calculation below. Consequently, a linear combination $\sum_k \boldsymbol{c}_k' \boldsymbol{S}_{n,k}$ has variance of order $n^6 \sum_{k=1}^K \sum_{l=1}^K \|\boldsymbol{c}_k\| \|\boldsymbol{c}_l\| \sqrt{p_{n,k}p_{n,l}}$, which depends on the thresholds on which $\boldsymbol{c}$ loads. No single scalar rate normalizes the joint object, and the matrix $\boldsymbol{D}_{n,\mathbf{y}}$ is used below to normalize the joint score covariance blockwise.

Following the same approach as in the single-threshold case, I establish the joint limit distribution by first obtaining the limit distribution of the normalized stacked score vector, and then transferring it to the estimator through the stacked mean-value expansion in \eqref{joint_expansion}. Because the blocks converge at different rates, the normalization is carried out block by block. The goal is to establish parts (i)--(iii) of Theorem \ref{theorem3}: the joint Gaussian approximation for the coordinatewise studentized estimators, the limit distribution of fixed studentized linear combinations, and the $\chi^2$ limit of the Wald statistic.

Similarly to before, the proof proceeds in different steps:

\noindent\textbf{Step 1.} Establish the joint asymptotic equivalence between the score vectors and their projections across all K thresholds. Specifically, characterize the joint projections $\boldsymbol{V}_{n,k}(\boldsymbol{\theta}_{y_{k,0}})$ for $k = 1, ..., K$, and show that the stacked score vector $(\boldsymbol{S}_{n,1}', ..., \boldsymbol{S}_{n,K}')'$ is asymptotically equivalent to the stacked projection vector $(\boldsymbol{V}_{n,1}', ..., \boldsymbol{V}_{n,K}')'$.

\noindent\textbf{Step 2.} Obtain the joint limit distribution of the projections by applying the Cram\'{e}r-Wold theorem across thresholds and a conditional version of Lyapunov's CLT, express this limit under the block-rate normalization by the conditional covariance matrix $\boldsymbol{\Upsilon}_{X,\mathbf{y}}$, and transfer it from the projections to the stacked score.

\noindent\textbf{Step 3.} Combine the results from Steps 1 and 2 with the stacked mean-value expansion of the first-order conditions around the true values to obtain the joint Gaussian approximation for the studentized estimators.

\noindent\textbf{Step 4.} Determine the threshold-specific convergence rates from the block orders of the joint sandwich covariance.

\noindent\textbf{Step 5.} Establish that the plug-in covariance may replace the truth-evaluated one, and derive parts (i)--(iii) of Theorem \ref{theorem3}.\\

The cross-threshold block orders of the score covariance are derived in Step 4; they are used 
earlier in the proof, where they are cited as established there. The diagonal blocks and the 
Hessian orders are inherited from the single-threshold case.\\

\noindent \textbf{Step 1.}\\
Given $K$ thresholds $\mathbf{y} = (y_1, ..., y_K)'$, recall that our model allows for threshold-varying fixed effects, where $\alpha_{i,y_k}$ and $\gamma_{j,y_k}$ can differ across thresholds. Define the threshold-specific information sets $\mathcal{F}_{n,k} = \{\{\boldsymbol{x}_{ij}\}_{n,n}, \{\alpha_{i,y_k}\}_n, \{\gamma_{j,y_k}\}_n\}$ for $k = 1, ..., K$. 

The projection of the score at threshold $y_k$, evaluated at the true parameters, is given by:
\begin{align} \label{hajek_k}
    \boldsymbol{V}_{n,k} (\boldsymbol{\theta}_{y_{k,0}}) &=\sum_{i=1}^n \sum_{j \neq i} \sum_{i' \neq i, j} \sum_{j' \neq i,j,i'} \mathbbm{E}(\boldsymbol{s}_k(\sigma \{i,i';j,j'\}; \boldsymbol{\theta}_{y_{k,0}}) \mid \tilde{y}_{k,ij}, \mathcal{F}_{n,k}) \nonumber \\
    &= \sum_{i=1}^n \sum_{j\neq i} \boldsymbol{v}_{k,ij}(\boldsymbol{\theta}_{y_{k,0}})
\end{align}
where $\boldsymbol{v}_{k,ij}(\boldsymbol{\theta}_{y_{k,0}}) = \sum_{i' \neq i, j} \sum_{j' \neq i,j,i'} \mathbbm{E}(\boldsymbol{s}_k(\sigma \{i,i';j,j'\}; \boldsymbol{\theta}_{y_{k,0}}) \mid \tilde{y}_{k,ij}, \mathcal{F}_{n,k})$ and $\tilde{y}_{k,ij} = \mathbf{1}\{y_{ij} \leq y_k\}$ is the binary indicator at threshold $y_k$. 

The stacked projection vector for all K thresholds is:
\begin{align}
\boldsymbol{V}_{n,\mathbf{y}}(\boldsymbol{\theta}_{\mathbf{y},0}) = \begin{pmatrix}
\boldsymbol{V}_{n,1}(\boldsymbol{\theta}_{y_{1,0}}) \\
\vdots \\
\boldsymbol{V}_{n,K}(\boldsymbol{\theta}_{y_{K,0}})
\end{pmatrix} = \sum_{i=1}^n \sum_{j\neq i} \begin{pmatrix}
\boldsymbol{v}_{1,ij}(\boldsymbol{\theta}_{y_{1,0}}) \\
\vdots \\
\boldsymbol{v}_{K,ij}(\boldsymbol{\theta}_{y_{K,0}})
\end{pmatrix}
\end{align}
Note that the intermediate results 1-4 hold for all thresholds. Furthermore, intermediate results 3 and 4 are extended to the joint threshold setting.\\
\\
\noindent \textit{Intermediate result 3'.}\\
Because link decisions are conditionally independent across dyads given the information sets, for any pair of thresholds $k, l \in \{1, ..., K\}$:
$$\begin{aligned} 
\mathbb{E}\left[\boldsymbol{v}_{k,ij}\left(\boldsymbol{\theta}_{y_{k,0}}\right) \boldsymbol{v}_{l,i'j'}\left(\boldsymbol{\theta}_{y_{l,0}}\right)'\right] 
&= \mathbb{E}\left[\mathbb{E}\left[\boldsymbol{v}_{k,ij}\left(\boldsymbol{\theta}_{y_{k,0}}\right) \boldsymbol{v}_{l,i'j'}\left(\boldsymbol{\theta}_{y_{l,0}}\right)' \mid \mathcal{F}_{n,k}, \mathcal{F}_{n,l}\right]\right] \\ 
&= \mathbb{E}\left[\mathbb{E}\left[\boldsymbol{v}_{k,ij}\left(\boldsymbol{\theta}_{y_{k,0}}\right) \mid \mathcal{F}_{n,k}, \mathcal{F}_{n,l}\right] \mathbb{E}\left[\boldsymbol{v}_{l,i'j'}\left(\boldsymbol{\theta}_{y_{l,0}}\right) \mid \mathcal{F}_{n,k}, \mathcal{F}_{n,l}\right]'\right]=0
\end{aligned}$$
unless $i=i'$ and $j=j'$. 

The main difference here with respect to intermediate result 3 is that the conditioning is done on both information sets. Moreover, this result takes into account that when $i=i'$ and $j=j'$, the projections at different thresholds are correlated because: (i) they depend on the same underlying outcome $y_{ij}$, such that naturally, the idiosyncratic error terms for a dyad across different thresholds can be correlated; (ii) the binary indicators $\tilde{y}_{k,ij}$ and $\tilde{y}_{l,ij}$ are related by monotonicity (if $y_k < y_l$, then $\tilde{y}_{k,ij} = 1 \implies \tilde{y}_{l,ij} = 1$); (iii) the fixed effects may be correlated across thresholds (e.g., $\alpha_{i,y_k}$ and $\alpha_{i,y_l}$ can be dependent). This cross-threshold correlation for the same dyad is precisely what generates the off-diagonal blocks $\boldsymbol{\Upsilon}_{kl}$ in the joint covariance matrix.\\
\\
\noindent \textit{Intermediate result 4'.}\\
From the conditional independence of links, it follows that for any pair of thresholds $k, l \in \{1, ..., K\}$:
$$\mathbb{E}\left[\boldsymbol{s}_{k,\sigma}\left(\boldsymbol{\theta}_{y_{k,0}}\right) \boldsymbol{s}_{l,\sigma'}\left(\boldsymbol{\theta}_{y_{l,0}}\right)'\right]=\mathbb{E}\left[\mathbb{E}\left[\boldsymbol{s}_{k,\sigma}\left(\boldsymbol{\theta}_{y_{k,0}}\right) \boldsymbol{s}_{l,\sigma'}\left(\boldsymbol{\theta}_{y_{l,0}}\right)' \mid \mathcal{F}_{n,k}, \mathcal{F}_{n,l}\right]\right]= 0$$
unless $\sigma$ and $\sigma'$ share at least one dyad in common.

When $\sigma$ and $\sigma'$ share exactly one dyad in common, the cross-threshold covariance is non-zero because both score contributions depend on the same underlying continuous outcome $y_{ij}$ for that shared dyad. This cross-dependence at the score level, aggregated over all quadruple pairs sharing dyads, contributes to the off-diagonal block $\boldsymbol{\Upsilon}_{kl}$ in the joint score covariance matrix.\\
\\
In order to establish the asymptotic equivalence between the vector of scores and the vector of projections, I establish that:
\begin{align}
    &\lim_{n \to \infty} \boldsymbol{\Upsilon}_{\mathbf{y}}^{-1/2} 
    \mathbb{E} \left[ 
        \begin{pmatrix} 
        \boldsymbol{V}_{n,1}(\boldsymbol{\theta}_{y_{1,0}}) - \boldsymbol{S}_{n,1}(\boldsymbol{\theta}_{y_{1,0}}) \\ 
        \vdots \\
        \boldsymbol{V}_{n,K}(\boldsymbol{\theta}_{y_{K,0}}) - \boldsymbol{S}_{n,K}(\boldsymbol{\theta}_{y_{K,0}}) 
        \end{pmatrix}
        \begin{pmatrix} 
        \boldsymbol{V}_{n,1}(\boldsymbol{\theta}_{y_{1,0}}) - \boldsymbol{S}_{n,1}(\boldsymbol{\theta}_{y_{1,0}}) \\ 
        \vdots \\
        \boldsymbol{V}_{n,K}(\boldsymbol{\theta}_{y_{K,0}}) - \boldsymbol{S}_{n,K}(\boldsymbol{\theta}_{y_{K,0}}) 
        \end{pmatrix}'
    \right]
    \boldsymbol{\Upsilon}_{\mathbf{y}}^{-1/2'} = 0
\end{align}

where $\boldsymbol{\Upsilon}_{\mathbf{y}}$ is the $(Kp) \times (Kp)$ joint population covariance matrix of the projections with blocks:
$$\boldsymbol{\Upsilon}_{\mathbf{y}} = \begin{pmatrix}
\boldsymbol{\Upsilon}_{11} & \boldsymbol{\Upsilon}_{12} & \cdots & \boldsymbol{\Upsilon}_{1K} \\
\boldsymbol{\Upsilon}_{21} & \boldsymbol{\Upsilon}_{22} & \cdots & \boldsymbol{\Upsilon}_{2K} \\
\vdots & \vdots & \ddots & \vdots \\
\boldsymbol{\Upsilon}_{K1} & \boldsymbol{\Upsilon}_{K2} & \cdots & \boldsymbol{\Upsilon}_{KK}
\end{pmatrix}$$

The diagonal blocks $\boldsymbol{\Upsilon}_{kk}$ denote the population variance of the projection for threshold $y_k$, which were defined previously. The off-diagonal blocks $\boldsymbol{\Upsilon}_{kl} = \mathbb{E}[\boldsymbol{V}_{n,k}(\boldsymbol{\theta}_{y_{k,0}}) \boldsymbol{V}_{n,l}(\boldsymbol{\theta}_{y_{l,0}})']$ denote the population covariance between projections at thresholds $y_k$ and $y_l$, with $\boldsymbol{\Upsilon}_{kl} = \boldsymbol{\Upsilon}_{lk}'$.

Using the intermediate result 3', the expression for the covariance matrix can be further pinned down. For any pair of thresholds $k, l \in \{1, ..., K\}$:
\begin{align}
    \boldsymbol{\Upsilon}_{kl} &= \mathbb{E} [\boldsymbol{V}_{n,k}(\boldsymbol{\theta}_{y_{k,0}}) \boldsymbol{V}_{n,l}(\boldsymbol{\theta}_{y_{l,0}})'] \nonumber \\
    &= \sum_{i=1}^n \sum_{j \neq i} \mathbb{E} [\boldsymbol{v}_{k,ij}(\boldsymbol{\theta}_{y_{k,0}}) \boldsymbol{v}_{l,ij}(\boldsymbol{\theta}_{y_{l,0}})'] \nonumber \\
    &= \sum_{i=1}^n \sum_{j \neq i} \mathbb{E} \Big[ \sum_{i' \neq i,j} \sum_{j' \neq i,j,i'} \mathbb{E} [\boldsymbol{s}_k(\sigma\{i,i';j,j'\}, \boldsymbol{\theta}_{y_{k,0}}) \mid \tilde{y}_{k,ij}, \mathcal{F}_{n,k}] \nonumber \\ 
    &\hspace{2cm}\sum_{i'' \neq i,j} \sum_{j'' \neq i,j,i''} \mathbb{E} [\boldsymbol{s}_l(\sigma\{i,i'';j,j''\}, \boldsymbol{\theta}_{y_{l,0}}) \mid \tilde{y}_{l,ij}, \mathcal{F}_{n,l}]' \Big] \nonumber \\
    & = \sum_{i=1}^n \sum_{j \neq i} \sum_{i' \neq i,j} \sum_{j' \neq i,j,i'} \sum_{i'' \neq i,j} \sum_{j'' \neq i,j,i''} \mathbb{E} \Big[ \mathbb{E} [\boldsymbol{s}_k(\sigma\{i,i';j,j'\}, \boldsymbol{\theta}_{y_{k,0}}) \mid \tilde{y}_{k,ij}, \mathcal{F}_{n,k}] \nonumber \\ 
    &\hspace{2cm} \mathbb{E} [\boldsymbol{s}_l(\sigma\{i,i'';j,j''\}, \boldsymbol{\theta}_{y_{l,0}}) \mid \tilde{y}_{l,ij}, \mathcal{F}_{n,l}]'\Big] 
\end{align}

This expression holds for all $K(K+1)/2$ unique blocks of the symmetric covariance matrix $\boldsymbol{\Upsilon}_{\mathbf{y}}$, with $\boldsymbol{\Upsilon}_{kl} = \boldsymbol{\Upsilon}_{lk}'$. Once again, due to the same reasoning as before, we can focus on the leading term of this covariance. The leading term arises from quadruples at thresholds $k$ and $l$ that share exactly one dyad in common, analogous to the single-threshold case. Thus, for any pair $k, l \in \{1, ..., K\}$:
\begin{align}
    \boldsymbol{\Upsilon}_{kl,\ell} &= \sum_{i=1}^n \sum_{j \neq i} \sum_{i' \neq i,j} \sum_{j' \neq i,j,i'} \sum_{i'' \neq i,j, i'} \sum_{j'' \neq i,j,j',i''} \mathbb{E} \Big[ \mathbb{E} [\boldsymbol{s}_k(\sigma\{i,i';j,j'\}, \boldsymbol{\theta}_{y_{k,0}}) \mid \tilde{y}_{k,ij}, \mathcal{F}_{n,k}] \nonumber \\ 
    &\hspace{2cm} \mathbb{E} [\boldsymbol{s}_l(\sigma\{i,i'';j,j''\}, \boldsymbol{\theta}_{y_{l,0}}) \mid \tilde{y}_{l,ij}, \mathcal{F}_{n,l}]'\Big] \nonumber \\ 
    &= \sum_{i=1}^n \sum_{j \neq i} \sum_{i' \neq i,j} \sum_{j' \neq i,j,i'} \sum_{i'' \neq i,j, i'} \sum_{j'' \neq i,j,j',i''}  \mathbb{E} \Big[ \mathbb{E} [\boldsymbol{s}_k(\sigma\{i,i';j,j'\}, \boldsymbol{\theta}_{y_{k,0}})  \nonumber \\ 
    &\hspace{2cm} \boldsymbol{s}_l(\sigma\{i,i'';j,j''\}, \boldsymbol{\theta}_{y_{l,0}}) \mid \tilde{y}_{k,ij}, \mathcal{F}_{n,k}, \tilde{y}_{l,ij}, \mathcal{F}_{n,l}]'\Big] \nonumber \\ 
    &= \sum_{i=1}^n \sum_{j \neq i} \sum_{i' \neq i,j} \sum_{j' \neq i,j,i'} \sum_{i'' \neq i,j, i'} \sum_{j'' \neq i,j,j',i''}  \mathbb{E} \Big[ \boldsymbol{s}_k(\sigma\{i,i';j,j'\}, \boldsymbol{\theta}_{y_{k,0}}) \boldsymbol{s}_l(\sigma\{i,i'';j,j''\}, \boldsymbol{\theta}_{y_{l,0}})' \Big] 
\end{align}
where the condition $i'' \neq i'$ and $j'' \neq j',i''$ ensures that the quadruples $\sigma\{i,i';j,j'\}$ and $\sigma\{i,i'';j,j''\}$ share only the dyad $(i,j)$ in common, making them conditionally independent given $(\tilde{y}_{k,ij}, \tilde{y}_{l,ij}, \mathcal{F}_{n,k}, \mathcal{F}_{n,l})$. This conditional independence is crucial for the second equality above.

Collecting these blocks gives the leading term of the joint projection covariance,
$$\boldsymbol{\Upsilon}_{\ell,\mathbf{y}} = \begin{pmatrix}
\boldsymbol{\Upsilon}_{11,\ell} & \boldsymbol{\Upsilon}_{12,\ell} & \cdots & \boldsymbol{\Upsilon}_{1K,\ell} \\
\boldsymbol{\Upsilon}_{21,\ell} & \boldsymbol{\Upsilon}_{22,\ell} & \cdots & \boldsymbol{\Upsilon}_{2K,\ell} \\
\vdots & \vdots & \ddots & \vdots \\
\boldsymbol{\Upsilon}_{K1,\ell} & \boldsymbol{\Upsilon}_{K2,\ell} & \cdots & \boldsymbol{\Upsilon}_{KK,\ell}
\end{pmatrix},$$
which is used here only to identify the leading covariance blocks and their factor-of-16 relation with the score covariance, established next. Formal nondegeneracy of the joint covariance is imposed in Assumption \ref{assumption_jointscorevar}, and its implications for the joint projection covariance are established in Step 2; the normalizations by $\boldsymbol{\Upsilon}_{\mathbf{y}}^{-1/2}$ used below are well defined on that basis for sufficiently large $n$.

The next step is to show that the leading term of the covariance matrix of the scores is asymptotically the same as the leading term of the covariance matrix for the projections given above. For any pair of thresholds $k, l \in \{1, ..., K\}$, the covariance of the scores is given by $\mathbb{E}[\boldsymbol{S}_{n,k}(\boldsymbol{\theta}_{y_{k,0}}) \boldsymbol{S}_{n,l}(\boldsymbol{\theta}_{y_{l,0}})']$. 

Following intermediate result 4', the leading term comprises pairs of quadruples (one at threshold $k$, one at threshold $l$) sharing exactly one dyad in common. In the projection covariance $\boldsymbol{\Upsilon}_{kl,\ell}$, the shared dyad $(i,j)$ is fixed in the first sender-receiver position for both thresholds. In the score covariance, this shared dyad can appear in any of the 4 possible positions within each quadruple. Since there are 4 positions for threshold $k$ and 4 positions for threshold $l$, there are $4 \times 4 = 16$ configurations that map to each term in the projection. Therefore, the leading term of the score covariance is:
\begin{adjustwidth}{-0.25in}{-0.25in}
\begin{align}
    \boldsymbol{\Upsilon}_{kl,s\ell} &= \sum_{i=1}^n \sum_{j \neq i} \sum_{i' \neq i,j} \sum_{j' \neq i,j,i'} \sum_{i'' \neq i,j, i'} \sum_{j'' \neq i,j,j',i''} 16 \times \mathbb{E} \Big[ \boldsymbol{s}_k(\sigma\{i,i';j,j'\}, \boldsymbol{\theta}_{y_{k,0}}) \boldsymbol{s}_l(\sigma\{i,i'';j,j''\}, \boldsymbol{\theta}_{y_{l,0}})' \Big] 
\end{align}
\end{adjustwidth}

Thus, $\boldsymbol{\Upsilon}_{kl,s\ell} = 16 \times \boldsymbol{\Upsilon}_{kl,\ell}$ for all pairs $(k,l)$, maintaining the same relationship between score and projection covariances as in the single-threshold case. This holds for both diagonal blocks ($k = l$) and off-diagonal blocks ($k \neq l$).

Two results follow immediately:
\begin{align}
\boldsymbol{\Upsilon}_{\mathbf{y}}^{-1/2} \mathbb{E} \left[ 
\begin{pmatrix} 
\boldsymbol{V}_{n,1}(\boldsymbol{\theta}_{y_{1,0}}) \\ 
\vdots \\
\boldsymbol{V}_{n,K}(\boldsymbol{\theta}_{y_{K,0}})  
\end{pmatrix}
\begin{pmatrix} 
\boldsymbol{V}_{n,1}(\boldsymbol{\theta}_{y_{1,0}}) \\ 
\vdots \\
\boldsymbol{V}_{n,K}(\boldsymbol{\theta}_{y_{K,0}})  
\end{pmatrix}'
\right]
\boldsymbol{\Upsilon}_{\mathbf{y}}^{-1/2'}  =  \boldsymbol{I}_{Kp} + o(1)
\end{align}
\begin{align}
\boldsymbol{\Upsilon}_{\mathbf{y}}^{-1/2} \mathbb{E} \left[ 
\begin{pmatrix} 
\tilde{\boldsymbol{S}}_{n,1}(\boldsymbol{\theta}_{y_{1,0}}) \\ 
\vdots \\
\tilde{\boldsymbol{S}}_{n,K}(\boldsymbol{\theta}_{y_{K,0}}) 
\end{pmatrix}
\begin{pmatrix} 
\tilde{\boldsymbol{S}}_{n,1}(\boldsymbol{\theta}_{y_{1,0}}) \\ 
\vdots \\
\tilde{\boldsymbol{S}}_{n,K}(\boldsymbol{\theta}_{y_{K,0}}) 
\end{pmatrix}'
\right]
\boldsymbol{\Upsilon}_{\mathbf{y}}^{-1/2'} = \boldsymbol{I}_{Kp} + o(1)
\end{align}
where, analogously to the single-threshold case, $\tilde{\boldsymbol{S}}_{n,k}(\boldsymbol{\theta}_{y_{k,0}}) = \frac{1}{4}\boldsymbol{S}_{n,k}(\boldsymbol{\theta}_{y_{k,0}})$ for $k = 1, ..., K$ are rescaled scores such that the factor of 16 is absorbed. Standard asymptotic equivalence follows directly. Note that this rescaling is merely for computational convenience, since the original vector of scores and projections remain asymptotically equivalent after proper normalization.

From analogous arguments, for any pair $k, l \in \{1, ..., K\}$:
\begin{align}
    \mathbb{E} [\boldsymbol{S}_{n,k}(\boldsymbol{\theta}_{y_{k,0}})\boldsymbol{V}_{n,l}(\boldsymbol{\theta}_{y_{l,0}})'] &= \sum_{i=1}^n \sum_{j \neq i} \sum_{i' \neq i,j} \sum_{j' \neq i,j,i'} \sum_{i'' \neq i,j, i'} \sum_{j'' \neq i,j,j',i''} \nonumber \\
    &\quad \mathbb{E} \Big[ 4 \times \boldsymbol{s}_k(\sigma\{i,i';j,j'\}, \boldsymbol{\theta}_{y_{k,0}}) \mathbb{E} [\boldsymbol{s}_l(\sigma\{i,i'';j,j''\}, \boldsymbol{\theta}_{y_{l,0}}) \mid \tilde{y}_{l,ij}, \mathcal{F}_{n,l}]'\Big]
\end{align}
Applying the same reasoning as before:
\begin{align}
    \boldsymbol{\Upsilon}_{\mathbf{y}}^{-1/2} \mathbb{E} \left[ 
        \begin{pmatrix} 
        \tilde{\boldsymbol{S}}_{n,1}(\boldsymbol{\theta}_{y_{1,0}}) \\ 
        \vdots \\
        \tilde{\boldsymbol{S}}_{n,K}(\boldsymbol{\theta}_{y_{K,0}}) 
        \end{pmatrix}
        \begin{pmatrix} 
        \boldsymbol{V}_{n,1}(\boldsymbol{\theta}_{y_{1,0}}) \\ 
        \vdots \\
        \boldsymbol{V}_{n,K}(\boldsymbol{\theta}_{y_{K,0}}) 
        \end{pmatrix}'
    \right]
    \boldsymbol{\Upsilon}_{\mathbf{y}}^{-1/2'} = \boldsymbol{I}_{Kp} + o(1)
\end{align}
And, analogously:
\begin{align}
    \boldsymbol{\Upsilon}_{\mathbf{y}}^{-1/2} \mathbb{E} \left[ 
        \begin{pmatrix} 
        \boldsymbol{V}_{n,1}(\boldsymbol{\theta}_{y_{1,0}}) \\ 
        \vdots \\
        \boldsymbol{V}_{n,K}(\boldsymbol{\theta}_{y_{K,0}}) 
        \end{pmatrix}
        \begin{pmatrix} 
        \tilde{\boldsymbol{S}}_{n,1}(\boldsymbol{\theta}_{y_{1,0}}) \\ 
        \vdots \\
        \tilde{\boldsymbol{S}}_{n,K}(\boldsymbol{\theta}_{y_{K,0}}) 
        \end{pmatrix}'
    \right]
    \boldsymbol{\Upsilon}_{\mathbf{y}}^{-1/2'} = \boldsymbol{I}_{Kp} + o(1)
\end{align}

This proves the asymptotic equivalence result for all $K$ thresholds simultaneously. This establishes that the leading terms of the score and projection covariance matrices are proportional, with $\boldsymbol{\Upsilon}_{s,\mathbf{y}} = 16 \times \boldsymbol{\Upsilon}_{\mathbf{y}}$. For every fixed pair of thresholds $(y_k, y_l)$, configurations in which the two quadruples share more than one dyad have at most five free node indices and therefore number $O(n^5)$, with aggregate contribution $O(n^5)$ by Cauchy--Schwarz and the moment bounds of Assumption \ref{assumption5}. The leading $(k,l)$ block scale being $n^6\sqrt{p_{n,k} p_{n,l}}$, as established in the convergence-rate calculation below, these configurations are negligible:
$$\frac{O(n^5)}{n^6\sqrt{p_{n,k} p_{n,l}}} = O\left(\frac{1}{\sqrt{(n p_{n,k})(n p_{n,l})}}\right) = o(1)$$
by Assumption \ref{assumption4}. Thus the one-shared-dyad blocks are the leading blocks at their threshold-specific rates, and the normalized scores and normalized projections converge to the same limiting distribution (after accounting for the scaling factor). This asymptotic equivalence allows us to work with either the scores or projections for establishing the joint limiting distribution across all $K$ thresholds.
\\
\\
\noindent \textbf{Step 2.}\\
\noindent \textit{Conditional joint covariance and the score-projection bridge.}\\
Recall that $\boldsymbol{v}_{k,ij}(\boldsymbol{\theta}_{y_{k,0}})$ for $k = 1, ..., K$ are zero mean random vectors of dimension $p$, independent across dyads $(i,j)$ conditional on the sequence of covariates and fixed effects. Define the joint information set:
$$\mathcal{F}_{n,\mathbf{y}} = \{\{\boldsymbol{x}_{ij}\}_{n,n}, \{\alpha_{i,y_1}\}_n, \{\gamma_{j,y_1}\}_n, ..., \{\alpha_{i,y_K}\}_n, \{\gamma_{j,y_K}\}_n\}$$
which contains all covariates and the fixed effects at all $K$ thresholds.

Define the conditional joint covariance matrix with blocks. For the diagonal blocks:
\begin{align}
\boldsymbol{\Upsilon}_{X,kk} &= \sum_{i=1}^n \sum_{j \neq i} E\left(\boldsymbol{v}_{k,ij}(\boldsymbol{\theta}_{y_{k,0}}) \boldsymbol{v}_{k,ij}(\boldsymbol{\theta}_{y_{k,0}})' \mid \mathcal{F}_{n,\mathbf{y}}\right), \quad k = 1, ..., K
\end{align}

And for the off-diagonal blocks:
\begin{align}
\boldsymbol{\Upsilon}_{X,kl} &= \sum_{i=1}^n \sum_{j \neq i} E\left(\boldsymbol{v}_{k,ij}(\boldsymbol{\theta}_{y_{k,0}}) \boldsymbol{v}_{l,ij}(\boldsymbol{\theta}_{y_{l,0}})' \mid \mathcal{F}_{n,\mathbf{y}}\right), \quad k \neq l \\
\boldsymbol{\Upsilon}_{X,lk} &= \boldsymbol{\Upsilon}_{X,kl}'
\end{align}
Let $\boldsymbol{\Upsilon}_{X,\mathbf{y}}$ denote the full $(Kp) \times (Kp)$ conditional joint covariance matrix with blocks $[\boldsymbol{\Upsilon}_{X,\mathbf{y}}]_{kl} = \boldsymbol{\Upsilon}_{X,kl}$. Assumption \ref{assumption_jointscorevar} imposes an eigenvalue floor on the block-normalized sample covariance matrix $\boldsymbol{D}_{n,\mathbf{y}}^{-1}\boldsymbol{\Upsilon}_{n,\mathbf{y}}(\boldsymbol{\theta}_{\mathbf{y},0})\boldsymbol{D}_{n,\mathbf{y}}^{-1}$. The floor will be carried to the two other covariance matrices used in the proof — the conditional 
matrix $\boldsymbol{\Upsilon}_{X,\mathbf{y}}$, which normalizes the projections in the conditional 
central limit theorem below, and the population matrix $\boldsymbol{\Upsilon}_{\mathbf{y}}$, which 
normalizes the asymptotic equivalence result of Step 1. The first ingredient is the joint analogue 
of the score--projection covariance bridge:
\begin{align} \label{joint_bridge}
\left\| \boldsymbol{D}_{n,\mathbf{y}}^{-1} \Big[ \boldsymbol{\Upsilon}_{n,\mathbf{y}}(\boldsymbol{\theta}_{\mathbf{y},0}) - 16\, \boldsymbol{\Upsilon}_{X,\mathbf{y}} \Big] \boldsymbol{D}_{n,\mathbf{y}}^{-1} \right\| \xrightarrow{p} 0.
\end{align}
The argument to show that the above holds is the one used in \ref{appendix_derivation}, applied to each block. It has two parts. First, the index restrictions $i'' \neq i'$ and $j'' \neq j'$ in \eqref{upsilon_n_kl} ensure that the two quadruples share only the dyad $(i,j)$, so that the conditional independence used in the derivation of $\boldsymbol{\Upsilon}_{kl,\ell}$ above gives, by iterated expectations,
$$\mathbb{E}\left[\boldsymbol{\Upsilon}_{n,kl}(\boldsymbol{\theta}_{y_{k,0}}, \boldsymbol{\theta}_{y_{l,0}}) \mid \mathcal{F}_{n,\mathbf{y}}\right] = 16\,\boldsymbol{\Upsilon}_{X,kl}$$
up to the configurations excluded by these restrictions, which number $O(n^5)$ and contribute $O_p(n^5)$ in aggregate. Second, $\boldsymbol{\Upsilon}_{n,kl}$ concentrates around this conditional mean: given $\mathcal{F}_{n,\mathbf{y}}$, each summand depends only on the outcomes of the dyads in its two quadruples, and these are independent across dyads, so two summands are conditionally uncorrelated unless their configurations share a dyad. There are $O(n^{10})$ such pairs, each contributing $O(1)$ in expectation under Assumption \ref{assumption5}, so a conditional Chebyshev inequality gives a deviation of order $O_p(n^5)$. Dividing both parts by the block scale $n^6\sqrt{p_{n,k}p_{n,l}}$ makes them $o_p(1)$ by Assumption \ref{assumption4}, and since $K$ is fixed, blockwise convergence gives convergence of the full matrix.\\

\noindent \textit{Eigenvalue transfers.}\\
Assumption \ref{assumption_jointscorevar}, \eqref{joint_bridge} and Weyl's inequality then imply, by the same transfer argument used in \ref{appendix_derivation} with pre- and post-multiplication by $\boldsymbol{D}_{n,\mathbf{y}}^{-1}$ in place of the scalar normalization,
$$\lambda_{\min}\left( \boldsymbol{D}_{n,\mathbf{y}}^{-1} \boldsymbol{\Upsilon}_{X,\mathbf{y}} \boldsymbol{D}_{n,\mathbf{y}}^{-1} \right) \geq \frac{c_{\mathbf{y}}}{16} - o_p(1)$$
with probability approaching one, where the factor $1/16$ arises because $\tfrac{1}{16}\boldsymbol{\Upsilon}_{n,\mathbf{y}}(\boldsymbol{\theta}_{\mathbf{y},0})$ is the matrix comparable to $\boldsymbol{\Upsilon}_{X,\mathbf{y}}$. Let $c_0 > 0$ be any constant with $c_0 < c_{\mathbf{y}}/16$, so that the $o_p(1)$ term is eventually dominated. Then $\boldsymbol{\Upsilon}_{X,\mathbf{y}}$ is positive definite with probability approaching one, and substituting $\boldsymbol{a} = \boldsymbol{D}_{n,\mathbf{y}}\boldsymbol{c}$ in the quadratic-form version of the eigenvalue bound gives, for every nonzero $\boldsymbol{c} \in \mathbb{R}^{Kp}$,
\begin{align} \label{joint_floor}
\boldsymbol{c}' \boldsymbol{\Upsilon}_{X,\mathbf{y}} \boldsymbol{c} \geq c_0 \|\boldsymbol{D}_{n,\mathbf{y}} \boldsymbol{c}\|^2 > 0
\end{align}
with probability approaching one. This is the form required for the Cram\'{e}r-Wold argument below, since $\boldsymbol{c}' \boldsymbol{\Upsilon}_{X,\mathbf{y}} \boldsymbol{c}$ is the conditional variance of the linear combination $W_n(\boldsymbol{\theta}_{\mathbf{y},0})$, and it demonstrates that this variance is bounded below at the rate determined by the blocks on which $\boldsymbol{c}$ loads. Since $\|\boldsymbol{D}_{n,\mathbf{y}}\boldsymbol{c}\|^2 = \sum_{k=1}^K n^6 p_{n,k}\|\boldsymbol{c}_k\|^2 \geq n^6 (\min_k p_{n,k}) \|\boldsymbol{c}\|^2$, \eqref{joint_floor} implies $\lambda_{\min}(\boldsymbol{\Upsilon}_{X,\mathbf{y}}) \geq c_0 n^6 \min_k p_{n,k}$, so that $\boldsymbol{\Upsilon}_{X,\mathbf{y}}$ is invertible and $\|\boldsymbol{\Upsilon}_{X,\mathbf{y}}^{-1/2}\|^2 \leq (c_0 n^6 \min_k p_{n,k})^{-1}$, both with probability approaching one. In the other direction, $\|\overline{\boldsymbol{\Upsilon}}_{X,\mathbf{y}}\| = O_p(1)$, since each block of $\boldsymbol{\Upsilon}_{X,\mathbf{y}}$ is of order $n^6\sqrt{p_{n,k}p_{n,l}}$ by the leading-term calculations of Step 1 and $K$ is fixed.

Taking expectations as in the pointwise proof, the normalized population matrix $\boldsymbol{D}_{n,\mathbf{y}}^{-1} \boldsymbol{\Upsilon}_{\mathbf{y}} \boldsymbol{D}_{n,\mathbf{y}}^{-1}$, where $\boldsymbol{\Upsilon}_{\mathbf{y}} = \mathbb{E}[\boldsymbol{\Upsilon}_{X,\mathbf{y}}]$,is positive definite for sufficiently large $n$, which justifies the normalizations by $\boldsymbol{\Upsilon}_{\mathbf{y}}^{-1/2}$ used in Step 1; moreover $\|\boldsymbol{D}_{n,\mathbf{y}}^{-1}\boldsymbol{\Upsilon}_{\mathbf{y}}\boldsymbol{D}_{n,\mathbf{y}}^{-1}\| = O(1)$, since each block of $\boldsymbol{\Upsilon}_{\mathbf{y}}$ is of order $n^6\sqrt{p_{n,k}p_{n,l}}$ by the leading-term calculations of Step 1 and $K$ is fixed.\\

\noindent \textit{Conditional central limit theorem.}\\
Joint asymptotic normality is established by applying the Cram\'{e}r-Wold theorem to the standardized stacked projection. By \eqref{joint_floor}, $\boldsymbol{\Upsilon}_{X,\mathbf{y}}$ is invertible with probability approaching one, so
$$\boldsymbol{Z}_{n,\mathbf{y}} := \boldsymbol{\Upsilon}_{X,\mathbf{y}}^{-1/2} \boldsymbol{V}_{n,\mathbf{y}}(\boldsymbol{\theta}_{\mathbf{y},0}) = \sum_{i=1}^n \sum_{j \neq i} \boldsymbol{\Upsilon}_{X,\mathbf{y}}^{-1/2} \boldsymbol{v}_{ij,\mathbf{y}}, \qquad \boldsymbol{v}_{ij,\mathbf{y}} := \left(\boldsymbol{v}_{1,ij}(\boldsymbol{\theta}_{y_{1,0}})', \ldots, \boldsymbol{v}_{K,ij}(\boldsymbol{\theta}_{y_{K,0}})'\right)',$$
is well defined. Standardizing by the conditional covariance of the vector itself gives, exactly,
$$\text{Var}(\boldsymbol{Z}_{n,\mathbf{y}} \mid \mathcal{F}_{n,\mathbf{y}}) = \boldsymbol{\Upsilon}_{X,\mathbf{y}}^{-1/2} \boldsymbol{\Upsilon}_{X,\mathbf{y}} \boldsymbol{\Upsilon}_{X,\mathbf{y}}^{-1/2} = \boldsymbol{I}_{Kp},$$
so the threshold-specific scales of the covariance blocks are absorbed without any further normalization at this stage.

For any non-zero vector $\boldsymbol{a} = (\boldsymbol{a}_1', ..., \boldsymbol{a}_K')' \in \mathbb{R}^{Kp}$, where each $\boldsymbol{a}_k \in \mathbb{R}^p$, consider the scalar linear combination
\begin{align}
W_n(\boldsymbol{\theta}_{\mathbf{y},0}) = \boldsymbol{a}' \boldsymbol{Z}_{n,\mathbf{y}} = \sum_{i=1}^n \sum_{j \neq i} w_{ij}, \qquad w_{ij} = \boldsymbol{a}' \boldsymbol{\Upsilon}_{X,\mathbf{y}}^{-1/2} \boldsymbol{v}_{ij,\mathbf{y}}.
\end{align}
Since $\boldsymbol{\Upsilon}_{X,\mathbf{y}}$ is $\mathcal{F}_{n,\mathbf{y}}$-measurable and, conditional on $\mathcal{F}_{n,\mathbf{y}}$, each $\boldsymbol{v}_{ij,\mathbf{y}}$ is a function of $y_{ij}$ alone, the scalars $w_{ij}$ are independent across dyads conditional on $\mathcal{F}_{n,\mathbf{y}}$, with zero mean by intermediate result 1 and
\begin{align}
\text{Var}(W_n \mid \mathcal{F}_{n,\mathbf{y}}) = \sum_{i=1}^n \sum_{j \neq i} \text{Var}(w_{ij} \mid \mathcal{F}_{n,\mathbf{y}}) = \boldsymbol{a}' \text{Var}(\boldsymbol{Z}_{n,\mathbf{y}} \mid \mathcal{F}_{n,\mathbf{y}}) \boldsymbol{a} = \|\boldsymbol{a}\|^2 .
\end{align}
The conditional fourth-moment Lyapunov condition for these summands can be verified directly from the definition of the projections. The condition to be established is
$$L_n := \frac{\sum_{i=1}^n \sum_{j \neq i} \mathbb{E}\left[|w_{ij}|^4 \mid \mathcal{F}_{n,\mathbf{y}}\right]}{\left(\sum_{i=1}^n \sum_{j \neq i} \text{Var}(w_{ij} \mid \mathcal{F}_{n,\mathbf{y}})\right)^2} \xrightarrow{p} 0,$$
which requires that no single dyad contribute a non-negligible share of the total variation, and which implies the conditional Lindeberg condition. By the variance identity above, the denominator equals $\|\boldsymbol{a}\|^4$, so it remains to bound the numerator.

By the Cauchy--Schwarz inequality and submultiplicativity of the operator norm,
$$|w_{ij}|^4 \leq \|\boldsymbol{a}\|^4 \, \|\boldsymbol{\Upsilon}_{X,\mathbf{y}}^{-1/2}\|^4 \, \|\boldsymbol{v}_{ij,\mathbf{y}}\|^4,$$
where both factors preceding $\|\boldsymbol{v}_{ij,\mathbf{y}}\|^4$ are $\mathcal{F}_{n,\mathbf{y}}$-measurable and therefore pass outside the conditional expectation. The factor $\|\boldsymbol{a}\|^4$ cancels against the denominator, so the resulting bound is uniform in $\boldsymbol{a}$ and the condition needs to be verified only once rather than separately for each direction.

For the stacked projection, note that $\|\boldsymbol{v}_{ij,\mathbf{y}}\| \leq \sum_{k=1}^K \|\boldsymbol{v}_{k,ij}(\boldsymbol{\theta}_{y_{k,0}})\|$. With $\mathcal{N}_{m_n,ij}$ as in \ref{appendix_derivation}, the envelope inequality established there, applied threshold by threshold, gives
$$\|\boldsymbol{v}_{ij,\mathbf{y}}\| \leq \sum_{k=1}^K \sum_{\sigma \in \mathcal{N}_{m_n,ij}} \|\boldsymbol{r}_\sigma\|.$$
Raising this bound to the fourth power and applying the power-mean inequality $(\sum_{m=1}^M a_m)^4 \leq M^3 \sum_{m=1}^M a_m^4$ twice, once to the sum over the $K$ thresholds and once to the sum over $\mathcal{N}_{m_n,ij}$, gives
$$\|\boldsymbol{v}_{ij,\mathbf{y}}\|^4 \leq C_K \, |\mathcal{N}_{m_n,ij}|^3 \sum_{\sigma \in \mathcal{N}_{m_n,ij}} \|\boldsymbol{r}_\sigma\|^4 = O(n^6) \sum_{\sigma \in \mathcal{N}_{m_n,ij}} \|\boldsymbol{r}_\sigma\|^4,$$
where $C_K$ depends only on $K$ and $|\mathcal{N}_{m_n,ij}| = O(n^2)$. The inner sum has $O(n^2)$ terms, each with uniformly bounded fourth moment under Assumption \ref{assumption5}, so its expectation is $O(n^2)$ and each dyad contributes $O(n^8)$. Summing over the $O(n^2)$ dyads gives an unconditional expectation of order $n^{10}$, and Markov's inequality then yields
$$\sum_{i=1}^n \sum_{j \neq i} \mathbb{E}\left[\|\boldsymbol{v}_{ij,\mathbf{y}}\|^4 \mid \mathcal{F}_{n,\mathbf{y}}\right] = O_p(n^{10}).$$

For the normalizing factor, since $\boldsymbol{\Upsilon}_{X,\mathbf{y}}$ is symmetric and positive definite with probability approaching one, $\|\boldsymbol{\Upsilon}_{X,\mathbf{y}}^{-1/2}\|^4 = \lambda_{\min}(\boldsymbol{\Upsilon}_{X,\mathbf{y}})^{-2}$, and the eigenvalue bound $\lambda_{\min}(\boldsymbol{\Upsilon}_{X,\mathbf{y}}) \geq c_0 n^6 \min_k p_{n,k}$ established above gives
$$\|\boldsymbol{\Upsilon}_{X,\mathbf{y}}^{-1/2}\|^4 \leq \left(c_0 n^6 \min_k p_{n,k}\right)^{-2}$$
with probability approaching one. Combining the two bounds,
$$L_n = O_p\left(\frac{n^{10}}{\left(c_0 n^6 \min_k p_{n,k}\right)^{2}}\right) = O_p\left(\left[n \min_k p_{n,k}\right]^{-2}\right) \longrightarrow 0,$$
since $n p_{n,k} \to \infty$ at every included threshold by Assumption \ref{assumption4}, and hence also for the minimum over the finitely many thresholds. 

A conditional version of the Lyapunov CLT (see, e.g., \citet{rao2009conditional}) therefore gives
\begin{align}
W_n(\boldsymbol{\theta}_{\mathbf{y},0}) = \boldsymbol{a}' \boldsymbol{Z}_{n,\mathbf{y}} \xrightarrow{d} N(0, \|\boldsymbol{a}\|^2)
\end{align}
conditional on $\mathcal{F}_{n,\mathbf{y}}$. Since the limit does not depend on the conditioning variables, the convergence also holds unconditionally, and since it holds for every fixed non-zero $\boldsymbol{a} \in \mathbb{R}^{Kp}$, the Cram\'{e}r-Wold theorem implies
\begin{align}
\boldsymbol{\Upsilon}_{X,\mathbf{y}}^{-1/2} \boldsymbol{V}_{n,\mathbf{y}}(\boldsymbol{\theta}_{\mathbf{y},0}) \xrightarrow{d} N(0, \boldsymbol{I}_{Kp}).
\end{align}

\noindent \textit{Normalization by the score covariance and passage to the score.}\\
The limit obtained above is stated for the projection, normalized by 
$\boldsymbol{\Upsilon}_{X,\mathbf{y}}^{-1/2}$. Two changes remain: the normalizer must be written 
in a form compatible with the bound available from Step 1, and the projection must then be 
replaced by the stacked score.

The difficulty is that Step 1 bounds the remainder $\boldsymbol{U}_{n,\mathbf{y}}$ under the normalization by $\boldsymbol{\Upsilon}_{\mathbf{y}}^{-1/2}$, the population matrix, while the limit above is normalized by $\boldsymbol{\Upsilon}_{X,\mathbf{y}}^{-1/2}$. Passing from one to the other requires bounding the product $\boldsymbol{\Upsilon}_{X,\mathbf{y}}^{-1/2}\boldsymbol{\Upsilon}_{\mathbf{y}}^{1/2}$, and by submultiplicativity of the operator norm, together with $\|\boldsymbol{A}^{-1/2}\| = \lambda_{\min}(\boldsymbol{A})^{-1/2}$ and $\|\boldsymbol{A}^{1/2}\| = \lambda_{\max}(\boldsymbol{A})^{1/2}$ for symmetric positive definite $\boldsymbol{A}$,
$$\left\|\boldsymbol{\Upsilon}_{X,\mathbf{y}}^{-1/2}\boldsymbol{\Upsilon}_{\mathbf{y}}^{1/2}\right\|^2 \leq \frac{\lambda_{\max}(\boldsymbol{\Upsilon}_{\mathbf{y}})}{\lambda_{\min}(\boldsymbol{\Upsilon}_{X,\mathbf{y}})}.$$
The two extreme eigenvalues sit at opposite ends of the block scales. Both matrices have $(k,l)$ blocks of order $n^6\sqrt{p_{n,k}p_{n,l}}$, so a unit vector loading only on the block with the largest sparsity parameter gives a quadratic form of order $n^6\max_k p_{n,k}$, while one loading only on the smallest gives $n^6\min_k p_{n,k}$; since $\lambda_{\max}$ is the supremum of the quadratic form over unit vectors and $\lambda_{\min}$ the infimum, the displayed ratio is of order $\max_k p_{n,k}/\min_k p_{n,k}$. This is bounded only if the threshold-specific sparsity rates are restricted relative to one another, which the framework does not assume: convergence and boundedness of each block at its own scale say nothing about the ratio between scales.

The solution is to write the normalizer in a form that separates the scales from the well-conditioned factors. Denote the block-rate-normalized conditional covariance by
$$\overline{\boldsymbol{\Upsilon}}_{X,\mathbf{y}} := \boldsymbol{D}_{n,\mathbf{y}}^{-1} \boldsymbol{\Upsilon}_{X,\mathbf{y}} \boldsymbol{D}_{n,\mathbf{y}}^{-1},$$
and, for any positive definite $\boldsymbol{M}$, define the rate-adapted normalizing matrix
$$\mathcal{C}_{\boldsymbol{D}_{n,\mathbf{y}}}(\boldsymbol{M}) := \left(\boldsymbol{D}_{n,\mathbf{y}}^{-1} \boldsymbol{M} \boldsymbol{D}_{n,\mathbf{y}}^{-1}\right)^{-1/2} \boldsymbol{D}_{n,\mathbf{y}}^{-1} = \overline{\boldsymbol{M}}^{-1/2}\boldsymbol{D}_{n,\mathbf{y}}^{-1},$$
which satisfies $\mathcal{C}_{\boldsymbol{D}_{n,\mathbf{y}}}(\boldsymbol{M}) \boldsymbol{M} \, \mathcal{C}_{\boldsymbol{D}_{n,\mathbf{y}}}(\boldsymbol{M})' = \boldsymbol{I}_{Kp}$. All rate heterogeneity is carried by the deterministic diagonal matrix $\boldsymbol{D}_{n,\mathbf{y}}$, while $\overline{\boldsymbol{M}}^{-1/2}$ acts on a matrix whose eigenvalues are bounded above and below by constants, by the eigenvalue transfers established above. Bounds obtained with this normalizer therefore have constants depending only on $c_0$, with no dependence on the relative sparsity rates.

Both $\mathcal{C}_{\boldsymbol{D}_{n,\mathbf{y}}}(\boldsymbol{\Upsilon}_{X,\mathbf{y}})$ and $\boldsymbol{\Upsilon}_{X,\mathbf{y}}^{-1/2}$ standardize $\boldsymbol{V}_{n,\mathbf{y}}(\boldsymbol{\theta}_{\mathbf{y},0})$ conditionally on $\mathcal{F}_{n,\mathbf{y}}$, so
$$\boldsymbol{O}_{n,\mathbf{y}} := \mathcal{C}_{\boldsymbol{D}_{n,\mathbf{y}}}(\boldsymbol{\Upsilon}_{X,\mathbf{y}}) \, \boldsymbol{\Upsilon}_{X,\mathbf{y}}^{1/2}$$
satisfies $\boldsymbol{O}_{n,\mathbf{y}}\boldsymbol{O}_{n,\mathbf{y}}' = \boldsymbol{I}_{Kp}$ and is therefore orthogonal. It is also $\mathcal{F}_{n,\mathbf{y}}$-measurable, being a function of $\boldsymbol{\Upsilon}_{X,\mathbf{y}}$ and $\boldsymbol{D}_{n,\mathbf{y}}$ alone. Since an orthogonal transformation that is fixed given the conditioning variables leaves a conditionally standard normal vector standard normal, the conditional limit established above also holds with $\mathcal{C}_{\boldsymbol{D}_{n,\mathbf{y}}}(\boldsymbol{\Upsilon}_{X,\mathbf{y}})$ in place of $\boldsymbol{\Upsilon}_{X,\mathbf{y}}^{-1/2}$:
$$\mathcal{C}_{\boldsymbol{D}_{n,\mathbf{y}}}(\boldsymbol{\Upsilon}_{X,\mathbf{y}}) \, \boldsymbol{V}_{n,\mathbf{y}}(\boldsymbol{\theta}_{\mathbf{y},0}) \xrightarrow{d} N(\boldsymbol{0}, \boldsymbol{I}_{Kp}).$$

It remains to extend this from the projection to the score. Write $\tilde{\boldsymbol{S}}_{n,\mathbf{y}}(\boldsymbol{\theta}_{\mathbf{y},0}) = \tfrac{1}{4}\boldsymbol{S}_{n,\mathbf{y}}(\boldsymbol{\theta}_{\mathbf{y},0})$ for the rescaled stacked score, which places it on the same scale as the projection, and $\boldsymbol{U}_{n,\mathbf{y}} = \tilde{\boldsymbol{S}}_{n,\mathbf{y}}(\boldsymbol{\theta}_{\mathbf{y},0}) - \boldsymbol{V}_{n,\mathbf{y}}(\boldsymbol{\theta}_{\mathbf{y},0})$ for the remainder. Factorizing so that Step 1's bound can be applied,
$$\mathcal{C}_{\boldsymbol{D}_{n,\mathbf{y}}}\!\left(\boldsymbol{\Upsilon}_{X,\mathbf{y}}\right)\boldsymbol{U}_{n,\mathbf{y}} = \left[\mathcal{C}_{\boldsymbol{D}_{n,\mathbf{y}}}\!\left(\boldsymbol{\Upsilon}_{X,\mathbf{y}}\right)\boldsymbol{\Upsilon}_{\mathbf{y}}^{1/2}\right]\left[\boldsymbol{\Upsilon}_{\mathbf{y}}^{-1/2}\boldsymbol{U}_{n,\mathbf{y}}\right],$$
where the first bracket satisfies
$$\left\|\mathcal{C}_{\boldsymbol{D}_{n,\mathbf{y}}}\!\left(\boldsymbol{\Upsilon}_{X,\mathbf{y}}\right)\boldsymbol{\Upsilon}_{\mathbf{y}}^{1/2}\right\|^2 \leq \left\|\overline{\boldsymbol{\Upsilon}}_{X,\mathbf{y}}^{-1/2}\right\|^2 \left\|\boldsymbol{D}_{n,\mathbf{y}}^{-1}\boldsymbol{\Upsilon}_{\mathbf{y}}\boldsymbol{D}_{n,\mathbf{y}}^{-1}\right\| = O_p(1)$$
by the eigenvalue floor on $\overline{\boldsymbol{\Upsilon}}_{X,\mathbf{y}}$ and the operator-norm bound on $\boldsymbol{D}_{n,\mathbf{y}}^{-1}\boldsymbol{\Upsilon}_{\mathbf{y}}\boldsymbol{D}_{n,\mathbf{y}}^{-1}$, both established above. This is where the factorization is used: with $\boldsymbol{\Upsilon}_{X,\mathbf{y}}^{-1/2}$ in place of $\mathcal{C}_{\boldsymbol{D}_{n,\mathbf{y}}}(\boldsymbol{\Upsilon}_{X,\mathbf{y}})$ the same bound would carry the ratio of sparsity rates described above. For the second bracket, combining the four normalized second-moment results of Step 1 gives $\boldsymbol{\Upsilon}_{\mathbf{y}}^{-1/2}\mathbb{E}[\boldsymbol{U}_{n,\mathbf{y}}\boldsymbol{U}_{n,\mathbf{y}}']\boldsymbol{\Upsilon}_{\mathbf{y}}^{-1/2} = o(1)$, so that $\boldsymbol{\Upsilon}_{\mathbf{y}}^{-1/2}\boldsymbol{U}_{n,\mathbf{y}} = o_p(1)$ by Markov's inequality. The product is therefore $o_p(1)$.

Since $\tilde{\boldsymbol{S}}_{n,\mathbf{y}} = \boldsymbol{V}_{n,\mathbf{y}} + \boldsymbol{U}_{n,\mathbf{y}}$, and since $\mathcal{C}_{\boldsymbol{D}_{n,\mathbf{y}}}(c\boldsymbol{M}) = c^{-1/2}\mathcal{C}_{\boldsymbol{D}_{n,\mathbf{y}}}(\boldsymbol{M})$ for scalar $c > 0$ gives
$$\mathcal{C}_{\boldsymbol{D}_{n,\mathbf{y}}}\!\left(16\,\boldsymbol{\Upsilon}_{X,\mathbf{y}}\right) \boldsymbol{S}_{n,\mathbf{y}}(\boldsymbol{\theta}_{\mathbf{y},0}) = \mathcal{C}_{\boldsymbol{D}_{n,\mathbf{y}}}\!\left(\boldsymbol{\Upsilon}_{X,\mathbf{y}}\right) \tilde{\boldsymbol{S}}_{n,\mathbf{y}}(\boldsymbol{\theta}_{\mathbf{y},0}),$$
so that the factor $4$ arising from the normalizer cancels against the factor $1/4$ in the rescaled score,
\begin{align} \label{joint_score_clt_cond}
d_{\mathrm{BL}}\left(\mathcal{L}\left(\mathcal{C}_{\boldsymbol{D}_{n,\mathbf{y}}}\!\left(16\,\boldsymbol{\Upsilon}_{X,\mathbf{y}}\right) \boldsymbol{S}_{n,\mathbf{y}}(\boldsymbol{\theta}_{\mathbf{y},0}) \,\middle|\, \mathcal{F}_{n,\mathbf{y}}\right), \, N(\boldsymbol{0}, \boldsymbol{I}_{Kp})\right) \xrightarrow{p} 0.
\end{align}
The limiting distribution here is fixed, so this is conditional convergence in distribution; the bounded--Lipschitz form is used because the corresponding statement in Step 3 has a target that varies with $n$.\\
\\
\noindent \textbf{Step 3.}\\
\textit{Joint limit distribution.}\\
The joint Gaussian approximation of Theorem \ref{theorem3} follows by combining the stacked mean-value expansion \eqref{joint_expansion} with \eqref{joint_score_clt_cond}. The expansion is an algebraic identity: it writes the estimation error as the score premultiplied by the inverse Hessian. Obtaining the limit for the studentized estimator from that of the score therefore requires multiplying by a matrix that combines the inverse Hessian with the normalization of the score covariance, denoted $\boldsymbol{L}_{X,\mathbf{y}}$ below and constructed with the normalizer $\mathcal{C}_{\boldsymbol{D}_{n,\mathbf{y}}}$ of Step 2, so that \eqref{joint_score_clt_cond} can be used as it stands. What shapes the argument is the mode of convergence.

Slutsky's theorem would apply if $\boldsymbol{L}_{X,\mathbf{y}}$ were deterministic, or converged 
in probability to a constant. It is neither: built from the score covariance, it is not 
independent of the standardized score it multiplies, so orthogonality alone does not preserve the 
limit. A deterministic counterpart would suffice and would not need to converge, but obtaining 
one would require the conditional score covariance to concentrate around its unconditional mean, 
which Assumption \ref{assumption_jointscorevar} does not deliver: it bounds the normalized 
score covariance away from degeneracy without restricting its behaviour across $n$. Such a 
condition would in any case concern the off-diagonal blocks, which describe which quadruples are 
informative at two thresholds at once, and would therefore restrict the joint design rather than 
each threshold separately.

Rather than restrict the design in this way, the argument proceeds conditionally on $\mathcal{F}_{n,\mathbf{y}}$, where $\boldsymbol{L}_{X,\mathbf{y}}$ is fixed, exactly as in the single-threshold case of \ref{appendix_derivation}; no limit for the score covariance is required in either. What differs is the normalization of the conclusion. Theorem \ref{theorem2} is stated with the full sandwich normalization $\boldsymbol{\Omega}_n(\boldsymbol{\theta}_n)^{-1/2}$, so its target is $N(\boldsymbol{0}, \boldsymbol{I}_p)$ and the conditional statement lifts unconditionally. Part (i) here studentizes only by the diagonal $\boldsymbol{\sigma}_{n,\mathbf{y}}(\boldsymbol{\theta}_{n,\mathbf{y}})$, retaining the cross-threshold correlations, so the target is $N(\boldsymbol{0}, \boldsymbol{P}_{n,\mathbf{y}}(\boldsymbol{\theta}_{n,\mathbf{y}}))$ and varies with $n$; the corresponding fully normalized statement is available, and is what parts (ii) and (iii) below use, but the remark following \eqref{joint_gauss} explains why it is not the form the simultaneous bands require.

This requires every object entering $\boldsymbol{L}_{X,\mathbf{y}}$ to be $\mathcal{F}_{n,\mathbf{y}}$-measurable, which is why the realized Hessian is replaced by its conditional expectation below, and why the standard errors used to studentize are computed from the resulting conditional sandwich covariance. These conditional quantities do not appear in the final result: in the inference discussion below they are replaced by their feasible counterparts, computed at the estimator, which yields the form stated in Theorem \ref{theorem3}(i).\\

\noindent \textit{Conditional sandwich covariance and Hessian bounds.}\\
Let $\boldsymbol{H}_{X,\mathbf{y}} := \mathbb{E}[\boldsymbol{H}_{n,\mathbf{y}}(\boldsymbol{\theta}_{\mathbf{y},0}) \mid \mathcal{F}_{n,\mathbf{y}}]$, which is block diagonal since estimation is carried out separately at each threshold, and define the $\mathcal{F}_{n,\mathbf{y}}$-measurable sandwich covariance
\begin{align} \label{omega_X}
\boldsymbol{\Omega}_{X,\mathbf{y}} := \boldsymbol{H}_{X,\mathbf{y}}^{-1} \left(16\,\boldsymbol{\Upsilon}_{X,\mathbf{y}}\right) \boldsymbol{H}_{X,\mathbf{y}}^{-1},
\end{align}
together with $\boldsymbol{\sigma}_{X,\mathbf{y}}$, the diagonal matrix of square roots of its diagonal entries, and $\boldsymbol{P}_{X,\mathbf{y}} := \boldsymbol{\sigma}_{X,\mathbf{y}}^{-1}\boldsymbol{\Omega}_{X,\mathbf{y}}\boldsymbol{\sigma}_{X,\mathbf{y}}^{-1}$. The matrix $\boldsymbol{\Omega}_{X,\mathbf{y}}$ is the conditional counterpart of the sandwich covariance, $\boldsymbol{\sigma}_{X,\mathbf{y}}$ collects the corresponding standard errors, and $\boldsymbol{P}_{X,\mathbf{y}}$ collects the conditional correlations between the studentized estimators across thresholds and covariates, given $\mathcal{F}_{n,\mathbf{y}}$. It is the correlation structure of the Gaussian approximation established below, and the object whose feasible counterpart determines the simultaneous critical value in Section \ref{joint}.

The variance bound for the Hessian established in Step 3 of \ref{appendix_derivation} applies to $\boldsymbol{H}_{X,\mathbf{y}}$ as well. Writing $H_{ab}$ for an entry of $\boldsymbol{H}_{n,\mathbf{y}}(\boldsymbol{\theta}_{\mathbf{y},0})$, the law of total variance gives
$$\operatorname{Var}(H_{ab}) = \underbrace{\operatorname{Var}\left(\mathbb{E}[H_{ab} \mid \mathcal{F}_{n,\mathbf{y}}]\right)}_{\text{entry of } \boldsymbol{H}_{X,\mathbf{y}}} + \underbrace{\mathbb{E}\left[\operatorname{Var}(H_{ab} \mid \mathcal{F}_{n,\mathbf{y}})\right]}_{\text{residual}},$$
and both terms are non-negative, so each is bounded by $\operatorname{Var}(H_{ab})$. In particular, the expected conditional variance, which is the second moment of the corresponding entry of $\boldsymbol{H}_{n,\mathbf{y}}(\boldsymbol{\theta}_{\mathbf{y},0}) - \boldsymbol{H}_{X,\mathbf{y}}$, is bounded by $\operatorname{Var}(H_{ab})$. For the block at threshold $y_k$ this is of order $n^3 m_n p_{n,k}$ by the calculation in Step 3 of \ref{appendix_derivation}, and since the dimension is fixed the same order holds for the matrix norms. Markov's inequality therefore gives
$$\left\|\boldsymbol{E}_{n,\mathbf{y}}^{-1}\left[\boldsymbol{H}_{n,\mathbf{y}}(\boldsymbol{\theta}_{\mathbf{y},0}) - \boldsymbol{H}_{X,\mathbf{y}}\right]\boldsymbol{E}_{n,\mathbf{y}}^{-1}\right\| \xrightarrow{p} 0, \qquad \boldsymbol{E}_{n,\mathbf{y}} := \operatorname{diag}\left((n^4 p_{n,1})^{1/2}\boldsymbol{I}_p, \ldots, (n^4 p_{n,K})^{1/2}\boldsymbol{I}_p\right).$$
This is what allows the realized Hessian in \eqref{joint_expansion} to be replaced by $\boldsymbol{H}_{X,\mathbf{y}}$.

Since $\boldsymbol{D}_{n,\mathbf{y}} = n\,\boldsymbol{E}_{n,\mathbf{y}}$, setting $\boldsymbol{G}_{n,\mathbf{y}} := \boldsymbol{Q}_{n,\mathbf{y}}^{-1}\boldsymbol{H}_{X,\mathbf{y}}^{-1}\boldsymbol{D}_{n,\mathbf{y}}$, definition \eqref{omega_X} gives
$$\boldsymbol{Q}_{n,\mathbf{y}}^{-1}\boldsymbol{\Omega}_{X,\mathbf{y}}\boldsymbol{Q}_{n,\mathbf{y}}^{-1} = \boldsymbol{G}_{n,\mathbf{y}}\left[16\,\overline{\boldsymbol{\Upsilon}}_{X,\mathbf{y}}\right]\boldsymbol{G}_{n,\mathbf{y}}'.$$
Since $\boldsymbol{Q}_{n,\mathbf{y}}^{-1} = (\kappa_n/n)\boldsymbol{E}_{n,\mathbf{y}}$ and $\boldsymbol{D}_{n,\mathbf{y}} = n\boldsymbol{E}_{n,\mathbf{y}}$, with $\kappa_n = [(n-1)/n]^{1/2} \to 1$, this is $\boldsymbol{G}_{n,\mathbf{y}} = \kappa_n\,\boldsymbol{E}_{n,\mathbf{y}}\boldsymbol{H}_{X,\mathbf{y}}^{-1}\boldsymbol{E}_{n,\mathbf{y}}$, so the bounds below on the normalized conditional Hessian are exactly two-sided bounds on $\boldsymbol{G}_{n,\mathbf{y}}$, and the normalizations are matched so that the rates cancel in every block. For the conditional Hessian, the rank condition of Assumption \ref{assumption4} applied at each threshold bounds the normalized limit of $\mathbb{E}[\boldsymbol{H}_{n,k}(\boldsymbol{\theta}_{y_{k,0}})]$ away from singularity; the pointwise convergence established in Step 3 of \ref{appendix_derivation} transfers this to $\boldsymbol{H}_{n,k}(\boldsymbol{\theta}_{y_{k,0}})$, and the equation above transfers it in turn to $\boldsymbol{H}_{X,\mathbf{y}}$, in each case by Weyl's inequality; the maintained moment bounds give the corresponding upper bound. The floor gives $\|\boldsymbol{E}_{n,\mathbf{y}}\boldsymbol{H}_{X,\mathbf{y}}^{-1}\boldsymbol{E}_{n,\mathbf{y}}\| = O_p(1)$ and the upper bound gives $\|\boldsymbol{E}_{n,\mathbf{y}}^{-1}\boldsymbol{H}_{X,\mathbf{y}}\boldsymbol{E}_{n,\mathbf{y}}^{-1}\| = O_p(1)$, and hence two-sided bounds on $\boldsymbol{G}_{n,\mathbf{y}}$. For the middle factor, the floor and the upper bound on $\overline{\boldsymbol{\Upsilon}}_{X,\mathbf{y}}$ established above apply. Hence there are constants $0 < c_\Omega < C_\Omega < \infty$ with
$$c_\Omega \boldsymbol{I}_{Kp} \preceq \boldsymbol{Q}_{n,\mathbf{y}}^{-1}\boldsymbol{\Omega}_{X,\mathbf{y}}\boldsymbol{Q}_{n,\mathbf{y}}^{-1} \preceq C_\Omega \boldsymbol{I}_{Kp}, \qquad \boldsymbol{Q}_{n,\mathbf{y}} := \operatorname{diag}\left([n(n-1)p_{n,1}]^{-1/2}\boldsymbol{I}_p, \ldots, [n(n-1)p_{n,K}]^{-1/2}\boldsymbol{I}_p\right),$$
with probability approaching one. Since $\boldsymbol{A} \preceq \boldsymbol{B}$ means $\boldsymbol{a}'\boldsymbol{A}\boldsymbol{a} \leq \boldsymbol{a}'\boldsymbol{B}\boldsymbol{a}$ for every $\boldsymbol{a} \in \mathbb{R}^{Kp}$, taking $\boldsymbol{a} = \boldsymbol{e}_j$ and using that $\boldsymbol{Q}_{n,\mathbf{y}}$ is diagonal, so that $\boldsymbol{Q}_{n,\mathbf{y}}^{-1}\boldsymbol{e}_j = [\boldsymbol{Q}_{n,\mathbf{y}}]_{jj}^{-1}\boldsymbol{e}_j$, gives, for each $j = 1, \ldots, Kp$,
$$c_\Omega \leq \frac{[\boldsymbol{\Omega}_{X,\mathbf{y}}]_{jj}}{[\boldsymbol{Q}_{n,\mathbf{y}}]_{jj}^2} \leq C_\Omega,$$
so that the diagonal entries of $\boldsymbol{\sigma}_{X,\mathbf{y}}$ are positive and of the threshold-specific order $[n(n-1)p_{n,k}]^{-1/2}$. Define the diagonal matrix $\boldsymbol{\Delta}_{n,\mathbf{y}}$ by
$$[\boldsymbol{\Delta}_{n,\mathbf{y}}]_{jj} := \frac{[\boldsymbol{\sigma}_{X,\mathbf{y}}]_{jj}}{[\boldsymbol{Q}_{n,\mathbf{y}}]_{jj}}, \qquad j = 1, \ldots, Kp,$$
so that $\boldsymbol{\sigma}_{X,\mathbf{y}} = \boldsymbol{Q}_{n,\mathbf{y}}\boldsymbol{\Delta}_{n,\mathbf{y}}$ and $\boldsymbol{\Delta}_{n,\mathbf{y}}$ collects the ratios of the standard errors to their nominal rates. Since $[\boldsymbol{\sigma}_{X,\mathbf{y}}]_{jj}^2 = [\boldsymbol{\Omega}_{X,\mathbf{y}}]_{jj}$, the equation above states that $[\boldsymbol{\Delta}_{n,\mathbf{y}}]_{jj}^2 \in [c_\Omega, C_\Omega]$, so that the entries of $\boldsymbol{\Delta}_{n,\mathbf{y}}$ lie in $[c_\Omega^{1/2}, C_\Omega^{1/2}]$. Since
$$\boldsymbol{P}_{X,\mathbf{y}} = \boldsymbol{\sigma}_{X,\mathbf{y}}^{-1}\boldsymbol{\Omega}_{X,\mathbf{y}}\boldsymbol{\sigma}_{X,\mathbf{y}}^{-1} = \boldsymbol{\Delta}_{n,\mathbf{y}}^{-1}\left[\boldsymbol{Q}_{n,\mathbf{y}}^{-1}\boldsymbol{\Omega}_{X,\mathbf{y}}\boldsymbol{Q}_{n,\mathbf{y}}^{-1}\right]\boldsymbol{\Delta}_{n,\mathbf{y}}^{-1},$$
it follows that, for any unit vector $\boldsymbol{a}$,
$$\boldsymbol{a}'\boldsymbol{P}_{X,\mathbf{y}}\boldsymbol{a} \geq c_\Omega\|\boldsymbol{\Delta}_{n,\mathbf{y}}^{-1}\boldsymbol{a}\|^2 \geq \frac{c_\Omega}{C_\Omega},$$
so that $\lambda_{\min}(\boldsymbol{P}_{X,\mathbf{y}}) \geq c_\Omega/C_\Omega > 0$.  

The matrix $\boldsymbol{P}_{X,\mathbf{y}}$ therefore varies with $n$ but remains in a fixed set of correlation matrices bounded away from singularity. This is what allows the arguments below to proceed without a limiting correlation matrix. It is used twice: in the inference discussion, where the replacement of $\boldsymbol{P}_{X,\mathbf{y}}$ by its feasible counterpart relies on the Gaussian law being Lipschitz in its correlation matrix, which holds uniformly only away from singularity; and in the passage from the bounded--Lipschitz approximation to the rectangular coverage probability underlying the simultaneous critical value of Section \ref{bands}.\\

\noindent \textit{Linear representation and Gaussian approximation.}\\
Combining \eqref{joint_expansion}, \eqref{joint_score_clt_cond} and the Hessian results above, the studentized estimator admits the linear representation
\begin{align} \label{joint_linear}
\boldsymbol{\sigma}_{X,\mathbf{y}}^{-1}\left(\boldsymbol{\theta}_{n,\mathbf{y}} - \boldsymbol{\theta}_{\mathbf{y},0}\right) = \boldsymbol{L}_{X,\mathbf{y}} \, \mathcal{C}_{\boldsymbol{D}_{n,\mathbf{y}}}\!\left(16\,\boldsymbol{\Upsilon}_{X,\mathbf{y}}\right)\boldsymbol{S}_{n,\mathbf{y}}(\boldsymbol{\theta}_{\mathbf{y},0}) + o_p(1),
\end{align}
where the remainder collects the replacement of the intermediate values $\boldsymbol{\theta}_{y_k,*}$ by $\boldsymbol{\theta}_{y_{k,0}}$ in the expansion and the replacement of the realized Hessian by $\boldsymbol{H}_{X,\mathbf{y}}$, both negligible at the block scales established above, and where
$$\boldsymbol{L}_{X,\mathbf{y}} := \boldsymbol{\sigma}_{X,\mathbf{y}}^{-1} \boldsymbol{H}_{X,\mathbf{y}}^{-1} \, \mathcal{C}_{\boldsymbol{D}_{n,\mathbf{y}}}\!\left(16\,\boldsymbol{\Upsilon}_{X,\mathbf{y}}\right)^{-1}.$$
Since $\mathcal{C}_{\boldsymbol{D}_{n,\mathbf{y}}}(16\,\boldsymbol{\Upsilon}_{X,\mathbf{y}})^{-1}\mathcal{C}_{\boldsymbol{D}_{n,\mathbf{y}}}(16\,\boldsymbol{\Upsilon}_{X,\mathbf{y}})^{-1\prime} = 16\,\boldsymbol{\Upsilon}_{X,\mathbf{y}}$, it follows from \eqref{omega_X} that $\boldsymbol{L}_{X,\mathbf{y}}\boldsymbol{L}_{X,\mathbf{y}}' = \boldsymbol{P}_{X,\mathbf{y}}$, so that $\|\boldsymbol{L}_{X,\mathbf{y}}\|^2 \leq \operatorname{tr}(\boldsymbol{P}_{X,\mathbf{y}}) = Kp$. 

By \eqref{joint_score_clt_cond}, the standardized score
$$\boldsymbol{W}_{n,\mathbf{y}} := \mathcal{C}_{\boldsymbol{D}_{n,\mathbf{y}}}\!\left(16\,\boldsymbol{\Upsilon}_{X,\mathbf{y}}\right)\boldsymbol{S}_{n,\mathbf{y}}(\boldsymbol{\theta}_{\mathbf{y},0})$$
satisfies $d_{\mathrm{BL}}\left(\mathcal{L}(\boldsymbol{W}_{n,\mathbf{y}} \mid \mathcal{F}_{n,\mathbf{y}}), N(\boldsymbol{0}, \boldsymbol{I}_{Kp})\right) \xrightarrow{p} 0$, and by \eqref{joint_linear} the studentized estimator is obtained from it as
$$\boldsymbol{\sigma}_{X,\mathbf{y}}^{-1}\left(\boldsymbol{\theta}_{n,\mathbf{y}} - \boldsymbol{\theta}_{\mathbf{y},0}\right) = \boldsymbol{L}_{X,\mathbf{y}}\boldsymbol{W}_{n,\mathbf{y}} + o_p(1).$$
Two properties of $\boldsymbol{L}_{X,\mathbf{y}}$ transfer the approximation from the score to the estimator. First, it is $\mathcal{F}_{n,\mathbf{y}}$-measurable, hence a fixed matrix given the conditioning, so that for any $f$ with $\|f\|_\infty \leq 1$ and $\operatorname{Lip}(f) \leq 1$ the composition $f \circ \boldsymbol{L}_{X,\mathbf{y}}$ satisfies $\|f \circ \boldsymbol{L}_{X,\mathbf{y}}\|_\infty \leq 1$ and $\operatorname{Lip}(f \circ \boldsymbol{L}_{X,\mathbf{y}}) \leq \|\boldsymbol{L}_{X,\mathbf{y}}\|$; since $\|\boldsymbol{L}_{X,\mathbf{y}}\|^2 \leq Kp$,
$$d_{\mathrm{BL}}\left(\mathcal{L}(\boldsymbol{L}_{X,\mathbf{y}}\boldsymbol{W}_{n,\mathbf{y}} \mid \mathcal{F}_{n,\mathbf{y}}), \, \mathcal{L}(\boldsymbol{L}_{X,\mathbf{y}}\boldsymbol{Z} \mid \mathcal{F}_{n,\mathbf{y}})\right) \leq \max\left(1, (Kp)^{1/2}\right) d_{\mathrm{BL}}\left(\mathcal{L}(\boldsymbol{W}_{n,\mathbf{y}} \mid \mathcal{F}_{n,\mathbf{y}}), N(\boldsymbol{0}, \boldsymbol{I}_{Kp})\right) \xrightarrow{p} 0,$$
where $\boldsymbol{Z} \sim N(\boldsymbol{0}, \boldsymbol{I}_{Kp})$. Second, $\boldsymbol{L}_{X,\mathbf{y}}\boldsymbol{Z} \sim N(\boldsymbol{0}, \boldsymbol{L}_{X,\mathbf{y}}\boldsymbol{L}_{X,\mathbf{y}}') = N(\boldsymbol{0}, \boldsymbol{P}_{X,\mathbf{y}})$, so the transformed target is the required Gaussian law. 

Finally, the $o_p(1)$ remainder in the linear representation is absorbed. Write $\boldsymbol{T}_{n,\mathbf{y}}^{X} := \boldsymbol{\sigma}_{X,\mathbf{y}}^{-1}(\boldsymbol{\theta}_{n,\mathbf{y}} - \boldsymbol{\theta}_{\mathbf{y},0})$ for the studentized estimator and $R_n := \min(\|\boldsymbol{T}_{n,\mathbf{y}}^{X} - \boldsymbol{L}_{X,\mathbf{y}}\boldsymbol{W}_{n,\mathbf{y}}\|, 2)$. Let $f$ satisfy $\|f\|_\infty \leq 1$ and $\operatorname{Lip}(f) \leq 1$. The Lipschitz property gives $|f(\boldsymbol{T}_{n,\mathbf{y}}^{X}) - f(\boldsymbol{L}_{X,\mathbf{y}}\boldsymbol{W}_{n,\mathbf{y}})| \leq \|\boldsymbol{T}_{n,\mathbf{y}}^{X} - \boldsymbol{L}_{X,\mathbf{y}}\boldsymbol{W}_{n,\mathbf{y}}\|$, while boundedness gives $|f(\boldsymbol{T}_{n,\mathbf{y}}^{X}) - f(\boldsymbol{L}_{X,\mathbf{y}}\boldsymbol{W}_{n,\mathbf{y}})| \leq 2$. The difference is therefore bounded by $R_n$, which does not depend on $f$. Taking conditional expectations and the supremum over such $f$, 
\begin{align}
d_{\mathrm{BL}}&\left(\mathcal{L}\left(\boldsymbol{\sigma}_{X,\mathbf{y}}^{-1}(\boldsymbol{\theta}_{n,\mathbf{y}} - \boldsymbol{\theta}_{\mathbf{y},0}) \mid \mathcal{F}_{n,\mathbf{y}}\right), \, \mathcal{L}\left(\boldsymbol{L}_{X,\mathbf{y}}\boldsymbol{W}_{n,\mathbf{y}} \mid \mathcal{F}_{n,\mathbf{y}}\right)\right) \leq \mathbb{E}\left[R_n \mid \mathcal{F}_{n,\mathbf{y}}\right]. \nonumber
\end{align}
Since $R_n \leq 2$ and $R_n \xrightarrow{p} 0$ by the linear representation, for any $\varepsilon > 0$
$$\mathbb{E}[R_n] \leq \varepsilon + 2\Pr(R_n > \varepsilon) \longrightarrow \varepsilon,$$
so that $\mathbb{E}[R_n] \to 0$; the truncation is what rules out a vanishing-probability event with a large value keeping the mean away from zero. By the law of iterated expectations and Markov's inequality,
$$\Pr\left(\mathbb{E}[R_n \mid \mathcal{F}_{n,\mathbf{y}}] > \varepsilon\right) \leq \frac{\mathbb{E}[R_n]}{\varepsilon} \longrightarrow 0,$$
so that $\mathbb{E}[R_n \mid \mathcal{F}_{n,\mathbf{y}}] = o_p(1)$. Combining this with the bound on the transformed law and the identification of the transformed target, the triangle inequality for $d_{\mathrm{BL}}$ gives
\begin{align} \label{joint_gauss}
d_{\mathrm{BL}}\left(\mathcal{L}\left(\boldsymbol{\sigma}_{X,\mathbf{y}}^{-1}(\boldsymbol{\theta}_{n,\mathbf{y}} - \boldsymbol{\theta}_{\mathbf{y},0}) \mid \mathcal{F}_{n,\mathbf{y}}\right), \, N(\boldsymbol{0}, \boldsymbol{P}_{X,\mathbf{y}})\right) \xrightarrow{p} 0.
\end{align}
This establishes the joint Gaussian approximation underlying Theorem \ref{theorem3}(i), stated here with the conditional quantities; the replacement by their feasible counterparts is carried out in the inference discussion below, where parts (ii) and (iii) are also derived from \eqref{joint_gauss}.

Whitening by $\boldsymbol{P}_{X,\mathbf{y}}^{-1/2}$ would deliver a fixed limit: since $\boldsymbol{P}_{X,\mathbf{y}}$ has eigenvalues bounded away from zero, $\boldsymbol{P}_{X,\mathbf{y}}^{-1/2}\boldsymbol{\sigma}_{X,\mathbf{y}}^{-1}(\boldsymbol{\theta}_{n,\mathbf{y}} - \boldsymbol{\theta}_{\mathbf{y},0}) \xrightarrow{d} N(\boldsymbol{0}, \boldsymbol{I}_{Kp})$ holds, and parts (ii) and (iii) below are of exactly this whitened form, which is why their limits are fixed. It is not, however, the form required by the simultaneous bands of Section \ref{bands}. Those bands are rectangles in the coordinates of $\boldsymbol{\theta}_{\mathbf{y}}$: sets defined by one constraint per coefficient, which is what makes them readable threshold by threshold and invertible into statements about individual coefficients. A diagonal transformation such as $\boldsymbol{\sigma}_{X,\mathbf{y}}^{-1}$ rescales each constraint without mixing coordinates and therefore maps rectangles to rectangles. The matrix $\boldsymbol{P}_{X,\mathbf{y}}^{-1/2}$ is not diagonal, so under it each constraint becomes a bound on a linear combination of coefficients across thresholds, and the resulting region says nothing about any individual $\theta_{y_k,d}$.

The cross-threshold dependence is not eliminated by whitening; it is moved. The hypotheses of interest, and the coefficients the bands report, remain in the original coordinates, so the dependence reappears as soon as the whitened statement is inverted back into a statement about $\boldsymbol{\theta}_{\mathbf{y}}$. In the sup-$t$ construction that dependence is instead priced in the critical value, which is computed from $\boldsymbol{P}_{n,\mathbf{y}}(\boldsymbol{\theta}_{n,\mathbf{y}})$ and falls towards $z_{1-\alpha/2}$ when the estimates are highly correlated across nearby thresholds; this is the source of the gains over Bonferroni-type constructions noted in Section \ref{bands}. Studentizing by $\boldsymbol{\sigma}_{X,\mathbf{y}}$ removes the threshold-specific scales, which would otherwise make the coordinates incomparable, and leaves the correlations available to be priced. The coordinatewise statement with a correlation matrix varying in $n$ is therefore the form the bands require.\\
\\
\noindent \textbf{Step 4.}\\ 
\textit{Convergence rates.}\\
The next step is to determine the convergence rate, given by the order of $\boldsymbol{\Omega}_{n,\mathbf{y}}(\boldsymbol{\theta}_{\mathbf{y},0})$. The order of the diagonal blocks $\boldsymbol{H}_{n,k}(\boldsymbol{\theta}_{y_{k,0}})^{-1}$ are given by the single-threshold case and are of order $O_p(1/(n^4p_{n,k}))$ for $k = 1, ..., K$. Similarly, the order of the diagonal blocks of $\boldsymbol{\Upsilon}_{n,\mathbf{y}}(\boldsymbol{\theta}_{\mathbf{y},0})$ are $O_p(n^6p_{n,k})$ for each $k$. 

I now establish the order of the off-diagonal blocks. For any pair $k, l \in \{1, ..., K\}$ with $k \neq l$, by an application of the Cauchy-Schwarz inequality:
\begin{align}
    &\text{Cov}(\boldsymbol{S}_{n,k}(\boldsymbol{\theta}_{y_k,0}), \boldsymbol{S}_{n,l}(\boldsymbol{\theta}_{y_l,0})) = \mathbb{E}[\boldsymbol{S}_{n,k}(\boldsymbol{\theta}_{y_k,0}) \boldsymbol{S}_{n,l}(\boldsymbol{\theta}_{y_l,0})'] \nonumber \\
    &= \mathbb{E}\left[\left(\sum_{\sigma \in \mathcal{N}_{m_n}}\boldsymbol{s}_{k,\sigma}(\boldsymbol{\theta}_{y_k,0})\right) \left(\sum_{\sigma' \in \mathcal{N}_{m_n}}\boldsymbol{s}_{l,\sigma'}(\boldsymbol{\theta}_{y_l,0})\right)\right] \nonumber \\
    &= \sum_{\sigma \in \mathcal{N}_{m_n}} \sum_{\sigma' \in \mathcal{N}_{m_n}} \mathbb{E}[\boldsymbol{s}_{k,\sigma}(\boldsymbol{\theta}_{y_k,0}) \boldsymbol{s}_{l,\sigma'}(\boldsymbol{\theta}_{y_l,0})'] \nonumber \\
    &= \sum_{\sigma \in \mathcal{N}_{m_n}} \sum_{\sigma' \in \mathcal{N}_{m_n}} 1 \{\text{dyads}(\sigma) \cap \text{dyads}(\sigma') \neq \emptyset\} \text{Cov}(\boldsymbol{s}_{k,\sigma}(\boldsymbol{\theta}_{y_k,0}), \boldsymbol{s}_{l,\sigma'}(\boldsymbol{\theta}_{y_l,0})) \nonumber \\
    &\leq \sum_{\sigma \in \mathcal{N}_{m_n}} \sum_{\sigma' \in \mathcal{N}_{m_n}} 1\{\text{dyads}(\sigma) \cap \text{dyads}(\sigma') \neq \emptyset\} \|\text{Cov}(\boldsymbol{s}_{k,\sigma}(\boldsymbol{\theta}_{y_k,0}), \boldsymbol{s}_{l,\sigma'}(\boldsymbol{\theta}_{y_l,0}))\| \nonumber \\
    &\leq \sum_{\sigma \in \mathcal{N}_{m_n}} \sum_{\sigma' \in \mathcal{N}_{m_n}} 1\{\text{dyads}(\sigma) \cap \text{dyads}(\sigma') \neq \emptyset\} \sqrt{\|\text{Var}(\boldsymbol{s}_{k,\sigma}(\boldsymbol{\theta}_{y_k,0}))\|} \sqrt{\|\text{Var}(\boldsymbol{s}_{l,\sigma'}(\boldsymbol{\theta}_{y_l,0}))\|} \nonumber \\
    &= \sum_{\sigma \in \mathcal{N}_{m_n}} \sqrt{\|\text{Var}(\boldsymbol{s}_{k,\sigma}(\boldsymbol{\theta}_{y_k,0}))\|} \sum_{\sigma' \in \mathcal{N}_{m_n}} 1\{\text{dyads}(\sigma) \cap \text{dyads}(\sigma') \neq \emptyset\} \sqrt{\|\text{Var}(\boldsymbol{s}_{l,\sigma'}(\boldsymbol{\theta}_{y_l,0}))\|} \nonumber \\
    &= \sum_{\sigma \in \mathcal{N}_{m_n}} O(\sqrt{p_{k,\sigma}}) \sum_{\sigma' \in \mathcal{N}_{m_n}} 1\{\text{dyads}(\sigma) \cap \text{dyads}(\sigma') \neq \emptyset\} O(\sqrt{p_{l,\sigma'}}) \nonumber \\
    &= O(n^6\sqrt{p_{n,k}p_{n,l}})
\end{align}
where each quadruple at threshold $y_k$ has informativeness probability $p_{k,\sigma} = \Pr(z_{k,\sigma} \in \{-1,1\} \mid \alpha_k, \gamma_k)$, and $p_{n,k} = \frac{1}{m_n}\sum_{\sigma \in \mathcal{N}_{m_n}} \Pr(z_{k,\sigma} \in \{-1,1\} \mid \alpha_k, \gamma_k)$. The final result follows from the same combinatorial reasoning as in the single-threshold case.

From the above analysis, the block orders of the joint sandwich covariance can be determined, and they are consistent with the threshold-specific rates of Theorem \ref{theorem2}. For the diagonal blocks, where $k = 1, ..., K$:
\begin{align}
[\boldsymbol{\Omega}_{n,\mathbf{y}}(\boldsymbol{\theta}_{\mathbf{y},0})]_{kk} &= \boldsymbol{H}_{n,k}(\boldsymbol{\theta}_{y_{k,0}})^{-1} \boldsymbol{\Upsilon}_{n,kk}(\boldsymbol{\theta}_{y_{k,0}}) \boldsymbol{H}_{n,k}(\boldsymbol{\theta}_{y_{k,0}})^{-1} \\
&= O_p\left(\frac{1}{n^4 p_{n,k}}\right) \times O_p(n^6 p_{n,k}) \times O_p\left(\frac{1}{n^4 p_{n,k}}\right) = O_p\left(\frac{1}{n^2 p_{n,k}}\right)
\end{align}
which is asymptotically equivalent to $O_p(\sqrt{n(n-1)p_{n,k}})$ and therefore consistent with the rate obtained in Theorem \ref{theorem2} at each threshold.
For the off-diagonal blocks, where $k \neq l$:
\begin{align}
[\boldsymbol{\Omega}_{n,\mathbf{y}}(\boldsymbol{\theta}_{\mathbf{y},0})]_{kl} &= \boldsymbol{H}_{n,k}(\boldsymbol{\theta}_{y_{k,0}})^{-1} \boldsymbol{\Upsilon}_{n,kl}(\boldsymbol{\theta}_{y_{k,0}}, \boldsymbol{\theta}_{y_{l,0}}) \boldsymbol{H}_{n,l}(\boldsymbol{\theta}_{y_{l,0}})^{-1} \\
&= O_p\left(\frac{1}{n^4 p_{n,k}}\right) \times O_p(n^6 \sqrt{p_{n,k}p_{n,l}}) \times O_p\left(\frac{1}{n^4 p_{n,l}}\right) \\
&= O_p\left(\frac{1}{n^2 \sqrt{p_{n,k}p_{n,l}}}\right)
\end{align}

The off-diagonal orders imply that the correlation between the estimators at thresholds $y_k$ and $y_l$,
$$\frac{[\boldsymbol{\Omega}_{n,\mathbf{y}}(\boldsymbol{\theta}_{\mathbf{y},0})]_{kl}}{\sqrt{[\boldsymbol{\Omega}_{n,\mathbf{y}}(\boldsymbol{\theta}_{\mathbf{y},0})]_{kk}[\boldsymbol{\Omega}_{n,\mathbf{y}}(\boldsymbol{\theta}_{\mathbf{y},0})]_{ll}}} = O_p(1),$$
is of order one whatever the relative magnitudes of $p_{n,k}$ and $p_{n,l}$, since the geometric-mean scaling of the off-diagonal blocks cancels against the diagonal ones. The cross-threshold dependence therefore does not vanish under rate heterogeneity.

The joint Gaussian approximation of Theorem \ref{theorem3}(i) accommodates these heterogeneous rates through the block structure of the sandwich covariance.\\
\\
\noindent \textbf{Step 5.}\\ \textit{Feasible inference.}\\
For inference, the joint asymptotic variance is estimated using the plug-in estimator:
$$\boldsymbol{\Omega}_{n,\mathbf{y}}(\boldsymbol{\theta}_{n,\mathbf{y}}) = \boldsymbol{H}_{n,\mathbf{y}}(\boldsymbol{\theta}_{n,\mathbf{y}})^{-1} \boldsymbol{\Upsilon}_{n,\mathbf{y}}(\boldsymbol{\theta}_{n,\mathbf{y}}) \boldsymbol{H}_{n,\mathbf{y}}(\boldsymbol{\theta}_{n,\mathbf{y}})^{-1}.$$
Replacing $\boldsymbol{\theta}_{\mathbf{y},0}$ by the estimator relies on the following results:
\begin{enumerate}
\item The stochastic equicontinuity bound for the Hessian established in Step 3 of \ref{appendix_derivation}, applied to each diagonal block, which together with consistency gives $(m_n p_{n,k})^{-1}\|\boldsymbol{H}_{n,k}(\boldsymbol{\theta}_{n,y_k}) - \boldsymbol{H}_{n,k}(\boldsymbol{\theta}_{y_{k,0}})\| \xrightarrow{p} 0$ for each $k$.
\item The corresponding bound for the score covariance blocks, established below, giving $\|\boldsymbol{D}_{n,\mathbf{y}}^{-1}[\boldsymbol{\Upsilon}_{n,\mathbf{y}}(\boldsymbol{\theta}_{n,\mathbf{y}}) - \boldsymbol{\Upsilon}_{n,\mathbf{y}}(\boldsymbol{\theta}_{\mathbf{y},0})]\boldsymbol{D}_{n,\mathbf{y}}^{-1}\| \xrightarrow{p} 0$.
\item The consistency of $\boldsymbol{\theta}_{n,y_k} \xrightarrow{p} \boldsymbol{\theta}_{y_{k,0}}$ for each $k$ (Theorem \ref{theorem1}).
\end{enumerate}

Item 2 follows from a stochastic equicontinuity bound for $\boldsymbol{\Upsilon}_{n,kl}(\cdot,\cdot)$, defined in \eqref{upsilon_n_kl}, at the block scale. These covariance estimators depend on estimators that may converge at different rates when $p_{n,k} \neq p_{n,l}$; the bound below is established for all pairs $k, l \in \{1, ..., K\}$.\\
\\
\textit{Stochastic equicontinuity.}\\
For any $(\boldsymbol{\theta}_{k,1}, \boldsymbol{\theta}_{l,1}), (\boldsymbol{\theta}_{k,2}, \boldsymbol{\theta}_{l,2}) \in \Theta \times \Theta$, applying the product rule to each summand and the mean value theorem to each score difference (as in Step 3 of Appendix C for the Hessian), and using the bounds $\left\|\frac{\partial \boldsymbol{s}_k}{\partial \boldsymbol{\theta}'}\right\| \leq \|\boldsymbol{r}_\sigma\|^2 \sup_{\varepsilon \in \mathbb{R}} |f(\varepsilon)| \cdot 1\{z_{k,\sigma} \in \{-1,1\}\}$ and $\|\boldsymbol{s}_l(\sigma'; \boldsymbol{\theta}_{l,1})\| \leq \|\boldsymbol{r}_{\sigma'}\| \cdot 1\{z_{l,\sigma'} \in \{-1,1\}\}$:
\begin{adjustwidth}{-0.25in}{-0.25in}
\begin{align}
    &\|\boldsymbol{\Upsilon}_{n,kl}(\boldsymbol{\theta}_{k,1}, \boldsymbol{\theta}_{l,1}) - \boldsymbol{\Upsilon}_{n,kl}(\boldsymbol{\theta}_{k,2}, \boldsymbol{\theta}_{l,2})\| \nonumber \\[4pt]
    &\leq 16\sup_{\varepsilon \in \mathbb{R}} |f(\varepsilon)| \left(\sum_{i} \sum_{j \neq i} \sum_{i' \neq i,j} \sum_{j' \neq i,j,i'} \sum_{i'' \neq i,j,i'} \sum_{j'' \neq i,j,j',i''} \|\boldsymbol{r}_\sigma\|^2 \|\boldsymbol{r}_{\sigma'}\| 1\{z_{k,\sigma} \in \{-1,1\}\} 1\{z_{l,\sigma'} \in \{-1,1\}\}\right) \nonumber \\
    &\hspace{10cm} \times \|\boldsymbol{\theta}_{k,1} - \boldsymbol{\theta}_{k,2}\| \nonumber \\[4pt]
    &\quad + 16\sup_{\varepsilon \in \mathbb{R}} |f(\varepsilon)| \left(\sum_{i} \sum_{j \neq i} \sum_{i' \neq i,j} \sum_{j' \neq i,j,i'} \sum_{i'' \neq i,j,i'} \sum_{j'' \neq i,j,j',i''} \|\boldsymbol{r}_\sigma\| \|\boldsymbol{r}_{\sigma'}\|^2 1\{z_{k,\sigma} \in \{-1,1\}\} 1\{z_{l,\sigma'} \in \{-1,1\}\}\right) \nonumber \\
    &\hspace{10cm} \times \|\boldsymbol{\theta}_{l,1} - \boldsymbol{\theta}_{l,2}\|
\end{align}
\end{adjustwidth}
where $\sigma = \sigma\{i,i';j,j'\}$ and $\sigma' = \sigma\{i,i'';j,j''\}$. Denote the first bracketed sum by $B_{n,kl}^{(1)}$. Note that $\sup_{\varepsilon \in \mathbb{R}} |f(\varepsilon)|$ is a finite constant, since the logistic density $f(\varepsilon) = F(\varepsilon)(1 - F(\varepsilon))$ attains its maximum of $1/4$ at $\varepsilon = 0$. Leading to the expression:

\begin{align}
    &\frac{\|\boldsymbol{\Upsilon}_{n,kl}(\boldsymbol{\theta}_{k,1}, \boldsymbol{\theta}_{l,1}) - \boldsymbol{\Upsilon}_{n,kl}(\boldsymbol{\theta}_{k,2}, \boldsymbol{\theta}_{l,2})\|}{n^6\sqrt{p_{n,k}p_{n,l}}} \nonumber \\[4pt]
    &\leq 16\sup_{\varepsilon \in \mathbb{R}} |f(\varepsilon)| \left(\frac{B_{n,kl}^{(1)}}{n^6\sqrt{p_{n,k}p_{n,l}}}\right) \|\boldsymbol{\theta}_{k,1} - \boldsymbol{\theta}_{k,2}\| + 16\sup_{\varepsilon \in \mathbb{R}} |f(\varepsilon)| \left(\frac{B_{n,kl}^{(2)}}{n^6\sqrt{p_{n,k}p_{n,l}}}\right) \|\boldsymbol{\theta}_{l,1} - \boldsymbol{\theta}_{l,2}\|
\end{align}
for any $(\boldsymbol{\theta}_{k,1}, \boldsymbol{\theta}_{l,1}), (\boldsymbol{\theta}_{k,2}, \boldsymbol{\theta}_{l,2}) \in \Theta \times \Theta$. Using the same arguments as those used to establish Theorem 1 it follows that:
\begin{align}
    \left(n^6\sqrt{p_{n,k}p_{n,l}}\right)^{-1} B_{n,kl}^{(1)} = O_p(1)
\end{align}
The proof proceeds in formally showing this statement. By Chebyshev's inequality:
\begin{align}
    &\Pr\left(\left|\frac{B_{n,kl}^{(1)}}{n^6\sqrt{p_{n,k}p_{n,l}}} - \mathbb{E}\left[\frac{B_{n,kl}^{(1)}}{n^6\sqrt{p_{n,k}p_{n,l}}}\right]\right| > \varepsilon\right) \leq \frac{\text{Var}\left(B_{n,kl}^{(1)}\right)}{(n^6\sqrt{p_{n,k}p_{n,l}})^2 \varepsilon^2}
\end{align}

I proceed with bounding the expectation:
\begin{align}
    \mathbb{E}\left[B_{n,kl}^{(1)}\right] &= \sum_{i} \sum_{j \neq i} \sum_{i' \neq i,j} \sum_{j' \neq i,j,i'} \sum_{i'' \neq i,j,i'} \sum_{j'' \neq i,j,j',i''} \mathbb{E}\left[\|\boldsymbol{r}_\sigma\|^2 1\{z_{k,\sigma} \in \{-1,1\}\} \cdot \|\boldsymbol{r}_{\sigma'}\| 1\{z_{l,\sigma'} \in \{-1,1\}\}\right] \nonumber \\[4pt]
    &\leq \sum_{i} \sum_{j \neq i} \sum_{i' \neq i,j} \sum_{j' \neq i,j,i'} \sum_{i'' \neq i,j,i'} \sum_{j'' \neq i,j,j',i''} \sqrt{\mathbb{E}\left[\|\boldsymbol{r}_\sigma\|^4 1\{z_{k,\sigma} \in \{-1,1\}\}\right]} \nonumber \\
    &\hspace{6cm} \cdot \sqrt{\mathbb{E}\left[\|\boldsymbol{r}_{\sigma'}\|^2 1\{z_{l,\sigma'} \in \{-1,1\}\}\right]} \nonumber \\[4pt]
    &\leq C'' \sum_{i} \sum_{j \neq i} \sum_{i' \neq i,j} \sum_{j' \neq i,j,i'} \sum_{i'' \neq i,j,i'} \sum_{j'' \neq i,j,j',i''} \sqrt{p_{k,\sigma}} \cdot \sqrt{p_{l,\sigma'}} \nonumber \\[4pt]
    &= C'' \sum_{i} \sum_{j \neq i} \left(\sum_{i' \neq i,j} \sum_{j' \neq i,j,i'} \sqrt{p_{k,\sigma}}\right) \left(\sum_{i'' \neq i,j,i'} \sum_{j'' \neq i,j,j',i''} \sqrt{p_{l,\sigma'}}\right) \nonumber \\[4pt]
    &= O\left(n^6\sqrt{p_{n,k}\, p_{n,l}}\right)
\end{align}
where $C'' < \infty$ is a finite constant. The first inequality follows from the Cauchy-Schwarz inequality applied to each summand. The second inequality follows from the conditional expectation decomposition (as in equation (C.14)): $\mathbb{E}\left[\|\boldsymbol{r}_\sigma\|^4 1\{z_{k,\sigma} \in \{-1,1\}\}\right] = \mathbb{E}\left[\|\boldsymbol{r}_\sigma\|^4 \mid z_{k,\sigma} \in \{-1,1\}\right] \cdot \Pr(z_{k,\sigma} \in \{-1,1\}) \leq C' \cdot p_{k,\sigma}$, where $C' < \infty$ since the conditional expectation is uniformly bounded by Assumption \ref{assumption5}; similarly $\mathbb{E}\left[\|\boldsymbol{r}_{\sigma'}\|^2 1\{z_{l,\sigma'} \in \{-1,1\}\}\right] \leq C''' \cdot p_{l,\sigma'}$ for some finite constant $C''' < \infty$. The fourth line separates the 6-fold sum into two factors sharing the dyad $(i,j)$. The final equality follows from Jensen's inequality (since $\sqrt{\cdot}$ is concave): $\sum_{i',j'} \sqrt{p_{k,\sigma}} \leq O(n^2\sqrt{p_{n,k}})$ and $\sum_{i'',j''} \sqrt{p_{l,\sigma'}} \leq O(n^2\sqrt{p_{n,l}})$, multiplied by $O(n^2)$ choices for $(i,j)$, following the same argument as in equations (C.1) and (C.20). Therefore:
\begin{align}
    \mathbb{E}\left[\frac{B_{n,kl}^{(1)}}{n^6\sqrt{p_{n,k}p_{n,l}}}\right] = O(1)
\end{align}

Next, I proceed with bounding the variance. Let $\xi = (i,j,i',j',i'',j'')$ index a position in the 6-fold sum, determining the pair of quadruples $\sigma = \sigma\{i,i';j,j'\}$ and $\sigma' = \sigma\{i,i'';j,j''\}$ sharing dyad $(i,j)$. Denote:
$$A_\xi = \|\boldsymbol{r}_\sigma\|^2 \|\boldsymbol{r}_{\sigma'}\| 1\{z_{k,\sigma} \in \{-1,1\}\} 1\{z_{l,\sigma'} \in \{-1,1\}\} - \mathbb{E}\left[\|\boldsymbol{r}_\sigma\|^2 \|\boldsymbol{r}_{\sigma'}\| 1\{z_{k,\sigma} \in \{-1,1\}\} 1\{z_{l,\sigma'} \in \{-1,1\}\}\right]$$
and write $1\{\xi \cap \xi' \neq \emptyset\}$ to indicate that the node sets of $\xi$ and $\xi'$ share at least one node. Then:
\begin{align}
    &\text{Var}\left(B_{n,kl}^{(1)}\right) = \mathbb{E}\left[\left|\sum_{\xi} A_\xi\right|^2\right] \nonumber \\[4pt]
    &= \mathbb{E}\left[\left(\sum_{\xi} A_\xi\right)\left(\sum_{\xi'} A_{\xi'}\right)\right] \nonumber \\[4pt]
    &= \sum_{\xi} \sum_{\xi'} 1\left\{\xi \cap \xi' \neq \emptyset\right\} \text{Cov}\left(\|\boldsymbol{r}_\sigma\|^2 \|\boldsymbol{r}_{\sigma'}\| 1\{z_{k,\sigma} \in \{-1,1\}\} 1\{z_{l,\sigma'} \in \{-1,1\}\},\right. \nonumber \\
    &\hspace{5cm} \left. \|\boldsymbol{r}_{\sigma_1}\|^2 \|\boldsymbol{r}_{\sigma_1'}\| 1\{z_{k,\sigma_1} \in \{-1,1\}\} 1\{z_{l,\sigma_1'} \in \{-1,1\}\}\right) \nonumber \\[4pt]
    &\leq \sum_{\xi} \sum_{\xi'} 1\left\{\xi \cap \xi' \neq \emptyset\right\} \left|\text{Cov}\left(\|\boldsymbol{r}_\sigma\|^2 \|\boldsymbol{r}_{\sigma'}\| 1\{z_{k,\sigma} \in \{-1,1\}\} 1\{z_{l,\sigma'} \in \{-1,1\}\},\right.\right. \nonumber \\
    &\hspace{5cm} \left.\left. \|\boldsymbol{r}_{\sigma_1}\|^2 \|\boldsymbol{r}_{\sigma_1'}\| 1\{z_{k,\sigma_1} \in \{-1,1\}\} 1\{z_{l,\sigma_1'} \in \{-1,1\}\}\right)\right| \nonumber \\[4pt]
    &\leq \sum_{\xi} \sum_{\xi'} 1\left\{\xi \cap \xi' \neq \emptyset\right\} \sqrt{\mathbb{E}\left[\|\boldsymbol{r}_\sigma\|^4 \|\boldsymbol{r}_{\sigma'}\|^2 1\{z_{k,\sigma} \in \{-1,1\}\} 1\{z_{l,\sigma'} \in \{-1,1\}\}\right]} \nonumber \\
    &\hspace{4.5cm} \times \sqrt{\mathbb{E}\left[\|\boldsymbol{r}_{\sigma_1}\|^4 \|\boldsymbol{r}_{\sigma_1'}\|^2 1\{z_{k,\sigma_1} \in \{-1,1\}\} 1\{z_{l,\sigma_1'} \in \{-1,1\}\}\right]} \nonumber \\[4pt]
    &\leq B' \sum_{\xi} \sum_{\xi'} 1\left\{\xi \cap \xi' \neq \emptyset\right\} (p_{k,\sigma}\, p_{l,\sigma'})^{1/4} \cdot (p_{k,\sigma_1}\, p_{l,\sigma_1'})^{1/4} \nonumber \\[4pt]
    &= O\left(n^{11}\sqrt{p_{n,k}\, p_{n,l}}\right)
\end{align}
where $B' < \infty$ is a finite constant. The third equality follows since 6-tuples involving entirely distinct sets of nodes are independent by Assumption 1, contributing zero to the covariance. The second inequality follows from the Cauchy-Schwarz inequality and $\text{Var}(X) \leq \mathbb{E}[X^2]$, using $1\{z_{k,\sigma} \in \{-1,1\}\}^2 = 1\{z_{k,\sigma} \in \{-1,1\}\}$. The third inequality follows from the conditional expectation decomposition applied to the joint event:

\begin{align}
    &\mathbb{E}\left[\|\boldsymbol{r}_\sigma\|^4 \|\boldsymbol{r}_{\sigma'}\|^2 1\{z_{k,\sigma} \in \{-1,1\}\} 1\{z_{l,\sigma'} \in \{-1,1\}\}\right] \nonumber \\[4pt]
    &= \Pr\left(z_{k,\sigma} \in \{-1,1\},\, z_{l,\sigma'} \in \{-1,1\}\right) \cdot \mathbb{E}\left[\|\boldsymbol{r}_\sigma\|^4 \|\boldsymbol{r}_{\sigma'}\|^2 \mid z_{k,\sigma} \in \{-1,1\},\, z_{l,\sigma'} \in \{-1,1\}\right] \nonumber \\[4pt]
    &\leq \sqrt{p_{k,\sigma}\, p_{l,\sigma'}} \cdot C''''
\end{align}
where $C'''' < \infty$ is a finite constant bounding the conditional expectation (which is uniformly bounded since $\mathbb{E}\left[\|\boldsymbol{r}_\sigma\|^4 \|\boldsymbol{r}_{\sigma'}\|^2\right] \leq C_6 < \infty$ by H\"{o}lder's inequality with exponents $3/2$ and $3$: $\mathbb{E}[\|\boldsymbol{r}_\sigma\|^4 \|\boldsymbol{r}_{\sigma'}\|^2] \leq (\mathbb{E}[\|\boldsymbol{r}_\sigma\|^6])^{2/3} (\mathbb{E}[\|\boldsymbol{r}_{\sigma'}\|^6])^{1/3}$, which is finite by Assumption \ref{assumption5}); and the joint probability is bounded by the Cauchy-Schwarz inequality: $\Pr(z_{k,\sigma} \in \{-1,1\},\, z_{l,\sigma'} \in \{-1,1\}) \leq \sqrt{\Pr(z_{k,\sigma} \in \{-1,1\}) \Pr(z_{l,\sigma'} \in \{-1,1\})} = \sqrt{p_{k,\sigma}\, p_{l,\sigma'}}$.

The final order $O(n^{11}\sqrt{p_{n,k}\, p_{n,l}})$ follows from the same reasoning as in the proof of Theorem 1 and equations (C.15)--(C.18), by applying Jensen's inequality to the sums over $\xi$ and $\xi'$ separately: there are $O(n^6)$ 6-tuples $\xi$, and for each fixed $\xi$, there are $O(n^5)$ 6-tuples $\xi'$ with $\xi \cap \xi' \neq \emptyset$ (each $\xi$ involves 6 nodes; sharing one node leaves 5 free indices).

Furthermore:
\begin{align}
    \frac{\text{Var}\left(B_{n,kl}^{(1)}\right)}{(n^6\sqrt{p_{n,k}p_{n,l}})^2} = \frac{O(n^{11}\sqrt{p_{n,k}\, p_{n,l}})}{n^{12} p_{n,k} p_{n,l}} = O\left(\frac{1}{n\, \sqrt{p_{n,k}\, p_{n,l}}}\right) \rightarrow 0
\end{align}
since $np_{n,k} \rightarrow \infty$ for all $k$ by Assumption \ref{assumption4}, which implies $n\sqrt{p_{n,k}\, p_{n,l}} \geq n \min_k p_{n,k} \rightarrow \infty$.\\
\\
\noindent\textit{Conclusion for stochastic equicontinuity.} Gathering these results, by Chebyshev's inequality, $\allowbreak(n^6\sqrt{p_{n,k}p_{n,l}})^{-1} B_{n,kl}^{(1)} = O_p(1)$. By symmetry in the roles of $k$ and $l$, the same holds for $B_{n,kl}^{(2)}$. Because $\sup_{\varepsilon \in \mathbb{R}} |f(\varepsilon)|$ is a finite constant, it follows that:
\begin{align}
    \frac{\|\boldsymbol{\Upsilon}_{n,kl}(\boldsymbol{\theta}_{k,1}, \boldsymbol{\theta}_{l,1}) - \boldsymbol{\Upsilon}_{n,kl}(\boldsymbol{\theta}_{k,2}, \boldsymbol{\theta}_{l,2})\|}{n^6\sqrt{p_{n,k}p_{n,l}}} \leq O_p(1) \left(\|\boldsymbol{\theta}_{k,1} - \boldsymbol{\theta}_{k,2}\| + \|\boldsymbol{\theta}_{l,1} - \boldsymbol{\theta}_{l,2}\|\right)
\end{align}
for any $(\boldsymbol{\theta}_{k,1}, \boldsymbol{\theta}_{l,1}), (\boldsymbol{\theta}_{k,2}, \boldsymbol{\theta}_{l,2}) \in \Theta \times \Theta$, with the $O_p(1)$ constant not depending on the parameter values. Thus, the normalized $\boldsymbol{\Upsilon}_{n,kl}$ is stochastically equicontinuous over the product parameter space.

Since $\boldsymbol{\theta}_{n,y_k} \xrightarrow{p} \boldsymbol{\theta}_{y_{k,0}}$ for each $k$ by Theorem \ref{theorem1}, the stochastic equicontinuity bound gives, for every pair $(k,l)$,
$$\frac{\|\boldsymbol{\Upsilon}_{n,kl}(\boldsymbol{\theta}_{n,y_k}, \boldsymbol{\theta}_{n,y_l}) - \boldsymbol{\Upsilon}_{n,kl}(\boldsymbol{\theta}_{y_{k,0}}, \boldsymbol{\theta}_{y_{l,0}})\|}{n^6\sqrt{p_{n,k}p_{n,l}}} \xrightarrow{p} 0,$$
and since $K$ is fixed this is item 2 above.\\

\noindent\textit{Eigenvalue bounds at the estimator.}\\
By Assumption \ref{assumption_jointscorevar}, item 2 above, and Weyl's inequality,
$$\lambda_{\min}\left(\boldsymbol{D}_{n,\mathbf{y}}^{-1}\boldsymbol{\Upsilon}_{n,\mathbf{y}}(\boldsymbol{\theta}_{n,\mathbf{y}})\boldsymbol{D}_{n,\mathbf{y}}^{-1}\right) \geq c_{\mathbf{y}} - o_p(1) \geq c_{\mathbf{y}}/2$$
with probability approaching one, so that $\boldsymbol{\Upsilon}_{n,\mathbf{y}}(\boldsymbol{\theta}_{n,\mathbf{y}})$ is positive definite and $\|\boldsymbol{D}_{n,\mathbf{y}}\boldsymbol{\Upsilon}_{n,\mathbf{y}}(\boldsymbol{\theta}_{n,\mathbf{y}})^{-1}\boldsymbol{D}_{n,\mathbf{y}}\| = O_p(1)$. For the Hessian, the rank condition of Assumption \ref{assumption4} applied at each threshold, together with item 1 and the argument of the Feasible inference passage of \ref{appendix_derivation} applied blockwise, gives $\|(m_n p_{n,k})\boldsymbol{H}_{n,k}(\boldsymbol{\theta}_{n,y_k})^{-1}\| = O_p(1)$ for each $k$. The feasible sandwich matrix $\boldsymbol{\Omega}_{n,\mathbf{y}}(\boldsymbol{\theta}_{n,\mathbf{y}})$ is therefore positive definite with probability approaching one, so that the diagonal entries of $\boldsymbol{\sigma}_{n,\mathbf{y}}(\boldsymbol{\theta}_{n,\mathbf{y}})$ are positive and $\boldsymbol{P}_{n,\mathbf{y}}(\boldsymbol{\theta}_{n,\mathbf{y}})$ is well defined.\\

\noindent\textit{Comparison with the conditional sandwich covariance.}\\
To compare the feasible sandwich matrix with $\boldsymbol{\Omega}_{X,\mathbf{y}}$, substitute one factor at a time:
\begin{adjustwidth}{-0.25in}{-0.25in}
\begin{align} \label{sandwich_decomp}
\boldsymbol{\Omega}_{n,\mathbf{y}}(\boldsymbol{\theta}_{n,\mathbf{y}}) - \boldsymbol{\Omega}_{X,\mathbf{y}} &= \boldsymbol{H}_{n,\mathbf{y}}(\boldsymbol{\theta}_{n,\mathbf{y}})^{-1}\left[\boldsymbol{\Upsilon}_{n,\mathbf{y}}(\boldsymbol{\theta}_{n,\mathbf{y}}) - 16\,\boldsymbol{\Upsilon}_{X,\mathbf{y}}\right]\boldsymbol{H}_{n,\mathbf{y}}(\boldsymbol{\theta}_{n,\mathbf{y}})^{-1} \nonumber \\[1ex]
&\quad + \left[\boldsymbol{H}_{n,\mathbf{y}}(\boldsymbol{\theta}_{n,\mathbf{y}})^{-1} - \boldsymbol{H}_{X,\mathbf{y}}^{-1}\right]\left(16\,\boldsymbol{\Upsilon}_{X,\mathbf{y}}\right)\boldsymbol{H}_{n,\mathbf{y}}(\boldsymbol{\theta}_{n,\mathbf{y}})^{-1} \nonumber \\[1ex]
&\quad + \boldsymbol{H}_{X,\mathbf{y}}^{-1}\left(16\,\boldsymbol{\Upsilon}_{X,\mathbf{y}}\right)\left[\boldsymbol{H}_{n,\mathbf{y}}(\boldsymbol{\theta}_{n,\mathbf{y}})^{-1} - \boldsymbol{H}_{X,\mathbf{y}}^{-1}\right].
\end{align}
\end{adjustwidth}
Since $\boldsymbol{H}_{n,\mathbf{y}}$ is block-diagonal and $\boldsymbol{D}_{n,\mathbf{y}}$, 
$\boldsymbol{E}_{n,\mathbf{y}}$ and $\boldsymbol{Q}_{n,\mathbf{y}}$ act as scalars within each 
block, each of the three normalizers commutes with $\boldsymbol{H}_{n,\mathbf{y}}$ and with the 
others, and the three scales are consistent:
$$\boldsymbol{Q}_{n,\mathbf{y}}^{-1}\boldsymbol{D}_{n,\mathbf{y}} = \kappa_n \boldsymbol{E}_{n,\mathbf{y}}^2, \qquad \kappa_n := \left(\frac{n-1}{n}\right)^{1/2} \longrightarrow 1.$$
Normalizing \eqref{sandwich_decomp} by $\boldsymbol{Q}_{n,\mathbf{y}}^{-1}$ on both sides and 
inserting $\boldsymbol{D}_{n,\mathbf{y}}\boldsymbol{D}_{n,\mathbf{y}}^{-1}$ around the score 
covariance in the first term, and $\boldsymbol{E}_{n,\mathbf{y}}\boldsymbol{E}_{n,\mathbf{y}}^{-1}$ 
around the Hessian differences in the remaining two, gives
\begin{adjustwidth}{-0.25in}{-0.25in}
\begin{align} \label{sandwich_decomp_norm}
&\boldsymbol{Q}_{n,\mathbf{y}}^{-1}\left[\boldsymbol{\Omega}_{n,\mathbf{y}}(\boldsymbol{\theta}_{n,\mathbf{y}}) - \boldsymbol{\Omega}_{X,\mathbf{y}}\right]\boldsymbol{Q}_{n,\mathbf{y}}^{-1} \nonumber \\[1ex]
&\quad = \kappa_n^2 \left[\boldsymbol{E}_{n,\mathbf{y}}\boldsymbol{H}_{n,\mathbf{y}}(\boldsymbol{\theta}_{n,\mathbf{y}})^{-1}\boldsymbol{E}_{n,\mathbf{y}}\right] \left[\boldsymbol{D}_{n,\mathbf{y}}^{-1}\left(\boldsymbol{\Upsilon}_{n,\mathbf{y}}(\boldsymbol{\theta}_{n,\mathbf{y}}) - 16\,\boldsymbol{\Upsilon}_{X,\mathbf{y}}\right)\boldsymbol{D}_{n,\mathbf{y}}^{-1}\right] \left[\boldsymbol{E}_{n,\mathbf{y}}\boldsymbol{H}_{n,\mathbf{y}}(\boldsymbol{\theta}_{n,\mathbf{y}})^{-1}\boldsymbol{E}_{n,\mathbf{y}}\right] \nonumber \\[1ex]
&\qquad - \kappa_n^2 \left[\boldsymbol{E}_{n,\mathbf{y}}\boldsymbol{H}_{n,\mathbf{y}}(\boldsymbol{\theta}_{n,\mathbf{y}})^{-1}\boldsymbol{E}_{n,\mathbf{y}}\right] \left[\boldsymbol{E}_{n,\mathbf{y}}^{-1}\left(\boldsymbol{H}_{n,\mathbf{y}}(\boldsymbol{\theta}_{n,\mathbf{y}}) - \boldsymbol{H}_{X,\mathbf{y}}\right)\boldsymbol{E}_{n,\mathbf{y}}^{-1}\right] \left[\boldsymbol{E}_{n,\mathbf{y}}\boldsymbol{H}_{X,\mathbf{y}}^{-1}\boldsymbol{E}_{n,\mathbf{y}}\right] \nonumber \\[1ex]
&\qquad \hspace{2cm} \times \left[\boldsymbol{D}_{n,\mathbf{y}}^{-1}\left(16\,\boldsymbol{\Upsilon}_{X,\mathbf{y}}\right)\boldsymbol{D}_{n,\mathbf{y}}^{-1}\right] \left[\boldsymbol{E}_{n,\mathbf{y}}\boldsymbol{H}_{n,\mathbf{y}}(\boldsymbol{\theta}_{n,\mathbf{y}})^{-1}\boldsymbol{E}_{n,\mathbf{y}}\right] \nonumber \\[1ex]
&\qquad - \kappa_n^2 \left[\boldsymbol{E}_{n,\mathbf{y}}\boldsymbol{H}_{X,\mathbf{y}}^{-1}\boldsymbol{E}_{n,\mathbf{y}}\right] \left[\boldsymbol{D}_{n,\mathbf{y}}^{-1}\left(16\,\boldsymbol{\Upsilon}_{X,\mathbf{y}}\right)\boldsymbol{D}_{n,\mathbf{y}}^{-1}\right] \left[\boldsymbol{E}_{n,\mathbf{y}}\boldsymbol{H}_{n,\mathbf{y}}(\boldsymbol{\theta}_{n,\mathbf{y}})^{-1}\boldsymbol{E}_{n,\mathbf{y}}\right] \nonumber \\[1ex]
&\qquad \hspace{2cm} \times \left[\boldsymbol{E}_{n,\mathbf{y}}^{-1}\left(\boldsymbol{H}_{n,\mathbf{y}}(\boldsymbol{\theta}_{n,\mathbf{y}}) - \boldsymbol{H}_{X,\mathbf{y}}\right)\boldsymbol{E}_{n,\mathbf{y}}^{-1}\right] \left[\boldsymbol{E}_{n,\mathbf{y}}\boldsymbol{H}_{X,\mathbf{y}}^{-1}\boldsymbol{E}_{n,\mathbf{y}}\right],
\end{align}
\end{adjustwidth}
where the minus signs in the last two terms come from
$$\boldsymbol{H}_{n,\mathbf{y}}(\boldsymbol{\theta}_{n,\mathbf{y}})^{-1} - \boldsymbol{H}_{X,\mathbf{y}}^{-1} = -\,\boldsymbol{H}_{n,\mathbf{y}}(\boldsymbol{\theta}_{n,\mathbf{y}})^{-1}\left[\boldsymbol{H}_{n,\mathbf{y}}(\boldsymbol{\theta}_{n,\mathbf{y}}) - \boldsymbol{H}_{X,\mathbf{y}}\right]\boldsymbol{H}_{X,\mathbf{y}}^{-1},$$
an exact identity valid whenever both inverses exist, which is used because the results available bound the difference of the Hessians rather than the difference of their inverses.

Each bracket in \eqref{sandwich_decomp_norm} is now bounded separately. Two of them contain a difference and are shown to be negligible; the remaining ones are bounded in probability at their own scales. Because $\boldsymbol{\Upsilon}_{X,\mathbf{y}}$ and $\boldsymbol{H}_{X,\mathbf{y}}$ are defined as conditional expectations evaluated at $\boldsymbol{\theta}_{\mathbf{y},0}$, while the feasible matrices are evaluated at $\boldsymbol{\theta}_{n,\mathbf{y}}$, each of the two differences involves a change of argument and a change from the realized matrix to its conditional expectation; these are treated in turn.

\noindent\textit{The score covariance difference.} Adding and subtracting $\boldsymbol{\Upsilon}_{n,\mathbf{y}}(\boldsymbol{\theta}_{\mathbf{y},0})$ and applying the triangle inequality,
\begin{align}
&\left\|\boldsymbol{D}_{n,\mathbf{y}}^{-1}\left[\boldsymbol{\Upsilon}_{n,\mathbf{y}}(\boldsymbol{\theta}_{n,\mathbf{y}}) - 16\,\boldsymbol{\Upsilon}_{X,\mathbf{y}}\right]\boldsymbol{D}_{n,\mathbf{y}}^{-1}\right\| \nonumber \\
&\qquad \leq \underbrace{\left\|\boldsymbol{D}_{n,\mathbf{y}}^{-1}\left[\boldsymbol{\Upsilon}_{n,\mathbf{y}}(\boldsymbol{\theta}_{n,\mathbf{y}}) - \boldsymbol{\Upsilon}_{n,\mathbf{y}}(\boldsymbol{\theta}_{\mathbf{y},0})\right]\boldsymbol{D}_{n,\mathbf{y}}^{-1}\right\|}_{\text{item 2: change of argument}} + \underbrace{\left\|\boldsymbol{D}_{n,\mathbf{y}}^{-1}\left[\boldsymbol{\Upsilon}_{n,\mathbf{y}}(\boldsymbol{\theta}_{\mathbf{y},0}) - 16\,\boldsymbol{\Upsilon}_{X,\mathbf{y}}\right]\boldsymbol{D}_{n,\mathbf{y}}^{-1}\right\|}_{\eqref{joint_bridge}: \text{ passage to the conditional matrix}} \xrightarrow{p} 0. \nonumber
\end{align}
The first term is item 2, established by the stochastic equicontinuity bound for $\boldsymbol{\Upsilon}_{n,kl}$ together with the consistency of $\boldsymbol{\theta}_{n,y_k}$ at each threshold. The second is the joint bridge \eqref{joint_bridge} of Step 2, which carries the factor $16$ because $\boldsymbol{\Upsilon}_{n,\mathbf{y}}$ is defined in \eqref{upsilon_n_kl} with that multiplicity while $\boldsymbol{\Upsilon}_{X,\mathbf{y}}$ is built from the projections and carries none.

\noindent\textit{The Hessian difference.} By the same decomposition, now with $\boldsymbol{H}_{n,\mathbf{y}}(\boldsymbol{\theta}_{\mathbf{y},0})$ as the intermediate matrix,
\begin{align}
&\left\|\boldsymbol{E}_{n,\mathbf{y}}^{-1}\left[\boldsymbol{H}_{n,\mathbf{y}}(\boldsymbol{\theta}_{n,\mathbf{y}}) - \boldsymbol{H}_{X,\mathbf{y}}\right]\boldsymbol{E}_{n,\mathbf{y}}^{-1}\right\| \nonumber \\
&\qquad \leq \underbrace{\left\|\boldsymbol{E}_{n,\mathbf{y}}^{-1}\left[\boldsymbol{H}_{n,\mathbf{y}}(\boldsymbol{\theta}_{n,\mathbf{y}}) - \boldsymbol{H}_{n,\mathbf{y}}(\boldsymbol{\theta}_{\mathbf{y},0})\right]\boldsymbol{E}_{n,\mathbf{y}}^{-1}\right\|}_{\text{item 1: change of argument}} + \underbrace{\left\|\boldsymbol{E}_{n,\mathbf{y}}^{-1}\left[\boldsymbol{H}_{n,\mathbf{y}}(\boldsymbol{\theta}_{\mathbf{y},0}) - \boldsymbol{H}_{X,\mathbf{y}}\right]\boldsymbol{E}_{n,\mathbf{y}}^{-1}\right\|}_{\text{Step 3: passage to the conditional matrix}} \xrightarrow{p} 0. \nonumber
\end{align}
The first term is item 1, which is the blockwise stochastic equicontinuity bound of Step 3 of \ref{appendix_derivation} combined with consistency; note that $\boldsymbol{E}_{n,\mathbf{y}}^2$ has blocks $n^4 p_{n,k}$, which is the scale $m_n p_{n,k}$ at which item 1 is stated, up to the factor $\kappa_n^2$. The second term is the law-of-total-variance bound established earlier in Step 3.

\noindent\textit{The remaining brackets.} The three normalized inverse Hessians and the normalized conditional score covariance are bounded in probability at their own scales:
\begin{itemize}
\item $\left\|\boldsymbol{E}_{n,\mathbf{y}}\boldsymbol{H}_{n,\mathbf{y}}(\boldsymbol{\theta}_{n,\mathbf{y}})^{-1}\boldsymbol{E}_{n,\mathbf{y}}\right\| = O_p(1)$, by the eigenvalue bound $\|(m_n p_{n,k})\boldsymbol{H}_{n,k}(\boldsymbol{\theta}_{n,y_k})^{-1}\| = O_p(1)$ established above, applied blockwise. This is the only bracket evaluated at the estimator.
\item $\left\|\boldsymbol{E}_{n,\mathbf{y}}\boldsymbol{H}_{X,\mathbf{y}}^{-1}\boldsymbol{E}_{n,\mathbf{y}}\right\| = O_p(1)$, shown in Step 3.
\item $\left\|\boldsymbol{D}_{n,\mathbf{y}}^{-1}\left(16\,\boldsymbol{\Upsilon}_{X,\mathbf{y}}\right)\boldsymbol{D}_{n,\mathbf{y}}^{-1}\right\| = 16\left\|\overline{\boldsymbol{\Upsilon}}_{X,\mathbf{y}}\right\| = O_p(1)$, by the block orders established in Step 2.
\end{itemize}

\noindent\textit{Collecting the terms.}\\
Each of the three terms in \eqref{sandwich_decomp_norm} replaces one factor at a time and therefore contains exactly one bracket holding a difference of the two matrices being compared: the score covariance difference in the first term, and the Hessian difference in the second and third, the latter appearing in the middle through the identity above. Each such bracket is $o_p(1)$ and every remaining bracket is $O_p(1)$, so submultiplicativity of the operator norm bounds each term by a product of $O_p(1)$ factors and at least one $o_p(1)$ factor. Summing the three terms,
\begin{align} \label{sandwich_bridge}
\left\|\boldsymbol{Q}_{n,\mathbf{y}}^{-1}\left[\boldsymbol{\Omega}_{n,\mathbf{y}}(\boldsymbol{\theta}_{n,\mathbf{y}}) - \boldsymbol{\Omega}_{X,\mathbf{y}}\right]\boldsymbol{Q}_{n,\mathbf{y}}^{-1}\right\| \xrightarrow{p} 0.
\end{align}

\noindent\textit{Standard errors and correlations.}\\
Since $|\boldsymbol{e}_j'\boldsymbol{A}\boldsymbol{e}_j| \leq \|\boldsymbol{A}\|$ for any matrix $\boldsymbol{A}$, and $\boldsymbol{Q}_{n,\mathbf{y}}$ is diagonal, so that $\boldsymbol{Q}_{n,\mathbf{y}}^{-1}\boldsymbol{e}_j = [\boldsymbol{Q}_{n,\mathbf{y}}]_{jj}^{-1}\boldsymbol{e}_j$, \eqref{sandwich_bridge} gives
$$\frac{[\boldsymbol{\Omega}_{n,\mathbf{y}}(\boldsymbol{\theta}_{n,\mathbf{y}})]_{jj} - [\boldsymbol{\Omega}_{X,\mathbf{y}}]_{jj}}{[\boldsymbol{Q}_{n,\mathbf{y}}]_{jj}^2} \xrightarrow{p} 0, \qquad j = 1, \ldots, Kp.$$

Since $[\boldsymbol{\Omega}_{X,\mathbf{y}}]_{jj}/[\boldsymbol{Q}_{n,\mathbf{y}}]_{jj}^2 = [\boldsymbol{\Delta}_{n,\mathbf{y}}]_{jj}^2 \geq c_\Omega$ by Step 3,
$$\frac{[\boldsymbol{\Omega}_{n,\mathbf{y}}(\boldsymbol{\theta}_{n,\mathbf{y}})]_{jj}}{[\boldsymbol{\Omega}_{X,\mathbf{y}}]_{jj}} - 1 = \frac{1}{[\boldsymbol{\Delta}_{n,\mathbf{y}}]_{jj}^2} \cdot \frac{[\boldsymbol{\Omega}_{n,\mathbf{y}}(\boldsymbol{\theta}_{n,\mathbf{y}})]_{jj} - [\boldsymbol{\Omega}_{X,\mathbf{y}}]_{jj}}{[\boldsymbol{Q}_{n,\mathbf{y}}]_{jj}^2} \xrightarrow{p} 0,$$
uniformly in $j$ since $Kp$ is fixed. Taking square roots and reciprocals, both continuous at one, gives $[\boldsymbol{\sigma}_{n,\mathbf{y}}(\boldsymbol{\theta}_{n,\mathbf{y}})^{-1}\boldsymbol{\sigma}_{X,\mathbf{y}}]_{jj} \xrightarrow{p} 1$ for each $j$, and since both matrices are diagonal,
\begin{align} \label{feasible_swap_se}
\left\|\boldsymbol{\sigma}_{n,\mathbf{y}}(\boldsymbol{\theta}_{n,\mathbf{y}})^{-1}\boldsymbol{\sigma}_{X,\mathbf{y}} - \boldsymbol{I}_{Kp}\right\| \xrightarrow{p} 0.
\end{align}
For the correlation matrices, the same one-factor-at-a-time substitution gives
\begin{align} \label{corr_decomp}
&\boldsymbol{P}_{n,\mathbf{y}}(\boldsymbol{\theta}_{n,\mathbf{y}}) - \boldsymbol{P}_{X,\mathbf{y}} \nonumber \\[1ex]
&\quad = \boldsymbol{\sigma}_{n,\mathbf{y}}(\boldsymbol{\theta}_{n,\mathbf{y}})^{-1}\left[\boldsymbol{\Omega}_{n,\mathbf{y}}(\boldsymbol{\theta}_{n,\mathbf{y}}) - \boldsymbol{\Omega}_{X,\mathbf{y}}\right]\boldsymbol{\sigma}_{n,\mathbf{y}}(\boldsymbol{\theta}_{n,\mathbf{y}})^{-1} \nonumber \\[1ex]
&\qquad + \left[\boldsymbol{\sigma}_{n,\mathbf{y}}(\boldsymbol{\theta}_{n,\mathbf{y}})^{-1} - \boldsymbol{\sigma}_{X,\mathbf{y}}^{-1}\right]\boldsymbol{\Omega}_{X,\mathbf{y}}\,\boldsymbol{\sigma}_{n,\mathbf{y}}(\boldsymbol{\theta}_{n,\mathbf{y}})^{-1} \nonumber \\[1ex]
&\qquad + \boldsymbol{\sigma}_{X,\mathbf{y}}^{-1}\boldsymbol{\Omega}_{X,\mathbf{y}}\left[\boldsymbol{\sigma}_{n,\mathbf{y}}(\boldsymbol{\theta}_{n,\mathbf{y}})^{-1} - \boldsymbol{\sigma}_{X,\mathbf{y}}^{-1}\right].
\end{align}
Inserting $\boldsymbol{Q}_{n,\mathbf{y}}\boldsymbol{Q}_{n,\mathbf{y}}^{-1} = \boldsymbol{I}_{Kp}$ between adjacent factors, the three terms of \eqref{corr_decomp} become
\begin{align} \label{corr_decomp_norm}
&\boldsymbol{P}_{n,\mathbf{y}}(\boldsymbol{\theta}_{n,\mathbf{y}}) - \boldsymbol{P}_{X,\mathbf{y}} \nonumber \\[1ex]
&\quad = \left[\boldsymbol{\sigma}_{n,\mathbf{y}}(\boldsymbol{\theta}_{n,\mathbf{y}})^{-1}\boldsymbol{Q}_{n,\mathbf{y}}\right]\left[\boldsymbol{Q}_{n,\mathbf{y}}^{-1}\left(\boldsymbol{\Omega}_{n,\mathbf{y}}(\boldsymbol{\theta}_{n,\mathbf{y}}) - \boldsymbol{\Omega}_{X,\mathbf{y}}\right)\boldsymbol{Q}_{n,\mathbf{y}}^{-1}\right]\left[\boldsymbol{Q}_{n,\mathbf{y}}\boldsymbol{\sigma}_{n,\mathbf{y}}(\boldsymbol{\theta}_{n,\mathbf{y}})^{-1}\right] \nonumber \\[1ex]
&\qquad + \left[\left(\boldsymbol{\sigma}_{n,\mathbf{y}}(\boldsymbol{\theta}_{n,\mathbf{y}})^{-1} - \boldsymbol{\sigma}_{X,\mathbf{y}}^{-1}\right)\boldsymbol{Q}_{n,\mathbf{y}}\right]\left[\boldsymbol{Q}_{n,\mathbf{y}}^{-1}\boldsymbol{\Omega}_{X,\mathbf{y}}\boldsymbol{Q}_{n,\mathbf{y}}^{-1}\right]\left[\boldsymbol{Q}_{n,\mathbf{y}}\boldsymbol{\sigma}_{n,\mathbf{y}}(\boldsymbol{\theta}_{n,\mathbf{y}})^{-1}\right] \nonumber \\[1ex]
&\qquad + \left[\boldsymbol{\sigma}_{X,\mathbf{y}}^{-1}\boldsymbol{Q}_{n,\mathbf{y}}\right]\left[\boldsymbol{Q}_{n,\mathbf{y}}^{-1}\boldsymbol{\Omega}_{X,\mathbf{y}}\boldsymbol{Q}_{n,\mathbf{y}}^{-1}\right]\left[\boldsymbol{Q}_{n,\mathbf{y}}\left(\boldsymbol{\sigma}_{n,\mathbf{y}}(\boldsymbol{\theta}_{n,\mathbf{y}})^{-1} - \boldsymbol{\sigma}_{X,\mathbf{y}}^{-1}\right)\right].
\end{align}
Every bracket in \eqref{corr_decomp_norm} is now controlled by a result already established, and each term contains exactly one bracket holding a difference.

For the standard-error factors, $\boldsymbol{\sigma}_{X,\mathbf{y}} = \boldsymbol{Q}_{n,\mathbf{y}}\boldsymbol{\Delta}_{n,\mathbf{y}}$ gives $\boldsymbol{\sigma}_{X,\mathbf{y}}^{-1}\boldsymbol{Q}_{n,\mathbf{y}} = \boldsymbol{\Delta}_{n,\mathbf{y}}^{-1}$, whose entries lie in $[C_\Omega^{-1/2}, c_\Omega^{-1/2}]$ by Step 3, so that $\|\boldsymbol{\Delta}_{n,\mathbf{y}}^{-1}\| \leq c_\Omega^{-1/2}$. Writing the other two standard-error brackets in terms of it,
\begin{align}
&\boldsymbol{\sigma}_{n,\mathbf{y}}(\boldsymbol{\theta}_{n,\mathbf{y}})^{-1}\boldsymbol{Q}_{n,\mathbf{y}} = \left[\boldsymbol{\sigma}_{n,\mathbf{y}}(\boldsymbol{\theta}_{n,\mathbf{y}})^{-1}\boldsymbol{\sigma}_{X,\mathbf{y}}\right]\boldsymbol{\Delta}_{n,\mathbf{y}}^{-1}, \nonumber \\ &\left[\boldsymbol{\sigma}_{n,\mathbf{y}}(\boldsymbol{\theta}_{n,\mathbf{y}})^{-1} - \boldsymbol{\sigma}_{X,\mathbf{y}}^{-1}\right]\boldsymbol{Q}_{n,\mathbf{y}} = \left[\boldsymbol{\sigma}_{n,\mathbf{y}}(\boldsymbol{\theta}_{n,\mathbf{y}})^{-1}\boldsymbol{\sigma}_{X,\mathbf{y}} - \boldsymbol{I}_{Kp}\right]\boldsymbol{\Delta}_{n,\mathbf{y}}^{-1},\nonumber
\end{align}
both of which follow from $\boldsymbol{\sigma}_{X,\mathbf{y}}^{-1}\boldsymbol{Q}_{n,\mathbf{y}} = \boldsymbol{\Delta}_{n,\mathbf{y}}^{-1}$ and are valid in this order because all four matrices are diagonal. By \eqref{feasible_swap_se}, $\|\boldsymbol{\sigma}_{n,\mathbf{y}}(\boldsymbol{\theta}_{n,\mathbf{y}})^{-1}\boldsymbol{\sigma}_{X,\mathbf{y}}\| \leq 1 + o_p(1)$ and the bracket in the second equation is $o_p(1)$, so
$$\left\|\boldsymbol{\sigma}_{n,\mathbf{y}}(\boldsymbol{\theta}_{n,\mathbf{y}})^{-1}\boldsymbol{Q}_{n,\mathbf{y}}\right\| = O_p(1), \qquad \left\|\left[\boldsymbol{\sigma}_{n,\mathbf{y}}(\boldsymbol{\theta}_{n,\mathbf{y}})^{-1} - \boldsymbol{\sigma}_{X,\mathbf{y}}^{-1}\right]\boldsymbol{Q}_{n,\mathbf{y}}\right\| \xrightarrow{p} 0.$$
For the sandwich factors, $\|\boldsymbol{Q}_{n,\mathbf{y}}^{-1}\boldsymbol{\Omega}_{X,\mathbf{y}}\boldsymbol{Q}_{n,\mathbf{y}}^{-1}\| \leq C_\Omega$ by Step 3, and the difference bracket is $o_p(1)$ by \eqref{sandwich_bridge}. Each term in \eqref{corr_decomp_norm} is therefore a product of $O_p(1)$ factors with one $o_p(1)$ factor, so
\begin{align} \label{feasible_swap}
\left\|\boldsymbol{P}_{n,\mathbf{y}}(\boldsymbol{\theta}_{n,\mathbf{y}}) - \boldsymbol{P}_{X,\mathbf{y}}\right\| \xrightarrow{p} 0.
\end{align}

Together, \eqref{feasible_swap_se} and \eqref{feasible_swap} are what permit the conditional quantities of Step 3, which are not computable, to be replaced by their feasible counterparts: the first says that the feasible standard errors agree with the conditional ones in ratio, threshold by threshold, and the second that the feasible correlation matrix, which determines the simultaneous critical value, agrees with the conditional one in norm. Neither requires either matrix to converge; only their difference is controlled.\\
\\
\noindent\textbf{The three parts of Theorem \ref{theorem3} now follow.}\\
\noindent\textit{\textbf{Part (i).}} Write $\boldsymbol{T}_{n,\mathbf{y}}^{X} := \boldsymbol{\sigma}_{X,\mathbf{y}}^{-1}(\boldsymbol{\theta}_{n,\mathbf{y}} - \boldsymbol{\theta}_{\mathbf{y},0})$ for the studentized estimator built from the conditional standard errors, so that $\boldsymbol{T}_{n,\mathbf{y}} = \left[\boldsymbol{\sigma}_{n,\mathbf{y}}(\boldsymbol{\theta}_{n,\mathbf{y}})^{-1}\boldsymbol{\sigma}_{X,\mathbf{y}}\right]\boldsymbol{T}_{n,\mathbf{y}}^{X}$. Their difference is
$$\boldsymbol{T}_{n,\mathbf{y}} - \boldsymbol{T}_{n,\mathbf{y}}^{X} = \left[\boldsymbol{\sigma}_{n,\mathbf{y}}(\boldsymbol{\theta}_{n,\mathbf{y}})^{-1}\boldsymbol{\sigma}_{X,\mathbf{y}} - \boldsymbol{I}_{Kp}\right]\boldsymbol{T}_{n,\mathbf{y}}^{X} = o_p(1),$$
by \eqref{feasible_swap_se} together with the tightness of $\boldsymbol{T}_{n,\mathbf{y}}^{X}$, which follows from \eqref{joint_gauss} and the bound $\|\boldsymbol{P}_{X,\mathbf{y}}\| \leq Kp$.

The triangle inequality for the bounded--Lipschitz distance then gives
\begin{align}
&d_{\mathrm{BL}}\left(\mathcal{L}\left(\boldsymbol{T}_{n,\mathbf{y}} \mid \mathcal{F}_{n,\mathbf{y}}\right), N\left(\boldsymbol{0}, \boldsymbol{P}_{n,\mathbf{y}}(\boldsymbol{\theta}_{n,\mathbf{y}})\right)\right) \nonumber \\[1ex]
&\quad \leq \underbrace{d_{\mathrm{BL}}\left(\mathcal{L}\left(\boldsymbol{T}_{n,\mathbf{y}} \mid \mathcal{F}_{n,\mathbf{y}}\right), \mathcal{L}\left(\boldsymbol{T}_{n,\mathbf{y}}^{X} \mid \mathcal{F}_{n,\mathbf{y}}\right)\right)}_{\text{(a)}} + \underbrace{d_{\mathrm{BL}}\left(\mathcal{L}\left(\boldsymbol{T}_{n,\mathbf{y}}^{X} \mid \mathcal{F}_{n,\mathbf{y}}\right), N\left(\boldsymbol{0}, \boldsymbol{P}_{X,\mathbf{y}}\right)\right)}_{\text{(b)}} \nonumber \\[1ex]
&\qquad + \underbrace{d_{\mathrm{BL}}\left(N\left(\boldsymbol{0}, \boldsymbol{P}_{X,\mathbf{y}}\right), N\left(\boldsymbol{0}, \boldsymbol{P}_{n,\mathbf{y}}(\boldsymbol{\theta}_{n,\mathbf{y}})\right)\right)}_{\text{(c)}}. \nonumber
\end{align}
Term (a) vanishes by the displayed $o_p(1)$ bound: for any $f$ with $\|f\|_\infty \leq 1$ and $\operatorname{Lip}(f) \leq 1$, the difference $|f(\boldsymbol{T}_{n,\mathbf{y}}) - f(\boldsymbol{T}_{n,\mathbf{y}}^{X})|$ is bounded both by $\|\boldsymbol{T}_{n,\mathbf{y}} - \boldsymbol{T}_{n,\mathbf{y}}^{X}\|$ and by $2$, hence by the truncated remainder $\min(\|\boldsymbol{T}_{n,\mathbf{y}} - \boldsymbol{T}_{n,\mathbf{y}}^{X}\|, 2)$, which does not depend on $f$. The argument given at the end of Step 3 then applies in the same manner in this case, giving $d_{\mathrm{BL}} \xrightarrow{p} 0$. Term (b) is \eqref{joint_gauss}. Term (c) vanishes because the bounded--Lipschitz distance between mean-zero Gaussian laws is Lipschitz in the correlation matrix, with a constant depending only on the smallest eigenvalue through the square-root perturbation bound used in Step 2, uniformly over correlation matrices bounded away from singularity; the floor established in Step 3 applies to $\boldsymbol{P}_{X,\mathbf{y}}$ and 
transfers to $\boldsymbol{P}_{n,\mathbf{y}}(\boldsymbol{\theta}_{n,\mathbf{y}})$ by Weyl's inequality and \eqref{feasible_swap}, which also supplies the vanishing difference. Hence
$$d_{\mathrm{BL}}\left(\mathcal{L}\left(\boldsymbol{T}_{n,\mathbf{y}} \mid \mathcal{F}_{n,\mathbf{y}}\right), \, N\left(\boldsymbol{0}, \boldsymbol{P}_{n,\mathbf{y}}(\boldsymbol{\theta}_{n,\mathbf{y}})\right)\right) \xrightarrow{p} 0,$$
which is Theorem \ref{theorem3}(i).\\

\noindent\textit{\textbf{Part (ii).}} Fix a non-zero $\boldsymbol{a} \in \mathbb{R}^{Kp}$ and set $\boldsymbol{b} := \boldsymbol{\sigma}_{X,\mathbf{y}}\boldsymbol{a}$, which is non-zero with probability approaching one since $\boldsymbol{\sigma}_{X,\mathbf{y}}$ is invertible by Step 3, so that
$$\frac{\boldsymbol{a}'(\boldsymbol{\theta}_{n,\mathbf{y}} - \boldsymbol{\theta}_{\mathbf{y},0})}{\sqrt{\boldsymbol{a}'\boldsymbol{\Omega}_{X,\mathbf{y}}\boldsymbol{a}}} = \frac{\boldsymbol{b}'\boldsymbol{T}_{n,\mathbf{y}}^{X}}{\sqrt{\boldsymbol{b}'\boldsymbol{P}_{X,\mathbf{y}}\boldsymbol{b}}}, \qquad \boldsymbol{T}_{n,\mathbf{y}}^{X} = \boldsymbol{\sigma}_{X,\mathbf{y}}^{-1}(\boldsymbol{\theta}_{n,\mathbf{y}} - \boldsymbol{\theta}_{\mathbf{y},0}).$$
Since $\boldsymbol{b}$ is $\mathcal{F}_{n,\mathbf{y}}$-measurable, the right-hand side is a fixed linear function of $\boldsymbol{T}_{n,\mathbf{y}}^{X}$ given the conditioning, divided by the standard deviation that this function has under the Gaussian approximation of \eqref{joint_gauss}. Writing $\ell(\boldsymbol{t}) := \boldsymbol{b}'\boldsymbol{t}/\sqrt{\boldsymbol{b}'\boldsymbol{P}_{X,\mathbf{y}}\boldsymbol{b}}$ for that map, $\ell$ is linear, so its Lipschitz constant is $\|\boldsymbol{b}\|/\sqrt{\boldsymbol{b}'\boldsymbol{P}_{X,\mathbf{y}}\boldsymbol{b}}$ by the Cauchy--Schwarz inequality. Since $\boldsymbol{b}'\boldsymbol{P}_{X,\mathbf{y}}\boldsymbol{b} \geq \lambda_{\min}(\boldsymbol{P}_{X,\mathbf{y}})\|\boldsymbol{b}\|^2$, the norm $\|\boldsymbol{b}\|$ cancels and
$$\operatorname{Lip}(\ell) \leq \lambda_{\min}(\boldsymbol{P}_{X,\mathbf{y}})^{-1/2} \leq (C_\Omega/c_\Omega)^{1/2}$$
by the eigenvalue floor of Step 3, a bound not depending on $\boldsymbol{b}$. Consequently, for any $f: \mathbb{R} \to \mathbb{R}$ with $\|f\|_\infty \leq 1$ and $\operatorname{Lip}(f) \leq 1$, the function $f \circ \ell$ on $\mathbb{R}^{Kp}$ satisfies $\|f \circ \ell\|_\infty \leq 1$ and $\operatorname{Lip}(f \circ \ell) \leq (C_\Omega/c_\Omega)^{1/2}$, so the approximation of \eqref{joint_gauss} transfers from the vector to the scalar, as in the treatment of $\boldsymbol{L}_{X,\mathbf{y}}$ in Step 3. Since a fixed linear function of an $N(\boldsymbol{0}, \boldsymbol{P}_{X,\mathbf{y}})$ vector divided by its own standard deviation is standard normal, and this limit is fixed,
$$\frac{\boldsymbol{a}'(\boldsymbol{\theta}_{n,\mathbf{y}} - \boldsymbol{\theta}_{\mathbf{y},0})}{\sqrt{\boldsymbol{a}'\boldsymbol{\Omega}_{X,\mathbf{y}}\boldsymbol{a}}} \xrightarrow{d} N(0,1)$$
conditionally on $\mathcal{F}_{n,\mathbf{y}}$. Self-standardization removes the dependence on $\boldsymbol{P}_{X,\mathbf{y}}$: whatever that matrix is, the standardized scalar is standard normal, so the target no longer varies with $n$ and no longer depends on the conditioning variables. The convergence therefore also holds unconditionally, and the remaining steps can be carried out without conditioning.

It remains to replace $\boldsymbol{\Omega}_{X,\mathbf{y}}$ by its feasible counterpart. Applying \eqref{sandwich_bridge} to the unit vector $\boldsymbol{Q}_{n,\mathbf{y}}\boldsymbol{a}/\|\boldsymbol{Q}_{n,\mathbf{y}}\boldsymbol{a}\|$ gives
$$\frac{\boldsymbol{a}'\left[\boldsymbol{\Omega}_{n,\mathbf{y}}(\boldsymbol{\theta}_{n,\mathbf{y}}) - \boldsymbol{\Omega}_{X,\mathbf{y}}\right]\boldsymbol{a}}{\|\boldsymbol{Q}_{n,\mathbf{y}}\boldsymbol{a}\|^2} \xrightarrow{p} 0,$$
while the lower bound $\boldsymbol{a}'\boldsymbol{\Omega}_{X,\mathbf{y}}\boldsymbol{a} \geq c_\Omega\|\boldsymbol{Q}_{n,\mathbf{y}}\boldsymbol{a}\|^2$ follows from the lower bound $c_\Omega \boldsymbol{I}_{Kp} \preceq \boldsymbol{Q}_{n,\mathbf{y}}^{-1}\boldsymbol{\Omega}_{X,\mathbf{y}}\boldsymbol{Q}_{n,\mathbf{y}}^{-1}$ of Step 3, evaluated at $\boldsymbol{Q}_{n,\mathbf{y}}\boldsymbol{a}$. Dividing the equation above by $\boldsymbol{a}'\boldsymbol{\Omega}_{X,\mathbf{y}}\boldsymbol{a}/\|\boldsymbol{Q}_{n,\mathbf{y}}\boldsymbol{a}\|^2 \geq c_\Omega$,
$$\frac{\boldsymbol{a}'\boldsymbol{\Omega}_{n,\mathbf{y}}(\boldsymbol{\theta}_{n,\mathbf{y}})\boldsymbol{a}}{\boldsymbol{a}'\boldsymbol{\Omega}_{X,\mathbf{y}}\boldsymbol{a}} - 1 \xrightarrow{p} 0,$$
and Slutsky's theorem gives Theorem \ref{theorem3}(ii). The bound holds for every $\boldsymbol{a}$, including those loading on several thresholds, since numerator and denominator are dominated by the same blocks. Slutsky's theorem applies here, unlike in the passage from the score to the estimator in Step 3, without requiring the normalized covariance to converge: the multiplying object is a ratio of the same quadratic form evaluated at $\boldsymbol{\theta}_{n,\mathbf{y}}$ and at $\boldsymbol{\theta}_{\mathbf{y},0}$, so that whatever drift the covariance exhibits affects numerator and denominator alike and cancels. The matrix $\boldsymbol{L}_{X,\mathbf{y}}$ in Step 3 is not a comparison of this kind but a product of distinct matrices, so no such cancellation occurs.\\

\noindent\textit{\textbf{Part (iii).}} Let $\boldsymbol{R}$ be a fixed $q \times Kp$ matrix of full row rank. Pre- and post-multiplying the lower bound $c_\Omega \boldsymbol{I}_{Kp} \preceq \boldsymbol{Q}_{n,\mathbf{y}}^{-1}\boldsymbol{\Omega}_{X,\mathbf{y}}\boldsymbol{Q}_{n,\mathbf{y}}^{-1}$ of Step 3 by $\boldsymbol{R}\boldsymbol{Q}_{n,\mathbf{y}}$ and its transpose gives
$$\boldsymbol{R}\boldsymbol{\Omega}_{X,\mathbf{y}}\boldsymbol{R}' \succeq c_\Omega (\boldsymbol{R}\boldsymbol{Q}_{n,\mathbf{y}})(\boldsymbol{R}\boldsymbol{Q}_{n,\mathbf{y}})',$$
which is positive definite since $\boldsymbol{R}$ has full row rank and $\boldsymbol{Q}_{n,\mathbf{y}}$ is nonsingular, so that $(\boldsymbol{R}\boldsymbol{\Omega}_{X,\mathbf{y}}\boldsymbol{R}')^{-1/2}$ is well defined with probability approaching one. Setting $\boldsymbol{M}_{n} := (\boldsymbol{R}\boldsymbol{\Omega}_{X,\mathbf{y}}\boldsymbol{R}')^{-1/2}\boldsymbol{R}\boldsymbol{Q}_{n,\mathbf{y}}$ and rearranging the equation above,
$$\boldsymbol{R}\boldsymbol{Q}_{n,\mathbf{y}}\boldsymbol{Q}_{n,\mathbf{y}}\boldsymbol{R}' \preceq c_\Omega^{-1}\,\boldsymbol{R}\boldsymbol{\Omega}_{X,\mathbf{y}}\boldsymbol{R}'.$$
Pre- and post-multiplying by $(\boldsymbol{R}\boldsymbol{\Omega}_{X,\mathbf{y}}\boldsymbol{R}')^{-1/2}$ preserves the inequality, since the quadratic form of each side at any vector becomes the quadratic form at the transformed vector, and gives $\boldsymbol{M}_{n}\boldsymbol{M}_{n}' \preceq c_\Omega^{-1}\boldsymbol{I}_q$, so that $\|\boldsymbol{M}_{n}\|^2 = \lambda_{\max}(\boldsymbol{M}_{n}\boldsymbol{M}_{n}') \leq c_\Omega^{-1}$.

By the argument of part (ii), applied to the $q$-dimensional vector $\boldsymbol{R}(\boldsymbol{\theta}_{n,\mathbf{y}} - \boldsymbol{\theta}_{\mathbf{y},0})$ rather than to a scalar, with the map $\boldsymbol{t} \mapsto (\boldsymbol{R}\boldsymbol{\Omega}_{X,\mathbf{y}}\boldsymbol{R}')^{-1/2}\boldsymbol{R}\boldsymbol{\sigma}_{X,\mathbf{y}}\boldsymbol{t}$ in place of $\ell$, whose norm is bounded by $\|\boldsymbol{M}_{n}\|\,\|\boldsymbol{Q}_{n,\mathbf{y}}^{-1}\boldsymbol{\sigma}_{X,\mathbf{y}}\| = \|\boldsymbol{M}_{n}\|\,\|\boldsymbol{\Delta}_{n,\mathbf{y}}\| \leq (C_\Omega/c_\Omega)^{1/2}$ by Step 3,
\begin{align} \label{wald_gauss}
\boldsymbol{\xi}_{n} := \left(\boldsymbol{R}\boldsymbol{\Omega}_{X,\mathbf{y}}\boldsymbol{R}'\right)^{-1/2}\boldsymbol{R}(\boldsymbol{\theta}_{n,\mathbf{y}} - \boldsymbol{\theta}_{\mathbf{y},0}) \xrightarrow{d} N(\boldsymbol{0}, \boldsymbol{I}_q).
\end{align}

The Wald statistic normalizes by $\boldsymbol{R}\boldsymbol{\Omega}_{n,\mathbf{y}}(\boldsymbol{\theta}_{n,\mathbf{y}})\boldsymbol{R}'$ rather than by $\boldsymbol{R}\boldsymbol{\Omega}_{X,\mathbf{y}}\boldsymbol{R}'$. Set
$$\boldsymbol{A}_n := \left(\boldsymbol{R}\boldsymbol{\Omega}_{X,\mathbf{y}}\boldsymbol{R}'\right)^{-1/2}\left(\boldsymbol{R}\boldsymbol{\Omega}_{n,\mathbf{y}}(\boldsymbol{\theta}_{n,\mathbf{y}})\boldsymbol{R}'\right)\left(\boldsymbol{R}\boldsymbol{\Omega}_{X,\mathbf{y}}\boldsymbol{R}'\right)^{-1/2},$$
which measures the discrepancy between the two normalizations. Since $$\boldsymbol{I}_q = (\boldsymbol{R}\boldsymbol{\Omega}_{X,\mathbf{y}}\boldsymbol{R}')^{-1/2}(\boldsymbol{R}\boldsymbol{\Omega}_{X,\mathbf{y}}\boldsymbol{R}')(\boldsymbol{R}\boldsymbol{\Omega}_{X,\mathbf{y}}\boldsymbol{R}')^{-1/2},$$ subtracting this from the definition of $\boldsymbol{A}_n$ and collecting the two middle factors gives
$$\boldsymbol{A}_n - \boldsymbol{I}_q = \left(\boldsymbol{R}\boldsymbol{\Omega}_{X,\mathbf{y}}\boldsymbol{R}'\right)^{-1/2}\boldsymbol{R}\left[\boldsymbol{\Omega}_{n,\mathbf{y}}(\boldsymbol{\theta}_{n,\mathbf{y}}) - \boldsymbol{\Omega}_{X,\mathbf{y}}\right]\boldsymbol{R}'\left(\boldsymbol{R}\boldsymbol{\Omega}_{X,\mathbf{y}}\boldsymbol{R}'\right)^{-1/2}.$$
Writing $\boldsymbol{R} = \boldsymbol{R}\boldsymbol{Q}_{n,\mathbf{y}}\boldsymbol{Q}_{n,\mathbf{y}}^{-1}$ on the left and $\boldsymbol{R}' = \boldsymbol{Q}_{n,\mathbf{y}}^{-1}\boldsymbol{Q}_{n,\mathbf{y}}\boldsymbol{R}'$ on the right, the outer factors combine into $\boldsymbol{M}_{n}$ and $\boldsymbol{M}_{n}'$, so that
\begin{align} \label{wald_swap}
\left\|\boldsymbol{A}_n - \boldsymbol{I}_q\right\| \leq \|\boldsymbol{M}_{n}\|^2 \left\|\boldsymbol{Q}_{n,\mathbf{y}}^{-1}\left[\boldsymbol{\Omega}_{n,\mathbf{y}}(\boldsymbol{\theta}_{n,\mathbf{y}}) - \boldsymbol{\Omega}_{X,\mathbf{y}}\right]\boldsymbol{Q}_{n,\mathbf{y}}^{-1}\right\| \xrightarrow{p} 0
\end{align}
by \eqref{sandwich_bridge}. By continuity of matrix inversion at $\boldsymbol{I}_q$, $\|\boldsymbol{A}_n^{-1} - \boldsymbol{I}_q\| \xrightarrow{p} 0$ as well.

Finally, the Wald statistic is
$$\mathcal{W}_n := \boldsymbol{R}(\boldsymbol{\theta}_{n,\mathbf{y}} - \boldsymbol{\theta}_{\mathbf{y},0})'\left(\boldsymbol{R}\boldsymbol{\Omega}_{n,\mathbf{y}}(\boldsymbol{\theta}_{n,\mathbf{y}})\boldsymbol{R}'\right)^{-1}\boldsymbol{R}(\boldsymbol{\theta}_{n,\mathbf{y}} - \boldsymbol{\theta}_{\mathbf{y},0}),$$
which involves only feasible quantities. Substituting $\boldsymbol{R}(\boldsymbol{\theta}_{n,\mathbf{y}} - \boldsymbol{\theta}_{\mathbf{y},0}) = (\boldsymbol{R}\boldsymbol{\Omega}_{X,\mathbf{y}}\boldsymbol{R}')^{1/2}\boldsymbol{\xi}_{n}$ from \eqref{wald_gauss} and using $\boldsymbol{A}_n^{-1} = (\boldsymbol{R}\boldsymbol{\Omega}_{X,\mathbf{y}}\boldsymbol{R}')^{1/2}(\boldsymbol{R}\boldsymbol{\Omega}_{n,\mathbf{y}}(\boldsymbol{\theta}_{n,\mathbf{y}})\boldsymbol{R}')^{-1}(\boldsymbol{R}\boldsymbol{\Omega}_{X,\mathbf{y}}\boldsymbol{R}')^{1/2}$ gives $\mathcal{W}_n = \boldsymbol{\xi}_{n}'\boldsymbol{A}_n^{-1}\boldsymbol{\xi}_{n}$. Since $\boldsymbol{\xi}_{n} \xrightarrow{d} N(\boldsymbol{0}, \boldsymbol{I}_q)$ and $\boldsymbol{A}_n^{-1} \xrightarrow{p} \boldsymbol{I}_q$, Slutsky's theorem gives $\mathcal{W}_n = \|\boldsymbol{\xi}_{n}\|^2 + o_p(1)$, and the continuous mapping theorem applied to the squared norm gives $\mathcal{W}_n \xrightarrow{d} \chi^2_q$. This is Theorem \ref{theorem3}(iii).

The proof of Theorem \ref{theorem3} is thus complete.
\section{Simulations}\label{simulation_appendix}

\subsection{Monte Carlo Simulations for a Single Threshold}\label{sim_single_threshold}

Building on the data generating processes (DGP) of \citet{jochmans2018semiparametric}, I consider additional specifications to better understand the implied growth rates of fixed effects that satisfy Assumption \ref{assumption4} and their connection to network sparsity. The outcome for a network with $n$ nodes is generally generated as
$$ y_{ij} = {1} \{x_{ij} \theta_0 + \alpha_i + \gamma_j - \epsilon_{ij} \geq 0\}, $$
where $\varepsilon_{ij} \sim \text{i.i.d.}\ 
\text{Logistic}(0,1)$ across the $n(n-1)$ dyads. Following Section \ref{asymptotics}, the fixed effects take the general form
$ \alpha_i = a_{n,i} g(n)$, and $\gamma_j = b_{n,j} g(n)$, where $a_{n,i}, b_{n,j} \in [0,1]$ captures cross-node 
heterogeneity and $g(n)$ controls the growth rate of the fixed effects with $n$. The distribution
of $(a_{n,i}, b_{n,j})$ and the rate $g(n)$ together determine the degree of network sparsity. The DGPs below differ in the specification of $(a_{n,i}, b_{n,j})$.

I set $\theta_0 = -1$ throughout, to ensure that covariates and fixed effects operate in the same direction, with the single regressor 
generated as $x_{ij} = -|u_i - u_j|$ where 
$u_i = \nu_i - 1/2$ for $\nu_i \sim Beta(2,2)$. Positive 
values of $g(n)$ push the fixed effects toward $+\infty$, 
generating networks where most pairs are linked and 
sparsity arises through few non-links. As $g(n)$ 
grows, both $p_n$ and $1 - q_n$ shrink to zero, reducing 
the share of informative quadruples available for 
estimation. This is the 
empirically relevant case for the trade application in 
Section~\ref{application}: bilateral trade is bounded below at zero with a mass point at zero, so the estimable thresholds lie above the bound where $\Pr(y_{ij}\le y)$ rises toward one, and sparsity arises only in the right tail. \citet{jochmans2018semiparametric} studies the 
left-tail case (few links); the right-tail analysis is 
new to this paper.

Sample sizes are $n \in \{25, 50, 100, 150\}$, with $1000$ Monte Carlo 
replications. The values of $g(n)$ 
range from $0$ (the dense case) through $\log\log n$, 
$\sqrt{\log n}$, $\log n$, $2\log n$, $(\log n)^{3/2}$, 
and $(\log n)^2$, in order of increasing sparsity. Tables \ref{tab:net_jochmans_right_full}, \ref{tab:est_jochmans_right_uniform_full} and \ref{tab:est_jochmans_right_beta_full} report the complete set of results.

\paragraph{DGP 1: Uniform}

Following \citet{jochmans2018semiparametric}, the fixed effects are a deterministic function of the sample size:
$$ \alpha_i = \frac{n-i}{n-1} g(n), \quad \gamma_i = \alpha_i,$$
yielding a sequence $a_{n,i}$ uniformly spaced between $[0,1]$. For large $g(n)$, the condition in Assumption \ref{assumption4} translates into the 
requirement that $g(n) = o(n^{1/4})$. Since all the growth rates $g(n)$ 
considered are logarithmic functions of $n$, they all 
satisfy Assumption~\ref{assumption4}, as any 
power of $\log n$ grows slower than $n^{1/4}$. However, the 
speed at which $np_n$ grows with $n$ varies substantially 
across specifications: for slower-growing $g(n)$ such as 
$\log\log n$ or $\sqrt{\log n}$, the growth of $np_n$ is 
apparent at moderate sample sizes, while for faster-growing 
$g(n)$ such as $(\log n)^2$, it may require sample sizes 
far beyond those considered here. 

\paragraph{DGP 2: Beta heterogeneity}

The main change in this specification is regarding the sequence of fixed effects. I consider a deterministic sequence of fixed effects, that is normalized to be in the range $[0,1]$, using a Beta quantile function such that
$$ a_{n,i} = F^{-1}_{Beta} \left( \frac{n - i}{n - 1}\right), $$
where $F^{-1}_{Beta}$ is the inverse CDF of $Beta(\alpha, \beta)$. 

For $\alpha = \beta <1$, this specification produces a U-shaped distribution of $a_{n,i}$, concentrating mass near zero and one. Compared to the uniform distribution above, a higher fraction of nodes has $a_{n,i} \approx 0$ and thus maintains moderate link probabilities even for large $|g(n)|$, generating more informative 
quadruples. At the same time, also a higher fraction of nodes has $a_{n,i} \approx 1$ relative to the uniform case. This greater effective heterogeneity 
allows for higher magnitudes of $|g(n)|$ than under uniform spacing, for a given general level of sparsity. 
 
More specifically, for large $g(n)$ and $\alpha = \beta = 0.5$, the condition in Assumption \ref{assumption4} translates into 
$g(n) = o(n^{1/2})$, which is less restrictive than in the uniform case. As in the uniform case, all 
logarithmic growth rates satisfy 
Assumption~\ref{assumption4} asymptotically, as any power 
of $\log n$ grows slower than $n^{1/2}$. However, the 
faster growth of $np_n$ under Beta heterogeneity means that 
it becomes empirically visible at smaller sample sizes and 
to $g(n)$ specifications that grow more rapidly. The sparsity 
level accommodated by the CMLE, as measured by $p_n$, is 
ultimately the same across both specifications; what differs 
is the magnitude of $g(n)$ required to reach a given level 
of sparsity.

\paragraph{Results}

\begin{table}[htbp]
\centering\footnotesize
\setlength{\tabcolsep}{4pt}
\caption{Network Statistics --- Jochmans DGP, Right Tail ($\theta_0 = -1$)}
\label{tab:net_jochmans_right_full}
\makebox[\linewidth][c]{%
\begin{tabular}{cc cccc cc c cccc cc}
\toprule
 & & \multicolumn{6}{c}{Uniform Heterogeneity} & & \multicolumn{6}{c}{Beta Heterogeneity} \\
\cmidrule(lr){3-8} \cmidrule(lr){10-15}
 & & \multicolumn{4}{c}{Network Characteristics} & \multicolumn{2}{c}{$np_n$} & & \multicolumn{4}{c}{Network Characteristics} & \multicolumn{2}{c}{$np_n$} \\
\cmidrule(lr){3-6} \cmidrule(lr){7-8} \cmidrule(lr){10-13} \cmidrule(lr){14-15}
$n$ & $g(n)$ & \% Quad & $p_n$ & \% Links & $1-q_n$ & Theor. & Empir. & & \% Quad & $p_n$ & \% Links & $1-q_n$ & Theor. & Empir. \\
\midrule
25 & $0$ & 12.02 & 1.0000 & 56.40 & 1.0000 & 25.00 & 3.01 & & 12.02 & 1.0000 & 56.40 & 1.0000 & 25.00 & 3.01 \\
50 &  & 12.04 & 1.0000 & 56.36 & 1.0000 & 50.00 & 6.02 & & 12.04 & 1.0000 & 56.36 & 1.0000 & 50.00 & 6.02 \\
100 &  & 12.04 & 1.0000 & 56.35 & 1.0000 & 100.00 & 12.04 & & 12.04 & 1.0000 & 56.35 & 1.0000 & 100.00 & 12.04 \\
150 &  & 12.04 & 1.0000 & 56.35 & 1.0000 & 150.00 & 18.06 & & 12.04 & 1.0000 & 56.35 & 1.0000 & 150.00 & 18.06 \\
\addlinespace[2pt]
25 & $\log\log n$ & 4.90 & 0.1232 & 79.45 & 0.3509 & 3.08 & 1.23 & & 4.90 & 0.1366 & 78.97 & 0.3696 & 3.42 & 1.22 \\
50 &  & 3.95 & 0.0898 & 82.01 & 0.2996 & 4.49 & 1.97 & & 3.98 & 0.1036 & 81.39 & 0.3219 & 5.18 & 1.99 \\
100 &  & 3.29 & 0.0696 & 83.84 & 0.2638 & 6.96 & 3.29 & & 3.36 & 0.0833 & 83.08 & 0.2885 & 8.33 & 3.36 \\
150 &  & 2.99 & 0.0612 & 84.73 & 0.2474 & 9.18 & 4.48 & & 3.07 & 0.0747 & 83.91 & 0.2733 & 11.20 & 4.60 \\
\addlinespace[2pt]
25 & $\sqrt{\log n}$ & 2.37 & 0.0486 & 86.44 & 0.2206 & 1.22 & 0.59 & & 2.46 & 0.0613 & 85.57 & 0.2476 & 1.53 & 0.61 \\
50 &  & 1.93 & 0.0369 & 87.95 & 0.1922 & 1.85 & 0.96 & & 2.06 & 0.0491 & 86.92 & 0.2215 & 2.45 & 1.03 \\
100 &  & 1.60 & 0.0291 & 89.15 & 0.1705 & 2.91 & 1.60 & & 1.76 & 0.0406 & 87.99 & 0.2015 & 4.06 & 1.76 \\
150 &  & 1.44 & 0.0256 & 89.76 & 0.1601 & 3.85 & 2.16 & & 1.60 & 0.0368 & 88.56 & 0.1917 & 5.51 & 2.41 \\
\addlinespace[2pt]
25 & $\log n$ & 0.46 & 0.0089 & 94.10 & 0.0946 & 0.22 & 0.12 & & 0.61 & 0.0166 & 92.66 & 0.1287 & 0.41 & 0.15 \\
50 &  & 0.24 & 0.0043 & 95.80 & 0.0654 & 0.21 & 0.12 & & 0.39 & 0.0101 & 94.19 & 0.1003 & 0.50 & 0.19 \\
100 &  & 0.14 & 0.0023 & 96.89 & 0.0475 & 0.23 & 0.14 & & 0.27 & 0.0067 & 95.23 & 0.0816 & 0.67 & 0.27 \\
150 &  & 0.10 & 0.0016 & 97.36 & 0.0401 & 0.24 & 0.15 & & 0.22 & 0.0054 & 95.70 & 0.0736 & 0.81 & 0.33 \\
\addlinespace[2pt]
25 & $2\log n$ & 0.03 & 0.0008 & 98.31 & 0.0288 & 0.02 & 0.01 & & 0.09 & 0.0036 & 96.82 & 0.0596 & 0.09 & 0.02 \\
50 &  & 0.02 & 0.0003 & 98.87 & 0.0184 & 0.02 & 0.01 & & 0.07 & 0.0022 & 97.45 & 0.0464 & 0.11 & 0.03 \\
100 &  & 0.01 & 0.0002 & 99.19 & 0.0127 & 0.02 & 0.01 & & 0.05 & 0.0014 & 97.85 & 0.0380 & 0.14 & 0.05 \\
150 &  & 0.01 & 0.0001 & 99.32 & 0.0105 & 0.02 & 0.01 & & 0.04 & 0.0012 & 98.04 & 0.0344 & 0.18 & 0.06 \\
\addlinespace[2pt]
25 & $(\log n)^{3/2}$ & 0.05 & 0.0012 & 97.91 & 0.0348 & 0.03 & 0.01 & & 0.12 & 0.0045 & 96.40 & 0.0668 & 0.11 & 0.03 \\
50 &  & 0.02 & 0.0004 & 98.84 & 0.0187 & 0.02 & 0.01 & & 0.07 & 0.0022 & 97.42 & 0.0469 & 0.11 & 0.03 \\
100 &  & 0.01 & 0.0001 & 99.29 & 0.0111 & 0.01 & 0.01 & & 0.04 & 0.0012 & 98.01 & 0.0353 & 0.12 & 0.04 \\
150 &  & 0.00 & 0.0001 & 99.46 & 0.0085 & 0.01 & 0.01 & & 0.03 & 0.0009 & 98.26 & 0.0306 & 0.14 & 0.05 \\
\addlinespace[2pt]
25 & $(\log n)^2$ & 0.00 & 0.0002 & 99.31 & 0.0130 & 0.00 & 0.00 & & 0.03 & 0.0014 & 98.11 & 0.0372 & 0.03 & 0.01 \\
50 &  & 0.00 & 0.0000 & 99.68 & 0.0056 & 0.00 & 0.00 & & 0.01 & 0.0006 & 98.74 & 0.0237 & 0.03 & 0.01 \\
100 &  & 0.00 & 0.0000 & 99.84 & 0.0027 & 0.00 & 0.00 & & 0.01 & 0.0003 & 99.10 & 0.0163 & 0.03 & 0.01 \\
150 &  & 0.00 & 0.0000 & 99.89 & 0.0018 & 0.00 & 0.00 & & 0.01 & 0.0002 & 99.25 & 0.0135 & 0.03 & 0.01 \\
\bottomrule
\end{tabular}
}
\par\smallskip
\parbox{\linewidth}{\footnotesize\textit{Notes:} Based on 1000 Monte Carlo replications. \% Quad is the share of informative quadruples; $p_n$ and $1-q_n$ are the theoretical non-link and quadruple density parameters; $np_n$ Theor.\ uses the population $p_n$, Empir.\ uses the simulated share. The dense case ($g(n)=0$) is identical across heterogeneity specifications.}
\end{table}

\begin{table}[htbp]
\centering\footnotesize
\setlength{\tabcolsep}{4pt}
\caption{Estimation Results --- Jochmans DGP, Right Tail ($\theta_0 = -1$), Uniform Heterogeneity}
\label{tab:est_jochmans_right_uniform_full}
\makebox[\linewidth][c]{%
\begin{tabular}{cc ccc ccc ccc ccc}
\toprule
 & & \multicolumn{3}{c}{Mean Bias} & \multicolumn{3}{c}{Median Bias} & \multicolumn{3}{c}{Size} & \multicolumn{3}{c}{se/sd} \\
\cmidrule(lr){3-5} \cmidrule(lr){6-8} \cmidrule(lr){9-11} \cmidrule(lr){12-14}
$n$ & $g(n)$ & MLE & BC & CMLE & MLE & BC & CMLE & MLE & BC & CMLE & MLE & BC & CMLE \\
\midrule
25 & $0$ & $-$0.098 & $-$0.004 & $-$0.014 & $-$0.103 & $-$0.010 & $-$0.016 & 0.063 & 0.037 & 0.032 & 0.93 & 1.02 & 1.05 \\
50 &  & $-$0.035 & 0.008 & 0.006 & $-$0.035 & 0.007 & 0.006 & 0.051 & 0.044 & 0.038 & 0.99 & 1.04 & 1.06 \\
100 &  & $-$0.019 & 0.002 & 0.001 & $-$0.019 & 0.002 & 0.001 & 0.060 & 0.057 & 0.053 & 0.97 & 0.99 & 1.00 \\
150 &  & $-$0.010 & 0.003 & 0.003 & $-$0.011 & 0.003 & 0.003 & 0.049 & 0.045 & 0.043 & 1.01 & 1.02 & 1.03 \\
\addlinespace[2pt]
25 & $\log\log n$ & $-$0.070 & 0.033 & 0.018 & $-$0.046 & 0.052 & 0.034 & 0.091 & 0.050 & 0.044 & 0.88 & 0.98 & 1.01 \\
50 &  & $-$0.020 & 0.027 & 0.024 & $-$0.005 & 0.041 & 0.042 & 0.059 & 0.039 & 0.036 & 0.97 & 1.02 & 1.05 \\
100 &  & $-$0.025 & $-$0.002 & $-$0.003 & $-$0.023 & 0.000 & $-$0.000 & 0.052 & 0.042 & 0.040 & 0.97 & 1.00 & 1.01 \\
150 &  & $-$0.009 & 0.006 & 0.006 & $-$0.006 & 0.009 & 0.006 & 0.047 & 0.042 & 0.039 & 1.01 & 1.03 & 1.04 \\
\addlinespace[2pt]
25 & $\sqrt{\log n}$ & $-$0.057 & 0.061 & 0.048 & $-$0.073 & 0.045 & 0.048 & 0.073 & 0.044 & 0.041 & 0.89 & 1.01 & 1.03 \\
50 &  & $-$0.038 & 0.016 & 0.012 & $-$0.041 & 0.013 & $-$0.002 & 0.045 & 0.039 & 0.035 & 0.98 & 1.03 & 1.06 \\
100 &  & $-$0.028 & $-$0.002 & $-$0.003 & $-$0.014 & 0.011 & 0.012 & 0.053 & 0.050 & 0.045 & 0.99 & 1.01 & 1.03 \\
150 &  & $-$0.009 & 0.008 & 0.008 & $-$0.009 & 0.008 & 0.006 & 0.050 & 0.048 & 0.046 & 0.98 & 0.99 & 1.00 \\
\addlinespace[2pt]
25 & $\log n$ & $-$0.091 & 0.105 & 0.061 & $-$0.102 & 0.088 & 0.057 & 0.066 & 0.026 & 0.025 & 0.87 & 1.04 & 1.05 \\
50 &  & $-$0.082 & 0.011 & $-$0.004 & $-$0.100 & $-$0.007 & $-$0.001 & 0.064 & 0.042 & 0.030 & 0.92 & 1.00 & 1.04 \\
100 &  & $-$0.056 & $-$0.009 & $-$0.013 & $-$0.059 & $-$0.011 & $-$0.016 & 0.058 & 0.049 & 0.042 & 0.97 & 1.01 & 1.03 \\
150 &  & $-$0.024 & 0.008 & 0.008 & $-$0.029 & 0.002 & 0.010 & 0.057 & 0.044 & 0.043 & 0.97 & 1.00 & 1.01 \\
\addlinespace[2pt]
25 & $2\log n$ & --- & --- & --- & --- & --- & --- & --- & --- & --- & --- & --- & --- \\
50 &  & $-$0.183 & 0.070 & 0.037 & $-$0.280 & $-$0.053 & $-$0.137 & 0.062 & 0.011 & 0.023 & 0.73 & 0.93 & 0.90 \\
100 &  & $-$0.093 & 0.015 & $-$0.003 & $-$0.046 & 0.058 & 0.042 & 0.060 & 0.032 & 0.031 & 0.93 & 1.03 & 1.05 \\
150 &  & $-$0.054 & 0.016 & 0.010 & $-$0.076 & $-$0.006 & 0.008 & 0.054 & 0.043 & 0.032 & 0.98 & 1.05 & 1.07 \\
\addlinespace[2pt]
25 & $(\log n)^{3/2}$ & --- & --- & --- & --- & --- & --- & --- & --- & --- & --- & --- & --- \\
50 &  & $-$0.152 & 0.090 & 0.059 & $-$0.254 & $-$0.019 & $-$0.109 & 0.067 & 0.014 & 0.024 & 0.73 & 0.92 & 0.90 \\
100 &  & $-$0.114 & 0.006 & $-$0.017 & $-$0.097 & 0.025 & $-$0.009 & 0.060 & 0.036 & 0.028 & 0.94 & 1.06 & 1.09 \\
150 &  & $-$0.044 & 0.036 & 0.029 & $-$0.072 & 0.007 & 0.011 & 0.054 & 0.040 & 0.038 & 0.95 & 1.03 & 1.05 \\
\addlinespace[2pt]
25 & $(\log n)^2$ & --- & --- & --- & --- & --- & --- & --- & --- & --- & --- & --- & --- \\
50 &  & --- & --- & --- & --- & --- & --- & --- & --- & --- & --- & --- & --- \\
100 &  & 0.887 & 1.015 & 0.950$^\dagger$ & $-$0.442 & $-$0.040 & $-$0.153 & 0.058 & 0.004 & 0.015 & 1.41 & 1.96 & 0.22 \\
150 &  & $-$0.219 & 0.057 & $-$0.002 & $-$0.220 & 0.021 & $-$0.077 & 0.078 & 0.016 & 0.026 & 0.82 & 1.04 & 1.03 \\
\bottomrule
\end{tabular}
}
\par\smallskip
\parbox{\linewidth}{\footnotesize\textit{Notes:} True parameter $\theta_0 = -1$; 1000 replications. MLE = maximum likelihood; BC = bias-corrected \citep{chernozhukov2020network}; CMLE = conditional maximum likelihood (pairwise differencing). Size = rejection frequency of the two-sided $t$-test at the 5\% nominal level. se/sd = ratio of average estimated standard error to simulation standard deviation. Entries marked~--- indicate that a non-negligible share of replications failed to converge due to insufficient variation in the binary outcomes, rendering 
the row unreliable.}
\end{table}

\begin{table}[htbp]
\centering\footnotesize
\setlength{\tabcolsep}{4pt}
\caption{Estimation Results --- Jochmans DGP, Right Tail ($\theta_0 = -1$), Beta Heterogeneity}
\label{tab:est_jochmans_right_beta_full}
\makebox[\linewidth][c]{%
\begin{tabular}{cc ccc ccc ccc ccc}
\toprule
 & & \multicolumn{3}{c}{Mean Bias} & \multicolumn{3}{c}{Median Bias} & \multicolumn{3}{c}{Size} & \multicolumn{3}{c}{se/sd} \\
\cmidrule(lr){3-5} \cmidrule(lr){6-8} \cmidrule(lr){9-11} \cmidrule(lr){12-14}
$n$ & $g(n)$ & MLE & BC & CMLE & MLE & BC & CMLE & MLE & BC & CMLE & MLE & BC & CMLE \\
\midrule
25 & $0$ & $-$0.098 & $-$0.004 & $-$0.014 & $-$0.103 & $-$0.010 & $-$0.016 & 0.063 & 0.037 & 0.032 & 0.93 & 1.02 & 1.05 \\
50 &  & $-$0.035 & 0.008 & 0.006 & $-$0.035 & 0.007 & 0.006 & 0.051 & 0.044 & 0.038 & 0.99 & 1.04 & 1.06 \\
100 &  & $-$0.019 & 0.002 & 0.001 & $-$0.019 & 0.002 & 0.001 & 0.060 & 0.057 & 0.053 & 0.97 & 0.99 & 1.00 \\
150 &  & $-$0.010 & 0.003 & 0.003 & $-$0.011 & 0.003 & 0.003 & 0.049 & 0.045 & 0.043 & 1.01 & 1.02 & 1.03 \\
\addlinespace[2pt]
25 & $\log\log n$ & $-$0.071 & 0.033 & 0.017 & $-$0.048 & 0.055 & 0.019 & 0.078 & 0.046 & 0.041 & 0.89 & 0.98 & 1.01 \\
50 &  & $-$0.022 & 0.025 & 0.022 & $-$0.002 & 0.045 & 0.040 & 0.051 & 0.036 & 0.035 & 0.99 & 1.03 & 1.06 \\
100 &  & $-$0.023 & 0.000 & $-$0.000 & $-$0.024 & $-$0.000 & $-$0.000 & 0.063 & 0.050 & 0.049 & 0.97 & 1.00 & 1.01 \\
150 &  & $-$0.009 & 0.007 & 0.006 & $-$0.007 & 0.009 & 0.008 & 0.048 & 0.044 & 0.042 & 1.01 & 1.03 & 1.04 \\
\addlinespace[2pt]
25 & $\sqrt{\log n}$ & $-$0.080 & 0.043 & 0.024 & $-$0.081 & 0.040 & 0.034 & 0.061 & 0.035 & 0.033 & 0.92 & 1.03 & 1.06 \\
50 &  & $-$0.033 & 0.022 & 0.018 & $-$0.044 & 0.010 & 0.010 & 0.052 & 0.046 & 0.036 & 0.98 & 1.03 & 1.07 \\
100 &  & $-$0.025 & 0.002 & 0.001 & $-$0.013 & 0.013 & 0.010 & 0.052 & 0.044 & 0.042 & 0.99 & 1.02 & 1.04 \\
150 &  & $-$0.012 & 0.006 & 0.006 & $-$0.009 & 0.008 & 0.009 & 0.060 & 0.055 & 0.054 & 0.97 & 0.99 & 1.00 \\
\addlinespace[2pt]
25 & $\log n$ & $-$0.083 & 0.101 & 0.069 & $-$0.073 & 0.093 & 0.046 & 0.062 & 0.021 & 0.021 & 0.88 & 1.06 & 1.05 \\
50 &  & $-$0.083 & 0.007 & $-$0.003 & $-$0.096 & $-$0.007 & $-$0.014 & 0.056 & 0.028 & 0.027 & 0.95 & 1.04 & 1.07 \\
100 &  & $-$0.056 & $-$0.010 & $-$0.012 & $-$0.051 & $-$0.005 & $-$0.001 & 0.065 & 0.048 & 0.042 & 0.96 & 1.00 & 1.02 \\
150 &  & $-$0.029 & 0.001 & 0.001 & $-$0.026 & 0.005 & 0.010 & 0.063 & 0.056 & 0.047 & 0.97 & 1.00 & 1.02 \\
\addlinespace[2pt]
25$^\dagger$ & $2\log n$ & $-$0.039 & 0.397 & 0.286 & $-$0.179 & 0.148 & 0.109 & 0.069 & 0.002 & 0.016 & 0.60 & 0.99 & 0.85 \\
50 &  & $-$0.095 & 0.060 & 0.044 & $-$0.147 & 0.014 & 0.008 & 0.070 & 0.027 & 0.029 & 0.85 & 0.99 & 1.01 \\
100 &  & $-$0.071 & 0.002 & $-$0.004 & $-$0.050 & 0.020 & 0.013 & 0.059 & 0.043 & 0.036 & 0.93 & 1.00 & 1.02 \\
150 &  & $-$0.042 & 0.006 & 0.005 & $-$0.047 & 0.000 & 0.003 & 0.049 & 0.036 & 0.035 & 1.00 & 1.05 & 1.07 \\
\addlinespace[2pt]
25$^\dagger$ & $(\log n)^{3/2}$ & $-$0.003 & 0.353 & 0.247 & $-$0.151 & 0.153 & 0.102 & 0.066 & 0.006 & 0.018 & 0.68 & 0.94 & 0.88 \\
50 &  & $-$0.096 & 0.058 & 0.041 & $-$0.130 & 0.021 & $-$0.004 & 0.073 & 0.029 & 0.031 & 0.85 & 0.99 & 1.01 \\
100 &  & $-$0.076 & 0.001 & $-$0.008 & $-$0.060 & 0.016 & 0.014 & 0.057 & 0.044 & 0.037 & 0.93 & 1.00 & 1.03 \\
150 &  & $-$0.047 & 0.005 & 0.003 & $-$0.054 & $-$0.004 & $-$0.016 & 0.039 & 0.030 & 0.028 & 0.99 & 1.04 & 1.07 \\
\addlinespace[2pt]
25 & $(\log n)^2$ & --- & --- & --- & --- & --- & --- & --- & --- & --- & --- & --- & --- \\
50 &  & $-$0.182 & 0.102 & 0.047 & $-$0.272 & 0.017 & $-$0.088 & 0.074 & 0.012 & 0.019 & 0.49 & 0.62 & 0.59 \\
100 &  & $-$0.120 & 0.009 & $-$0.019 & $-$0.126 & 0.001 & $-$0.035 & 0.066 & 0.034 & 0.037 & 0.90 & 1.02 & 1.05 \\
150 &  & $-$0.026 & 0.055 & 0.046 & $-$0.038 & 0.045 & 0.027 & 0.050 & 0.034 & 0.028 & 0.99 & 1.07 & 1.10 \\
\bottomrule
\end{tabular}
}
\par\smallskip
\parbox{\linewidth}{\footnotesize\textit{Notes:} True parameter $\theta_0 = -1$; 1000 replications. MLE = maximum likelihood; BC = bias-corrected \citep{chernozhukov2020network}; CMLE = conditional maximum likelihood (pairwise differencing). Size = rejection frequency of the two-sided $t$-test at the 5\% nominal level. se/sd = ratio of average estimated standard error to simulation standard deviation. $^\dagger$Results in this row are based on 999 out of 1000 replications, where all estimators produced finite 
estimates. Entries marked~--- indicate that a non-negligible share of replications failed to converge due to insufficient variation in the binary outcomes, rendering 
the row unreliable.}
\end{table}

Table~\ref{tab:net_jochmans_right_full} reports network statistics. The DGP predictions are reflected in finite-sample 
behavior: under uniform heterogeneity, growth in $np_n$ 
is visible through $g(n) = \sqrt{\log n}$, becomes modest at 
$g(n) = \log n$ (from $0.12$ to $0.15$), and is effectively 
undetectable beyond. Under Beta heterogeneity, growth extends 
through $g(n) = (\log n)^{3/2}$ (from $0.03$ to $0.05$), with 
stagnation only at $g(n) = (\log n)^2$. Although Assumption~\ref{assumption4} is satisfied in all 
cases asymptotically, the simulations illustrate that the 
practical feasibility of estimation depends on whether $np_n$ 
grows with $n$ at the sample sizes available.

Notice that as $g(n)$ increases, the theoretical and empirical 
$np_n$ move closer together. The theoretical $p_n$ in the 
table is computed using the tail approximation derived in 
\ref{appendix_sparsity}, which is valid when fixed 
effects tend to $\pm\infty$. In the dense case ($g(n) = 0$), 
fixed effects are at zero and this approximation is outside 
its regime of validity, producing $p_n = 1$ while only about 
$12\%$ of quadruples are empirically informative. As $g(n)$ 
grows and the fixed effects move into the regime where the 
approximation applies, the theoretical and empirical measures 
align, both shrinking toward zero. This is also the regime 
where the rate restriction $\sqrt{n}\,(1 - q_n) \to \infty$ becomes 
operative, with $1 - q_n$ visibly shrinking under both 
heterogeneity patterns as $g(n)$ grows.

Tables \ref{tab:est_jochmans_right_uniform_full} and \ref{tab:est_jochmans_right_beta_full} report estimation results across the same sparsity configurations, for uniform heterogeneity and for Beta. The MLE exhibits bias away from zero across most configurations, as expected with two-way fixed effects. While the bias decreases with
$n$, it remains non-negligible at $n=150$ in sparser configurations. In the two most extreme sparsity rows at $n=25$ under Beta heterogeneity, the mean bias is near zero but the median bias remains substantial ($-$0.179 and $-$0.151). The MLE's rejection frequencies do not account for this bias, since the estimator is centered at the wrong value.

The bias correction removes most of the 
incidental-parameter bias at moderate sparsity levels but 
deteriorates in sparse networks. At $g(n) = \log n$ and 
$n = 25$, the BC mean bias is $0.105$ under uniform 
heterogeneity and $0.101$ under Beta, comparable in 
magnitude to the MLE bias it is designed to remove. Under 
Beta heterogeneity at $g(n) = 2\log n$ and 
$(\log n)^{3/2}$ with $n = 25$, the BC mean bias increases 
further to $0.397$ and $0.353$, while the CMLE remains at 
$0.286$ and $0.247$ with smaller median bias throughout. 
The gap between BC and CMLE widens as sparsity increases, 
and BC's size becomes extremely conservative (rejection 
rates near zero) at the most extreme sparsity levels.

The CMLE performs well across configurations where $np_n$
continues to grow with $n$. Under uniform heterogeneity, it
delivers reliable inference through $g(n) = \log n$ (over
$97\%$ of pairs linked); under Beta heterogeneity, through
$g(n) = (\log n)^{3/2}$ (over $98\%$ of pairs linked). The
boundary of reliable CMLE performance aligns with the range
where $np_n$ grows visibly with $n$, established above.
Beyond this range, the CMLE remains reasonable at larger $n$
even where $np_n$ growth is slow. Throughout, the CMLE is
slightly conservative: rejection frequencies at or below
nominal and se/sd ratios above one, indicating mild
overestimation of variability. Overall, the two estimators perform comparably at moderate sparsity levels, but BC deteriorates more than the CMLE at the sparse extremes, particularly at small sample sizes.

\clearpage
\subsection{Additional Tables and Figures for Subsection \ref{monte_carlo_DR}}

\begin{figure}[htbp]
\centering
\captionsetup{skip=0pt}
\captionsetup[subfigure]{font=footnotesize, skip=0pt}
    \includegraphics[width=0.45\textwidth]{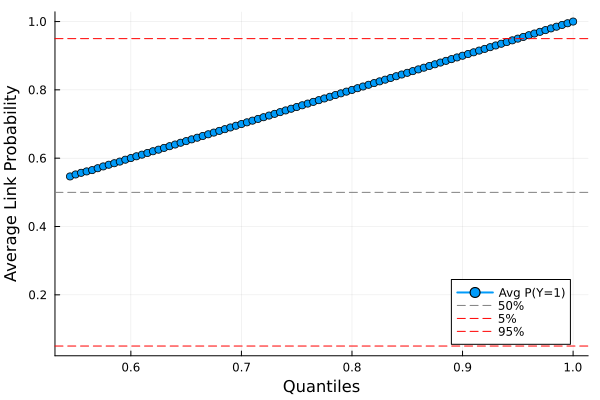}
    \caption{Average probability 
$\Pr(\tilde{y}_{ij,k} = 1)$ at each threshold $y_k$ in the empirically calibrated DGP. 
Dashed lines indicate reference levels at 5\%, 50\%, and 
95\%. Thresholds correspond to empirical
quantiles in [0.545, 0.990] at intervals of 0.005.}
\label{fig:sparsity}
\end{figure}

\begin{figure}[htbp]
\centering
\captionsetup{skip=0pt}
\captionsetup[subfigure]{font=footnotesize, skip=0pt}
\begin{subfigure}[b]{0.32\textwidth}
    \centering
    \includegraphics[width=\textwidth]{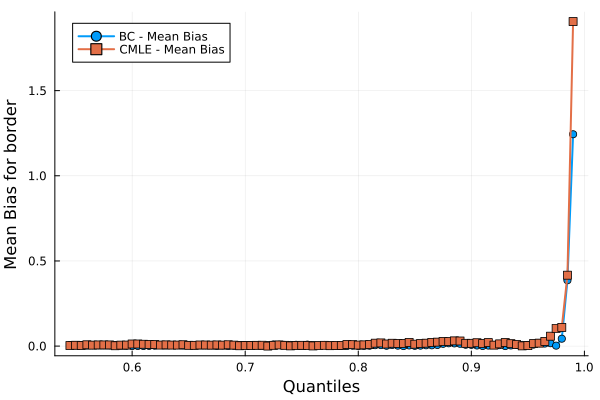}
    \caption{Border}
    \label{fig:sim_bias_border}
\end{subfigure}
\hfill
\begin{subfigure}[b]{0.32\textwidth}
    \centering
    \includegraphics[width=\textwidth]{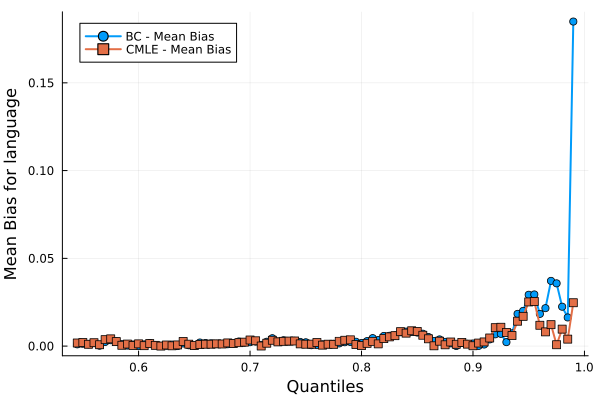}
    \caption{Language}
    \label{fig:sim_bias_language}
\end{subfigure}
\hfill
\begin{subfigure}[b]{0.32\textwidth}
    \centering
    \includegraphics[width=\textwidth]{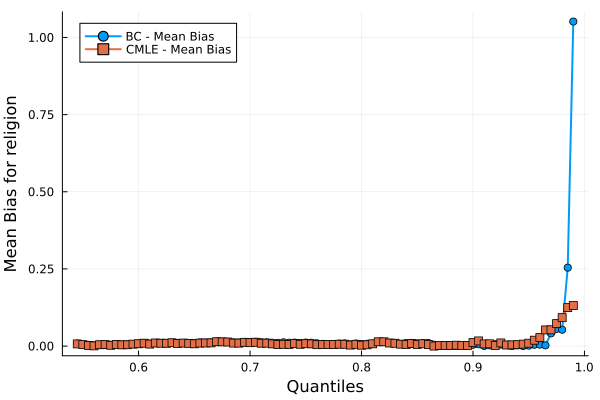}
    \caption{Religion}
    \label{fig:sim_bias_religion}
\end{subfigure}
\caption{Mean bias of 
the bias-corrected estimator (BC) and the conditional 
maximum likelihood estimator (CMLE) across quantiles of 
the trade distribution, based on 500 replications using 
the empirically calibrated DGP. Thresholds correspond to empirical
quantiles in [0.545, 0.990] at intervals of 0.005.}
\label{fig:sim_bias_appendix}
\end{figure}

\begin{figure}[htbp]
\centering
\captionsetup{skip=0pt}
\captionsetup[subfigure]{font=footnotesize, skip=0pt}
\begin{subfigure}[b]{0.32\textwidth}
    \centering
    \includegraphics[width=\textwidth]{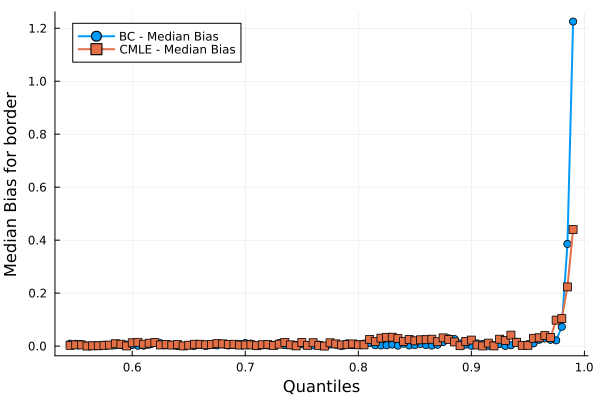}
    \caption{Border}
    \label{fig:sim_medbias_border}
\end{subfigure}
\hfill
\begin{subfigure}[b]{0.32\textwidth}
    \centering
    \includegraphics[width=\textwidth]{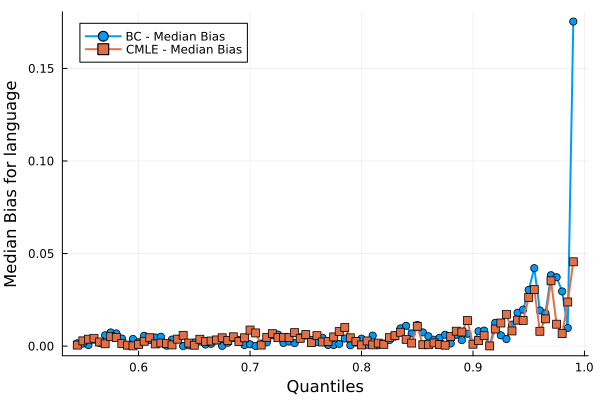}
    \caption{Language}
    \label{fig:sim_medbias_language}
\end{subfigure}
\hfill
\begin{subfigure}[b]{0.32\textwidth}
    \centering
    \includegraphics[width=\textwidth]{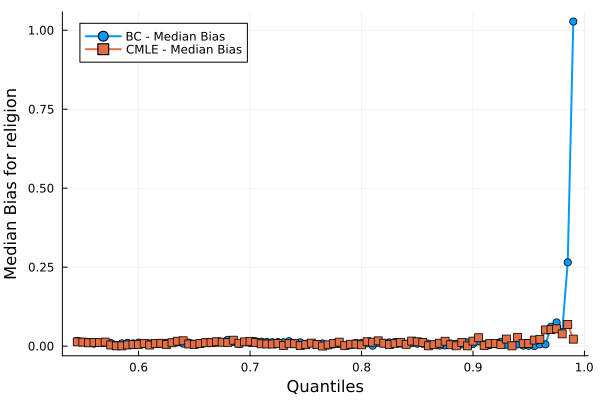}
    \caption{Religion}
    \label{fig:sim_medbias_religion}
\end{subfigure}
\caption{Median bias of 
the bias-corrected estimator (BC) and the conditional 
maximum likelihood estimator (CMLE) across quantiles of 
the trade distribution, based on 500 replications using 
the empirically calibrated DGP. Thresholds correspond to empirical
quantiles in [0.545, 0.990] at intervals of 0.005.}
\label{fig:sim_medbias_appendix}
\end{figure}

\begin{figure}[htbp]
\centering
\captionsetup{skip=0pt}
\captionsetup[subfigure]{font=footnotesize, skip=0pt}
\begin{subfigure}[b]{0.32\textwidth}
    \centering
    \includegraphics[width=\textwidth]{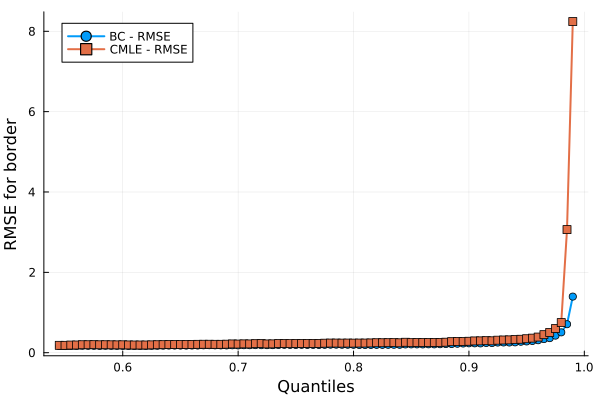}
    \caption{Border}
    \label{fig:sim_rmse_border}
\end{subfigure}
\hfill
\begin{subfigure}[b]{0.32\textwidth}
    \centering
    \includegraphics[width=\textwidth]{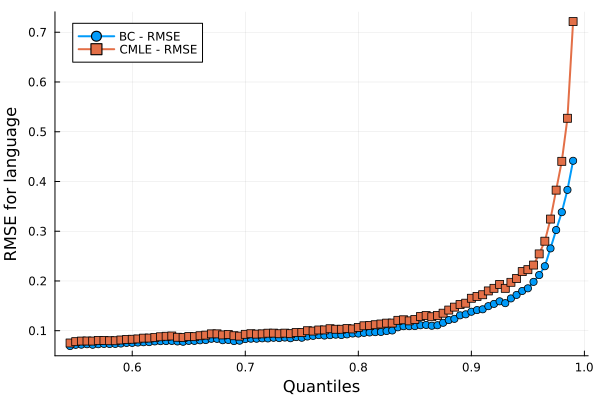}
    \caption{Language}
    \label{fig:sim_rmse_language}
\end{subfigure}
\hfill
\begin{subfigure}[b]{0.32\textwidth}
    \centering
    \includegraphics[width=\textwidth]{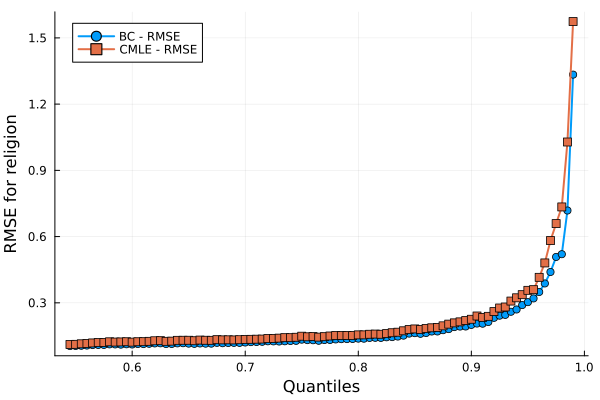}
    \caption{Religion}
    \label{fig:sim_rmse_religion}
\end{subfigure}
\caption{RMSE for 
the bias-corrected estimator (BC) and the conditional 
maximum likelihood estimator (CMLE) across quantiles of 
the trade distribution, based on 500 replications using 
the empirically calibrated DGP. Thresholds correspond to empirical
quantiles in [0.545, 0.990] at intervals of 0.005.}
\label{fig:sim_rmse_appendix}
\end{figure}

\begin{figure}[htbp]
\centering
\captionsetup{skip=0pt}
\captionsetup[subfigure]{font=footnotesize, skip=0pt}
\begin{subfigure}[b]{0.32\textwidth}
    \centering
    \includegraphics[width=\textwidth]{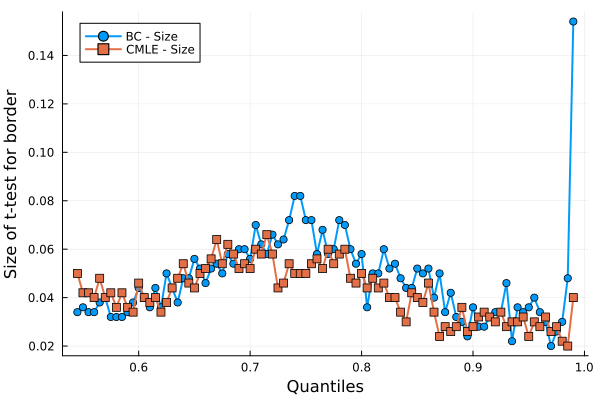}
    \caption{Border}
    \label{fig:sim_size_border}
\end{subfigure}
\hfill
\begin{subfigure}[b]{0.32\textwidth}
    \centering
    \includegraphics[width=\textwidth]{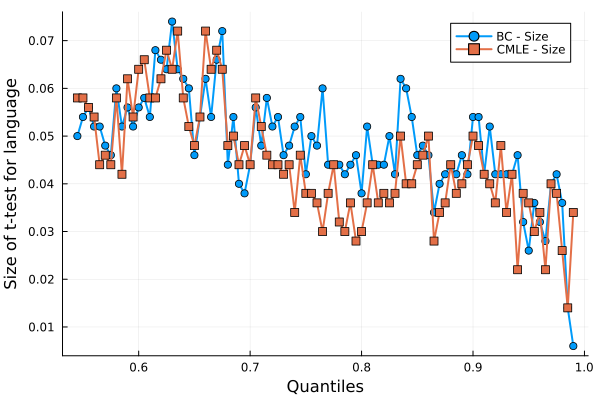}
    \caption{Language}
    \label{fig:sim_size_language}
\end{subfigure}
\hfill
\begin{subfigure}[b]{0.32\textwidth}
    \centering
    \includegraphics[width=\textwidth]{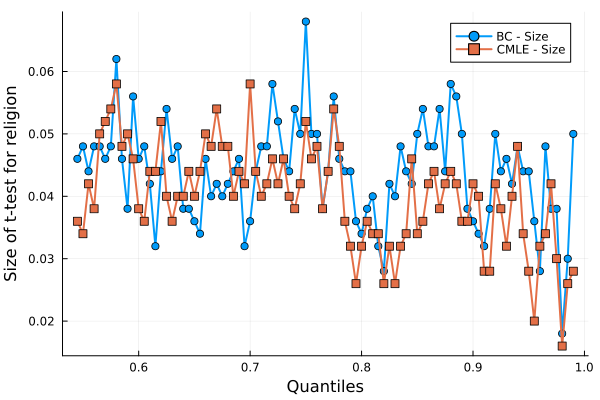}
    \caption{Religion}
    \label{fig:sim_size_religion}
\end{subfigure}
\caption{Rejection 
frequency of the two-sided $t$-test at the 5\% nominal 
level for the bias-corrected estimator (BC) and the 
conditional maximum likelihood estimator (CMLE) across 
quantiles of the trade distribution, based on 500 
replications using the empirically calibrated DGP. Thresholds correspond to empirical
quantiles in [0.545, 0.990] at intervals of 0.005.}
\label{fig:sim_size_appendix}
\end{figure}

\clearpage
\subsection{Additional Tables and Figures for Subsection \ref{monte_carlo_simultaneous}}
\label{app_monte_carlo_simultaneous}

\begin{table}[htbp]
\centering
\captionsetup{skip=0pt}
\captionsetup[subfigure]{font=footnotesize, skip=0pt}
\caption{Coverage Comparison: Pointwise vs Sup-$t$ Bands ($\tau_{\max} = 0.99$, Equally Spaced)}
\label{tab:main_coverage_099}
\begin{threeparttable}
\footnotesize
\begin{tabular}{ll cc cc cc cc cc}
\toprule
& & \multicolumn{2}{c}{Distance} & \multicolumn{2}{c}{Legal} & \multicolumn{2}{c}{Border} & \multicolumn{2}{c}{Language} & \multicolumn{2}{c}{Religion} \\
\cmidrule(lr){3-4} \cmidrule(lr){5-6} \cmidrule(lr){7-8} \cmidrule(lr){9-10} \cmidrule(lr){11-12}
$K$ & DGP & PW & Sup-$t$ & PW & Sup-$t$ & PW & Sup-$t$ & PW & Sup-$t$ & PW & Sup-$t$ \\
\midrule
5 & Varying & 82.80 & 97.60 & 85.60 & 97.00 & 83.60 & 93.20 & 80.80 & 95.00 & 83.60 & 96.40 \\
5 & Constant & 85.80 & 97.00 & 84.80 & 96.80 & 82.60 & 96.40 & 82.00 & 95.80 & 84.80 & 96.00 \\
\addlinespace
15 & Varying & 66.40 & 96.20 & 69.60 & 96.80 & 71.20 & 92.60 & 64.80 & 96.00 & 69.20 & 95.40 \\
15 & Constant & 70.60 & 96.60 & 70.40 & 96.40 & 69.00 & 96.40 & 65.60 & 95.80 & 67.80 & 96.60 \\
\addlinespace
50 & Varying & 49.60 & 96.20 & 50.40 & 97.40 & 56.00 & 93.00 & 46.20 & 95.60 & 51.40 & 96.80 \\
50 & Constant & 52.20 & 96.00 & 51.80 & 97.20 & 54.80 & 95.20 & 47.40 & 94.80 & 53.20 & 97.40 \\
\bottomrule
\end{tabular}
\begin{tablenotes}
\smallskip\footnotesize
\item \textit{Notes:} Empirical coverage rates (in \%) of 95\% simultaneous confidence bands. ``PW'' (Pointwise) uses $z_{0.975} = 1.96$ at each threshold; ``Sup-$t$'' uses simulated critical values. Target coverage is 95\%. ``Constant'' = coefficients identical across thresholds ($H_0$ true); ``Varying'' = coefficients follow empirical trade data pattern ($H_0$ false). Based on 500 Monte Carlo simulations.
\end{tablenotes}
\end{threeparttable}
\end{table}

\begin{table}[htbp]
\centering
\captionsetup{skip=0pt}
\captionsetup[subfigure]{font=footnotesize, skip=0pt}
\caption{Coverage Comparison: Pointwise vs Sup-$t$ Bands ($\tau_{\max} = 0.95$, First $K$ and Last $K$)}
\label{tab:appendix_coverage_095}
\begin{threeparttable}
\footnotesize
\begin{tabular}{lll cc cc cc cc cc}
\toprule
& & & \multicolumn{2}{c}{Distance} & \multicolumn{2}{c}{Legal} & \multicolumn{2}{c}{Border} & \multicolumn{2}{c}{Language} & \multicolumn{2}{c}{Religion} \\
\cmidrule(lr){4-5} \cmidrule(lr){6-7} \cmidrule(lr){8-9} \cmidrule(lr){10-11} \cmidrule(lr){12-13}
$K$ & Selection & DGP & PW & Sup-$t$ & PW & Sup-$t$ & PW & Sup-$t$ & PW & Sup-$t$ & PW & Sup-$t$ \\
\midrule
5 & First $K$ & Varying & 90.80 & 95.20 & 92.20 & 96.20 & 91.80 & 96.00 & 90.80 & 94.60 & 92.60 & 96.60 \\
5 & First $K$ & Constant & 91.60 & 95.20 & 92.00 & 96.00 & 91.20 & 96.20 & 90.60 & 95.60 & 93.20 & 95.80 \\
5 & Last $K$ & Varying & 90.60 & 95.60 & 91.00 & 96.60 & 90.80 & 97.40 & 89.80 & 97.20 & 90.00 & 97.00 \\
5 & Last $K$ & Constant & 91.60 & 95.80 & 91.60 & 97.40 & 91.60 & 97.80 & 90.20 & 97.40 & 90.40 & 97.60 \\
\addlinespace
15 & First $K$ & Varying & 85.40 & 95.40 & 87.00 & 95.20 & 87.00 & 96.60 & 82.80 & 93.80 & 86.40 & 95.80 \\
15 & First $K$ & Constant & 86.60 & 94.80 & 87.80 & 96.20 & 86.80 & 97.20 & 86.60 & 95.40 & 88.60 & 95.60 \\
15 & Last $K$ & Varying & 80.60 & 97.40 & 81.40 & 97.60 & 82.80 & 97.40 & 80.20 & 95.60 & 79.60 & 97.60 \\
15 & Last $K$ & Constant & 80.60 & 96.60 & 80.60 & 97.20 & 82.00 & 98.40 & 78.40 & 96.40 & 77.40 & 96.80 \\
\addlinespace
50 & First $K$ & Varying & 66.60 & 95.60 & 69.20 & 96.60 & 75.80 & 94.40 & 67.00 & 94.60 & 69.20 & 96.20 \\
50 & First $K$ & Constant & 69.40 & 95.40 & 73.80 & 97.00 & 73.60 & 95.20 & 70.40 & 95.00 & 73.60 & 97.40 \\
50 & Last $K$ & Varying & 61.80 & 96.00 & 61.40 & 96.80 & 64.40 & 95.60 & 59.80 & 96.80 & 63.40 & 96.20 \\
50 & Last $K$ & Constant & 60.00 & 96.00 & 59.20 & 95.60 & 62.80 & 96.60 & 59.00 & 97.40 & 64.20 & 97.20 \\
\bottomrule
\end{tabular}
\begin{tablenotes}
\smallskip\footnotesize
\item \textit{Notes:} Empirical coverage rates (in \%) of 95\% simultaneous confidence bands. ``PW'' (Pointwise) uses $z_{0.975} = 1.96$ at each threshold; ``Sup-$t$'' uses simulated critical values. Target coverage is 95\%. ``Constant'' = coefficients identical across thresholds ($H_0$ true); ``Varying'' = coefficients follow empirical trade data pattern ($H_0$ false). First $K$ selects the lowest $K$ thresholds; Last $K$ selects the highest $K$ thresholds up to $\tau_{\max}$. Based on 500 Monte Carlo simulations.
\end{tablenotes}
\end{threeparttable}
\end{table}

\begin{table}[htbp]
\centering
\captionsetup{skip=0pt}
\captionsetup[subfigure]{font=footnotesize, skip=0pt}
\caption{Coverage Comparison: Pointwise vs Sup-$t$ Bands ($\tau_{\max} = 0.99$, First $K$ and Last $K$)}
\label{tab:appendix_coverage_099}
\begin{threeparttable}
\footnotesize
\begin{tabular}{lll cc cc cc cc cc}
\toprule
& & & \multicolumn{2}{c}{Distance} & \multicolumn{2}{c}{Legal} & \multicolumn{2}{c}{Border} & \multicolumn{2}{c}{Language} & \multicolumn{2}{c}{Religion} \\
\cmidrule(lr){4-5} \cmidrule(lr){6-7} \cmidrule(lr){8-9} \cmidrule(lr){10-11} \cmidrule(lr){12-13}
$K$ & Selection & DGP & PW & Sup-$t$ & PW & Sup-$t$ & PW & Sup-$t$ & PW & Sup-$t$ & PW & Sup-$t$ \\
\midrule
5 & First $K$ & Varying & 90.80 & 95.20 & 92.20 & 96.20 & 91.80 & 96.00 & 90.80 & 94.60 & 92.60 & 96.60 \\
5 & First $K$ & Constant & 91.60 & 95.20 & 92.00 & 96.00 & 91.20 & 96.20 & 90.60 & 95.60 & 93.20 & 95.80 \\
5 & Last $K$ & Varying & 90.00 & 98.20 & 87.80 & 98.40 & 89.00 & 94.80 & 88.00 & 98.40 & 87.80 & 97.00 \\
5 & Last $K$ & Constant & 91.00 & 98.00 & 86.00 & 98.20 & 92.00 & 98.40 & 86.00 & 98.40 & 84.80 & 97.60 \\
\addlinespace
15 & First $K$ & Varying & 85.40 & 95.40 & 87.00 & 95.20 & 87.00 & 96.60 & 82.80 & 93.80 & 86.40 & 95.80 \\
15 & First $K$ & Constant & 86.60 & 94.80 & 87.80 & 96.20 & 86.80 & 97.20 & 86.60 & 95.40 & 88.60 & 95.60 \\
15 & Last $K$ & Varying & 76.00 & 97.80 & 76.60 & 98.00 & 76.80 & 94.40 & 73.40 & 96.60 & 72.40 & 97.80 \\
15 & Last $K$ & Constant & 75.80 & 98.00 & 77.40 & 98.00 & 72.60 & 96.80 & 75.00 & 98.20 & 70.80 & 98.20 \\
\addlinespace
50 & First $K$ & Varying & 66.60 & 95.60 & 69.20 & 96.60 & 75.80 & 94.40 & 67.00 & 94.60 & 69.20 & 96.20 \\
50 & First $K$ & Constant & 69.40 & 95.40 & 73.80 & 97.00 & 73.60 & 95.20 & 70.40 & 95.00 & 73.60 & 97.40 \\
50 & Last $K$ & Varying & 57.40 & 97.40 & 54.80 & 97.00 & 57.80 & 93.40 & 52.40 & 97.60 & 55.40 & 95.80 \\
50 & Last $K$ & Constant & 56.20 & 96.80 & 52.40 & 94.20 & 55.60 & 97.60 & 52.00 & 97.60 & 58.00 & 96.20 \\
\bottomrule
\end{tabular}
\begin{tablenotes}
\smallskip\footnotesize
\item \textit{Notes:} Empirical coverage rates (in \%) of 95\% simultaneous confidence bands. ``PW'' (Pointwise) uses $z_{0.975} = 1.96$ at each threshold; ``Sup-$t$'' uses simulated critical values. Target coverage is 95\%. ``Constant'' = coefficients identical across thresholds ($H_0$ true); ``Varying'' = coefficients follow empirical trade data pattern ($H_0$ false). First $K$ selects the lowest $K$ thresholds; Last $K$ selects the highest $K$ thresholds up to $\tau_{\max}$. Based on 500 Monte Carlo simulations.
\end{tablenotes}
\end{threeparttable}
\end{table}

Table \ref{tab:appendix_coverage_095} provides the results for the coverage of the pointwise confidence intervals and the sup-$t$ simultaneous confidence bands for the specifications \textit{first $K$} and \textit{last $K$} when considering the maximum quantile to be $\tau_{max} = 0.95$, and Tables \ref{tab:main_coverage_099} and \ref{tab:appendix_coverage_099} provide the results for all the designs considering the maximum quantile to be $\tau_{max} = 0.99$.

Again, while the results for the pointwise confidence intervals show under-coverage, the sup-$t$ bands maintain coverage at or slightly above 95\% level across all configurations. Moreover, the pointwise coverage varies substantially across threshold selection schemes. First $K$ tends to achieve higher 
pointwise coverage, particularly at larger $K$, while 
equally spaced and last $K$ show more pronounced 
under-coverage, reflecting increased estimation 
uncertainty when thresholds extend into sparser regions 
of the distribution. In all cases, however, pointwise confidence bands exhibit substantial under-coverage that worsens with $K$. For the sup-$t$ simultaneous bands, while \textit{last $K$} shows slightly more conservative coverage compared to \textit{equally spaced} and \textit{first $K$}, this deviation from the nominal target is small, whereas pointwise under-coverage is far more pronounced.

Importantly, the choice of the maximum quantile $\tau_{max}$ has negligible impact on sup-$t$ coverage. At 
$\tau_{\max} = 0.99$, Border shows mild undercoverage 
(approximately 93\%) in some configurations, which may 
reflect higher estimation uncertainty at the most extreme 
thresholds for this covariate; all other covariates maintain 
coverage at or above 95\%. For pointwise coverage, however, extending to more extreme quantiles has noticeable effects. \textit{First $K$} results are identical across different maximal thresholds since this scheme always selects the lowest thresholds. \textit{Equally Spaced} shows deterioration at large $K$, while \textit{last $K$} shows consistent deterioration across all $K$ values because this scheme selects thresholds closer to the more extreme boundary ($\tau_{\max} = 0.99$). These patterns reflect increased estimation uncertainty at extreme quantiles, which the sup-$t$ simulated critical values accommodate but pointwise inference does not.

\begin{table}[htbp]
\centering
\captionsetup{skip=0pt}
\captionsetup[subfigure]{font=footnotesize, skip=0pt}
\caption{Size and Power by Covariate: Sup-$t$ vs Wald ($\tau_{\max} = 0.99$, Equally Spaced)}
\label{tab:main_by_covariate_099}
\begin{threeparttable}
\footnotesize
\begin{tabular}{l cc cc cc cc cc}
\toprule
& \multicolumn{2}{c}{Distance} & \multicolumn{2}{c}{Legal} & \multicolumn{2}{c}{Border} & \multicolumn{2}{c}{Language} & \multicolumn{2}{c}{Religion} \\
\cmidrule(lr){2-3} \cmidrule(lr){4-5} \cmidrule(lr){6-7} \cmidrule(lr){8-9} \cmidrule(lr){10-11}
$K$ & Sup-$t$ & Wald & Sup-$t$ & Wald & Sup-$t$ & Wald & Sup-$t$ & Wald & Sup-$t$ & Wald \\
\midrule
\multicolumn{11}{l}{\textit{Panel A: Constant Coefficients (Size, target: 5\%)}} \\
\addlinespace
5 & 2.00 & 2.20 & 3.40 & 3.00 & 2.60 & 3.00 & 3.20 & 3.80 & 4.20 & 3.40 \\
15 & 2.40 & 3.00 & 3.80 & 2.00 & 3.20 & 3.80 & 2.80 & 2.40 & 3.60 & 1.40 \\
50 & 3.40 & 2.40 & 2.40 & 1.00 & 4.60 & 14.20 & 4.60 & 2.20 & 3.40 & 2.60 \\
\midrule
\multicolumn{11}{l}{\textit{Panel B: Varying Coefficients (Power)}} \\
\addlinespace
5 & 99.80 & 100.00 & 100.00 & 100.00 & 97.60 & 99.20 & 26.20 & 24.40 & 67.20 & 77.20 \\
15 & 100.00 & 100.00 & 100.00 & 100.00 & 99.00 & 98.60 & 58.80 & 55.40 & 69.20 & 88.60 \\
50 & 100.00 & 100.00 & 100.00 & 100.00 & 99.60 & 100.00 & 53.80 & 77.80 & 88.80 & 93.40 \\
\addlinespace
\bottomrule
\end{tabular}
\begin{tablenotes}
\smallskip\footnotesize
\item \textit{Notes:} Rejection rates (in \%) for testing $H_0: \theta_d(\tau_1) = \cdots = \theta_d(\tau_K)$ at the 5\% level. Sup-$t$ = Sup-$t$ test; Wald = Wald test. ``Constant'' = coefficients identical across thresholds ($H_0$ true); ``Varying'' = coefficients follow empirical trade data pattern ($H_0$ false). Based on 500 Monte Carlo simulations.
\end{tablenotes}
\end{threeparttable}
\end{table}

\begin{table}[htbp]
\centering
\caption{Size and Power by Covariate: First $K$ and Last $K$ ($\tau_{\max} = 0.95$)}
\label{tab:appendix_by_covariate_095}
\begin{threeparttable}
\footnotesize
\begin{tabular}{ll cc cc cc cc cc}
\toprule
& & \multicolumn{2}{c}{Distance} & \multicolumn{2}{c}{Legal} & \multicolumn{2}{c}{Border} & \multicolumn{2}{c}{Language} & \multicolumn{2}{c}{Religion} \\
\cmidrule(lr){3-4} \cmidrule(lr){5-6} \cmidrule(lr){7-8} \cmidrule(lr){9-10} \cmidrule(lr){11-12}
$K$ & Selection & Sup-$t$ & Wald & Sup-$t$ & Wald & Sup-$t$ & Wald & Sup-$t$ & Wald & Sup-$t$ & Wald \\
\midrule
\multicolumn{12}{l}{\textit{Panel A: Constant Coefficients (Size, target: 5\%)}} \\
\addlinespace
5 & First $K$ & 3.40 & 3.60 & 4.20 & 4.20 & 5.20 & 6.00 & 5.80 & 5.40 & 3.20 & 3.40 \\
5 & Last $K$ & 1.80 & 1.60 & 2.20 & 1.60 & 2.40 & 1.80 & 2.40 & 1.80 & 2.40 & 2.80 \\
\addlinespace
15 & First $K$ & 2.80 & 1.80 & 3.40 & 2.80 & 5.20 & 8.20 & 5.60 & 3.40 & 2.80 & 3.00 \\
15 & Last $K$ & 2.40 & 2.60 & 2.40 & 0.40 & 1.80 & 2.60 & 2.40 & 2.20 & 2.40 & 1.60 \\
\addlinespace
50 & First $K$ & 3.20 & 1.60 & 2.00 & 1.20 & 5.60 & 22.60 & 4.20 & 3.20 & 3.60 & 2.60 \\
50 & Last $K$ & 3.00 & 1.80 & 2.40 & 0.60 & 2.60 & 9.40 & 2.00 & 2.00 & 1.80 & 1.80 \\
\midrule
\multicolumn{12}{l}{\textit{Panel B: Varying Coefficients (Power)}} \\
\addlinespace
5 & First $K$ & 87.60 & 83.60 & 14.20 & 14.40 & 10.00 & 4.80 & 53.80 & 38.80 & 6.00 & 7.20 \\
5 & Last $K$ & 5.40 & 4.80 & 3.60 & 3.60 & 16.20 & 11.00 & 58.00 & 45.20 & 10.80 & 9.00 \\
\addlinespace
15 & First $K$ & 100.00 & 97.60 & 40.80 & 44.40 & 20.60 & 11.60 & 65.60 & 45.40 & 12.80 & 26.20 \\
15 & Last $K$ & 5.80 & 29.60 & 55.20 & 84.20 & 33.80 & 31.00 & 51.40 & 43.00 & 9.40 & 21.60 \\
\addlinespace
50 & First $K$ & 100.00 & 98.60 & 87.40 & 95.00 & 91.40 & 90.40 & 64.80 & 58.00 & 16.00 & 82.40 \\
50 & Last $K$ & 13.60 & 55.00 & 99.40 & 100.00 & 62.20 & 90.40 & 47.00 & 67.60 & 48.00 & 87.60 \\
\addlinespace
\bottomrule
\end{tabular}
\begin{tablenotes}
\smallskip\footnotesize
\item \textit{Notes:} Rejection rates (in \%) for testing $H_0: \theta_d(\tau_1) = \cdots = \theta_d(\tau_K)$ at the 5\% level. Sup-$t$ = Sup-$t$ test; Wald = Wald test. ``Constant'' = coefficients identical across thresholds ($H_0$ true); ``Varying'' = coefficients follow empirical trade data pattern ($H_0$ false). First $K$ selects the lowest $K$ thresholds; Last $K$ selects the highest $K$ thresholds up to $\tau_{\max}$. Based on 500 Monte Carlo simulations.
\end{tablenotes}
\end{threeparttable}
\end{table}

\begin{table}[htbp]
\centering
\captionsetup{skip=0pt}
\captionsetup[subfigure]{font=footnotesize, skip=0pt}
\caption{Size and Power by Covariate: First $K$ and Last $K$ ($\tau_{\max} = 0.99$)}
\label{tab:appendix_by_covariate_099}
\begin{threeparttable}
\footnotesize
\begin{tabular}{ll cc cc cc cc cc}
\toprule
& & \multicolumn{2}{c}{Distance} & \multicolumn{2}{c}{Legal} & \multicolumn{2}{c}{Border} & \multicolumn{2}{c}{Language} & \multicolumn{2}{c}{Religion} \\
\cmidrule(lr){3-4} \cmidrule(lr){5-6} \cmidrule(lr){7-8} \cmidrule(lr){9-10} \cmidrule(lr){11-12}
$K$ & Selection & Sup-$t$ & Wald & Sup-$t$ & Wald & Sup-$t$ & Wald & Sup-$t$ & Wald & Sup-$t$ & Wald \\
\midrule
\multicolumn{12}{l}{\textit{Panel A: Constant Coefficients (Size, target: 5\%)}} \\
\addlinespace
5 & First $K$ & 3.40 & 3.60 & 4.20 & 4.20 & 5.20 & 6.00 & 5.80 & 5.40 & 3.20 & 3.40 \\
5 & Last $K$ & 0.40 & 0.80 & 2.60 & 1.80 & 1.00 & 1.20 & 2.20 & 1.20 & 3.00 & 1.60 \\
\addlinespace
15 & First $K$ & 2.80 & 1.80 & 3.40 & 2.80 & 5.20 & 8.20 & 5.60 & 3.40 & 2.80 & 3.00 \\
15 & Last $K$ & 1.40 & 0.40 & 1.80 & 0.40 & 1.20 & 1.40 & 1.20 & 0.60 & 3.80 & 1.80 \\
\addlinespace
50 & First $K$ & 3.20 & 1.60 & 2.00 & 1.20 & 5.60 & 22.60 & 4.20 & 3.20 & 3.60 & 2.60 \\
50 & Last $K$ & 1.80 & 1.40 & 2.20 & 1.60 & 2.20 & 9.00 & 2.40 & 1.00 & 3.40 & 1.00 \\
\midrule
\multicolumn{12}{l}{\textit{Panel B: Varying Coefficients (Power)}} \\
\addlinespace
5 & First $K$ & 87.60 & 83.60 & 14.20 & 14.40 & 10.00 & 4.80 & 53.80 & 38.80 & 6.00 & 7.20 \\
5 & Last $K$ & 2.60 & 9.80 & 7.40 & 24.60 & 12.80 & 22.00 & 3.40 & 2.40 & 2.80 & 3.20 \\
\addlinespace
15 & First $K$ & 100.00 & 97.60 & 40.80 & 44.40 & 20.60 & 11.60 & 65.60 & 45.40 & 12.80 & 26.20 \\
15 & Last $K$ & 7.60 & 23.60 & 14.40 & 20.20 & 61.20 & 84.80 & 3.80 & 38.20 & 6.60 & 15.40 \\
\addlinespace
50 & First $K$ & 100.00 & 98.60 & 87.40 & 95.00 & 91.40 & 90.40 & 64.80 & 58.00 & 16.00 & 82.40 \\
50 & Last $K$ & 13.20 & 72.00 & 67.00 & 100.00 & 76.60 & 98.00 & 4.00 & 65.80 & 15.20 & 92.00 \\
\addlinespace
\bottomrule
\end{tabular}
\begin{tablenotes}
\smallskip\footnotesize
\item \textit{Notes:} Rejection rates (in \%) for testing $H_0: \theta_d(\tau_1) = \cdots = \theta_d(\tau_K)$ at the 5\% level. Sup-$t$ = Sup-$t$ test; Wald = Wald test. ``Constant'' = coefficients identical across thresholds ($H_0$ true); ``Varying'' = coefficients follow empirical trade data pattern ($H_0$ false). First $K$ selects the lowest $K$ thresholds; Last $K$ selects the highest $K$ thresholds up to $\tau_{\max}$. Based on 500 Monte Carlo simulations.
\end{tablenotes}
\end{threeparttable}
\end{table}

Table \ref{tab:appendix_by_covariate_095} provides the results for size and power of the joint equality testing procedures for $\tau_{max} = 0.95$ for the specifications \textit{first $K$} and \textit{last $K$} when considering the maximum quantile to be $\tau_{max} = 0.95$, and Tables \ref{tab:main_coverage_099} and \ref{tab:appendix_coverage_099} provide the results for all the designs considering the maximum quantile to be $\tau_{max} = 0.99$. 

Overall, the sup-$t$ test maintains valid size across the configurations, with rejection rates typically between 2\% and 6\% under the null hypothesis for $\tau_{max} = 0.95$. However, when thresholds are concentrated in the upper 
tail (\textit{last $K$}), the sup-$t$ test tends to be 
conservative, with size below the nominal level, 
particularly at $\tau_{\max} = 0.99$ where the selected 
thresholds lie in sparser regions. Moreover, power is reduced for covariates whose coefficients vary primarily at 
lower quantiles. This reflects both the wider confidence bands at 
extreme thresholds, and the limited heterogeneity across the coefficients in the selected range. However, in comparison, the Wald test shows more substantial size distortion that varies by covariate and worsens with $K$. For Border, the Wald test over-rejects in all configurations when $K=50$, with the distortion most severe for \textit{first $K$} at 22.6\%, exceeding the 5\% nominal level substantially. The Wald test's poor performance at large $K$ is consistent with its known theoretical limitations when testing many restrictions jointly. Furthermore, note that the choice of $\tau_{\max}$ (0.95 vs 0.99) does not affect size for either test substantially under the equally 
spaced configuration. 

Turning to more details on the results for power, the sup-$t$ test exhibits high rejection rates for Distance, Legal, and Border across all specifications for the \textit{equally spaced} design, which is consistent with the larger heterogeneity across the coefficients for these covariates in Figures \ref{fig:coef_paths_main} and \ref{fig:coef_paths_main_099}. Language and Religion show more moderate power, reflecting their smaller coefficient variation across thresholds. For Religion, power increases substantially from $K = 5$ 
to $K = 15$ and remains high at $K = 50$. For Language, 
power increases from $K = 5$ to $K = 15$ but declines 
slightly at $K = 50$, suggesting that the additional 
thresholds introduce estimation noise without capturing 
further coefficient variation. The choice of $\tau_{\max}$ (0.95 vs 0.99) does not yield a consistent pattern for power; the effect is covariate-specific and depends on where each covariate's coefficients vary most. 

Across the different specifications for threshold selection, power depends on which region of the distribution captures the most coefficient variation. For instance, Distance varies mostly at lower quantiles, yielding high power for \textit{first $K$} but substantially lower power for \textit{last $K$}. Legal shows the opposite pattern at larger $K$, with higher 
power for \textit{last $K$} where its coefficients vary most (Figures \ref{fig:coef_paths_appendix} and \ref{fig:coef_paths_appendix_099}). Comparing the two tests, Wald shows higher power than sup-$t$ in several configurations, particularly for Language and Religion at large $K$. However, given the Wald test's size distortion for other covariates at the same levels of $K$ and its known theoretical limitations when testing many restrictions jointly, the sup-$t$ test offers more reliable inference overall.

\begin{figure}[htbp]
\centering
\captionsetup{skip=0pt}
\captionsetup[subfigure]{font=footnotesize, skip=0pt}
\includegraphics[width=\textwidth]{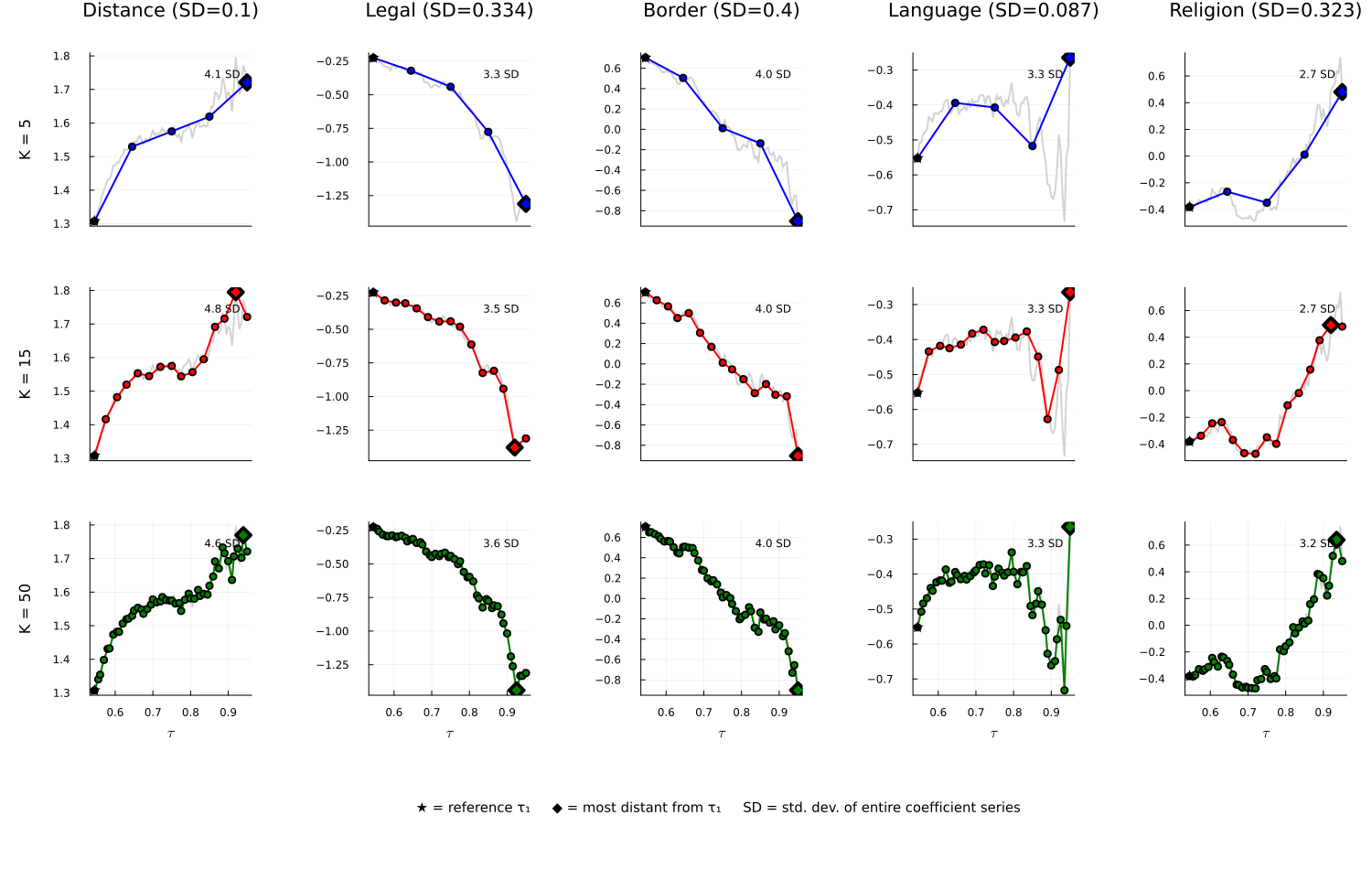}
\caption{True coefficient paths from empirical trade data with equally spaced threshold selection ($\tau_{\max} = 0.95$). Gray line: full coefficient series $\{\theta_d(\tau)\}_{\tau \leq \tau_{\max}}$. Colored markers: selected thresholds for $K = 5$ (blue), $K = 15$ (red), $K = 50$ (green). Star ($\bigstar$) marks reference threshold $\tau_1$; diamond ($\blacklozenge$) marks most distant coefficient from $\tau_1$. Each panel reports the maximum deviation in standard deviation units, where SD is computed over the entire coefficient series.}
\label{fig:coef_paths_main}
\end{figure}

\begin{figure}[htbp]
\centering
\captionsetup{skip=0pt}
\captionsetup[subfigure]{font=footnotesize, skip=0pt}
\includegraphics[width=\textwidth]{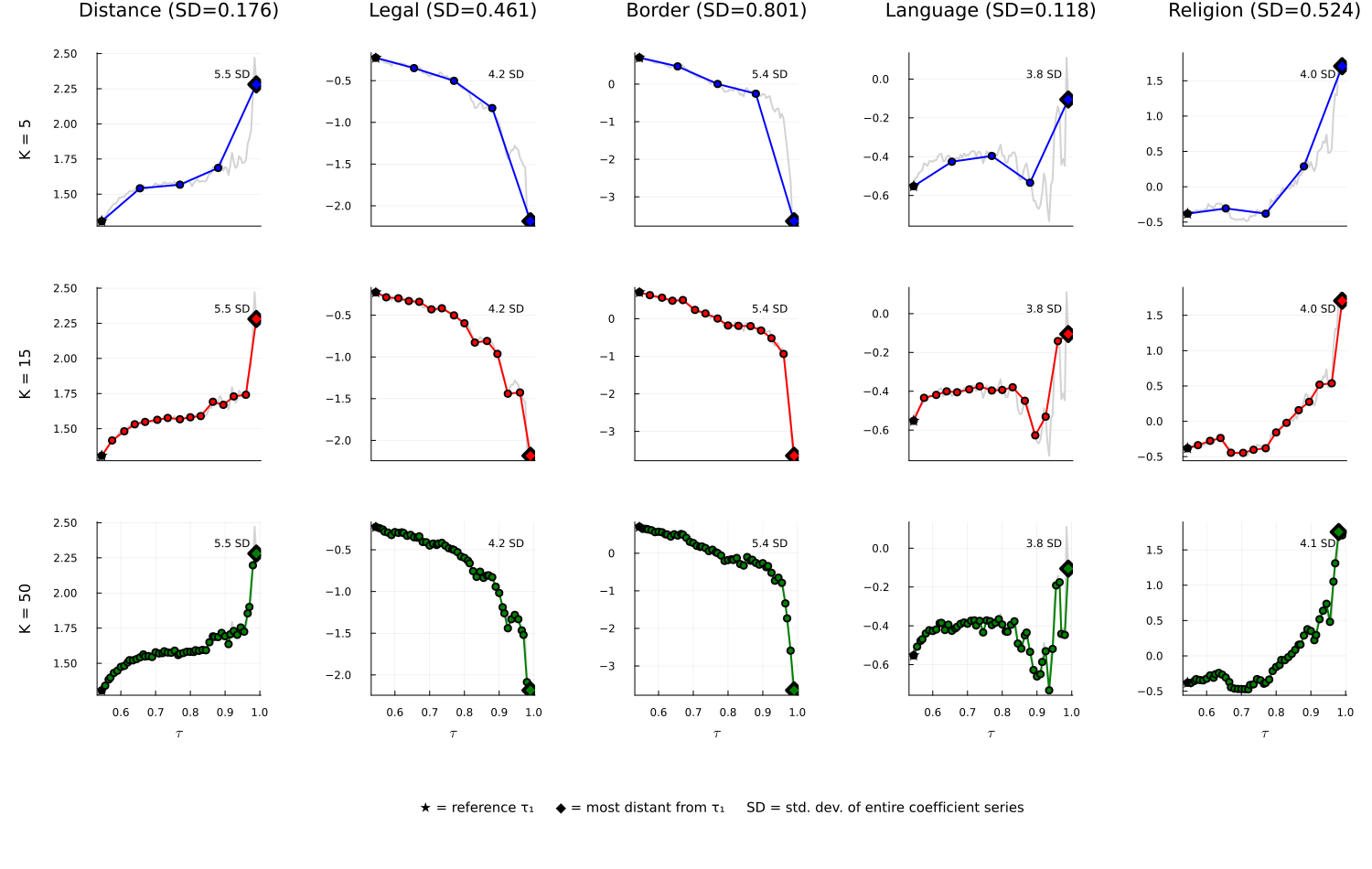}
\caption{True coefficient paths from empirical trade data with equally spaced threshold selection ($\tau_{\max} = 0.99$). Gray line: full coefficient series $\{\theta_d(\tau)\}_{\tau \leq \tau_{\max}}$. Colored markers: selected thresholds for $K = 5$ (blue), $K = 15$ (red), $K = 50$ (green). Star ($\bigstar$) marks reference threshold $\tau_1$; diamond ($\blacklozenge$) marks most distant coefficient from $\tau_1$. Each panel reports the maximum deviation in standard deviation units, where SD is computed over the entire coefficient series.}
\label{fig:coef_paths_main_099}
\end{figure}

\begin{figure}[htbp]
\centering
\captionsetup{skip=0pt}
\captionsetup[subfigure]{font=footnotesize, skip=0pt}
\includegraphics[width=\textwidth]{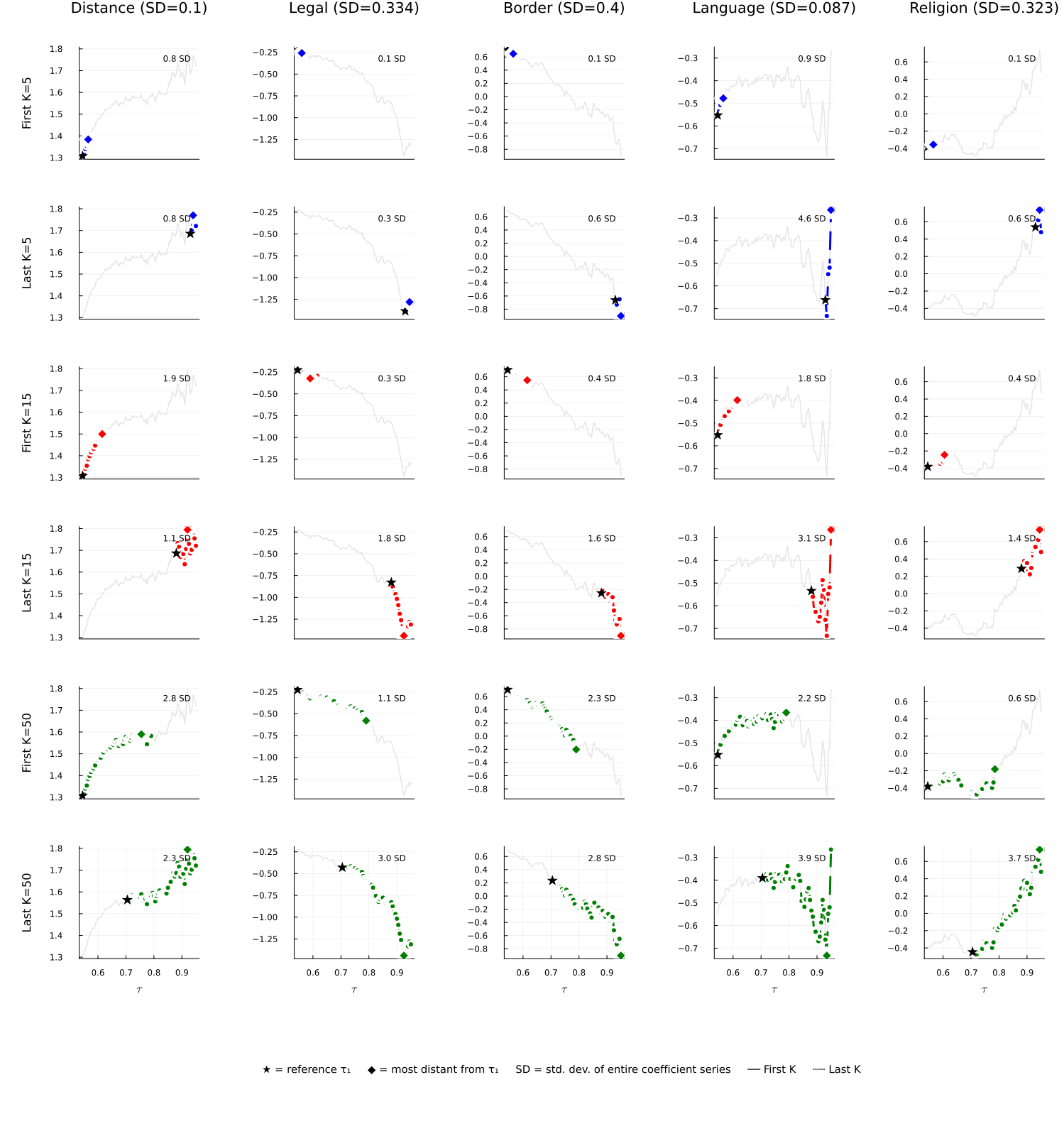}
\caption{True coefficient paths from empirical trade data with First $K$ and Last $K$ threshold selection ($\tau_{\max} = 0.95$). Gray line: full coefficient series $\{\theta_d(\tau)\}_{\tau \leq \tau_{\max}}$. First $K$ selects the lowest $K$ thresholds; Last $K$ selects the highest $K$ thresholds. Colored markers indicate $K = 5$ (blue), $K = 15$ (red), $K = 50$ (green). Star ($\bigstar$) marks reference threshold $\tau_1$; diamond ($\blacklozenge$) marks most distant coefficient. Lower power in these configurations reflects reduced coefficient variation within the selected threshold ranges compared to equally spaced selection.}
\label{fig:coef_paths_appendix}
\end{figure}

 \begin{figure}[htbp]
\centering
\captionsetup{skip=0pt}
\captionsetup[subfigure]{font=footnotesize, skip=0pt}
\includegraphics[width=\textwidth]{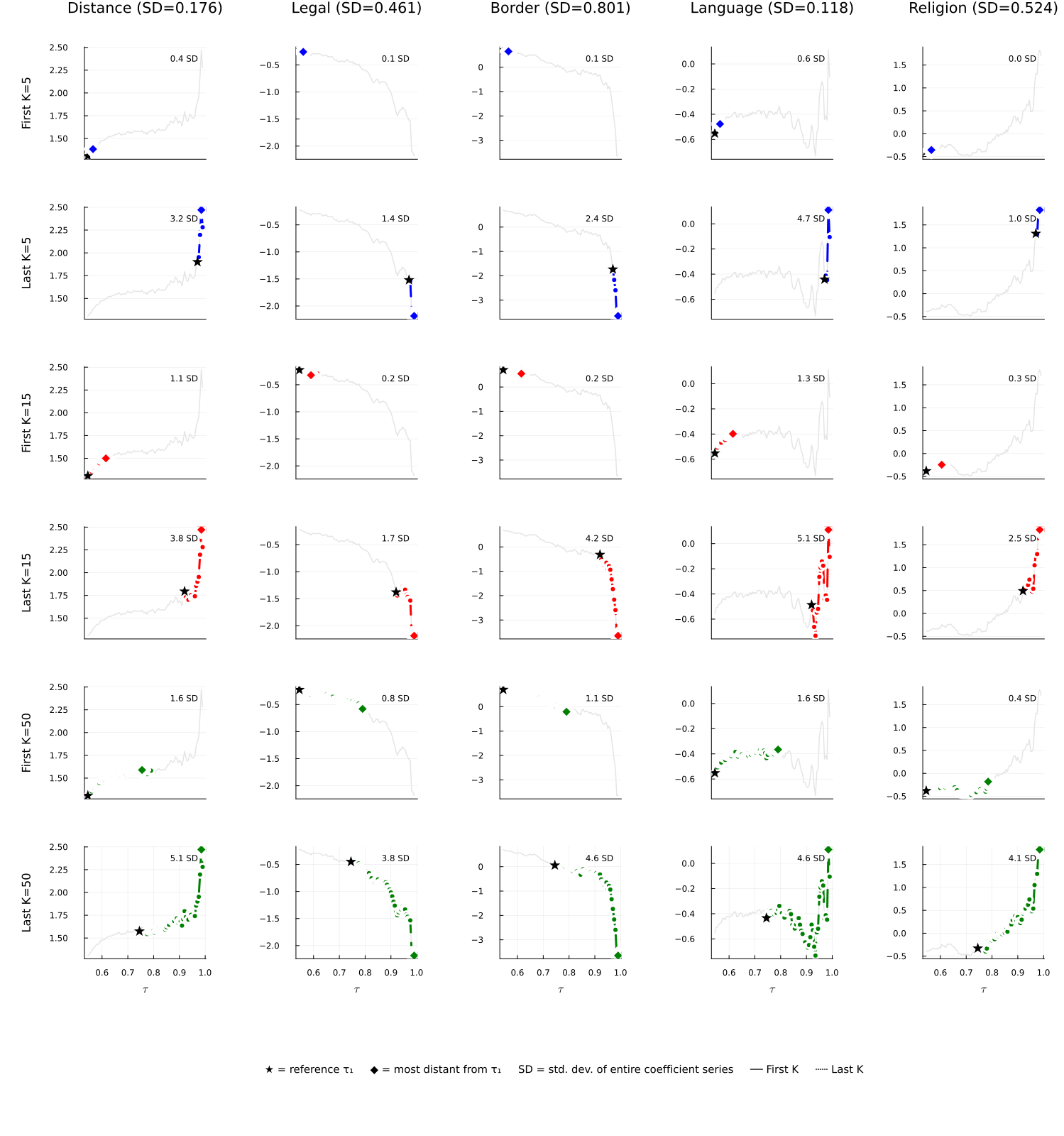}
\caption{True coefficient paths from empirical trade data with First $K$ and Last $K$ threshold selection ($\tau_{\max} = 0.99$). Gray line: full coefficient series $\{\theta_d(\tau)\}_{\tau \leq \tau_{\max}}$. First $K$ selects the lowest $K$ thresholds; Last $K$ selects the highest $K$ thresholds. Colored markers indicate $K = 5$ (blue), $K = 15$ (red), $K = 50$ (green). Star ($\bigstar$) marks reference threshold $\tau_1$; diamond ($\blacklozenge$) marks most distant coefficient. Lower power in these configurations reflects reduced coefficient variation within the selected threshold ranges compared to equally spaced selection.}
\label{fig:coef_paths_appendix_099}
\end{figure}
\clearpage
\subsection{Additional Tables and Figures for Section \ref{application}}\label{application_appendix}
\begin{table}[H]
    \centering
    \captionsetup{skip=0pt}
\captionsetup[subfigure]{font=footnotesize, skip=0pt}
    \caption{Descriptive statistics. Source: \cite{helpman2008estimating}.}
    \label{tab:descriptive}
    \footnotesize
      \begin{tabular}{rrrrr}
            &       &       &       &  \\
  \cmidrule{2-4}          &       & \multicolumn{1}{l}{Mean} & \multicolumn{1}{l}{Std. Dev.} &  \\
  \cmidrule{2-4}          & \multicolumn{1}{l}{Trade} & 0.45  & 0.50  &  \\
            & \multicolumn{1}{l}{Trade volume} & 84.54 & 1,082,219 &  \\
            & \multicolumn{1}{l}{Log distance} & 4.18  & 0.78  &  \\
            & \multicolumn{1}{l}{Legal} & 0.37  & 0.48  &  \\
            & \multicolumn{1}{l}{Language} & 0.29  & 0.45  &  \\
            & \multicolumn{1}{l}{Religion} & 0.17  & 0.25  &  \\
            & \multicolumn{1}{l}{Border} & 0.02  & 0.13  &  \\
            & \multicolumn{1}{l}{Currency} & 0.01  & 0.09  &  \\
            & \multicolumn{1}{l}{FTA} & 0.01  & 0.08  &  \\
            & \multicolumn{1}{l}{Colony} & 0.01  & 0.10  &  \\
  \cmidrule{2-4}          &       &       &       &  \\
      \end{tabular}%
  \end{table}%

  \begin{figure}[htbp]
\centering
\captionsetup{skip=0pt}
\captionsetup[subfigure]{font=footnotesize, skip=0pt}
\begin{subfigure}[b]{0.32\textwidth}
    \centering
    \includegraphics[width=\textwidth]{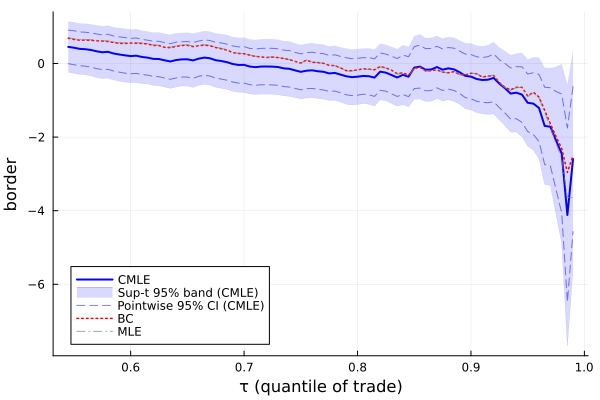}
    \caption{Border}
    \label{fig:band_border}
\end{subfigure}
\hfill
\begin{subfigure}[b]{0.32\textwidth}
    \centering
    \includegraphics[width=\textwidth]{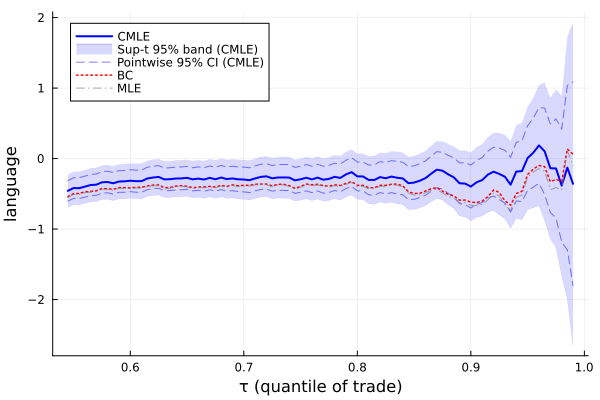}
    \caption{Language}
    \label{fig:band_language}
\end{subfigure}
\hfill
\begin{subfigure}[b]{0.32\textwidth}
    \centering
    \includegraphics[width=\textwidth]{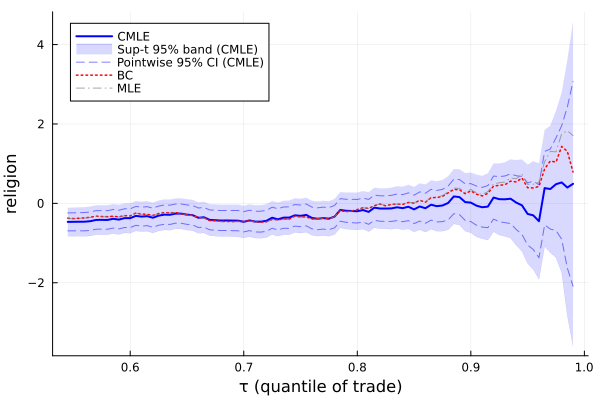}
    \caption{Religion}
    \label{fig:band_religion}
\end{subfigure}
\caption{Distribution regression estimates of the effect of remaining covariates on bilateral trade. Solid blue: CMLE; dotted red: BC; dash-dotted 
gray: MLE. Shaded region and dashed lines show the 95\% simultaneous sup-$t$ 
confidence band and pointwise confidence intervals for the CMLE, respectively. 
Thresholds correspond to empirical quantiles in $[0.545, 0.990]$ at intervals 
of $0.005$.}
\label{fig:bands_main_appendix}
\end{figure}

  \begin{figure}[htbp]
\centering
\captionsetup{skip=0pt}
\captionsetup[subfigure]{font=footnotesize, skip=0pt}
\begin{subfigure}[b]{0.48\textwidth}
    \centering
    \includegraphics[width=\textwidth]{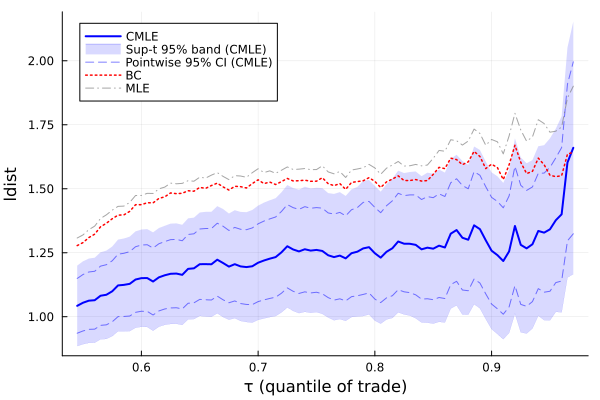}
    \caption{Log distance}
    \label{fig:band2_ldist}
\end{subfigure}
\hfill
\begin{subfigure}[b]{0.48\textwidth}
    \centering
    \includegraphics[width=\textwidth]{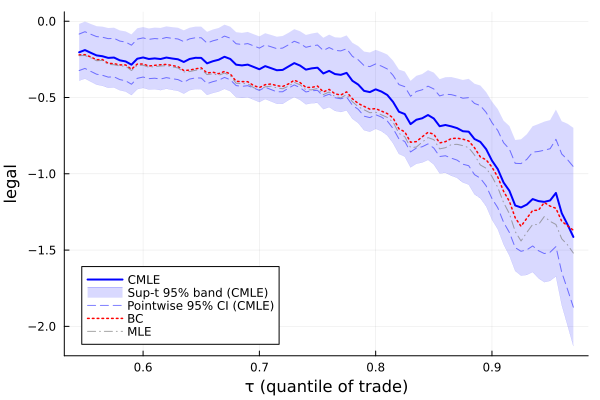}
    \caption{Legal system}
    \label{fig:band2_legal}
\end{subfigure}
\hfill
\begin{subfigure}[b]{0.48\textwidth}
    \centering
    \includegraphics[width=\textwidth]{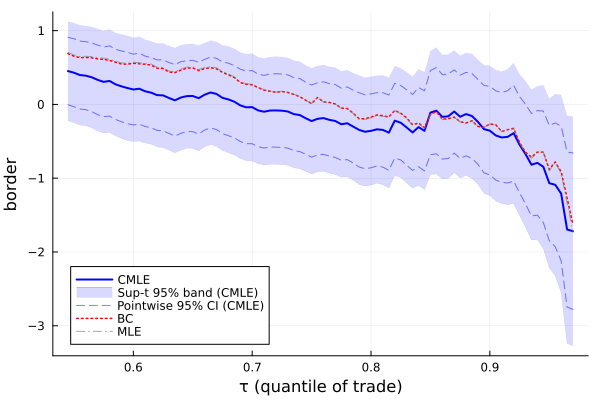}
    \caption{Border}
    \label{fig:band2_legal}
\end{subfigure}
\hfill
\begin{subfigure}[b]{0.48\textwidth}
    \centering
    \includegraphics[width=\textwidth]{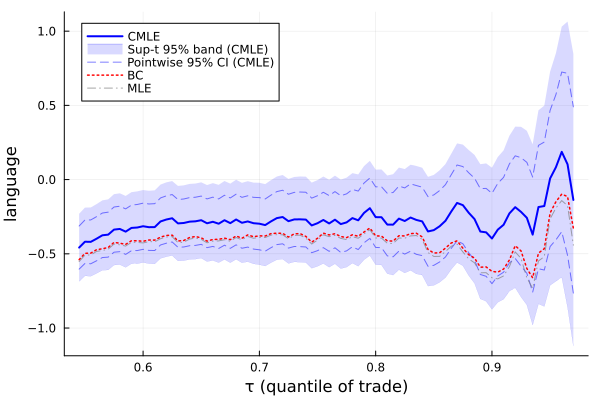}
    \caption{Language}
    \label{fig:band2_language}
\end{subfigure}
\hfill
\begin{subfigure}[b]{0.48\textwidth}
    \centering
    \includegraphics[width=\textwidth]{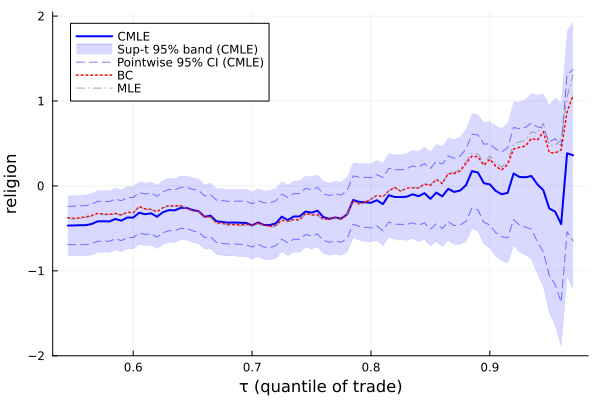}
    \caption{Religion}
    \label{fig:band2_religion}
\end{subfigure}
\caption{Distribution regression estimates of the effect of all covariates on bilateral trade. Solid blue: CMLE; dotted red: BC; dash-dotted 
gray: MLE. Shaded region and dashed lines show the 95\% simultaneous sup-$t$ 
confidence band and pointwise confidence intervals for the CMLE, respectively. 
Thresholds correspond to empirical quantiles in $[0.545, 0.970]$ at intervals 
of $0.005$.}
\label{fig:bands_main_097_appendix}
\end{figure}

\begin{figure}[htbp]
\centering
\captionsetup{skip=0pt}
\captionsetup[subfigure]{font=footnotesize, skip=0pt}
\begin{subfigure}[b]{0.48\textwidth}
    \centering
    \includegraphics[width=\textwidth]{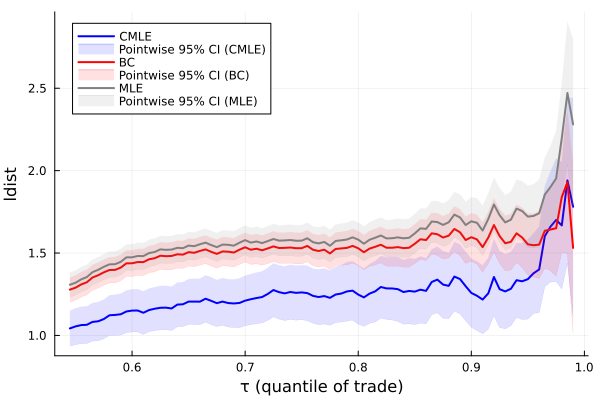}
    \caption{Log distance}
    \label{fig:band3_ldist}
\end{subfigure}
\hfill
\begin{subfigure}[b]{0.48\textwidth}
    \centering
    \includegraphics[width=\textwidth]{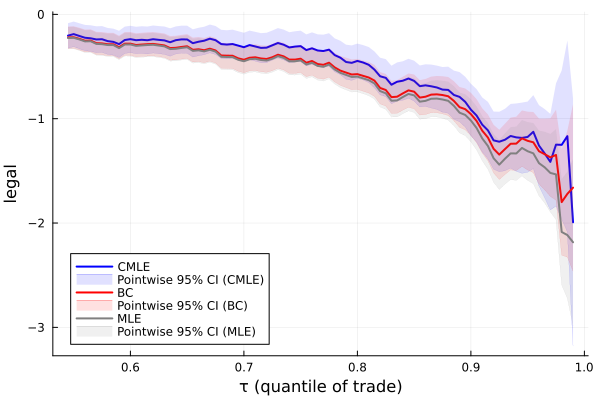}
    \caption{Legal system}
    \label{fig:band3_legal}
\end{subfigure}
\hfill
\begin{subfigure}[b]{0.48\textwidth}
    \centering
    \includegraphics[width=\textwidth]{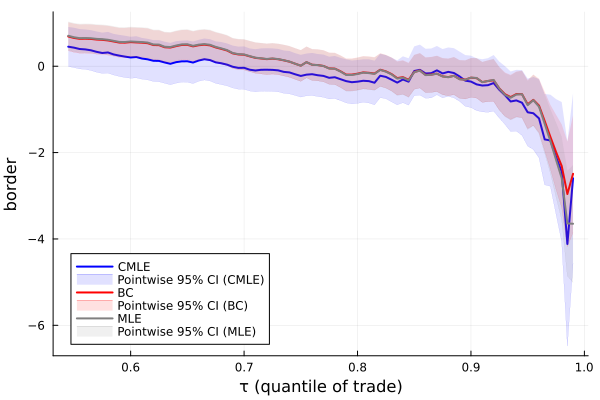}
    \caption{Border}
    \label{fig:band3_border}
\end{subfigure}
\hfill
\begin{subfigure}[b]{0.48\textwidth}
    \centering
    \includegraphics[width=\textwidth]{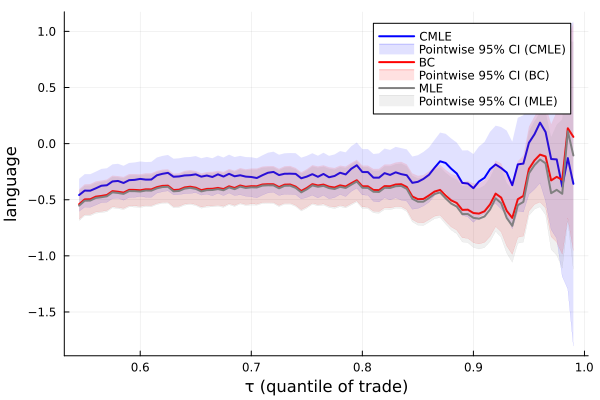}
    \caption{Language}
    \label{fig:band3_language}
\end{subfigure}
\hfill
\begin{subfigure}[b]{0.48\textwidth}
    \centering
    \includegraphics[width=\textwidth]{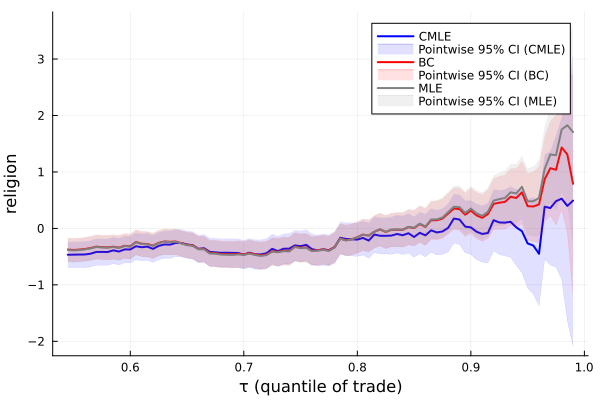}
    \caption{Religion}
    \label{fig:band3_religion}
\end{subfigure}
\caption{Pointwise comparison of CMLE, BC, and MLE estimates with 95\% 
pointwise confidence intervals for all covariates. Solid blue: CMLE; 
solid red: BC; solid gray: MLE. Shaded regions show pointwise 95\% 
confidence intervals for each estimator. Thresholds correspond to 
empirical quantiles in $[0.545, 0.990]$ at intervals of $0.005$.}
\label{fig:pw_comparison_099}
\end{figure}

Figure~\ref{fig:pw_comparison_099} displays pointwise 95\% 
confidence intervals for all three estimators. The CMLE 
confidence intervals are wider than those of the BC estimator 
throughout the distribution, reflecting the efficiency cost of 
conditioning out the fixed effects rather than estimating and 
correcting for them. This efficiency loss is the expected 
trade-off for eliminating the incidental parameter problem 
entirely. Note that the uncorrected MLE confidence intervals 
have similar width to those of the BC, since both are based on 
the same logit variance estimate. 
All three sets of confidence intervals widen in the upper 
tail. For the CMLE, this reflects the smaller number of 
informative quadruples at extreme thresholds. For the MLE 
and BC, which use all observations regardless of the 
threshold, the widening reflects 
the reduced information content at extreme thresholds, 
where predicted probabilities approach zero or one.

\begin{figure}[htbp]
\centering
\captionsetup{skip=0pt}
\captionsetup[subfigure]{font=footnotesize, skip=0pt}
\begin{subfigure}[b]{0.48\textwidth}
    \centering
    \includegraphics[width=\textwidth]{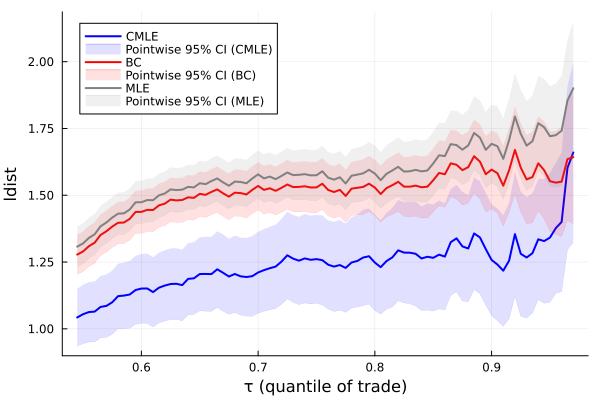}
    \caption{Log distance}
    \label{fig:band4_ldist}
\end{subfigure}
\hfill
\begin{subfigure}[b]{0.48\textwidth}
    \centering
    \includegraphics[width=\textwidth]{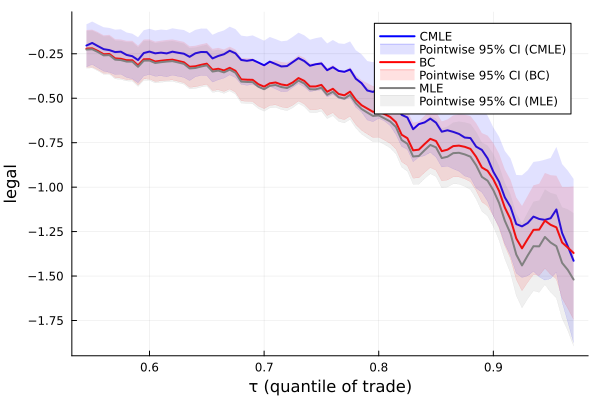}
    \caption{Legal system}
    \label{fig:band4_legal}
\end{subfigure}
\hfill
\begin{subfigure}[b]{0.48\textwidth}
    \centering
    \includegraphics[width=\textwidth]{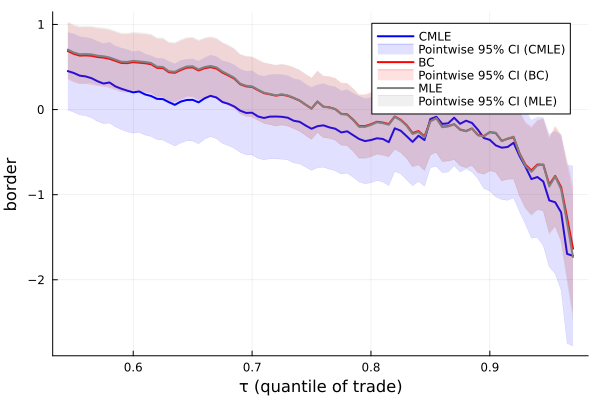}
    \caption{Border}
    \label{fig:band4_border}
\end{subfigure}
\hfill
\begin{subfigure}[b]{0.48\textwidth}
    \centering
    \includegraphics[width=\textwidth]{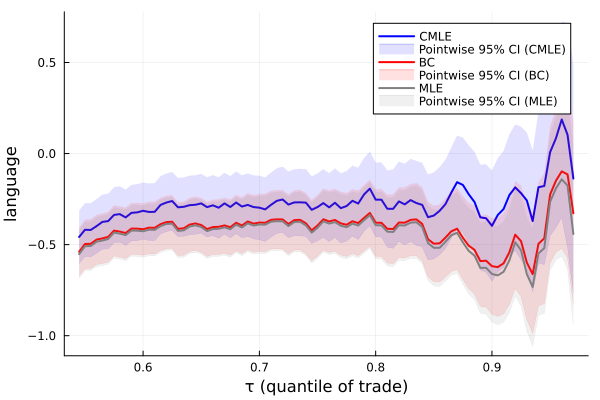}
    \caption{Language}
    \label{fig:band4_language}
\end{subfigure}
\hfill
\begin{subfigure}[b]{0.48\textwidth}
    \centering
    \includegraphics[width=\textwidth]{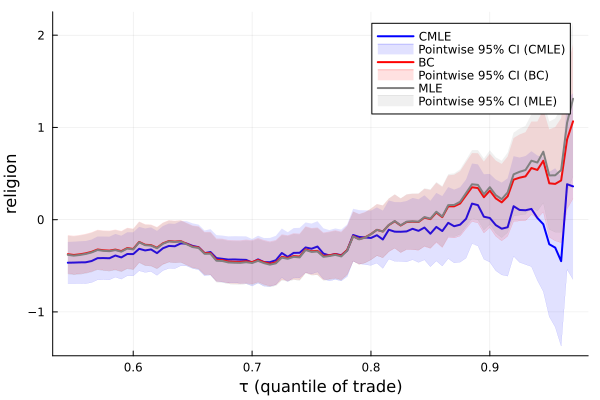}
    \caption{Religion}
    \label{fig:band4_religion}
\end{subfigure}
\caption{Pointwise comparison of CMLE, BC, and MLE estimates with 95\% 
pointwise confidence intervals for all covariates. Solid blue: CMLE; 
solid red: BC; solid gray: MLE. Shaded regions show pointwise 95\% 
confidence intervals for each estimator. Thresholds correspond to 
empirical quantiles in $[0.545, 0.970]$ at intervals of $0.005$.}
\label{fig:pw_comparison_097}
\end{figure}

\begin{figure}[htbp]
\centering
\captionsetup{skip=0pt}
\captionsetup[subfigure]{font=footnotesize, skip=0pt}
\begin{subfigure}[b]{0.48\textwidth}
    \centering
    \includegraphics[width=\textwidth]{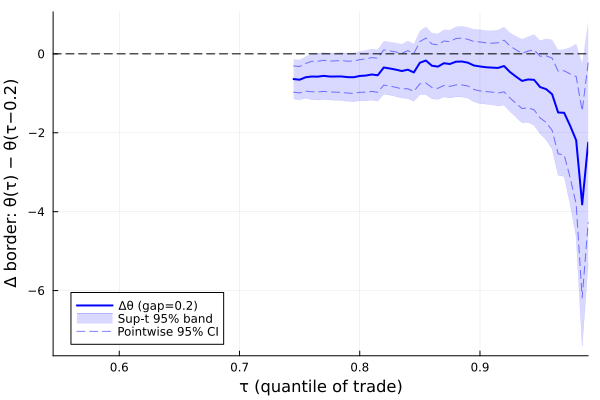}
    \caption{Border}
    \label{fig:diff_border}
\end{subfigure}
\hfill
\begin{subfigure}[b]{0.48\textwidth}
    \centering
    \includegraphics[width=\textwidth]{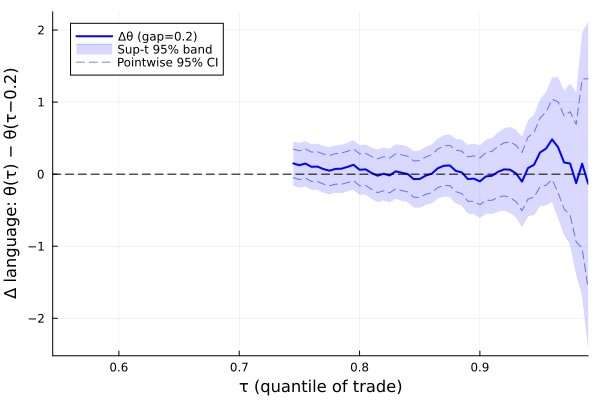}
    \caption{Language}
    \label{fig:diff_legal}
\end{subfigure}
\hfill
\begin{subfigure}[b]{0.48\textwidth}
    \centering
    \includegraphics[width=\textwidth]{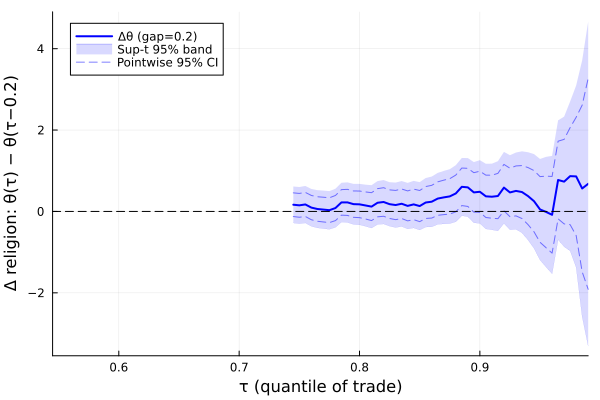}
    \caption{Religion}
    \label{fig:diff_religion}
\end{subfigure}
\caption{Differences $\theta_{n,d}(\tau) - \theta_{n,d}(\tau - 0.20)$ from the CMLE 
estimator for log distance and legal system. Shaded region and dashed lines show 
the 95\% simultaneous sup-$t$ confidence band and pointwise confidence intervals 
for the differences, respectively. The horizontal dashed line marks zero. Maximal quantile considered is 0.990.}
\label{fig:diffs_appendix}
\end{figure}

\begin{figure}[htbp]
\centering
\captionsetup{skip=0pt}
\captionsetup[subfigure]{font=footnotesize, skip=0pt}
\begin{subfigure}[b]{0.48\textwidth}
    \centering
    \includegraphics[width=\textwidth]{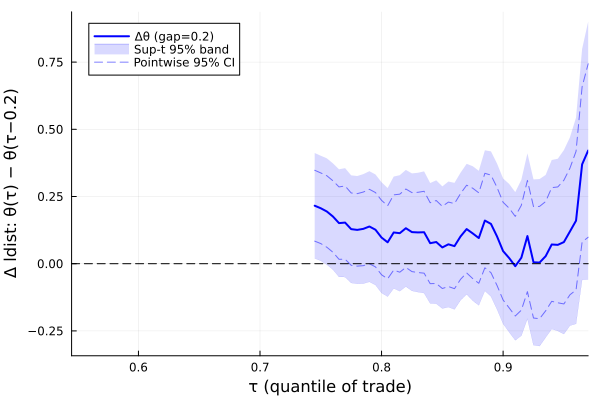}
    \caption{Log distance}
    \label{fig:diff2_ldist}
\end{subfigure}
\hfill
\begin{subfigure}[b]{0.48\textwidth}
    \centering
    \includegraphics[width=\textwidth]{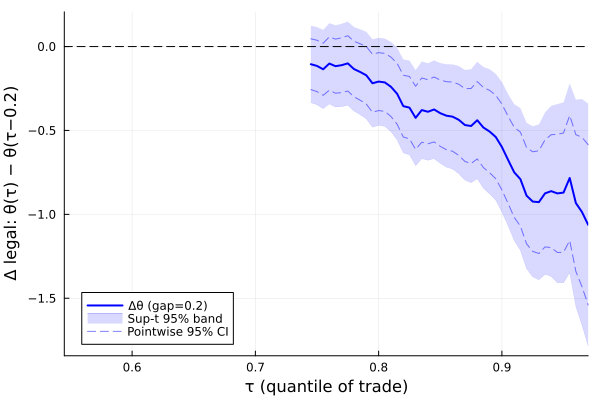}
    \caption{Legal system}
    \label{fig:diff2_legal}
\end{subfigure}
\hfill
\begin{subfigure}[b]{0.48\textwidth}
    \centering
    \includegraphics[width=\textwidth]{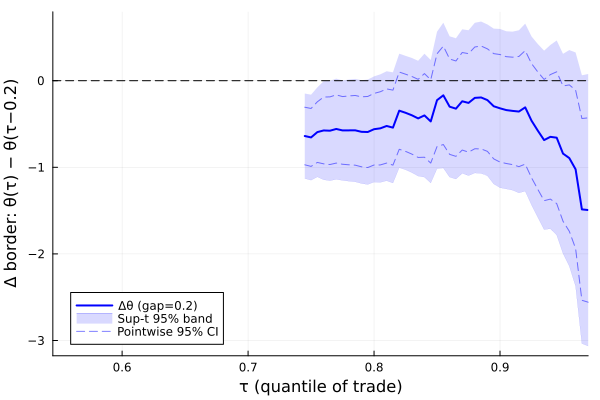}
    \caption{Border}
    \label{fig:diff2_border}
\end{subfigure}
\hfill
\begin{subfigure}[b]{0.48\textwidth}
    \centering
    \includegraphics[width=\textwidth]{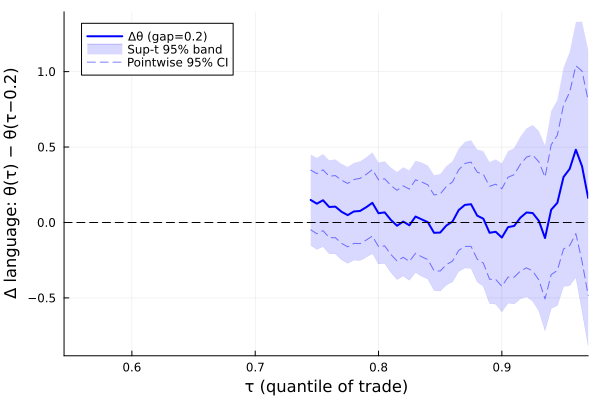}
    \caption{Language}
    \label{fig:diff2_language}
\end{subfigure}
\hfill
\begin{subfigure}[b]{0.48\textwidth}
    \centering
    \includegraphics[width=\textwidth]{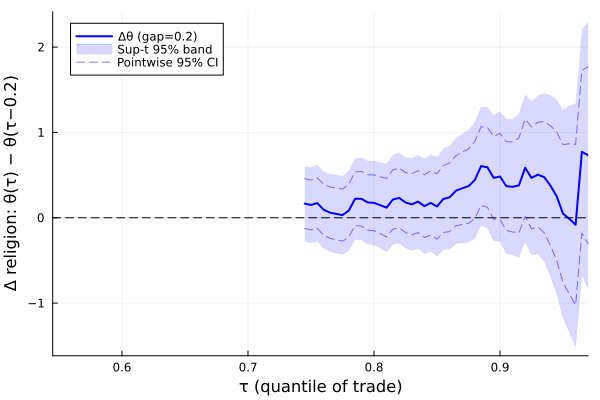}
    \caption{Religion}
    \label{fig:diff2_religion}
\end{subfigure}
\caption{Differences $\theta_{n,d}(\tau) - \theta_{n,d}(\tau - 0.20)$ from the CMLE 
estimator for log distance and legal system. Shaded region and dashed lines show 
the 95\% simultaneous sup-$t$ confidence band and pointwise confidence intervals 
for the differences, respectively. The horizontal dashed line marks zero. Maximal quantile considered is 0.970.}
\label{fig:diffs_appendix2}
\end{figure}


\end{document}